%% file: generatorRB.tex
\documentclass[aps,prx,reprint,superscriptaddress,nofootinbib,10pt,showpacs,floatfix,longbibliography]{revtex4-2}

\pdfoutput=1

\usepackage[T1]{fontenc} %256 bit font encoding
\usepackage[english]{babel} %english language
\usepackage[utf8]{inputenc}

\usepackage[table,usenames,dvipsnames]{xcolor}
\usepackage[psdextra]{hyperref}
\usepackage{bookmark}
\usepackage{amsmath,amssymb,amsfonts,mathtools,amsthm}
\usepackage{graphicx}
\usepackage{bbm}
\usepackage[printonlyused,withpage]{acronym}
\usepackage{enumitem}
\usepackage[normalem]{ulem}

\usepackage{tikz}
\usepackage{pgfplots}
\usepackage{booktabs}
\pgfplotsset{compat=1.17}

\usepgfplotslibrary{external}
\usepackage{csquotes} %to avoid warning

\mathtoolsset{showonlyrefs} 

\providecommand{\nat}{Nature}

\definecolor{martin}{rgb}{0,.4,1}

\definecolor{markus}{HTML}{006600}

\definecolor{ingo}{rgb}{.8,.5,0}

\input{mymath}

\input{mk_macros}

\newcommand{\hhu}{Institute for Theoretical Physics,
	Heinrich-Heine-Universit{\"a}t D{\"u}sseldorf, 
	Germany
}

\newcommand{\tuhh}{%
    Institute for Quantum Inspired and Quantum Optimization,
    Technische Universit{\"a}t Hamburg, Germany
}

\hypersetup{
       pdfsubject = {Quantum computing},
       pdfkeywords = {%
            filtered, character, generator, randomized, cross entropy, benchmarking, 
            gate-dependent noise, state, preparation, measurement, SPAM, 
            random, quantum, circuits, walk, 
            convergence, moment, operator, 
            unitary, Weyl-Heisenberg, local, Clifford, group, design, spectral gap, 
            non-commutative, Fourier, transform, analysis, 
            irreducible, representation, irrep, 
            perturbation, theory, subspaces, 
            frame, constant, 
            leakage, crosstalk, average gate fidelity%
            }
}

\begin{document}

\title{Randomized benchmarking with random quantum circuits}

\author{Markus Heinrich}\thanks{corresponding author}
\email{markus.heinrich@hhu.de}
\affiliation{\hhu}
\author{Martin Kliesch}
\email{martin.kliesch@tuhh.de}
\affiliation{\hhu}
\affiliation{\tuhh}
\author{Ingo Roth}
\affiliation{Quantum research centre, Technology Innovation Institute, Abu Dhabi, United Arab Emirates}
\email{ingo.roth@tii.ae}

\begin{abstract}
In its many variants, randomized benchmarking (RB) is a broadly used technique for assessing the quality of gate implementations on quantum computers. 
A detailed theoretical understanding and general guarantees exist for the functioning and interpretation of RB protocols if the gates under scrutiny are drawn \emph{uniformly} at random from a compact group. 
In contrast, many practically attractive and scalable RB protocols implement random quantum circuits with local gates randomly drawn from some gate-set. 
Despite their abundance in practice, for those \emph{non-uniform} RB protocols, general guarantees for gates from arbitrary compact groups under experimentally plausible assumptions are missing. 
In this work, we derive such guarantees for a large class of RB protocols for random circuits that we refer to as \emph{filtered RB}. 
Prominent examples include linear cross-entropy benchmarking, character benchmarking, Pauli-noise tomography and variants of simultaneous RB. 
Building upon recent results for random circuits, we show that many relevant filtered RB schemes can be realized with random quantum circuits in linear depth, and we provide explicit small constants for common instances. 
We further derive general sample complexity bounds for filtered RB. 
We show filtered RB to be sample-efficient for several relevant groups, including protocols addressing higher-order cross-talk. 
Our theory for non-uniform filtered RB is, in principle, flexible enough to design new protocols for non-universal and analog quantum simulators. 
\end{abstract}

\maketitle

%%% ==========================
\section{Introduction}
%%% ==========================

Assessing the quality of quantum gate implementations is a crucial task in developing quantum computers \cite{Eisert2020QuantumCertificationAnd,Kliesch2020TheoryOfQuantum}. 
Arguably, the most widely employed protocols for this task are 
\acf{RB} \cite{EmeAliZyc05,Levi07EfficientErrorCharacterization,DanCleEme09, EmeSilMou07,KniLeiRei08,Magesan2012} and its many variants 
(see Ref.~\cite{helsen_general_2022} for a recent overview)
including linear \ac{XEB} \cite{Arute2019QuantumSupremacy}. 
The basic idea of \ac{RB} is to measure the accuracy of random gate sequences of different lengths. 
Typically, this results in an experimental signal described by (a mixture of) exponential decays.  
Stronger noise results in faster decays with smaller decay parameters. 
Hence, those decay parameters are used to capture the average fidelity of the implemented quantum gates. 
A crucial advantage of these methods besides their experimental efficiency is that the reported decay parameters are robust against \ac{SPAM} errors. 

Generally speaking, many experimental signatures can be rather well fitted by an exponential decay. 
Experimentally observing an exponential decay in an \ac{RB} experiment does by itself not justify the interpretation of the decay parameter as a measure for the quality of the gates. 
In addition, \ac{RB} requires a well-controlled theoretical model that explains the observed decays under realistic assumptions and provides the desired interpretation of the decay parameters. 

Extensive research has already established a solid theoretical foundation for \ac{RB}, particularly when the gates comprising the sequences are drawn \emph{uniformly} at random from a compact group. 
Generalizing the arguments of Refs.~\cite{MagGamEme11,Magesan2012,proctor2017WhatRandomizedBenchmarking,wallman2018randomized,Merkel18}, 
\textcite{helsen_general_2022} 
derived general guarantees for the signal form of the entire zoo of \ac{RB} protocols with finite groups (which readily generalizes to compact groups \cite{kong_framework_2021}): 
If the noise of the gate implementation is sufficiently small, % (in a precise sense), 
each decay parameter is associated to an irreducible representation (irrep) of the group generated by the gates. 
Thus, the decay parameter indeed quantifies the average deviation of the gate implementation from their ideal action on the subspace carrying the irrep. 
For example, the `standard' RB protocol draws random multi-qubit gates uniformly from the Clifford group, except for the last gate of the sequence, which is supposed to restore the initial state. 
This protocol results in a single decay parameter associated with the irreducible action on traceless matrices and related to the average gate fidelity.

In practice, however, the suitability of \emph{uniform} \ac{RB} protocols for holistically assessing the quality of noisy and intermediate-scale quantum (NISQ) hardware is restricted. 
On currently available hardware, sufficiently long sequences of multi-qubit Clifford unitaries lead to way too fast decays to be accurately estimated for already moderate qubit counts. 
More scalable \ac{RB} protocols \emph{directly} draw sequences of local random gates, implementing a \emph{random circuit} \cite{Franca2018ApproximateRB,Proctor2019DirectRandomized,chasseur_complete_2015,chasseur_hybrid_2017}. 
We refer to those protocols that use a non-uniform probability distribution over a compact group 
as \emph{non-uniform} \ac{RB} protocols. 
Arguably, the most prominent example of non-uniform \ac{RB} is the \emph{linear \ac{XEB}} protocol, which  was used for the first demonstration of a quantum computational advantage in sampling tasks \cite{Arute2019QuantumSupremacy,hangleiter_computational_2023}. 

Establishing theoretical guarantees for non-uniform \ac{RB} is considerably more subtle. 
Roughly speaking, the interpretation of the decay parameter is more complicated as one additionally witnesses the convergence of the non-uniform distribution to the uniform one with the sequence length---causing a superimposed decay in the experimental data. 
These obstacles are 
well-known in the RB literature \cite{boone2019randomized,Proctor2019DirectRandomized} 
and have raised suspicion in the context of linear \ac{XEB} 
\cite{RinottShohamKalai:2022,BarakEtAl:2020:Spoofing}. 
If not carefully considered, one easily ends up significantly overestimating the fidelity of the gate implementations. 
In the context of their \emph{universal randomized benchmarking} framework, \textcite{chen_randomized_2022} have given a comprehensive analysis of non-uniform \ac{RB} protocols using random circuits which form approximate unitary 2-designs.
As such, the results in Refs.~\cite{chen_randomized_2022} are e.g.~applicable to linear \ac{XEB} with universal gate sets or with gates from the Clifford group \cite{chen_linear_2022}.

The original theoretical analysis of linear \ac{XEB} relies on the assumption that for every circuit, one observes an ideal implementation up to global depolarizing noise \cite{Arute2019QuantumSupremacy}. 
Building more trust in linear \ac{XEB} has motivated a line of theoretical research, introducing different heuristic estimators 
\cite{RinottShohamKalai:2022} 
and analyzing the behaviour of different noise models in random circuits \cite{Liu21BenchmarkingNear-term,Dalzell21RandomQuantumCircuits} using mappings of random circuits to statistical models \cite{Hun19}. 
But general guarantees that work under minimal plausible assumptions on the gate implementation and for random circuits generating arbitrary compact groups---akin to the framework \cite{helsen_general_2022,kong_framework_2021} for uniform RB and going beyond unitary 2-designs \cite{chen_randomized_2022,chen_linear_2022} --- are missing. 
Moreover, the sampling complexity of protocols like linear \ac{XEB} for non-uniform distributions and general gate-dependent noise remains unclear.

In this work, we close these gaps by developing a general theory of \emph{`filtered' randomized benchmarking} with \emph{random circuits} using gates from \emph{arbitary compact groups} under arbitrary gate-dependent (Markovian and time-stationary) noise. 
Under minimal assumptions, we guarantee the functioning of the protocol, and give explicit bounds on sufficient sequence lengths as well as on the number of samples.
Moreover, we specialize our general findings to concrete groups and random circuits, and give explicit constants.

{
Concretely, the filtered \ac{RB} protocol requires the execution of random circuit instances with a varying number of layers, c.f.~Fig.~\ref{fig:intro-protocol}.
Deviating from standard \ac{RB},  the last gate inverting the sequence is omitted and a simple computational basis measurement is performed instead.
This approach simplifies the experimental procedure and is arguably a core requirement for experimentally scalable \emph{non-uniform} \ac{RB}. 
The inversion of the circuit is effectively performed in the classical post-processing of the data.  
At this stage, the experimental data is additionally filtered to show only the specific decay associated with a single irrep of the group generated by the random circuits.
This step is the motivation for the name `filtered \ac{RB}' \cite{helsen_general_2022}.
It is especially useful if the relevant group decomposes into many irreps which would otherwise result in multiple, overlapping decays.
}

Besides  linear \ac{XEB}, \emph{filtered \ac{RB}} \cite{helsen_general_2022} encompasses character benchmarking \cite{HelsenEtAl:2019:character}, matchgate benchmarking \cite{Helsen20MatchgateBenchmarking}, and Pauli-noise tomography \cite{Flammia2019EfficientEstimation} as well as variants of simultaneous \cite{Gambetta2012CharacterizationAddressability} and correlated \cite{McKay2020CorrelatedRB} {RB} as additional examples.

The filtering allows for a more fine-grained perspective on the perturbative argument at the heart of the framework of Ref.~\cite{helsen_general_2022}, in that the different irreps of a group can be analyzed individually. 
In this way, we derive new perturbative bounds based on the harmonic analysis of compact groups that can be naturally combined with results from the theory of random circuits.
Thereby, we go significantly beyond previous works and treat uniform and non-uniform filtered RB on the same footing.

\begin{figure}
 \centering
 \includegraphics[width=\linewidth]{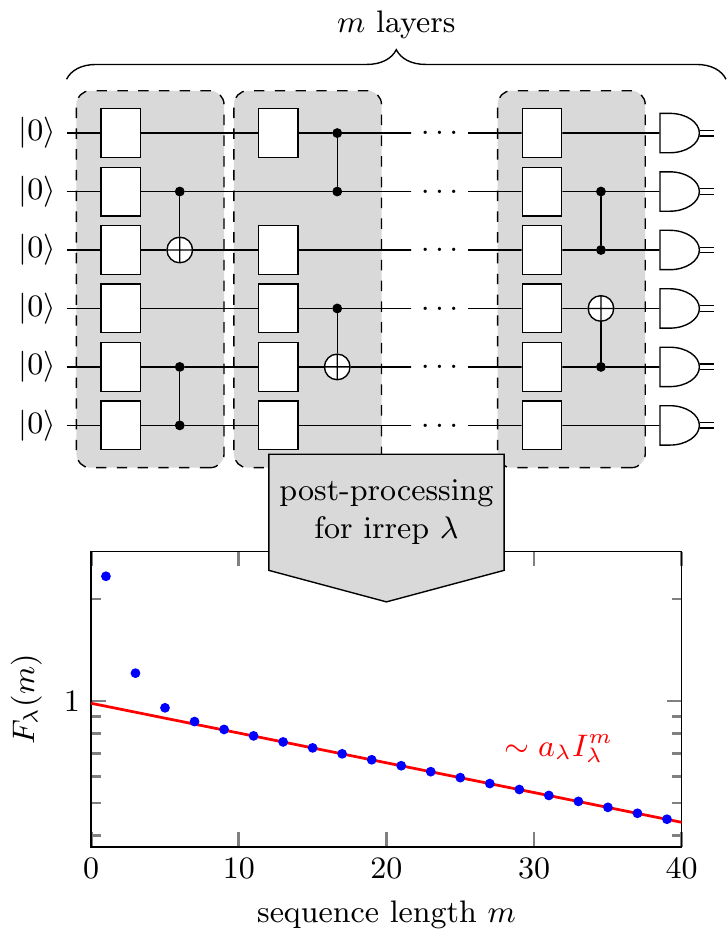}
 \caption{{
  Illustration of the filtered randomized benchmarking protocol.
  Random quantum circuits with varying number of layers $m$ are applied to the all-zero state and subsequently measured in the computational basis.
  The experimental data is post-processed depending on a parameter $\lambda$, labelling the irreps of the group generated by the random quantum circuits.
  {In a perturbative noise regime and }for sufficiently large $m$, the resulting signal decays exponentially with $m$ and the decay parameter $I_\lambda$ quantifies the average quality of the used gates on the considered irrep.
 }}
 \label{fig:intro-protocol}
\end{figure}

More precisely, our guarantees assume that the \emph{error of the average implementation} (per irrep) of the gates appearing in the random circuit is sufficiently small compared to the spectral gap of the random circuit.  
Then, the signal of filtered RB is well-described by a suitable exponential decay after a sufficient circuit depth. 
The required depth depends inversely on the spectral gap and logarithmically on the dimension of the irrep. 
We show that for practically relevant examples, our results imply that a \emph{linear} circuit depth in the number of qubits suffices for filtered \ac{RB}. 
{Furthermore, a sufficiently small average implementation error is ensured if the noise rate per gate scales reciprocally with the system size.}

Omitting the inversion gate comes at the price that the simple arguments for the sample-efficiency of standard randomized benchmarking do not longer apply to filtered \ac{RB}. 
As in shadow estimation for quantum states \cite{Huang2020Predicting}, the post-processing introduces estimators that are generally only bounded exponentially in the number of qubits. 
Thus, the precise convergence of estimators calculated from polynomially many samples is a priori far from clear. 

Generalizing our perturbative analysis of the filtered \ac{RB} signal to its variance, we derive general expressions for the sample complexity of filtered \ac{RB}. 
In particular and under essentially the same assumptions that guarantee the signal form of the protocol, filtered \ac{RB} is as sample-efficient as the analogous protocol that uses uniformly distributed unitaries.
Again important examples are found to be already sample-efficient using linear circuit depth.  
Perhaps surprisingly, we find that filtered \ac{RB} without entangling gates has constant sampling complexity independent of the non-trivial support of the irreps. 
This finding is in contrast to the related results in state shadow estimation. 

To showcase the general results, we explicitly discuss the cases where the random circuit generates the Clifford group, the local Clifford group, or the Pauli group.
Moreover, we discuss common families of random circuits and summarize spectral gap bounds with explicit, small constants from the literature and our own considerations \cite{harrow_random_2009,brown_convergence_2010,haferkamp_efficient_2023,haferkamp_improved_2021,HarMeh18}.

Finally, it is an open question whether the post-processing of filtered \ac{RB} can be modified so that meaningful decay constants can be extracted already from \emph{constant depth} circuits. 
In the context of linear \ac{XEB}, Ref.~\cite{RinottShohamKalai:2022} introduces a heuristic so-called `unbiased' estimator to this end. 
Using the general perspective of filtered \ac{RB}, we sketch two general approaches to construct modified linear estimators for constant-depth circuits. 
The first approach introduces a more costly computational task in the classical post-processing.  
The second approach requires that the random distribution of circuits is locally invariant of local Clifford gates. 
We formally argue that these estimators work under the assumption of global depolarizing noise, putting them at least on comparable footing as existing theoretical proposals. 
A detailed perturbative analysis is left to future work. 

We expect that the theory of non-uniform filtered \ac{RB} can be applied to many other practically relevant benchmarking schemes and bootstraps the development of new \ac{RB} schemes. 
In fact, one of our main motivations for deriving the flexible theoretical framework  
is its applications for the characterization and benchmarking of non-universal and analog quantum computing devices---consolidating and extending existing proposals \cite{Derbyshire2019RBanalogue,shaffer2021practical} in forthcoming and future work. 

On a technical level, we develop tools to analyze noisy random circuits  using harmonic analysis on compact groups and matrix perturbation theory. 
We expect that this perturbative description also 
finds applications in quantum computing beyond the randomized benchmarking of quantum gates. 
The tools and results might, in principle, be applicable to analyze the noise-robustness of any scheme involving random circuits, e.g.\ randomized compiling \cite{Wallman16NoiseTailoringFor}, shadow tomography and randomized measurements \cite{elben_randomized_2022} or error mitigation \cite{Temme_2017}. 
As a by-product, our variance bounds take a more direct representation-theoretic approach working with tensor powers of the adjoint representation rather than exploiting vector space isomorphisms and invoking Schur-Weyl duality \cite{huang_predicting_2020}.  
This approach also opens up a complimentary, illuminating perspective on the sample-efficiency of estimation protocols based on random sequences of gates more generally.

\paragraph*{Prior and related work.} 
Already one of the first \ac{RB} proposals, NIST \ac{RB} \cite{KniLeiRei08} classifies as non-uniform \ac{RB} and was later thoroughly analyzed and compared to standard Clifford \ac{RB} \cite{boone2019randomized}. 
A first discussion of the obstacles arising from decays associated with the convergence to the uniform measure was then given in Ref.~\cite{boone2019randomized}.
Further non-uniform \ac{RB}
protocols are approximate 
RB  \cite{Franca2018ApproximateRB} 
and direct \ac{RB} \cite{Proctor2019DirectRandomized} (sometimes called \emph{generator \ac{RB}}). 
The original guarantees for these protocols rely on the closeness of the probability distribution to the uniform one in total variance distance, thus generally requiring long sequences.  
Direct \ac{RB} ensures this closeness by starting from a random stabilizer state as the initial state---assumed noiseless in the analysis, which is additionally restricted to Pauli-noise. 
The restriction can be justified with randomized compiling \cite{knill05,Wallman16NoiseTailoringFor,Ware21ExperimentalPauli-frame},  which essentially requires the perfect implementation of Pauli unitaries. 
{After the publication of a preprint of this paper, a more thorough analysis of direct \ac{RB} under general gate-dependent noise was published \cite{polloreno_theory_2023}, using techniques which are similar to ours.}

The work by \textcite{helsen_general_2022} unifies and generalizes the guarantees for these RB protocols to gate-dependent noise but still works with convergence of the probability distribution to the uniform distribution in total variation distance.
The approach of Ref.~\cite{helsen_general_2022} extends previous arguments for the analysis of gate-dependent noise by Wallman \cite{wallman2018randomized} using the language of Fourier transforms of finite groups introduced to \ac{RB} by Merkel \emph{et al.} \cite{Merkel18}.  
The argument straightforwardly carries over to compact groups \cite{kong_framework_2021}. 
The assumptions on the gate implementation required for the guarantees of Ref.~\cite{helsen_general_2022}, closeness in average diamond norm error over all irreps, are too strong to yield practical circuit depths for \ac{RB} with random circuits. 

This obstacle has been overcome in the \emph{universal randomized benchmarking} framework by \textcite{chen_randomized_2022}.
There, the authors are able to relax the assumption on the probability measure for the above protocols (``twirling schemes'' in Ref.~\cite{chen_randomized_2022}) and only require that the channel twirl over this measure is within unit distance from the Haar-random channel twirl (in induced diamond norm or spectral norm).
Hence, it is sufficient for these schemes to implement random circuits which form approximate unitary 2-designs w.r.t.~the relevant norm.
As such, it is necessary that the used distributions have support on groups which are unitary 2-designs, such as the unitary or the Clifford group \cite{chen_linear_2022}.

Filtered \ac{RB}, as formulated in Ref.~\cite{helsen_general_2022}, is a variant of character \ac{RB} \cite{HelsenEtAl:2019:character}. 
Linear \ac{XEB} \cite{Arute2019QuantumSupremacy}, 
when averaged over multiple circuits, 
can be seen as the special case of filtered \ac{RB} when the group generated by the circuits is a unitary $2$-design. 
Ref.~\cite{helsen_general_2022} analyzes linear \ac{XEB}, including variance bounds, but only for uniform measures, not for random circuits. 
Ref.~\cite{Liu21BenchmarkingNear-term} puts forward a different perturbative analysis for filtered randomized benchmarking schemes by carefully tracing the effect of individual Pauli-errors in random circuits. 
To our understanding, the argument, however, crucially relies on the heuristic estimator proposed in Ref.~\cite{RinottShohamKalai:2022}.  See also the review~\cite{hangleiter_computational_2023} for a detailed literature overview on linear \ac{XEB}. 
Hybrid benchmarking \cite{chasseur_hybrid_2017} 
puts forward another approach to avoid the linear inversion using random Pauli observables.
%For sequences from arbitrary probability measures and gate-dependent noise, 
In contrast to other randomized benchmarking schemes, the hybrid benchmarking signal consists of linear combinations of (exponentially) many decays with complex poles \cite{chasseur_complete_2015,chasseur_hybrid_2017}.
Estimating these poles, however, is typically infeasible, 
see the detailed discussion in Ref.~\cite{helsen_general_2022}.

The here discussed filtered \ac{RB} protocols for circuits generating local groups is an alternative to simultaneous \cite{Gambetta2012CharacterizationAddressability} and correlated \cite{McKay2020CorrelatedRB} \ac{RB} but is in addition capable of estimating higher-order correlations. 
A randomized benchmarking scheme with the Heisenberg-Weyl group is also proposed in Ref.~\cite{StilckFranca2020EfficientBenchmarkingAnd}.

Alternative approaches to filtered non-uniform RB aiming at better scalability of randomized benchmarking protocols are cycle \ac{RB} \cite{Erhard2019CharacterizingLarge-scale,Zhang22ScalableFastBenchmarking}, 
average circuit eigenvalue sampling \cite{flammia_averaged_2021} and the recent \ac{RB} with mirror circuits \cite{Proctor21ScalableRB}. 

\paragraph*{Outline.}

The remainder of this work is structured as follows:
We start by introducing and discussing the filtered \ac{RB} protocol in Sec.~\ref{sec:overview-protocol}.
Afterwards, in Sec.~\ref{sec:result-summary}, we give a non-technical overview of our main results and highlight the central messages of this work.
The technical part begins with Sec.~\ref{sec:preliminaries}, where we introduce necessary background and definitions.
This section is self-contained and gives a general introduction into the techniques used in this paper.
We then proceed by stating and proving our results in Sec.~\ref{sec:results}.
This section is structured into nine subsections which address the central assumptions of our work, the above described main results, some auxillary results, and the specialization to specific examples. 
In Sec.~\ref{sec:related-works}, we give a precise comparison to the technical assumptions and conclusions of related works.
Finally, the conclusion is given in Sec.~\ref{sec:conclusion}.

The appendices contain a summary 
of the relevant matrix perturbation theory in App.~\ref{sec:perturbation-theory}, a discussion of single-shot versus multi-shot estimators in App.~\ref{sec:estimators}, and a self-contained computation of noise-free second moments for various groups in App.~\ref{sec:sampling_complexity_ideal_case}.

% ----------------------------------------------------------
\section{The filtered randomized benchmarking protocol}
\label{sec:overview-protocol}
% ----------------------------------------------------------

We start by describing and motivating the general protocol of \emph{non-uniform filtered \acf{RB}}. 

We consider a quantum device with state space modelled by a $d$-dimensional Hilbert space $\mathcal H$. 
{Randomized benchmarking} aims at assessing the quality of the implementation of a set of coherent operations on the device that constitute a compact group $G<\U(\mathcal H)$. 
The random unitaries that are actually applied in the experiment are specified by a probability measure $\nu$ on $G$. 
For example, $\nu$ can be a uniform measure on a subset of operations generating $G$ that are `native' to the device. 
{
The quantum device is prepared in a fixed initial state $\rho$ and we consider measurements in a fixed basis $\{\ket{i}\}_{i\in[d]}$ where $[d]\coloneqq\{1,\dots,d\}$.
Usually, $\rho$ is taken to be one of the basis elements.
}
{It is instructive to  briefly recall the standard uniform \ac{RB} protocol for general compact groups $G$ \cite{Franca2018ApproximateRB,helsen_general_2022,kong_framework_2021} first.}

{
\emph{Standard randomized benchmarking} samples a sequence of gates $g_1,\dots, g_m$ uniformly from the Haar measure $\nu=\mu$ on $G$.
After applying the sequence and the inversion gate $g_\mathrm{inv}=(g_m \cdots g_1)^{-1}$ to the initial state $\rho$, the resulting state is measured in the given basis.
Let $\hat p_i(m)$ be the frequency of observing outcome $i$ and
\begin{equation}
\label{eq:rb-signal}
 p_i(m) \coloneqq \EE[\hat p_i(m)] \,,
\end{equation}
be the expected Born probabilities.
It is well-known \cite{Franca2018ApproximateRB,helsen_general_2022} that the probabilities $p_i(m)$ can be well-approximated by a linear combination of (potentially complex) exponential decays:
\begin{equation}
\label{eq:standard-rb-signal}
 p_i(m) \approx \sum_\lambda a_{i,\lambda} z_\lambda^m \,.
\end{equation}
Here, the {so-called} poles $z_\lambda$ can be identified with the (possibly repeated) irreducible subrepresentations (irreps) of the group $G$.
{Intuitively speaking, the poles are an effective depolarizing strength of the average noise acting on a specific irrep.}
In many cases, the $z_\lambda$ are real, for instance if all irreps are multiplicity-free (and of real type).
Then, one can observe the typical exponential decays in the RB signal.
Fitting Eq.~\eqref{eq:standard-rb-signal} becomes challenging if $G$ has many relevant irreps, especially if all $z_\lambda$ are real \cite{helsen_general_2022}.
Moreover, even if a reliable fit is possible, it is impossible to associate the poles $z_\lambda$ with the correct irreps if more than two irreps contribute.
}
{
\emph{Filtered \ac{RB}} is designed to address these and other problems in the standard approach to \ac{RB} \cite{helsen_general_2022,helsen_estimating_2021}.
In particular, it allows to isolate single poles in Eq.~\eqref{eq:standard-rb-signal} by ``filtering'' onto the irrep $\lambda$ of interest.
}

\paragraph{The protocol of non-uniform filtered \ac{RB}.}
The protocol can be divided into two distinct phases:
First, the \emph{data acquisition phase} involving a simple experimental prescription that is already routinely implemented in many experiments.
Second, the \emph{post-processing phase} in which this data is processed and the decay parameters are extracted. 

\begin{enumerate}[label=(\Roman*)]
\item \textbf{Data acquisition.} 
Repeat the following primitive for different sequence lengths $m$: 
Prepare the state $\rho$, apply independent and identically distributed gates $g_1,\dots, g_m \sim \nu$ and measure in the basis $\{\ket{i}\}_{i\in[d]}$. 
The output of a single run of this primitive is a tuple $(i, g_1, \ldots, g_m) \in [d] \times G^m$ where $i$ is the observed measurement outcome. 

\item \textbf{Post-processing.} Using the samples acquired before, we compute the mean estimator of a soon-to-be-defined \emph{filter function} $f_\lambda: [d] \times G^m \to \R$ depending on {the irrep $\lambda$.}
Assuming we collected a number of $N$ samples $(i^{(l)}, g_1^{(l)},\dots,g_m^{(l)})$, we thus evaluate 
\begin{align}
\label{eq:filtered-rb-data-estimator}
  \hat{\FD}_\lambda(m)%
  &= \frac{1}{N} \sum_{l=1}^N f_\lambda(i^{(l)}, g_1^{(l)}, \dots, g_m^{(l)}) \, .
\end{align}
We refer to $\hat{\FD}_\lambda$ as the filtered \ac{RB} signal and denote its expectation value as 
\begin{equation}\label{eq:filtered-rb-data}
    \FD_\lambda(m) \coloneqq \EE[\hat{\FD}_\lambda(m)]\,. 
\end{equation}
The main result of this work lies in deriving conditions that guarantee an exponential fitting model for the expected \ac{RB} signal $\FD_\lambda(m)$ and the variance of {its estimator} $\hat{\FD}_\lambda(m)$ for a suitable choice of filter function $f_\lambda$. 
This justifies to fit an exponential $a_\lambda r_\lambda^m$ to the \ac{RB} signal and 
obtain the $r_\lambda$, the result of the \ac{RB} protocol. 
\end{enumerate}

The setting differs from the standard \ac{RB} protocol in two major aspects: First, we allow that the gates $g_i$ are drawn from a suitable probability measure which does not need to be the Haar measure nor a unitary 2-design. 
In particular, it is explicitly allowed to draw them from a set of generators of the group $G$. 
Precise conditions on the measure will be formulated and discussed later.
Second, we omit the \emph{end} or \emph{inversion gate} $g_\mathrm{end} = (g_m\cdots g_1)^{-1}$ at the end of the sequence and record a basis measurement outcome. 
From a practical point of view, this is advantageous since even if the individual gates $g_1,\dots,g_m$ have short circuit implementations, this is not necessarily true for $g_\mathrm{end}$. 
Instead, the inversion gate is effectively accounted for in the post-processing phase \cite{helsen_general_2022}.
As we see shortly, the evaluation of $f_\lambda$ essentially requires the simulation of the gate sequence.
This may require run-time and memory scaling exponentially in the number of qubits. 

\paragraph{The choice of filter function.}
The device is intended to implement a certain \emph{target} or \emph{reference} representation $\omega$ of $G$. 
For all practical purposes, $\omega$ is given as the representation $\omega(g) = U_g( \cdot ) U_g\ad$  acting on $\End(\mathcal H)$, the {space of} linear operators on $\mathcal H$, and $U_g$ is the representation of $G$ on the Hilbert space $\mathcal H$ (i.e.~$\omega(g)$ is the unitary channel associated to $U_g$).
{Associated to $G$ and the measurement basis $E_i\coloneqq\ketbra{i}{i}$}, we introduce the quantum channels
\begin{align}
\label{eq:measurement-frame-operator}
 S
 &\coloneqq \int_{G}\omega(g)^\dagger M \omega(g) \dd\mu(g)\, , \;
 M \coloneqq { \sum_{i\in[d]} \oketbra{E_i}{E_i}\, .}
\end{align}
Here, $\mu$ is the Haar (uniform) probability measure on $G$ {and $\oketbra{A}{B}$ denotes the outer product of operators, i.e.~the superoperator $X\mapsto \obraket{B}{X} A = \tr(B^\dagger X) A$.}
In words, $S$ is the \emph{channel twirl} w.r.t.~to $G$ applied to the completely dephasing channel $M$ in the measurement basis $\{\ket{i}\}_{i\in[d]}$.

The representation $\omega$ has a decomposition into irreducible subrepresentations (irreps).
For example for $G = \U(d)$, $\omega$ has two irreps: The trivial action on the subspace spanned by the identity and the action on the subspace of traceless matrices. 
Given an irrep of $G$ labelled by $\lambda$, we denote by $P_\lambda$ the projector onto the irrep in $\End(\mathcal H)$.\footnote{More precisely, $P_\lambda$ projects onto the isotypic component, i.e.~onto the span of all irreps isomorphic to $\lambda$.}
Finally, we define the filter function by 
\begin{align}
    f_\lambda(i,g_1,\dots,g_m) &\equiv f_\lambda(i,g)
    \coloneqq \osandwich{\rho}{P_\lambda S^+ \omega(g)^\dagger }{E_i} \, , \label{eq:filter-function} \\
    % \tr[E_i \omega(g) S^+P_\lambda(\rho)] \, ,
    % g &\coloneqq g_m\cdots g_1 \, ,
\end{align}
where we abused notation a bit to indicate that $f_\lambda$ only depends on the product of gates $g = g_m\cdots g_1$ and the round `bra-ket' notation again refers to the trace inner product of operators.
Moreover, $S^+$ is the Moore-Penrose pseudoinverse of $S$.

\paragraph{The ideal signal.} To motivate our choice of filter function, let us consider an ideal and noise-free implementation, and gates that are drawn uniformly from $G$.
For a given sequence of gates, the data acquisition phase produces samples from the distribution given by the Born probabilities 
\begin{align} 
  p(i|g_1, \dots, g_m) 
  &= 
  {
  \osandwich{E_i}{\omega(g_m) \cdots \omega(g_1)}{\rho}
  = 
  \osandwich{E_i}{\omega(g)}{\rho}} \,, &
\end{align}
where again $g = g_m \cdots g_1$.
Hence, we are effectively measuring $\rho$ with respect to the \ac{POVM} $(i,g)\mapsto \omega(g)^\dagger(E_i)\dd\mu(g)$. 
Let us, for the sake of the argument, assume that the \ac{POVM} is \emph{informationally complete}, i.e.~the operators $\omega(g)^\dagger(E_i)$ span the full operator space $\End(\mathcal H)$.
As this span is exactly the range of $S$, it is invertible and the pseudoinverse is the inverse, $S^+=S^{-1}$. 
We observe that the filtered \ac{RB} signal \eqref{eq:filtered-rb-data} becomes 
\begin{align}
 \FD_\lambda(m) 
 &= 
 \sum_{i \in [d]} \int_{G^{m}} 
 \osandwich{\rho}{P_\lambda S^{-1}\omega(g_1\cdots g_m)^\dagger}{E_i} \times \\
 & \qquad\times
 \osandwich{E_i}{\omega(g_1\cdots g_m)}{\rho} \dd\mu(g_1,\dots,g_m) \\
 &=
 \osandwich{\rho}{P_\lambda S^{-1} \int_{G}\omega(g)^\dagger M \omega(g) \dd\mu(g)}{\rho} \\
 &= 
 \osandwich{\rho}{P_\lambda}{\rho} \,. 
 \label{eq:ideal-signal-intro}
\end{align}
Hence, we have found that $\FD_\lambda(m)$ is the overlap of $\rho$ with the irrep $\tau_\lambda$.
In the case that the \ac{POVM} is not informationally complete, the result is $\osandwich{\rho}{P_\lambda S^+ S}{\rho} $, where $S^+ S$ is the projector onto the span of the \ac{POVM}.

\paragraph{Example: Linear XEB as filtered randomized benchmarking.} 
\label{par:linear-XEB}

The perhaps most prominent example of filtered RB is linear \acf{XEB}, as already observed in Ref.~\cite[Sec.~VIII.C]{helsen_general_2022}. 
Originally, the linear cross-entropy was proposed as a proxy to the cross-entropy between the output probability distribution of an \emph{individual} generic quantum circuit and its experimental implementation, designed to specifically discriminate against a uniform output distribution  \cite{Arute2019QuantumSupremacy}. 
The theoretical motivation, however, already stems from considering an ensemble of random unitaries and the original analysis of linear \ac{XEB} is based on typicality statements that hold on average or (by concentration of measure) with high-probability over the ensemble. 
When also explicitly \emph{taking the average} of the linear cross-entropy estimates of random instances from an ensemble of unitaries, linear \ac{XEB} becomes a randomized benchmarking scheme, more precisely, a \emph{non-uniform filtered randomized benchmarking for the full unitary group}. 
Let us reproduce this argument in a slightly more general form.

In linear \ac{XEB} on $n$ qubits, a {random quantum circuit, described by a unitary $U$,}
is applied to the initial state $\rho\coloneqq\ketbra{0}{0}$, followed by a computational basis measurement described by projectors $\{E_x\coloneqq\ketbra{x}{x}\}_{x\in\F_2^n}$, where $\FF_2$ denotes the binary field.
Having observed outcomes $x^{(1)},\dots,x^{(N)}$ for unitaries $U^{(1)},\dots,U^{(N)}$, one computes the estimator
\begin{equation}
\label{eq:XEB-estimator}
\hat F_\mathrm{XEB} = \frac{1}{N} \sum_{i=1}^N ( d \, p_\ideal(x^{(i)}|U^{(i)}) - 1),
\end{equation}
{
where $d=2^n$ and $p_\ideal(x|U) \coloneqq \abs{\sandwich{x}{U}{0}}^2 = \osandwich{E_x}{\omega(U)}{E_0}$ is the ideal, noiseless outcome distribution of the circuit $U$.

In this setting, we can readily compute our filter function defined in Eq.~\eqref{eq:filter-function}, using that the projector onto the traceless irrep {$\lambda=\mathrm{ad}$} of the unitary group is $P_\Ad(X) = X - \tr(X)\one/d$ and the operator $S$ is such that $S^{-1}P_\Ad = (d+1)P_\Ad$ (as derived later in Sec.~\ref{sec:frame-operator}):
\begin{align}
 f_\Ad(x,U)
 &=
 (d+1) \osandwich{E_x}{\omega(U)P_\Ad}{E_0} \\
 &=
 \frac{d+1}{d} \big( d \, \osandwich{E_x}{\omega(U)}{E_0} - 1 \big) \\
 &=
 \frac{d+1}{d}  \big(d\, p_\ideal(x|U) - 1\big) \,.
\end{align}
Suppose the random circuit has $m$ layers such that $U = U_m\cdots U_1$, where every layer $U_i$ is sampled independently according to a measure $\nu$ on the unitary group $\U(2^n)$.
Then, our estimator for the filtered \ac{RB} signal \eqref{eq:filtered-rb-data-estimator} is identical, up to a global factor, to the linear \ac{XEB} estimator \eqref{eq:XEB-estimator}:
\begin{equation}
 \hat{F}_\mathrm{XEB} = \frac{d}{d+1} \hat\FD_\Ad(m) \,.
\end{equation}
Thus, we observe that linear \ac{XEB} is in fact performing a filtered \ac{RB} protocol for the unitary group $G=\U(2^n)$ and the measure $\nu$.
The different normalization is negligible for a moderate number of qubits, $\frac{d}{d+1} \approx 1$.
}

%%% ==========================
\section{Overview of results}
\label{sec:result-summary}
%%% ==========================

Having introduced the protocol, we give an overview of our main results in this section.
The results 
can be grouped into three categories:
Guarantees for the \emph{signal form} of the filtered \ac{RB} signal $\FD_\lambda(m)$ and the required \emph{sequence lengths},
bounds on the \emph{number of samples} needed to estimate the signal (\emph{sampling complexity}),
and proposals of \emph{protocol modifications} for short-depth circuits.

Our results rely on the following assumptions.
We model the imperfect implementation of gates on the quantum device by a so-called \emph{implementation map} $\phi$ on the group $G$ such that $\phi(g)$ is completely positive and trace non-increasing for all $g\in G$.
Importantly, this model allows for highly gate-dependent noise and thus overcomes the common, but unrealistic assumption of gate-independent noise in the literature.
{If only a set of generator of $G$ can be implemented, it is admissable to choose $\phi(g)$ e.g.\ equal to the identity channel for every `non-native gate' in the group.}
The existence of an implementation map $\phi$ requires that the gate noise is \emph{Markovian and time-stationary}.
Moreover, we assume that the probability measure $\nu$ is sufficiently well-behaved in a later to be defined sense (e.g.~it has support on a set of generators of $G$).

\paragraph{Signal form guarantees.}
Similar to Eq.~\eqref{eq:ideal-signal-intro}, we start by showing that the filtered \ac{RB} signal $\FD_\lambda(m)$ is the {linear} contraction of {the $m$th power} of a linear operator $\tilde T_\lambda$.
This immediately implies that $\FD_\lambda(m)$ is a linear combination of exponentially many, possibly complex poles:
\begin{equation}
 \FD_\lambda(m) = \osandwich{Q_\lambda}{\tilde T_\lambda^m}{M_\lambda} = \sum_{i} a_{\lambda,i}  z_{\lambda,i}^m \,.
\end{equation}
Here, $Q_\lambda$ and $M_\lambda$ are suitable operators and the $z_{\lambda,i}$ are the eigenvalues of $\tilde T_\lambda$.
{Our central result now gives general conditions on the noise and minimal sequence length under which only `few' poles dominate in the expression of $\FD_\lambda(m)$.} 
{The simplest signal form arises for multiplicity-free irreps where there exist a single real dominant pole $I_\lambda = z_{\lambda,0}$. Then, the signal is well-approximated by the exponential decay $a_{\lambda}  I_{\lambda}^m$ for sufficiently large $m$.}

To this end, it is instructive to consider the situation in the absence of noise:
{Then, the implementation map should be exactly given by the \emph{reference representation} $\omega$ of $G$ on the support of $\nu$.
As noted above, the reference representation is %simply 
$\omega(g) = U_g( \argdot ) U_g\ad$, where $U_g$ is the representation of $G$ on the Hilbert space $\mathcal H$.
In the noise-free setting $\tilde T_\lambda$ becomes $T_\lambda (\mathcal X) \coloneqq P_\lambda \int_G  \omega(g)^\dagger \mathcal X \omega(g) \dd\nu(g)$, the well-known \emph{channel twirl} w.r.t.~the measure $\nu$, subsequently projected onto the irrep $\tau_\lambda$.}
In the theory of random circuits, $T_\lambda$ is called a \emph{moment operator} and its spectral properties control the convergence of random sequences generated by $\nu$ to the uniform measure on $G$.
Its largest eigenvalue is always 1 with multiplicity given by the multiplicity $n_\lambda$ of $\tau_\lambda$ in $\omega$.
The difference to the second-largest eigenvalue, the spectral gap $\Delta_\lambda$, measures the convergence rate.
For $\FD_\lambda(m)$, this implies that the dominant pole is $I_\lambda = 1$ and the subdominant poles are suppressed by $1-\Delta_\lambda$. 

In the presences of noise, the relevant operator is $\tilde T_\lambda (\mathcal X) =  P_\lambda \int_G  \omega(g)^\dagger \mathcal X \phi(g) \dd\nu(g)$.
From a harmonic analysis point of view, $\tilde T_\lambda {\cong\widehat{\phi\nu}[\tau_\lambda^{\otimes n_\lambda}]}$ is an operator-valued Fourier transform of $\phi$.
In particular, if $\nu$ is the uniform (Haar) measure, $\tilde T_\lambda$ contains information about the action of $\phi$ on the irrep $\lambda$.
We then consider $\tilde T_\lambda$ as a perturbation of $T_\lambda$:
As long as $\tilde T_\lambda - T_\lambda$ is small compared to the spectral gap $\Delta_\lambda$, the gap cannot close and thus the subdominant poles are still suppressed.
{This formulates a condition on the admissable strength of the noise.}
The $n_\lambda$-fold degeneracy of the largest eigenvalue is generally lifted under the perturbation and the resulting poles may even become complex.
As a result, we may encounter signal forms such as damped oscillations.
In the multiplicity-free case, $n_\lambda = 1$, the largest eigenvalue has to stay real and a single exponential decay can be ensured.
Formally, we show the following result:

\begin{theorem*}[Signal form of non-uniform filtered \ac{RB}, informal version]
 Suppose there is a $\delta_\lambda >0$ such that
 \begin{equation}
  \snormb{\tilde T_\lambda - T_\lambda} \leq \delta_\lambda < \frac{\Delta_\lambda}{5}\,.
 \end{equation}
 Then, 
 \begin{equation}
  \FD_\lambda(m) = \tr(A_\lambda I_\lambda^m) + \tr\left( B_\lambda O_\lambda^m\right) \,,
 \end{equation}
 % $\FD_\lambda(m) = \tr(A_\lambda I_\lambda^m) + \tr\left( B_\lambda O_\lambda^m\right)$ 
 where $I_\lambda$ is a real $n_\lambda \times n_\lambda$ matrix captures the average gate noise, independent of SPAM, and the complex eigenvalues of $I_\lambda$ fulfil $|z|\leq 1$ and $|z-1|\leq 2\delta$.
 The second term is suppressed as 
\begin{equation}
\abs[\big]{ \tr\left( B_\lambda O_\lambda^m\right) }
 \leq
 c_\lambda\,
 (1-\Delta_\lambda + 2\delta_\lambda)^m\,,
 \label{eq:thm8-subdominant-decay-bound}
\end{equation}
with a constant $c_\lambda$ depending on the irrep, measurement basis, and SPAM.
Typically, we have $c_\lambda = O(d_\lambda)$.
\end{theorem*}

An illustration of a typical filtered \ac{RB} signal for a multiplicity-free irrep, $n_\lambda=1$, is shown in Fig.~\ref{fig:rb-signal}.
{The theorem guarantees that in the perturbative noise regime }the filtered \ac{RB} signal is well-described by an exponential decay $a_\lambda I_\lambda^m$ if the sequence length $m$ is chosen large enough.
For small $m$, the second subdominant term $\tr(B_\lambda O_\lambda^m)$ leads to a complicated non-exponential behavior.
This initial regime can be understood as a mixing process of the (noisy) random circuit by Eq.~\eqref{eq:thm8-subdominant-decay-bound}.
Hence, it is imperative to understand the extent of this initial mixing regime to reliably extract decay rates.

The assumption that $\tilde T_\lambda$ is close to $T_\lambda$ is a statement about the \emph{error of the average implementation} of the gates being small compared to spectral gap of the random circuit.  
Importantly, only the error of the gates which actually appear in the random circuit matter.
It should be clear that such an assumption is necessary as strong noise will eventually close the gap of $T_\lambda$, thereby preventing us from extracting the relevant information from the signal. 
Instead, the decay will be dominated by the noisy mixing process in the strong noise regime.
We provide a detailed discussion of the assumption in Sec.~\ref{sec:implementation-quality} and connect the assumption to error measures of the individual gates.
{In particular, for local noise the noise per gate needs to improve as $O(1/n)$ with the system size $n$ in relevant examples, e.g.\ linear XEB.}

{Compared to the assumption 
formulated by 
\textcite{helsen_general_2022}, our perturbative analysis is able to treat individual irreps in isolation.}
Thereby, our guarantees depend only on `irrep-specific' quantities{, such as the restricted implementation error, multiplicity and dimension.}
Moreover, we only require that the measure $\nu$ `approximates Haar moments' of the irrep of $G$, a crucially weaker assumption than, e.g.~approximation in total variation distance used in Ref.~\cite{helsen_general_2022}.
At the same time our result shows that \emph{filtered \ac{RB} yields the same decay parameters as standard \ac{RB}}, as we obtain the  same dominant matrix $I_\lambda$ as \textcite{helsen_general_2022}.

\begin{figure}
\centering
\includegraphics{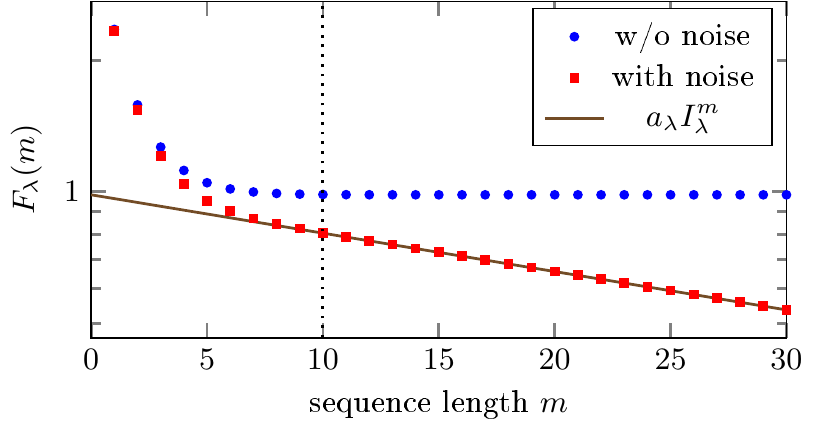}
\caption{\label{fig:rb-signal}
Illustration of a typical filtered \ac{RB} signal for a multiplicity-free irrep.
In the absence of noise (blue dots), the signal converges to a constant value.
With weak noise (red squares), the signal is asymptotically given by an exponential decay. 
In both cases, we observe an initial non-exponential regime which reflects a mixing process of the underlying (noisy) random circuit.
The dotted vertical line indicates the approximate extent of this regime and the point from which on the signal is well-described by an exponential decay.
}
\end{figure}

Deviations from  {monotonously decaying} signals  in \ac{RB} %with multiplicity-free irreps 
have been attributed to non-Markovian noise effects (or temporal drifts) in the literature \cite{wallman2018randomized,figueroa-romero_towards_2022}.
Indeed, if the considered irreps $\lambda$ of the group $G$ are multiplicity-free (the settings analysed in the literature), a monotonous decay (real pole) is guaranteed by our results. 
A non-monotonous signal requires that at least one of our assumptions (such as Markovianity) is broken.

\paragraph{Sequence length bounds.}
Provided that the perturbation assumption holds, the filtered \ac{RB} signal is well-described by a matrix exponential decay $\tr(A_\lambda I_\lambda^m)$, as long as the circuit is sufficiently deep to suppress the second sub-dominant term $\tr(B_\lambda O_\lambda^m)$ by Eq.~\eqref{eq:thm8-subdominant-decay-bound}.
The prefactor $c_\lambda$ is typically of the order $O(d_\lambda)$ and can be slightly improved for concrete examples, such as unitary 2-groups (e.g.~the Clifford group), where we find $c_\Ad =  O(d^{3/2})$ for the traceless irrep ($d_\Ad = d^2-1$).
Based on Eq.~\eqref{eq:thm8-subdominant-decay-bound}, we work out explicit sufficient conditions on the sequence length in Sec.~\ref{sec:sequence-lengths}.
We find that the following sequence length is typically sufficient to suppress the subdominant terms by $\alpha$:
\begin{align}
 m \geq 2 \Delta_\lambda^{-1}\big( \log(d_\lambda) + \log\tfrac{1}{\alpha} + 1.8 \big)\,,
 \label{eq:sequence-length-bound-1-summary}
\end{align} 
again with improved constants for concrete examples.

By evaluating the bound \eqref{eq:sequence-length-bound-1-summary} using results for the spectral gap of common random quantum circuits with universal or Clifford gates \cite{harrow_random_2009,brown_convergence_2010,haferkamp_efficient_2023,haferkamp_improved_2021,HarMeh18}, we arrive at concrete scalings of the circuit depth for these examples in Sec.~\ref{sec:application-to-random-circuits}.
The derived scalings in the number of qubits $n$ are summarized in Tab.~\ref{tab:circuit-length-summary}.
Our result implies that for brickwork circuits \emph{linear} circuit depth $m=O(n)$ suffices for filtered \ac{RB} even if one directly draws generators, either from the unitary or from the Clifford group.

\begin{table}[tb]
\centering
\begin{tabular}{lr}
\toprule
Brickwork circuit (BWC) & $9.8 n$ \\
Clifford generators BWC$^*$ & $470 n$ \\
Local random circuit (LRC) & $4.2 n^2$ \\
LRC nearest-neighbor (NN) & $17.5 n^2$ \\
Clifford generator LRC & $49 n^2$ \\
 Clifford generator LRC NN & $49.5 n^2$ \\
\bottomrule
\end{tabular}
\caption{\label{tab:circuit-length-summary}
Sufficient sequence lengths for filtered RB with different circuit architectures using $2$-qubit gates on $n$ qubits.
The constant in $*$ is expected to be highly non-optimal.
The bounds are based on own calculations and Refs.~\cite{haferkamp_improved_2021,brown_convergence_2010,Hun19}, see also Sec.~\ref{sec:application-to-random-circuits}.
}
\end{table}

One might argue that the linear bound for universal or Clifford random circuits is too pessimistic.
Indeed, in the noiseless setting, we know that $\FD_\Ad(m)$ converges to the Haar-random constant value $1-1/d$ in logarithmic depth, at least for random circuits on linear nearest-neighbor and all-to-all architectures \cite{BarakEtAl:2020:Spoofing,dalzell_random_2022}.
This convergence is somewhat surprising since $\FD_\Ad(m)$ is a second moment of the used measure and generally linear depth is required for the convergence of arbitrary second moments.
Thus, we recover this generic scaling in our results.
Nevertheless, one may hope that the logarithmic scaling persists in the presence of weak noise.
Indeed, for the mentioned architectures subject to \emph{gate-independent and local noise}, it is possible to show convergence in logarithmic depth using intermediate results from Ref.~\cite{Dalzell21RandomQuantumCircuits}.\footnote{To the best of our knowledge, this has not been explicitly shown in the literature before. The statement follows by directly bounding the occurring partition functions in $\FD_\Ad(m)$ using Ref.~\cite[{Lem.~1 and 2}]{Dalzell21RandomQuantumCircuits}.} 
However, there is no rigorous argument for general architectures, such as the two-dimensional layouts used in superconducting devices.
Moreover, it is unclear whether the used results persist in the presence of correlated noise \cite{hunter_hones_private}.
The noise model used in this work is considerably more general and explicitly allows for correlated, gate-dependent noise.
We do not make any explicit assumptions beyond the implementation map model and the perturbation assumption.
In particular, it may be conceivable that our assumptions allow for adversarial noise that makes a linear scaling necessary.
Unfortunately, we were not able to construct such an example, nor could we improve the scaling in the bound \eqref{eq:thm8-subdominant-decay-bound}.
We think that further insights into the properties of random circuits beyond spectral gaps are necessary and leave the resolution of this problem for future work.
In practice, this means that the regime in which the filtered \ac{RB} signal is well-described by the dominant term has to be determined empirically.

Despite the above discussion, we expect that the bound \eqref{eq:sequence-length-bound-1-summary} is a reasonable approximation for small to moderate number of qubits, see also the discussion in Sec.~\ref{sec:random-walk-numerics}.
Moreover, our bound should be rather tight if smaller groups $G$ are used. 
In particular, if $G$ is the local Clifford group or the Pauli group, we find system-size-independent bounds for the sequence lengths.\footnote{Assuming that the irreps of interest do not have a too large support in the case of the local Clifford group.}

\paragraph{Sampling complexity.}

For `large' irreps, the range of the estimator $\hat{\FD}_\lambda(m)$ for $\FD_\lambda(m)$ can scale exponentially in the number of qubits.
For this reason, additional effort is required to establish efficient sample complexity bounds.
To this end, we derive the variance of $\hat{\FD}_\lambda(m)$ through a perturbative expansion of the second moment of the filter function, Thm.~\ref{thm:second-moment}.
This allows us to give general bounds on the sample complexity of filtered \ac{RB} in terms of the corresponding second moments of the noise-free and uniformly random implementation. 
In particular, our results show that using \emph{non-uniform sampling and the presence of noise does not significantly change the sampling complexity}.
To this end, we again have to ensure that sub-dominant terms appearing in the perturbative expansion are suppressed by choosing the sequence length $m$ large enough.
Typically, this requires that $m$ has to be chosen approximately \emph{twice as large} compared to the bound \eqref{eq:sequence-length-bound-1-summary}, but in most relevant cases the overhead is smaller.
Denote by $\EE[f_\lambda^2]_\ideal$ the second-moment of the filter function when the implementation is noise-free and the gates are drawn uniformly from the group.
We prove the following statement:
\begin{theorem*}[Sampling complexity of filtered \ac{RB}, informal]
  Choose the sequence length $m$ such that the subdominant terms are bounded by $\alpha$. If the number of samples fulfills $N \geq (\EE[f_\lambda^2]_\ideal + \alpha )\varepsilon^{-2}\delta^{-1}$, then the mean estimator $\hat{\FD}_\lambda(m)$ 
  {has an additive error bounded by $\varepsilon$}
  with probability at least ${1-}\delta$.
\end{theorem*}
The result allows us to derive sample complexity bounds for filtered \ac{RB} by calculating the moments of 
the analogous protocol using noise-free, Haar-distributed unitaries.
We give the results for groups that form global unitary 3-designs, local unitary 3-designs, and the Heisenberg-Weyl group in Prop.~\ref{prop:second-moments-ideal}.
Perhaps surprisingly, we find that filtered \ac{RB} with single-\emph{qubit} gates coming from a unitary 3-design (e.g.~the Clifford group) has constant sampling complexity irrespective of the size of the non-trivial support of the irreps.
Interestingly, a similar result in local dimension $q>2$ does not hold.
More generally, if the group $G$ contains the Heisenberg-Weyl group, Prop.~\ref{prop:2nd-moment-bounds} gives an upper and lower bound for the ideal second moment, generalizing the explicit calculations leading to Prop.~\ref{prop:second-moments-ideal}.

We find that the role of SPAM in the derivation of the sampling complexity is considerably more intricate than in the one of the signal form.
Technically, we have to assume a non-negativity condition on the involved SPAM coefficients.
Prop.~\ref{prop:Csigma} ensures non-negativity in the absence of SPAM and we argue that even with SPAM noise this is likely to still hold.
Furthermore, for the case that the Heisenberg-Weyl group is a subgroup of $G$, we show in Prop.~\ref{prop:2nd-moment-vis-leq-one} that the effect of SPAM is to reduce the absolute value of the involved coefficients

\paragraph{Modified filter functions.}

Finally, in Sec.~\ref{sec:other-filters}, we propose two potential modifications of filtered \ac{RB}
that can improve over the above scalings and yield meaningful decay parameters already from \emph{constant-depth} circuits, at least under simplifying assumptions on the noise. 
A detailed study of the properties of these modified protocols is left for future work.

The first modification is motivated from \emph{classical state shadows}:
We argue that the filtered \ac{RB} protocol can be improved by using the exact frame operator (also called \emph{measurement channel} in the literature) of the length-$m$ random circuit ensemble.
This should exactly correct for the non-uniformity of constant-depth circuits, thereby improving the estimates.
However, the computation of the frame operator of random circuits is an important open problem in classical shadows and in general computationally intractable \cite{hu_classical_2021,bu_classical_2022,akhtar_scalable_2022,bertoni_shallow_2022,arienzo_closed-form_2022}.

The second modification works for random circuits that are invariant under local Clifford unitaries or even only under Pauli operators.
In this case, a simple modification of the filter function is able to project onto the relevant subspace of the signal by `simulating' a trace inner product.
Again, this may lead to a significant reduction in the required sequence length.

\paragraph{Practical advice for designing RB protocols.}
Our results suggest the following blue-print for the randomized benchmarking of specific quantum devices.
To this end, our work enables an instructive `bottom-up' approach:
Practioners can now start directly from a set of native gates and the connectivity provided by the platform. 
Thereby, the benchmark directly reflects the implementation quality of the native gates.
Fast-scrambling random circuits with a large constant spectral gap such as brickwork circuits are particulary suited.
Even for non-uniform RB protocols
the \emph{underlying} group $G$ is
 of central importance as its representation theory determines the `simple' form of the {dominant} \ac{RB} signal.

We propose that a filtered \ac{RB} experiment should be preceded by a suitable theoretical and numerical analysis which ensures that the extracted decay rates are {reflecting the quality of the implementation} and assesses the scalibilty of the scheme.
{The following recipe summarizes the `bottom-up' approach:}
\begin{enumerate}[label=(\alph*)]
 \item Choose a gate set $\mathcal G$ and a layout for the random circuits, as well as a sampling scheme.
 \item Determine the generated group $G$ and the relevant irreps and multiplicities. Compute the frame operator $S$ and its (pseudo-)inverse.
 \item Estimate the spectral gap $\Delta$ of the measure defined in (a). 
    By adjusting the gate set $\mathcal{G}$ and the sampling scheme, the spectral gap can be optimized.
    To this end, a combination of numerical methods to compute `local' gaps and literature results may be useful (c.f.~Sec.~\ref{sec:application-to-random-circuits}).
    Moreover, numerical computation of spectral gaps for small system sizes ($n\leq 10$) and extrapolation to larger systems may be feasible and give good results.
 \item The theoretical bound \eqref{eq:sequence-length-bound-1-summary} gives a first estimate on the required sequence lengths.
 To get a more precise estimate, simulate the noiseless \ac{RB} experiment.
    The signal will quickly converge to a known constant $C_\lambda$ (c.f.~Eq.~\eqref{eq:ideal-signal-intro}) at a rate which matches (c), c.f.~Fig.~\ref{fig:rb-signal}.
    From this simulation, one can estimate the length of this mixing phase.
    More precisely, for a given error $\varepsilon$, determine the sequence length $m_0$ such that $|F_\lambda(m)-C_\lambda|\leq \varepsilon$ for $m\geq m_0$.
    This bound is an estimate for minimum sequence length required in the actual filtered \ac{RB} experiment.
 \item If the noise is too strong, the dominant decay in the filtered \ac{RB} experiment may originate in {an uncontrolled} mixing process  instead of desired scrambling of gate noise. 
    To ensure that we are in the {admissible} low-noise regime, {it is instructive to analyse the scaling of the perturbative assumption for realistic noise models (from theoretical consideration or experimental characterization) analytically or numerically.}
\end{enumerate}

%%% ==========================
\section{Preliminaries}
\label{sec:preliminaries}
%%% ==========================

In the following we introduce the mathematical definitions required for the precise statement and derivation of our results.  

\subsection{Operators, superoperators, and norms}
\label{sec:operators-norms}

\paragraph{Linear operators.}
Consider a finite-dimensional Hilbert space $\mathcal H$ over the field $\F = \C$ or $\F = \R$.
Then, the vector space $\End(\mathcal H)$ of linear operators on $\mathcal H$ is by itself a finite-dimensional Hilbert space over $\F$ with the \emph{Hilbert-Schmidt inner product}:
\begin{equation}
 \obraket{X}{Y} \coloneqq \tr( X^\dagger Y ).
\end{equation}
Here, $X^\dagger$ is the adjoint operator defined w.r.t.~the (complex or real) inner product on $\mathcal H$.
As in the usual Dirac notation, we can use the Hilbert-Schmidt inner product to define \emph{operator kets and bras} by $\oket{Y} \equiv Y$ and $\obra{X}: \, Y \mapsto \obraket{X}{Y}$.
Likewise, we can define outer products $\oketbra{X}{Y}$ which form linear maps on $\End(\mathcal H)$ acting as $A \mapsto \obraket{Y}{A} X$.
Following a common nomenclature, we refer to such linear maps as \emph{superoperators} (on $\mathcal H$).
As $\End(\mathcal H)$ is again a Hilbert space, it should not come as a surprise that the vector space of superoperators, $\End\End(\mathcal H)=\End^2(\mathcal H)$, can again be endowed with a Hilbert space structure using an analogue inner product.
By slightly overloading notation, we use $\obraket{\mathcal X}{\mathcal Y}$ to also denote the Hilbert-Schmidt inner product between superoperators $\mathcal X,\mathcal Y\in\End^2(\mathcal H)$.
Likewise, we denote outer products by $\oketbra{\mathcal X}{\mathcal Y}$, which are linear operators on superoperators and thus lie in $\End^3(\mathcal H)$.\footnote{We resist the urge to call these \emph{super duper operators} in public.}

We also consider linear maps $\mathcal V \rightarrow \mathcal W$ between different Hilbert spaces $\mathcal V$ and $\mathcal W$ over $\F$.
Analogue to above, these form a Hilbert space $\Hom(\mathcal V, \mathcal W)$ with the Hilbert-Schmidt inner product $\obraket{X}{Y}\coloneqq \tr (X^\dagger Y)$ where $X^\dagger: \mathcal W \rightarrow \mathcal V$ is the adjoint of $X$ defined by $\langle w, X(v)\rangle_\mathcal{W} = \langle X^\dagger(w), v \rangle_\mathcal{V}$.

\paragraph{Moore-Penrose pseudoinverse.}
Given a linear operator $X\in\End{\mathcal H}$, the restricted linear map $\tilde X:\, (\ker X)^\perp \rightarrow \ran X$ is an isomorphism and we define the \emph{Moore-Penrose pseudoinverse}, or simply \emph{pseudoinverse} of $X$ to be the linear operator $X^+$ which is $\tilde X^{-1}$ on $\ran X$ and identically zero on $(\ran X)^\perp$.
In a basis, $X^+$ can be computed using the singular value decomposition $X  = U\Sigma V^\dagger$ as the matrix $X^+ \coloneqq V \Sigma^+ U^\dagger$, where $\Sigma^+$ is the diagonal matrix obtained from $\Sigma$ by inverting all non-zero singular values.
Note that if $X$ is a real matrix, then the singular value decomposition is $X  = O\Sigma T^\dagger$ where $O$ and $T$ are orthogonal matrices; in particular, $X^+$ is a real matrix, too.

\paragraph{Norms.}
Throughout this paper, we use \emph{Schatten $p$-norms} which are defined for any linear map $X \in \Hom( \mathcal V, \mathcal W)$ between Hilbert spaces $\mathcal V$ and $\mathcal W$ and $p\in [1,\infty]$ as
\begin{equation}
 \norm{X}_p \coloneqq \Big(\tr \abs{X}^p \Big)^{\frac 1 p} = \left( \sum_{i=1}^d \sigma_i^p \right)^{\frac 1 p},
\end{equation}
where $|X|\coloneqq\sqrt{X^\dagger X}\in\End(\mathcal V)$ and $\sigma_i\geq 0$ are the singular values of $X$, i.e.~the square roots of the eigenvalues of the positive semidefinite operator $X^\dagger X$.
In particular, we use the \emph{trace norm} $p=1$, the \emph{spectral norm} $p=\infty$, as well as the \emph{Hilbert-Schmidt norm} $p=2$ which is simply the norm induced by the Hilbert-Schmidt inner product. 
The definition of Schatten norms only relies on the Hilbert space structure of the underlying vector space, thus these norms can be defined for operators, superoperators, and even higher-order operators alike.

\paragraph{Hermiticity-preserving maps.}
The Hilbert space of linear operators $\End(\mathcal H)$ over $\F=\C$ has a real structure in the sense that it decomposes as a direct sum $\End(\mathcal H) = \Herm(\mathcal H) \oplus i \Herm(\mathcal H)$, 
where $\Herm(\mathcal H)$ is the real Hilbert space of Hermitian matrices on $\mathcal H$.
The associated antilinear involution is given by the adjoint $\dagger$ with $\Herm(\mathcal H)$ as its fixed point space.
We call a linear map $\phi:\,\End(\mathcal H) \rightarrow \End(\mathcal H)$ \emph{Hermiticity-preserving} if it commutes with $\dagger$, i.e.~it maps $\Herm(\mathcal H)$ to itself.
Naturally, such a map induces a real linear map $\phi_\R$ on $\Herm(\mathcal H)$ by restriction.
Note that $\phi$ is Hermiticity-preserving if and only if it is represented by a real matrix in some basis of Hermitian matrices for $\End(\mathcal H)$ (and $\phi$ and $\phi_\R$ have the same matrix representation).

Finally, suppose $\phi:\,\End(\mathcal H) \rightarrow \End(\mathcal H)$ is Hermiticity-preserving, then so is its pseudoinverse $\phi^+$.
To see this, choose a basis of Hermitian matrices for $\End(\mathcal H)$ and let $A$ be the matrix representation of $\phi$ in this basis.
Since $\phi$ is Hermiticity-preserving, $A$ is real-valued and we find the singular value decomposition $A=O\Sigma T^\dagger$ with orthogonal matrices $O$ and $T$.
Next, note that the action of $O$ and $T$ on the complex Hilbert space $\End(\mathcal H)$ is unitary and hence this is also the singular value decomposition of $A$ seen as a complex matrix.
In particular, the matrix representation of $\phi^+$ is the real matrix $A^+ = T \Sigma^+ O^\dagger$ and hence $\phi^+$ is Hermiticity-preserving.

\subsection{Representation theory}

In this section, we briefly review some basic concepts from the representation theory of compact groups and introduce the relevant notation.
For more details, we refer the interested reader to standard text books \cite{FouHar91,brocker_representations_1985,goodman_symmetry_2009,bump_lie_2004}.

A \emph{topological group} $G$ is a group which is endowed with a topology such that group multiplication and inversion are continuous maps.
We call a topological group \emph{compact} if it is a compact topological Hausdorff space. 
A compact group comes with a unique Borel measure $\mu$, called the \emph{Haar measure}, which is left and right invariant under group multiplication, $\mu(g A) = \mu (A) = \mu(A g)$ for all $g\in G$ and open sets $A\subset G$, and normalized as $\mu(G) = 1$.

\subsubsection{Unitary representations}

Given a compact group $G$, a finite-dimensional \emph{unitary representation} of $G$ is a pair $(\rho,V)$ where $V$ is a finite-dimensional Hilbert space and $\rho:\,G\rightarrow \U(V)$ is a group homomorphism such that the map $G\times V \rightarrow V$ given by $(g,v)\mapsto \rho(g)v$ is continuous.
In general, we call two representations $(\rho,V)$ and $(\rho',V')$ \emph{isomorphic} or \emph{equivalent} if there is a unitary isomorphism $U:\, V\rightarrow V'$ such that $\rho(g) = U^\dagger \rho'(g) U$ for all $g\in G$.
A subspace $W\subset V$ is called \emph{invariant} w.r.t.~$\rho$ if $\rho(g)(W) = W$ for all $g\in G$.
We call $\rho$ an \emph{irreducible representation} or short \emph{irrep} if the only invariant subspaces are $\{0\}$ and $V$ itself.
Otherwise, we call $\rho$ \emph{reducible}.
It is well-known that any finite-dimensional unitary representation $(\rho,V)$ of $G$ is \emph{completely reducible}, i.e.~we can write the vector space $V$ as a direct sum of invariant subspaces $V_i$,
\begin{equation}
 V = \bigoplus_i V_i,
 \label{eq:complete-reducibility}
\end{equation}
such that each restriction $\rho_i \coloneqq \rho|_{V_i}$ is irreducible.
We can then write $\rho = \oplus_i \rho_i$.
However, the decomposition \eqref{eq:complete-reducibility} is in general not unique.
This is the case if two irreps $\rho_i$ and $\rho_j$ with $i\neq j$ are \emph{isomorphic}. 
Then, the possible ways of decomposing the representation $\rho_i\oplus\rho_j$ into irreps corresponds exactly to a $\U(2)$ symmetry. 

More generally, we call a representation \emph{isotypic} if it is a direct sum of mutually isomorphic irreps.
Let us denote by $\Irr(G)$ the set of \emph{inequivalent} irreducible unitary representations of $G$ and note that these are necessarily finite-dimensional for a compact group $G$.
Given an irrep $\tau\in\Irr(G)$, the $\tau$-isotype of a representation $(\rho,V)$ is defined as the subspace $V(\tau) \subset V$ given as the ordinary sum of all irreducible subspaces isomorphic to $\tau$.
One can show that the orthogonal projection onto $V(\tau)$ is given by the formula
\begin{equation}
 P_\tau \coloneqq \dim(\tau) \int_G \overline{\chi_\tau}(g) \rho(g) \dd\mu(g).
 \label{eq:projection-isotypes}
\end{equation}
Here, $\chi_\tau(g) \coloneqq \tr(\tau(g))$ is the character of the irrep $\tau$ and $\mu$ is the Haar measure on $G$.
In particular, we have the canonical decomposition into isotypes as follows
\begin{equation}
 V = \bigoplus_{\tau\in\Irr(G)} V(\tau).
 \label{eq:isotypic-decomposition}
\end{equation}
Note that $V(\tau) = \{0\}$ if $\tau$ is not contained in $\rho$.
Hence, the sum actually runs over the inequivalent irreps of $\rho$, which we denote by $\Irr(\rho)$.
The dimension of $V(\tau)$ is given as $n_\tau \dim(\tau)$ where $n_\tau$ is the \emph{multiplicity} of $\tau$: it is the unique number of copies of $\tau$ that appear in any decomposition of $\rho$.
For some choice of irrep decomposition we have
\begin{equation}
 V(\tau) \simeq V_\tau^{\oplus n_\tau} \simeq V_\tau \otimes \C^{n_\tau},
\end{equation}
where $V_\tau$ is the Hilbert space on which $\tau$ acts.
The corresponding decomposition of $\rho|_{V(\tau)}$ is $\tau^{\oplus n_\tau}$ under the first identification and $\tau\otimes\id_{n_\tau}$ under the second one.
The factor $\C^{n_\tau}$ is sometimes called the \emph{multiplicity space}.

For any compact group $G$, the vector space of square-integrable complex functions on $G$, $L^2(G,\mu)\equiv L^2(G)$, is a Hilbert space endowed with the inner product
\begin{equation}
 \langle f, g \rangle \coloneqq  \int_G \overline{f(t)} g(t) \dd\mu(t).
\end{equation}
An important example of functions in $L^2(G)$ are the characters $\chi_\rho(g) = \tr(\rho(g))$, where $\rho$  is a finite-dimensional (not necessarily irreducible) unitary representation. 
Characters are very useful in the representation theory of compact groups, hence let us summarize a few important facts.
Here, $\rho$ and $\rho'$ are two finite-dimensional unitary representations.
\begin{enumerate}[label=(\roman*)]
 \item $\rho$ and $\rho'$ are isomorphic if and only if their characters agree.
 \item $\rho$ is irreducible if and only if $\langle \chi_\rho, \chi_\rho \rangle = 1$.
 \item Characters of inequivalent irreps are orthogonal:  $\langle \chi_\tau, \chi_{\tau'} \rangle = 0$ for all $\tau,\tau'\in\Irr(G)$ with $\tau\neq \tau'$.
 \item For any $\tau\in\Irr(G)$, $n_\tau = \langle \tau, \rho\rangle$ is the multiplicity of $\tau$ in $\rho$ (see e.g.\ \cite[Corollary~2.16]{FouHar91}).
\end{enumerate}
The Hilbert space $L^2(G)$ has more interesting properties on which we comment in more detail in Sec.~\ref{sec:fourier-transform}.

\paragraph{Notation.}
For $G = \U(d)$, we denote by $g \mapsto U_g$ its defining representation as unitary matrices on $\C^d$. 
Likewise, we use the same notation for the restriction to a subgroup $G \subset \U(d)$.
In the following and throughout this paper, we assume that all representations are finite-dimensional and unitary, if not stated otherwise.

\subsubsection{Real representations}
\label{sec:real-representations}

{

More generally, for any vector space $W$ over a field $\F$, a continuous group homomorphism $\rho:\, G \rightarrow \GL(W)$ is called a \emph{linear representation}, or simply \emph{representation}.
A particularly relevant case is $\F=\R$ in which case we call $\rho$ a \emph{real} representation.
The \emph{complexification} $\rho^\C$ of $\rho$ is a representation on the complex vector space $W^\C \coloneqq W \oplus i W$.
Conversely, if $\omega$ is a representation on the \emph{complex} vector space $V$ endowed with a real structure $V = W \oplus i W$, and a representation $\rho:\, G \rightarrow \GL(W)$ such that $\rho^\C = \omega$, then $\omega$ is also called a \emph{real representation}.
We then also write $\omega_\R=\rho$ for the restriction of $\omega$ to $W$.

These concepts can be used to study the relations between the real and complex irreps of a compact group $G$. 
Suppose that the real representation $\rho$ decomposes into irreps as $\rho=\oplus_\lambda \tau_\lambda$, then we have $\rho^\C=\oplus_\lambda \tau_\lambda^\C$, however the complexified irreps $\tau_\lambda^\C$ may now be reducible.
As we see shortly, the possible cases that can occur are exactly captured by Schur's lemma:
In its general form, it says that the commutant of an irrep over a field $\F$ is a division algebra over $\F$.
The division algebras over $\F=\R$ are classified by Frobenius' theorem and given by $\R$, $\C$, and the quaternions $\mathbb H$.
Since the commutant of $\tau_\lambda^\C$ is given by the complexification of the commutant of $\tau_\lambda$, we find that it is either $\R \oplus i \R \simeq \C$, $\C\oplus i\C \simeq \C^2$, or $\mathbb H \oplus i \mathbb H \simeq \C^{2 \times 2}$.
This restricts the possible reductions of $\tau_\lambda^\C$ to either $\sigma$, $\sigma\oplus\overline\sigma$, or $\sigma\oplus\sigma$, where $\sigma$ is a suitable complex irrep.
Consequently, we say that $\sigma$ is of \emph{real, complex, or quaternionic type}, respectively.
The results are summarized in the following table:

\begin{table}[h]
\centering
\begin{tabular}{lccc}
\toprule
 Commutant & $\tau^\C$ & $\dim_\R\tau$ & $\sigma\simeq\overline\sigma$ \\
 \midrule
 $\R$ & $\sigma$ & $k$ & yes \\
 $\C$ & $\sigma\oplus\overline\sigma$ & $2k$ & no \\
 $\mathbb H$ & $\sigma\oplus\sigma$ & $4k$ & yes \\
 \bottomrule
\end{tabular} 
\caption{
  Summary of the possible cases for the commutant of a real irrep $\tau$.
  These cases also correspond to the possible reductions of the complexified representation $\tau^\C$ in terms of a complex irrep $\sigma$.
  In the real and quaternionic case, $\sigma$ is self-conjugate, in the complex case, it is not.
}
\label{tab:division-algebras}
\end{table}

}

\subsection{Fourier transform on compact groups}
\label{sec:fourier-transform}

Given a function $f\in L^2(G)$, we can construct its \emph{Fourier transform} $\widehat f$ which maps an irrep $\tau_\lambda$ of $G$ to an operator on the representation space $V_\lambda$ as follows \cite{folland_course_2015}:
\begin{equation}
  \label{eq:fourier-transform-of-function}
  \widehat f[\tau_\lambda] \coloneqq  \int f(g) \tau_\lambda(g)^\dagger \dd\mu(g) \in \End(V_\lambda).
\end{equation}
It is a classic result in harmonic analysis that the map $f\mapsto\widehat{f}$ induces an algebra isomorphism $L^2(G)\simeq \bigoplus_{\lambda\in\Irr(G)} \End(V_\lambda)$; this is one incarnation of the famous \emph{Peter-Weyl theorem}.

In the following, we introduce a generalization of the Fourier transform to operator-valued functions $\phi:\,G\rightarrow\End(V)$ which can be understood as the component-wise application of the above Fourier transform.
This has been studied in the mathematical literature \cite{gowers_inverse_2017} and introduced to the randomized benchmarking literature in Refs.~\cite{Merkel18,helsen_general_2022}.

\begin{definition}[Fourier transform]
 \label{def:fourier}
 Let $\phi:\, G\rightarrow \End(V)$ be a square-integrable operator-valued function on a compact group $G$ with Haar measure $\mu$.
 Let $\lambda\in\Irr(G)$ label the inequivalent irreducible representations $(V_\lambda,\tau_\lambda)$ of $G$.
 Then, for any $\lambda\in\Irr(G)$, we define a linear operator $\Hom(V,V_\lambda) \rightarrow \Hom(V,V_\lambda)$ as follows:
 \begin{equation}
 \label{eq:def-fourier-transform}
  \widehat\phi[\tau_\lambda] \coloneqq   \int_G  \tau_\lambda(g)^\dagger (\argdot) \phi(g) \dd\mu(g).
 \end{equation}
\end{definition}

Note that our definition of Fourier transform differs a bit from the one in Refs.~\cite{Merkel18,helsen_general_2022}.
The reason for this is that there is a certain ambiguity in defining the Fourier transform, even for ordinary functions as in Eq.~\eqref{eq:fourier-transform-of-function}.
We think that, for our purposes, Definition \ref{def:fourier} is suited best, although other, equivalent formulations can be useful in certain contexts.
Thus, let us briefly comment on these.
We have a canonical isomorphism 
\begin{align}
    \End(\Hom(V,V_\lambda)) 
    &\simeq (V_\lambda\otimes V^*)\otimes(V_\lambda\otimes V^*)^* \\
    &\simeq (V_\lambda\otimes V_\lambda^*)\otimes(V\otimes V^*) \\
    &\simeq \End(V_\lambda\otimes V) \,,
\end{align}
with respect to which Eq.~\eqref{eq:def-fourier-transform} becomes
\begin{equation}
 \label{eq:def-fourier-transform-dagger}
 \widehat\phi[\tau_\lambda] \simeq \int_G  \tau_\lambda(g)^\dagger\otimes\phi(g) \dd\mu(g) \in \End(V_\lambda\otimes V).
\end{equation}
However, under this isomorphism, the order of composition is exchanged on the first factor compared to the second.
This can be accounted for by applying a suitable algebra anti-automorphism to the first factor which, up to a choice of basis, is uniquely given by the transposition:
\begin{equation}
\label{eq:def-fourier-transform-dagger-cc}
 \widehat\phi[\tau_\lambda] \simeq \int_G  \bar\tau_\lambda(g)\otimes\phi(g) \dd\mu(g) \in \End(V_\lambda\otimes V)
\end{equation}
This is exactly the definition of the Fourier transform used in Refs.~\cite{Merkel18,helsen_general_2022}. 
Note that the proper isomorphism is similar to the \emph{vectorization} which is prominently used in quantum information to represent a unitary channel $\rho \mapsto U \rho\, U^\dagger$ as the operator $U \otimes \bar U$.

It is convenient to slightly extend the notation above and use the defining Eq.~\eqref{eq:def-fourier-transform} even if the argument is not an irrep $\tau_\lambda$.
For an arbitrary representation $\rho\simeq\bigoplus_{\lambda}\tau_\lambda^{\oplus n_\lambda}\simeq\bigoplus_{\lambda} \tau_\lambda \otimes \id_{n_\lambda}$, we then obtain, in the light of the above isomorphisms, 
\begin{equation}
\label{eq:fourier-transform-reducible-rep}
 \widehat\phi[\rho]  \simeq \bigoplus_\lambda \widehat\phi[\tau_\lambda]^{\oplus n_\lambda} \simeq \bigoplus_\lambda \widehat\phi[\tau_\lambda] \otimes \id_{n_\lambda}\,. 
\end{equation}
For example, if $\omega(g)= U_g (\argdot)U_g\ad$ is acting by conjugation, $\widehat \omega[\omega]$ is a convenient notation of what in quantum information is referred to as the \emph{channel twirl}.

We proceed by proving a property of the Fourier transform which we will frequently use. 
It is a direct consequence of Schur's lemma.
\begin{proposition}[Fourier transform of representations]
\label{prop:fourier-transform}
Let $\rho:\, G \rightarrow \U(V)$ be a representation of $G$. 
Then, $\widehat\rho[\tau_\lambda]$ is an orthogonal projection and its rank is the multiplicity $n_\lambda$ of the irrep $\tau_\lambda$ in $\rho$.
More precisely, if the $\tau_\lambda$-isotypic component in $V$ is $V(\lambda) \simeq V_\lambda \otimes \C^{n_\lambda}$ with $n_\lambda \neq 0$, then $\widehat\rho[\tau_\lambda]$ is block-diagonal w.r.t.~the induced decomposition and the only non-zero block is
\begin{equation}
\label{eq:fourier-transform-representation}
 \widehat\rho[\tau_\lambda] \simeq 
  d_\lambda^{-1}\oketbra{\id_\lambda}{\id_\lambda}\otimes \id_{n_\lambda}^*,
\end{equation}
where $\id_{n_\lambda}^*$ is the identity on the dual multiplicity space $(\C^{n_\lambda})^*$.
\end{proposition}
\begin{proof}
If $\rho:\, G \rightarrow \U(V)$ is a representation, then we have
\begin{align}
 \widehat\rho[\tau_\lambda]^2 
 &= \int\int \tau_\lambda(g)^\dagger\tau_\lambda(g')^\dagger (\argdot) \rho(g')\rho(g) \dd\mu(g) \dd\mu(g') \\
 &= \int_G  \tau_\lambda(g)^\dagger (\argdot) \rho(g) \dd\mu(g) = \widehat\rho[\tau_\lambda]. 
\end{align}
Hence, $\widehat\rho[\tau_\lambda]$ is a projector.
Moreover, $\widehat\rho[\tau_\lambda]$ is also self-adjoint:
\begin{align}
 \obraket{\widehat\rho[\tau_\lambda](X)}{Y} 
 &= \tr\left( \int_G  \rho(g)^\dagger X^\dagger \tau_\lambda(g) \dd\mu(g) Y\right) \\
 &= \tr\left(X^\dagger \int_G    \tau_\lambda(g) Y \rho(g)^\dagger \dd\mu(g) \right) \\
 &= \obraket{X}{\widehat\rho[\tau_\lambda](Y)}.
\end{align}
In the last step, we use that the Haar measure is invariant under inversion.
Thus, $\widehat\rho[\tau_\lambda]$ is an \emph{orthogonal} projector and its rank is
\begin{align}
 \rank \widehat\rho[\tau_\lambda] 
 &= \tr \widehat\rho[\tau_\lambda] 
 = \int_G \overline{\tr \tau_\lambda(g)} \tr \rho(g) \dd\mu(g) \\
 &= \langle \chi_\lambda, \tr\rho \rangle = n_\lambda \,.
\end{align}
We can now explicitly compute $\widehat\rho[\tau_\lambda]$ by decomposing $\rho$ into irreps as follows:
\begin{align}
 \widehat\rho[\tau_\lambda]
  &= \int_G \tau_\lambda(g)^\dagger (\argdot) \rho(g) \dd\mu(g) \\
  &\simeq \bigoplus_{\lambda'\in\Irr(\rho)} \int_G \tau_\lambda(g)^\dagger (\argdot)  \tau_{\lambda'}(g) \otimes \id_{n_{\lambda'}} \dd\mu(g) . \label{eq:fourier-projector}
\end{align}
On every $\lambda'$ block, the integral on the right hand side is exactly the projection onto operators $X$ which are equivariant, i.e.~$\tau_\lambda(g)X  = X( \tau_{\lambda'}(g) \otimes \id_{n_{\lambda'}} )$ for all $g\in G$.
By Schur's lemma, such an operator has to be trivial if $\lambda\neq\lambda'$ and otherwise we can write it as $X = \id_\lambda \otimes x_{\lambda}^*$ for some $x_\lambda^*\equiv \bra{x_\lambda} \in (\C^{n_\lambda})^*$. 
An orthogonal basis for the subspace of these operators is given by $\id_\lambda\otimes e_{\lambda,i}^*$ where $(e_{\lambda,i})_{i \in [n_\lambda]}$ is a basis for $\C^{n_\lambda}$.
With respect to the second isomorphism in Eq.~\eqref{eq:fourier-transform-reducible-rep}, the Fourier transform then becomes
\begin{align}
 \widehat\rho[\tau_\lambda] 
 \simeq
 \frac{1}{d_\lambda}\sum_{i=1}^{n_\lambda}\oketbra{\id_\lambda \otimes e_{\lambda,i}^*}{\id_\lambda \otimes e_{\lambda,i}^*}
 =
 \frac{\oketbra{\id_\lambda}{\id_\lambda }}{d_\lambda}\otimes \id_{n_\lambda}^*,
\end{align}
where $\id_{n_\lambda}^*$ denotes the identity on the dual multiplicity space $(\C^{n_\lambda})^*$.
\end{proof}

For our purposes, it will be convenient to consider Fourier transforms where the integration is performed with respect to another measure than the Haar measure $\mu$.
As we will see in a moment, this is captured by a generalization of Fourier analysis from functions to measures.
For simplicity, let us return to the ordinary Fourier transform of a function $f\in L^2(G)$.
We can effectively change the measure by multiplying $f$ with a density $\varphi\in L^1(G)$:
\begin{align}
  \widehat{f\varphi}[\tau_\lambda] 
  % &= \int \tau_\lambda(g)^\dagger f(g)\varphi(g)\dd\mu(g) \\
  &= \int \tau_\lambda(g)^\dagger f(g) \dd(\varphi\mu)(g).
\end{align}
However, this does not allow us to consider measures on $G$ which do not have a density w.r.t.~the Haar measure $\mu$, for instance discrete measures.\footnote{Admittedly, the conceptually difficult bit would be the singular continuous part of the measure.} 
Nevertheless, we can instead interpret $\widehat f$ as a transformation of the complex measure $f\mu$ and notice that this transformation is still well-defined if we replace $f\mu$ by a suitable (complex) measure $\nu$ on $G$.
More precisely, we denote by $\Borel(G)$ the Borel $\sigma$-algebra on $G$, i.e.~the one generated by the open sets of $G$.
We consider measures on $\Borel(G)$ called (complex) \emph{Radon measures}
which are exactly those measures $\nu$ that yield continuous linear functionals on $C(G)$ by integration $f \mapsto \int_G f\dd\nu$.
The vector space of complex Radon measures is denoted by $\MA(G)$.
Note that for the typical groups which we are considering, namely finite or compact Lie groups, any \emph{finite} measure on $\Borel(G)$ is a Radon measure.\footnote{More generally, this is true if our compact group $G$ is also second-countable.}
Then, we define the Fourier transform of $\nu\in \MA(G)$ as
\begin{equation}
  \label{eq:fourier-transform-of-measure}
  \widehat \nu[\tau_\lambda] \coloneqq  \int  \tau_\lambda(g)^\dagger \dd\nu(g).
\end{equation}
In this way, we can treat continuous and discrete measures on the same footing and recover the Fourier transform of functions for $\nu = f\mu$.
To generalize this discussion to operator-valued functions, we define \emph{operator-valued measures} in a natural way.

\begin{definition}[Operator-valued measure]
 \label{def:ovm}
 A map $\nu:\,\Borel(G) \rightarrow \End(V)$ taking values in the linear operators on a finite-dimensional real/complex Hilbert space $V$ is called an \emph{\ac{OVM}} if for all $v,w\in V$ the function
 \begin{equation}
  \Borel(G)\ni A \longmapsto \sandwich{v}{\nu(A)}{w} ,
 \end{equation}
 is a real/complex Radon measure on $G$.
\end{definition}

Suppose $f$ is a function on $G$ which is integrable w.r.t.~the measures $\nu_{ij}(A) \coloneqq  \sandwich{e_i}{\nu(A)}{e_j}$ for some orthonormal basis $\{e_i\}$ of $V$.
For instance, this is always the case if $f$ is continuous.
Then, we can define the integral of $f$ with respect to the operator-valued measure $\nu$,
\begin{equation}
 \int_G f(g) \dd\nu(g) \in \End(V),
\end{equation}
which is a well-defined linear operator on $V$.
{This definition extends in the obvious way to integrals of matrix-valued functions.}

\begin{definition}[Fourier transform of \acp{OVM}]
 \label{def:fourier-ovm}
 Let $\nu$ be an \acl{OVM} on a compact group $G$ taking values in $\End(V)$.
 As above, let $\lambda\in\Irr(G)$ label the inequivalent irreducible representations $(V_\lambda,\tau_\lambda)$ of $G$.
 Then, for any $\lambda\in\Irr(G)$, we define the following operator on $\Hom(V,V_\lambda)$:
 \begin{equation}
  \widehat\nu[\tau_\lambda] \coloneqq  \int_G  \tau_\lambda(g)^\dagger (\argdot) \dd\nu(g).
 \end{equation}
\end{definition}

We recover the Fourier transform for operator-valued functions $\phi$ in Def.~\ref{def:fourier} by considering the OVM $\nu_\phi \coloneqq \phi \mu$ where $\mu$ is the  Haar measure on $G$.
For this paper, we exclusively consider OVMs of the form $\nu_\phi \coloneqq \phi \nu$ where $\nu\in\MA(G)$ is a Radon probability measure on $G$.

Let us emphasize that Prop.~\ref{prop:fourier-transform} only holds for the Fourier transform in the sense of Def.~\ref{def:fourier} since it relies on the properties of the Haar measure.
In the next section, we discuss the structure of the Fourier transform for specific measures $\nu$ more generally.

\subsection{\texorpdfstring{$\rho$}{ρ}-designs, moment operators, and random walks}
\label{sec:rho-designs}

In this section, we introduce the concept of \emph{designs} and thereby adapt a quite general definition of this term from Ref.~\cite{bannai_explicit_2020}.
We then proceed by discussing the relation of designs to the previously introduced Fourier transform and summarize a few properties of \emph{moment operators}.

For reference, recall that a probability measure $\nu\in\MA(\U(d))$ on the unitary group $\U(d)$ is called a \emph{unitary $t$-design} if the following equality holds:
\begin{equation}
  \int_{\U(d)} U_g^{\otimes t} (\argdot) ( U_g^{\otimes t})^\dagger \dd\nu(g) = \int_{\U(d)} U_g^{\otimes t} (\argdot) ( U_g^{\otimes t})^\dagger \dd\mu(g), 
\end{equation}
where $\mu$ is again the normalized Haar measure on $\U(d)$.
In the following, we use a generalization of this definition where $\U(d)$ is replaced by an arbitrary compact group $G$ and likewise the representation $\tau^{(t)}(g)\coloneqq U_g^{\otimes t} (\argdot) ( U_g^{\otimes t})^\dagger$ is replaced by an arbitrary representation $\rho$ of $G$.

\begin{definition}[$\rho$-design]
\label{def:rho-design}
 Let $(\rho,V)$ be a representation of $G$ and let $\nu\in \MA(G)$ be a probability measure on $G$.
 \begin{enumerate}[label=(\roman*)]
  \item $\nu$ is called a \emph{$\rho$-design} if and only if
   \begin{equation}
    \MO_\rho(\nu) \coloneqq \int_G \rho(g) \dd\nu(g) = \int_G \rho(g) \dd\mu(g) = \MO_\rho(\mu).
   \end{equation}
   The linear operator $\MO_\rho(\nu)$ is called the \emph{moment operator} of $\nu$ (w.r.t.~$\rho$).
   For any set $\Lambda$ of $G$-representations, we call $\nu$ a $\Lambda$-design if $\nu$ is a $\rho$-design for all $\rho\in\Lambda$. 
  \item Let $\norm{\argdot}$ be an arbitrary norm on $\End(V)$ and $\varepsilon >0$.
    Then, we call $\nu$ an \emph{$\varepsilon$-approximate $\rho$-design} w.r.t.~$\norm{\argdot}$ if $\norm{\MO_\rho(\nu)-\MO_\rho(\mu)} < \varepsilon$.
    If the norm is not further specified, we mean the spectral norm $\norm{\argdot}=\snorm{\argdot}$.
 \end{enumerate} 
\end{definition}

If the representation $\rho$ is reducible, say $\rho \simeq \sigma\oplus\eta$, then the moment operator $\MO_\rho(\nu)$ is block-diagonal:
\begin{equation}
 \MO_\rho(\nu) \simeq \MO_\rho(\sigma) \oplus \MO_\rho(\eta).
\end{equation}
Hence, $\nu$ is a $\rho$-design if and only if it is a $\{\sigma,\eta\}$-design, and more generally a $\kappa$-design for any subrepresentation $\kappa$ of $\rho$.
In particular, if $\Irr(\rho)$ labels the irreps appearing in $\rho$, \ie\ if $\rho\simeq\bigoplus_{\lambda\in\Irr(\rho)} \rho_\lambda^{\oplus n_\lambda}$, then $\MO_\rho(\nu) \simeq \bigoplus_{\lambda\in\Irr(\rho)} \MO_\lambda(\nu)^{\oplus n_\lambda}$ and a $\rho$-design is exactly a $\Irr(\rho)$-design.

Let us briefly specialize this notion to unitary $t$-designs which, in the notation of Def.~\ref{def:rho-design}, are exactly $\tau^{(t)}$-designs for the compact group $G=\U(d)$ and the representation $\tau^{(t)}(g)= U_g^{\otimes t} (\argdot) ( U_g^{\otimes t})^\dagger$.
From the representation theory of the unitary group, we know that each of its irreps is labelled by a non-increasing integer sequence $\lambda_1 \geq \dots \geq \lambda_d$ (see e.g.~Ref.~\cite[Thm.~38.3]{bump_lie_2004}).
In particular, the irreps appearing in $\tau^{(t)}$ are exactly labelled by those sequences for which the sum of positive elements $\lambda^+$ is equal to the absolute sum of negative elements $\lambda^-$, and bounded by $t$, \ie~$\lambda^+=\lambda^-\leq t$ \cite{stembridge_rational_1987,stembridge_combinatorial_1989,benkart_tensor_1994} (see also \cite{roy_unitary_2009}).
Let us define the set of such sequences as
\begin{equation}
 \diagup^t_d \coloneqq  \big\{ \lambda \in \Z^d \; | \; \lambda_1 \geq \dots \geq \lambda_d, \, \lambda^+=\lambda^- \leq t \big\}.
\end{equation}
Then, a unitary $t$-design is exactly a $\diagup^t_d$-design in the sense of Def.~\ref{def:rho-design}.
From the definition, we directly find $\diagup^{t-1}_d \subset \diagup^t_d$, and hence we recovered the well-known fact that any unitary $t$-design is also a $(t-1)$-design.

Consider a fixed representation $\omega$ of a compact group $G$ (later called the \emph{reference representation}) and some probability measure $\nu\in \MA(G)$.
The Fourier transform of the operator-valued measure $\omega\nu$ is by Def.~\ref{def:fourier-ovm} given as
\begin{equation}
  \widehat{\omega\nu}[\tau_\lambda] = \int_G \tau_\lambda(g)^\dagger (\argdot)\, \omega(g)\dd\nu(g) \equiv \MO_\rho(\nu),
\end{equation}
i.e.~$\widehat{\omega\nu}[\tau_\lambda]$ is the \emph{moment operator} of $\nu$ w.r.t.~the representation $\rho=\tau_\lambda^\dagger (\argdot)\,\omega$.
Hence, $\nu$ is a $\rho$-design if and only if $\widehat{\omega\nu}[\tau_\lambda] = \widehat{\omega}[\tau_\lambda]$.

In the following, we discuss some properties of moment operators $\MO_\rho(\nu)$ for a certain family of probability measures $\nu$.

\begin{definition}
\label{def:random-walks}
Let $\nu\in \MA(G)$ be a Radon measure on a compact group $G$. 
\begin{enumerate}[label=(\roman*)]
  \item $\nu$ is called \emph{symmetric} if it is invariant under inversion.
  \item We say that $\nu$ has \emph{support on generators} if there is a set $\mathcal{G}$ of generators for $G$ such that every open neighbourhood of a generator $g\in\mathcal G$ has non-zero measure, i.e.~$\mathcal{G}\subset\supp\nu$.
  \item Suppose $\nu$ is a probability measure. The \emph{random walk generated by $\nu$} is the stochastic process $(g_m)_{m\in\N}$ on $G$ defined by the transition rule $g_{m+1} = h g_m$ where $h\sim\nu$. 
  Equivalently, the random variables $g_m$ are distributed as $\nu^{*m}$ for all $m\in \N$.
\end{enumerate}
\end{definition}

Note that having support on generators is not much of a restriction on $\nu$ since we can always consider the subgroup $G'\subset G$ generated by $\supp(\nu)$ instead.
If $G'$ is dense in $G$, then we will abuse terminology a bit and still say that $\nu$ has support on generators of $G$, because this detail does not affect our arguments.
In fact, this is a practically relevant situation as e.g.~$G=\U(d)$ is not finitely generated, but has finitely generated dense subgroups.
For instance, a measure which has support on generators of $\U(d)$ in the broad sense, is given by drawing from a finite, universal gate set.

Furthermore, we will encounter important examples of measures that are \emph{not} symmetric.
Most of the following discussion and all results in Sec.~\ref{sec:results} also hold for non-symmetric probability measures.
Nevertheless, symmetric ones have some nice properties which we need for certain conclusions.
To this end, we introduce the \emph{symmetrization trick} in the end of this section.

Let us now explore more properties of the moment operator $\MO_\rho(\nu)$, given a representation $(\rho,V)$ of $G$, and how the concepts of Def.~\ref{def:random-walks} are reflected in its spectrum (see e.g.~Refs.~\cite{diaconis_comparison_1993,brown_convergence_2010}).
Suppose $v\in V$ is a right eigenvector of $\MO_\rho(\nu)$ with eigenvalue $\lambda\in\C$, then we find:
\begin{align}
  |\lambda| \twonorm{v}^2
  &= \abs[\big]{\sandwich{v}{\MO_\rho(\nu)}{v} }
  \leq \int_G \big| \braket{v}{\rho(g)v} \big| \dd\nu(g) \\
  &\leq \int_G \twonorm{v} \twonorm{\rho(g)v} \dd\nu(g) 
  = \twonorm{v}^2\,,
\end{align}
where we have used the Cauchy-Schwarz inequality and that the representation is unitary. 
This shows that the spectrum of $\MO_\rho(\nu)$ is contained in the unit disk $\{z\in\C\,|\,|z|\leq 1\}$.
Next suppose $v\in V$ is a right eigenvector with eigenvalue 1.
Then, we find
\begin{align}
 \MoveEqLeft[2] \int_G \sandwich{v}{(\id-\rho(g))^\dagger (\id-\rho(g))}{v}\,\dd\nu(g) \\
 &=
 2\twonorm{v}^2 - \int_G \sandwich{v}{\rho(g)^\dagger + \rho(g)}{v} \,\dd\nu(g) \\
 &= 
 2\twonorm{v}^2 - \sandwich{v}{\MO_\rho(\nu)^\dagger + \MO_\rho(\nu)}{v} \\
 &=
 2 \big( \twonorm{v}^2  - \Re(\sandwich{v}{\MO_\rho(\nu)}{v}) \big) = 0.
\end{align}
Since $(\id-\rho(g))^\dagger (\id-\rho(g))$ is positive semidefinite, the left hand side is an integral of a non-negative function.
This can only be zero if the function is zero $\nu$-almost everywhere.
Hence, $v\in V$ is a right eigenvector with eigenvalue 1 if and only if $\rho(g)v = v$ $\nu$-almost everywhere.
Put differently, the right 1-eigenspace of $\MO_\rho(\nu)$ is the fixed point space of $\supp(\nu)$.
Likewise, we find that the left 1-eigenspace corresponds to the fixed point space of $\supp(\tilde\nu)$ where $\tilde\nu(A)\coloneqq \nu(A^{-1})$ is the inverted measure.

If $\nu$ has support on generators, the fixed point space of $\supp(\nu)$ coincides with the fixed point space of the whole group $G$.
The same holds for the fixed point space of $\supp(\tilde\nu)$, hence the left and right 1-eigenspaces of $\MO_\rho(\nu)$ coincide and agree with $V(1)$, the trivial isotype of the representation $\rho$ (possibly $\{0\}$).
Hence, the moment operator $\MO_\rho(\nu)$ can be unitarily block-diagonalized as follows:
\begin{equation}
\label{eq:moment-op-spectral-decomposition}
 \MO_\rho(\nu) 
 \simeq
 \begin{bmatrix}
  \id & 0 \\
  0 & \MO_{\rho\ominus 1}(\nu)
 \end{bmatrix} ,
\end{equation}
where $\rho\ominus 1$ denotes the representation on the orthocomplement of $V(1)$.
To quantify how much $\MO_\rho(\nu)$ differs from the moment operator $\MO_\rho(\mu)$ w.r.t.~the Haar measure $\mu$, one typically uses the spectral distance.
Note that the $\MO_\rho(\mu)$ is exactly the orthogonal projection onto the trivial isotype $V(1)$ and hence we have the following relation
\begin{align}
\label{eq:def-spectral-gap}
 \snorm{ \MO_\rho(\nu) - \MO_\rho(\mu) } &= 1 - \Delta_\rho(\nu), \\
 \Delta_\rho(\nu) &\coloneqq 1 - \snorm{\MO_{\rho\ominus 1}(\nu)}\, .
\end{align}
The number $\Delta_\rho(\nu)$ is called the \emph{spectral gap} of $\nu$.
Its importance lies in the following observation:
Suppose we draw repeatedly and independently from the measure $\nu$, such that the product of our samples performs a \emph{random walk} on the group $G$.
After $k$ steps, the distribution of the product is described by the $k$-fold convolution power $\nu^{*k}$ with moment operator $\MO_\rho(\nu^{*k}) = \MO_\rho(\nu)^k$.
Moreover, is is straightforward to check that 
\begin{equation}
 \MO_\rho(\nu)\MO_\rho(\mu) = \MO_\rho(\mu) = \MO_\rho(\mu) \MO_\rho(\nu),
\end{equation}
using the left and right invariance of the Haar measure.
Hence, we find
\begin{align}
 \snorm{ \MO_\rho(\nu)^k - \MO_\rho(\mu) } 
 &=  \snormb{ \left[\MO_\rho(\nu) - \MO_\rho(\mu) \right]^k } \\
 &\leq (1 - \Delta_\rho(\nu))^k, \label{eq:powers-of-moment-operator}
\end{align}
and hence the spectral distance decays exponentially with $k$, provided that $\Delta_\rho(\nu) > 0$.
Therefore, a random walk $\nu$ converges to an approximate $\rho$-design at a $\Delta_\rho(\nu)$-dependent rate.

Finally, let us assume that $\nu$ is in addition \emph{symmetric}.
This implies that the moment operator is self-adjoint:
\begin{equation}
 \MO_\rho(\nu)^\dagger = \int_G \rho(g^{-1}) \dd\nu(g) = \int_G \rho(g) \dd\nu(g^{-1}) = \MO_\rho(\nu).
\end{equation}
Hence, $\MO_\rho(\nu)$ has a spectral decomposition and real eigenvalues $\lambda\in[-1,1]$.
The operator $\MO_{\rho\ominus 1}(\nu)$ defined above is in this case given by the spectral decomposition for eigenvalues $\lambda < 1$.
In particular, the spectral gap of a symmetric measure is given as $\Delta_\rho(\nu) = 1 - \max\{ \lambda_2, |\lambda_\mathrm{min}| \}$, where $\lambda_2$ is the second-largest eigenvalue of $\MO_\rho(\nu)$ and $\lambda_\mathrm{min}$ is the smallest one.
Moreover, for symmetric $\nu$, we have equality in Eq.~\eqref{eq:powers-of-moment-operator} as the involved operators are self-adjoint. 

If $\nu$ is not symmetric, we can define the measure 
\begin{equation}\label{eq:inversted_measure}
  \tilde\nu(A)\coloneqq \nu(A^{-1})
\end{equation}
for any measurable set $A$. Then $\tilde\nu\ast\nu$ is symmetric and $\MO_\rho(\tilde\nu\ast\nu) = \MO_\rho(\tilde\nu)\cdot \MO_\rho(\nu) = \MO_\rho(\nu)^\dagger\cdot \MO_\rho(\nu)$.
In particular, we have 
\begin{align}
\MoveEqLeft[2] \snorm{\MO_\rho(\tilde\nu\ast\nu) - \MO_\rho(\mu)}\\
&=
\snorm{(\MO_\rho(\nu)^\dagger - \MO_\rho(\mu))(\MO_\rho(\nu) - \MO_\rho(\mu))}\\
&=
\snorm{\MO_\rho(\nu)- \MO_\rho(\mu)}^2,
\end{align}
where we have used $\MO_\rho(\nu)^\dagger\MO_\rho(\mu)=\MO_\rho(\mu)=\MO_\rho(\mu)\MO_\rho(\nu)$ by the left and right invariance of the Haar measure.
Hence, we find that $1- \Delta_\rho(\tilde\nu\ast\nu) = ( 1 - \Delta_\rho(\nu))^2$ for the respective spectral gaps which implies the following relation:
\begin{equation}
\label{eq:symmetrized_spectral_gap}
 \Delta_\rho(\tilde\nu\ast\nu) \geq \Delta_\rho(\nu) \geq \Delta_\rho(\tilde\nu\ast\nu) /2.
\end{equation}
This is a common trick used for non-symmetric probability measures, cf.~Ref.~\cite{varju_random_2015}.

The above discussion holds for any representation $\rho$.
In particular, it also applies to moment operators of the form $\widehat{\omega\nu}[\omega_\lambda]$ where $\omega_\lambda = \tau_\lambda^{\oplus n_\lambda}$ is a $\tau_\lambda$-isotypic representation.
However, the moment operators then have the form $\widehat{\omega\nu}[\omega_\lambda] \simeq  \widehat{\omega\nu}[\tau_\lambda]\otimes \id_{n_\lambda}$, see Eq.~\eqref{eq:fourier-transform-reducible-rep}.
Hence, the spectral gap is in this case
\begin{align}
\label{eq:spectral-gap-multiplicities}
 \snorm{ \widehat{\omega\nu}[\omega_\lambda] - \widehat{\omega}[\omega_\lambda] } 
 &= 
 \snorm{ \big(\widehat{\omega\nu}[\tau_\lambda] - \widehat{\omega}[\tau_\lambda]\big) \otimes \id_{n_\lambda} } \\
 &=
 \snorm{ \widehat{\omega\nu}[\tau_\lambda] - \widehat{\omega}[\tau_\lambda] }.
\end{align}

In general, the moment operator $\MO_\rho(\nu)$ can have negative eigenvalues.
If these are too negative, this can make the spectral gap very small, or even zero if $-1$ is an eigenvalue.
However, by adapting the measure $\nu$, it is possible to evade this problem.
To this end, note that the eigenvalue equation $v= -\rho(g)v$ is not fulfilled for $g=\id$.
Hence, if there is a neighborhood of $\id$ in $\supp\nu$, then $v= -\rho(g)v$ cannot hold $\nu$-almost everywhere and $-1$ is not an eigenvalue of $\MO_\rho(\nu)$.
More precisely, one can generalize Lem.~1 in Ref.~\cite{diaconis_comparison_1993} from finite to compact groups to obtain a lower bound on the smallest eigenvalue.

\begin{proposition}
\label{prop:spectral-gap}
Let $\nu\in \MA(G)$ be a symmetric probability measure.
Suppose there is an open neighborhood $\Omega$ of $\id\in G$ with $\nu(\Omega)<1$.
Then, the smallest eigenvalue of $\MO_\rho(\nu)$ obeys $\lambda_\mathrm{min} \geq -1 + 2\nu(\Omega)$.
\end{proposition}
\begin{proof}
Note that the statement is true for $\nu(\Omega)=0$.
Hence, let us assume that $1>\nu(\Omega)>0$.
Then define the symmetric measure $\nu_\Omega (A) \coloneqq \nu(A\cap \Omega)$, and
\begin{equation}
 \xi \coloneqq \frac{1}{1-\nu(\Omega)}\left( \nu - \nu_\Omega \right).
\end{equation}
Note that $\xi$ is by construction a symmetric probability measure and thus 
\begin{align}
 -1 \geq \lambda_\mathrm{min}(\MO_\rho(\xi)) 
 &= \frac{1}{1-\nu(\Omega)}\big( \lambda_\mathrm{min} - \lambda_\mathrm{max}(\MO_\rho(\nu_\Omega)) \big) \\
 &\geq \frac{1}{1-\nu(\Omega)}\big( \lambda_\mathrm{min} - \nu(\Omega) \big).
\end{align}
In the last step, we used that $\nu_\Omega$ is not a probability measure as $\nu_\Omega(G) = \nu(\Omega)$ and thus, the largest eigenvalue of its moment operator is at most $\nu(\Omega)$.
Rewriting the above inequality proves the claim.
\end{proof}

\subsection{The Heisenberg-Weyl and Clifford group}
\label{sec:clifford-group}

Consider the Hilbert space $\mathcal H = (\C^p)^{\otimes n}$ of $n$ qudits of local dimension $p$, where we assume that $p$ is prime.
We label the computational basis $\ket{x}=\bigotimes_{i=1}^n\ket{x_i}$ by vectors $x=(x_1,\dots,x_n)$ in the discrete vector space $\F_p^n$. 
Here, $\F_p$ is the finite field of $p$ elements which, for concreteness, can be chosen as the residue field $\Z/p\Z$ of integers modulo $p$.
This section roughly follows the presentation in Ref.~\cite{heinrich_2021}; we refer the reader to this reference for more details.

\paragraph{The Heisenberg-Weyl group.}
Let $\xi=e^{2\pi i/p}$ be a $p$-th root of unity.
We define the $n$-qudit $Z$ and $X$ operators by their action on the computational basis:
\begin{equation}
\label{eq:pauli-zx}
 Z(z)\ket{y} \coloneqq \xi^{z\cdot y}\ket{y}, \quad X(x)\ket{y}\coloneqq \ket{y+x}, \quad z,x,y\in\F_p^n.
\end{equation}
Here, all operations take place in the finite field $\F_p$ (i.e.~modulo $p$), if not stated otherwise.
Note that the operators $Z(z)$ and $X(x)$ are unitary and have order $p$ except if $z=0$ or $x=0$, respectively.
To unify the slightly different definitions for $p=2$ and $p > 2$ in the following, we define $\tau \coloneqq  (-1)^p e^{i\pi/p}$ as a suitable square root of $\xi$.
Note that for $p=2$, $\tau=i$ has order 4 while for $p > 2$, $\tau = -e^{i\pi/p}$ has order $p$.
Next, we group the $Z$ and $X$ operators and their coordinates to define the so-called \emph{Weyl operators}
indexed by $a=(a_z,a_x)\in\F_p^{2n}$, 
\begin{align}
 \label{eq:weyl-ops}
 w(a) &\coloneqq  \tau^{-a_z\cdot a_x} Z(a_z) X(a_x) \, .
\end{align}
Here, it is understood that the exponent is computed modulo 4 for $p=2$ and modulo $p$ in the case $p > 2$.
Note that the definition in the case $p=2$ exactly reproduces the $n$-qubit Pauli operators.
In the quantum information literature, the Weyl operators for $p>2$ are thus also sometimes called \emph{generalized Pauli operators}.
It is straightforward to check the following commutation relation
\begin{equation}
 \label{eq:weyl-commutation-relation}
 w(a)w(b) = \xi^{[a,b]} w(b)w(a), \quad [a,b] \coloneqq  a_z\cdot b_x - a_x\cdot b_z.
\end{equation}
The non-degenerate, alternating form $[\argdot,\argdot]$ is the standard \emph{symplectic form} on $\F_p^{2n}$.
Furthermore, the Weyl operators are unitary, have order $p$, are traceless except for $w(0) = \one$, and form an orthogonal operator basis of $\End(\C^p)^{\otimes n}$:
\begin{equation}
\label{eq:weyl-orthogonality}
\obraket{w(a)}{w(b)} = p^n \delta_{a,b}.
\end{equation}
Finally, the \emph{Heisenberg-Weyl group} is the group generated by Weyl operators and thus given by:
\begin{align}
 \HW{n}{p} 
  &= \{ \tau^k w(a) \; | \; k\in\Z_D, a\in\F_p^{2n} \}, 
\end{align}
where $D$ is 4 if $p=2$ and $p$ otherwise.

Note that the center of $\HW{n}{p}$ is $\Z_D$.
The inequivalent irreducible representations of $\HW{n}{p}$ are either labelled by the additive characters of the center $\Z_D$, or by additive characters of the vector space $\F_p^{2n}$, see e.g.~Ref.~\cite{neuhauser_explicit_2002}.
In this paper, we encounter the second type as the irreps of the conjugation representation $\omega(g) = U_g(\argdot)U_g^\dagger$ restricted to $\HW{n}{p}$.
From Eq.~\eqref{eq:weyl-commutation-relation}, it is evident that any element $\tau^k w(a) \in \HW{n}{p}$ acts as $\xi^{[a,b]}$ on the one-dimensional vector spaces spanned by the operators $w(b)$.
Since these are orthogonal and span $\End(\C^p)^{\otimes n}$, we have the decomposition into irreps
\begin{equation}
\label{eq:irrep-decomp-HW-group}
 \omega|_{\HW{n}{p}} \simeq \bigoplus_{b\in\F_p^{2n}} \xi^{[\argdot,b]}.
\end{equation}
The characters $\xi^{[\argdot,b]}$ are mutually orthogonal, thus these irreps are mutually inequivalent.

\paragraph{The Clifford group.} is defined as the group of unitary symmetries of the Heisenberg-Weyl group:
\begin{equation}
 \label{eq:def-clifford-group}
 \Cl{n}{p} \coloneqq  \big\{ U \in U(p^n) \; | \; U \HW{n}{p} U^\dagger = \HW{n}{p} \big\} \, / \, U(1).
\end{equation}
We take the quotient with respect to irrelevant global phases in order to render the Clifford group a finite group.\footnote{Strictly speaking, we are here considering the \emph{projective} Clifford group. While it is possible to define a finite, non-projective version by restricting the matrix entries to $\QQ[\chi]$ where $\chi$ is a suitable root of unity depending on $p$ \cite{nebe_invariants_2001,ZhuKueGra16}, these details are not needed for this paper.}
It is well-known that the Clifford group forms a unitary 2-design for all primes $p$ and even a unitary 3-design for $p=2$ \cite{DiVincenzo985948,DanCleEme09,Web15,Zhu15,HelWalWeh16,ZhuKueGra16}.
In fact, the Clifford group is the canonical example for a unitary design forming a group \cite{guralnick_decompositions_2005,BannaiEtAl:2020:tgroups}, see also Ref.~\cite[Sec.~V]{haferkamp_efficient_2023}.

The above definition of the Clifford group can be generalized to the case when the local dimension is a prime power, $q=p^k$, by using arithmetic in the finite field $\F_q$.
The so-obtained groups $\Cl{n}{p^k}$ are subgroups of $\Cl{nk}{p}$, and form unitary 2-designs, but not 3-designs for any $p$ and $k>1$ \cite{Zhu15}.

It is well-known that $\omega(g) = U_g(\argdot)U_g^\dagger$ decomposes as a representation of $\U(d)$ as $\omega = 1 \oplus \Ad$ where $1$ is the trivial irrep supported on the identity matrix $\one$, and $\Ad$ is the \emph{adjoint irrep} supported on the traceless matrices.
The unitary 2-design property implies that $\omega$ decomposes into the same irreps over the Clifford group than over the unitary group.

The general representation theory of the Clifford group is significantly more difficult than the one of the unitary group.
For tensor power representations, a duality theory has been developed in a recent series of papers \cite{montealegre-mora_rank-deficient_2021,montealegre-mora_duality_2022}.

%%% =============================================
\section{Results}
%%% =============================================
\label{sec:results}

To set the stage for our results we first briefly state the setting and noise model we are considering.
If not mentioned otherwise, these hold throughout the remainder of this work.

%%% ---------------------------------------------------
\subsection{Setting and noise model}
\label{sec:setting}
%%% ---------------------------------------------------
In the following, $G$ is a compact group (finite or infinite), $\mu\in\MA(G)$ is the normalized Haar measure on $G$, and $\nu\in \MA(G)$ is a probability measure with support on generators of $G$.

We fix a finite-dimensional unitary representation $\omega$, called the \emph{reference representation}, on the operator space $V\coloneqq\End(\mathcal H)$ where $\mathcal H \simeq \C^d$ is a suitable $d$-dimensional Hilbert space.
We want $\omega$ to represent possible unitary dynamics of the system $\mathcal H$, hence it has the form $\omega(g) = \eta(g)(\argdot)\eta(g)^\dagger$, where $\eta$ is a suitable unitary representation of $G$ on $\mathcal{H}$.
Typically, we have $G\subset\U(\mathcal H)$ and $\eta(g) = U_g$ is simply the defining representation of $\U(\mathcal H)$ restricted to $G$.
The representation $\omega$ has an isotypic decomposition $\omega = \bigoplus_{\lambda} \omega_\lambda$ where $V=\bigoplus_\lambda V(\lambda)$ and the subrepresentations are of the form $\omega_\lambda \simeq \tau_\lambda^{\oplus n_\lambda}$ for irreps $\tau_\lambda$ with multiplicities $n_\lambda$. 
The dimension of $\tau_\lambda$ is $d_\lambda \coloneqq \dim(\tau_\lambda)$.
Note that by construction, the trivial irrep is contained in $\omega$ at least once.

Importantly, $\omega$ is Hermiticity-preserving and thus a real representation w.r.t.~the real structure $\End(\mathcal H) = \Herm(\mathcal H) \oplus i \Herm(\mathcal H)$.
Eventually, we are only interested in the action of $\omega$ and its irreps on the real subspace $\Herm(\mathcal H)$.
{
To avoid further technicalities in the statement of our results, we assume that all irreps of $\omega$ are of \emph{real type} such that they are simply the complexification of the real irreps of $\omega_\R$ on $\Herm(\mathcal H)$, see also Sec.~\ref{sec:real-representations}.
}
In particular, the complex representation spaces $V(\lambda)$ can be written as $V(\lambda) = H(\lambda) \oplus i H(\lambda)$ where $H(\lambda) \subset \Herm(\mathcal H)$ carries the real part of the representation $\omega_\lambda$.
{
However, our results readily generalize to the case when irreps of $\omega$ are of arbitrary type, and we comment on this briefly in Sec.~\ref{sec:data-guarantees}.
}

We assume the existence of an \emph{implementation function} $\phi:\, G \rightarrow \End(V)$ which should be understood as a noisy implementation of $\omega$.
This assumes in particular that the noisy implementation does not depend on the history of the experiment, nor does it change with time.
Hence, we assume \emph{Markovian, time-stationary} (but otherwise arbitrary) noise.
We assume that $\phi$ is integrable w.r.t.~the measure $\nu$ and that $\phi(g)$ is a completely positive, trace non-increasing superoperator for any $g\in G$.
Finally, we fix an \emph{initial state} $\rho\in V$ and a measurement basis $\ket{i}\in\mathcal H$ and write $E_i\coloneqq \ketbra{i}{i}\in V$.
We define $M = \sum_{i=1}^d \oketbra{E_i}{E_i}$ to be the completely dephasing channel in this basis.
We denote by $\tilde\rho = \ESP(\rho)$, $\tilde E_i = \EM(E_i)$, and $\tilde M = M\EM$ the noisy versions of the initial state and measurement.
Here, we assume that $\EM^\dagger$ and $\ESP$ are trace non-increasing quantum channels.

\subsection{The effective measurement frame}
\label{sec:frame-operator}

Besides the projection to a specific irrep, the only non-trivial ingredient to the filter function \eqref{eq:filter-function} is the pseudo-inverse of the superoperator $S$. 
Therefore, it is instructive to first analyze the structure of $S$ in the following. 

In filtered randomized benchmarking, one effectively approximates measurements of the \ac{POVM} given by $(i,g)\mapsto \omega(g)^\dagger\oket{E_i}\dd\mu(g)$. 
In Eq.~\eqref{eq:measurement-frame-operator} we defined an associated superoperator
\begin{align}
 S 
 &\coloneqq \sum_{i\in[d]} \int_G \omega(g)^\dagger \oketbra{E_i}{E_i} \omega(g) \dd\mu(g) \\
 &= \int_G \omega(g)^\dagger M\omega(g) \dd\mu(g) \, .
 \label{eq:measurement-frame-operator-1}
\end{align}
Recall that $M$ is defined as $M = \sum_{i\in[d]}\oketbra{E_i}{E_i}$.
From the definition, it is evident that $S$ is a quantum channel.
Moreover, it is easy to see that $S$ is a positive semidefinite operator.

In frame theory, the superoperator $S$ is called the \emph{frame operator} associated with the set of operators 
$\{\omega(g)^\dagger\oket{E_i}\}_{i\in [d], g\in G}$ \cite{waldron_introduction_2018}.
Strictly speaking, these might fail to form a proper \emph{frame} since their span might not be all of $\End(\mathcal H)$ (i.e.~the \ac{POVM} might not be \emph{informationally complete}).
Since the range of $S$ is the span of $\omega(g)^\dagger\oket{E_i}$, $S$ would then not be of full rank.
Despite the possible lack of invertibility, we still call $S$ a \emph{frame operator}.
A full rank is guaranteed if $G$ and its representation $\omega$ fulfill certain properties \cite[Ch.~10]{waldron_introduction_2018}.
For instance, this is the case if $G$ forms a unitary 2-design.
However, one can readily check that this is not always the case: Take $d=2^n$,  measurements in the computational basis, and $G$ the $n$-qubit Pauli group.
Then, $S$ is not of full rank, as shown explicitly below.

\paragraph{Decomposition of the frame operator.}
Recall that $\omega$ acts on the vector space of linear operators $V=\End(\mathcal H)$.
Let us consider an irreducible subrepresentation $\tau_\lambda$ of $\omega$ and denote the $\tau_\lambda$-isotypic subrepresentation of $\omega$ by $\omega_\lambda\simeq \tau_\lambda^{\oplus n_\lambda}$, where $n_\lambda$ is its multiplicity. 
Let $P_\lambda$ be the orthogonal projection onto the associated isotypic component $V(\lambda)\subset V$, c.f.~Eq.~\eqref{eq:projection-isotypes}.
We can write $P_\lambda$ using a suitable isometric embedding $X_\lambda: V(\lambda) \rightarrow V$ as 
\begin{equation}
  P_\lambda = X_\lambda X_\lambda^\dagger \, .
\end{equation}
In representation-theoretic terms, the frame operator \eqref{eq:measurement-frame-operator-1} is exactly the projection of the quantum channel $M$ onto the commutant of $\omega$.
In particular, Schur's lemma implies that
\begin{equation}
\label{eq:frame-op-block-decomposition}
 S = \bigoplus_{\lambda \in \Irr\omega} S_\lambda \simeq \bigoplus_{\lambda \in \Irr\omega} \id_\lambda \otimes s_\lambda,
\end{equation}
where $S_\lambda = X_\lambda^\dagger S X_\lambda$ and $s_\lambda\in\C^{n_\lambda\times n_\lambda}$ is a positive semidefinite matrix acting on the multiplicity space of $\tau_\lambda$ in $\omega$.
Alternatively, we can deduce this from the Fourier theory introduced in Sec.~\ref{sec:fourier-transform} by noting that $S = \widehat\omega[\omega](M)$.
Then, Eq.~\eqref{eq:fourier-transform-reducible-rep} implies $S = \bigoplus_{\lambda \in \Irr\omega} \widehat\omega[\omega_\lambda](M)$ and Prop.~\ref{prop:fourier-transform} gives the form of $\widehat\omega[\omega_\lambda](M)$.
In particular, we have $S_\lambda = \widehat\omega[\omega_\lambda](M)X_\lambda$.

In the case that $\tau_\lambda$ is multiplicity-free, we find $S_\lambda = s_\lambda\id_\lambda$ for a scalar $s_\lambda = \tr(P_\lambda M)/d_\lambda \geq 0$. 
In the language of frame theory, $\{\tau_\lambda(g)\ad X_\lambda\ad \oket{E_i}\}_{i \in [d], g \in G}$ then constitute a tight-frame for $V_\lambda$ if $s_\lambda \neq 0$.
In the case of multiplicities, $n_\lambda > 1$, we have the relation $\tr(s_\lambda) = \tr(P_\lambda M)/d_\lambda$ instead.
For some choice of irrep decomposition of $V(\lambda)$, let $P_\lambda^{(i)}=X_\lambda^{(i)}X_\lambda^{(i)\dagger}$ be the projection onto the $i$-th copy of the irrep $\tau_\lambda$; then we can write the matrix $s_\lambda$ explicitly as
\begin{equation}
\label{eq:frame-op-constant}
  s_\lambda = d_\lambda^{-1} \sum_{i,j=1}^{n_\lambda}  \tr\big(X_\lambda^{(i)\dagger} M X_\lambda^{(j)}\big) \ketbra{i}{j}.
\end{equation}
Note that since $s_\lambda$ is positive semi-definite, it can always be unitarily diagonalized.
The corresponding diagonalizing transformation can be understood as a change of the irrep decomposition of $V(\lambda)$ such that $\tr\big(X_\lambda^{(i)\dagger} M X_\lambda^{(j)}\big) = 0$ for all $i\neq j$.
In particular, $\{\tau_\lambda(g)^{(j)\dagger} X_\lambda^{(j)\dagger} \oket{E_i}\}_{i \in [d], g \in G}$ is a tight-frame for $V_\lambda^{(j)}$ if $s_\lambda^{(j)} = \tr(P_\lambda^{(j)} M) \neq 0$

\paragraph{Examples} For concreteness, let us discuss a few important examples.
We take $G$ to be a subgroup of the unitary group $\U(d)$, and the representation $\omega(g)=U_g (\argdot) U_g^\dagger$.
First, for $G=\U(d)$, we have the multiplicity-free trivial and adjoint irrep which we label by $1$ and $\Ad$, respectively.
The associated blocks of the frame operator are proportional to the scalars
\begin{align}
  s_1 
  &= 
   \tr\left(P_1 M\right)
  =
  \sum_{j\in[d]} d^{-1} \obraket{E_j}{\one}\obraket{\one}{E_j} = 1, \\
  s_\Ad &= \frac{1}{d^2-1} \tr\left(P_\Ad M\right) \\
  &= \frac{1}{d^2-1} \left( \sum_{j\in[d]} \obraket{E_j}{E_j} -  s_1 \right) = \frac{1}{d+1}, 
\end{align}
where we have used that $P_1 = \frac 1d \oketbra\1\1$ is the projector onto the trivial irrep, $P_\Ad = \id - P_1$ and $d_\lambda=\tr[P_\lambda]$. 
We can use this information to write the frame operator in a more recognizable form, namely as a convex combination of the identity and the \emph{completely depolarizing channel} $\mathcal{D}=P_1$:
\begin{align}
 S &= 
 P_1 + \frac{1}{d+1} P_\Ad = \left(1 - \frac{1}{d+1}\right) \frac{1}{d}\,\oketbra\1\1 + \frac{1}{d+1}\, \id \\
 &= 
 \frac{1}{d+1}\left( d \,\mathcal{D} + \id \right)\eqqcolon \mathcal{D}_{\frac{1}{d+1}}. \label{eq:frame-operator-as-dep-channel}
\end{align}
In other words, $1/(d+1)$ is the \emph{effective depolarizing parameter} of the dephasing channel $M$.
The same result holds for the multi-qudit Clifford group $\Cl{n}{p}$ with $d=p^n$ where $p$ is prime, and more generally, for any unitary 2-group since it has the same irreps as $\U(d)$.

Next, let us consider the \emph{local Clifford group} $G=\Cl{1}{p}^{\otimes n}$ (for $p$ prime).
Its irreps are multiplicity-free and given as all possible tensor products of the single-qudit trivial and adjoint irreps, i.e.\ there are $2^n$ many.
Let us label such an irrep by a binary string $b\in \{0,1\}^n$ where the `0' corresponds to an adjoint irrep on the respective system, and `1' to the trivial irrep. 
It is straightforward to see that $s_b$ is then 
% just 
the product 
\begin{equation}
    s_b = s_1^{|b|} s_\Ad^{n-|b|} = (p+1)^{-(n-|b|)}
\end{equation}
where $|b|$ denotes the Hamming weight of $b$.
Indeed, we have
\begin{align}
 s_b 
 &= \left(\frac{1}{p^2-1}\right)^{n-|b|} \sum_{y\in\F_p^n} \prod_{i=1}^n \osandwich{E_{y_i}}{P_{b_i}}{E_{y_i}} \\
 &= 
 p^n \left(\frac{1}{p^2-1}\right)^{n-|b|}   \left(\frac{1}{p}\right)^{|b|} \left( 1 - \frac{1}{p} \right)^{n-|b|} \\
 &= \left( \frac{1}{p+1} \right)^{n-|b|}.
\end{align}
The same argument holds more generally for a locally acting group $G=G_1\otimes\dots\otimes G_k$ when every factor $G_i$ is a unitary 2-design, possibly of different local dimension.

As a final example, consider the \emph{Heisenberg-Weyl} (or generalized Pauli) group $G=\HW{n}{p}$.
Its irreps are one-dimensional and given as the span of the Weyl operators $w(z,x)$ for $(z,x)\in\F_p^{2n}$, see Eq.~\eqref{eq:weyl-ops} and \eqref{eq:irrep-decomp-HW-group}.
Let us label the elements of the measurement basis by $\ket y$ with $y\in\F_p^n$. 
We find that
\begin{align}
 s_{x,z} 
 &= \sum_{y\in\F_p^n} d^{-1} |\sandwich{y}{w(z,x)}{y}|^2 \\
 &= \sum_{y\in\F_p^n} d^{-1} |\braket{y}{y+x}|^2 = \delta_{x,0}.
\end{align}
Hence, the frame operator vanishes on every irrep with $x\neq 0$.

\paragraph{Eigenvalues of the frame operator.} \label{par:eigenvalues-frame-op}
For our analysis of the filtered \ac{RB} protocol for arbitrary compact groups $G$ in Secs.~\ref{sec:data-guarantees} and \ref{sec:sampling-complexity}, the properties of the associated frame operator $S$ and, in particular, $\snorm{S_\lambda^+}=\snorm{s_\lambda^+}$ will be of some importance. 
This is exactly the inverse of the smallest non-zero eigenvalue of $s_\lambda$, as $s_\lambda$ is positive semidefinite. 
By Eq.~\eqref{eq:frame-op-constant}, the eigenvalues of $s_\lambda$ are of the form $\tr(P_\lambda^{(i)} M)/d_\lambda$ where $P_\lambda^{(i)}$ are the projections associated to an irrep decomposition of $V(\lambda)$ which diagonalizes $s_\lambda$.
Using Hölder's inequality, we can bound these eigenvalues as follows:
\begin{equation}
 \frac{\tr(P_\lambda^{(i)} M)}{d_\lambda} 
 \leq 
 \frac{\min\{d_\lambda,d-1\}}{d_\lambda} 
 = 
 \min\left\{ 1,\frac{d-1}{d_\lambda} \right\} .
\end{equation}
Here, we have used the bound $d-1$ since, if $\lambda$ is not the trivial irrep, $P_\lambda$ is supported on the traceless subspace, and thus we can consider the corresponding restriction of $M$ of rank $d-1$.
Hence, if $\snorm{S_\lambda^+}\neq 0$, we have the lower bound
\begin{align}
 \snorm{S_\lambda^+}
 \geq
 \max\left\{ 1 , \frac{d_\lambda}{d-1}\right\} .
 \label{eq:lower-bound-Slambda}
\end{align}
This lower bound is tight, e.g.~for unitary 2-designs.

In the following discussion, we often restrict to the case where $[P_\lambda, M] = 0$. 
We refer to this condition as the measurement \emph{$M$ being aligned with the irrep $\lambda$}.
Moreover, many calculations are simplified by the assumption that $\lambda$ is multiplicity-free in $\omega$.
Both restrictions are met by the Clifford group, the local Clifford group and the Heisenberg-Weyl group together with the basis measurement in the associated $Z$-basis.

For any irrep $\lambda$ aligned with $M$ we can also give an upper bound on the spectral norm $\snorm{S_\lambda^+}$.
Let $P_\lambda^{(i)}$ be as above.
Then, we have $[P_\lambda,M]=0$ if and only if $[P_\lambda^{(i)},M]=0$ for all $i=1,\dots,n_\lambda$.
Assuming that $s_\lambda \neq 0$, the smallest non-zero eigenvalue is thus given as $\tr(P_\lambda^{(i)} M)/d_\lambda=\rank(P_\lambda^{(i)} M)/d_\lambda\geq 1/d_\lambda$ for some $i$.
Thus, we obtain the bound: 
\begin{align}
\snorm{S_\lambda^+} 
 \leq
 d_\lambda 
 \qquad
 \text{(for $\lambda$ aligned with $M$).}
 \label{eq:upper-bound-Slambda}
\end{align}
This bound is tight, for instance for the Heisenberg-Weyl group.
Note however, that the bound is overestimating $\snorm{S_\lambda^+}$ for large irreps, and it is possible to give tighter bounds under additional assumptions.

\subsection{Signal guarantees for filtered randomized benchmarking}
\label{sec:data-guarantees}

We now derive a general guarantee for the expected signal $\FD_\lambda(m)$ of filtered \ac{RB} that is then discussed and analyzed in detail in the remainder of this work. 

With the notation introduced in the preliminaries, we can, as a first step, very compactly write the expected signal for a general implementation function $\phi$.
Recall that to isolate the decay parameter associated with a given irreducible representation $\tau_\lambda$ we have introduced the filter function $f_\lambda$ in Eq.~\eqref{eq:filter-function}.
In the notation of Sec.~\ref{sec:operators-norms}, it reads:
\begin{equation}
    f_\lambda(i,g_1,\dots,g_m) = \osandwich{E_i}{\omega(g_1\cdots g_m)S^+P_\lambda}{\rho}.
\end{equation}
Note that $f_\lambda$ is real-valued, as all involved superoperators are in fact Hermiticity-preserving:
For the projector $P_\lambda$, this follows from the definition \eqref{eq:projection-isotypes} and the fact that real representations have real characters.
Moreover, $S$ is a quantum channel, hence $S^+$ is Hermiticity-preserving as argued in Sec.~\ref{sec:operators-norms}.
Thus, we can ommit the complex conjugation in the computation of the estimator $\hat{\FD}_\lambda(m)$.

Let us define the noisy quantum channel $\tilde M \coloneqq \sum_{i \in [d]} \oketbra{E_i}{\tilde E_i}$.
Then, by linearity of the expected value and the assumption of i.i.d.\ experiments in the data acquisition, the filtered \ac{RB} signal becomes
\begin{align}
  \MoveEqLeft[1] \FD_\lambda(m) \\
  &= \sum_{i \in [d]} \int_{G^{m}} f_\lambda(i, g_1,\dots,g_m) \, \\
  &\qquad \times p(i| g_1, \ldots, g_m) \dd\nu(g_1)\cdots\dd\nu(g_m) \\
  &= \sum_{i \in [d]} \osandwich{\rho}%
  {P_\lambda S^+ 
  \int \omega(g_1)^\dagger\cdots\omega(g_m)^\dagger}{E_i} \\
  &\qquad\times\osandwich{\tilde E_i}{\dd(\phi\nu)(g_m)\dots\dd(\phi\nu)(g_1)}{\tilde\rho} \\
  &= \osandwichb{\rho}%
  {P_\lambda S^+ 
   \left( \int \omega(g)^\dagger (\argdot) \dd(\phi\nu)(g)\right)^m \left( \tilde M \right) }%
  {\tilde\rho} 
  \\
  &= 
  \osandwichb{\rho}%
  {X_\lambda S_\lambda^+\,
  \widehat{\phi\nu}[\omega_{\lambda}]^m \big( X_\lambda^\dagger \tilde M \big) }%
  {\tilde\rho}\,. \label{eq:Flambda-as-Fourier-transform}
\end{align}
We use in the last line that $S^+$ and $\omega(g)$ are block-diagonal, hence $X_\lambda^\dagger S^+ = S_\lambda^+ X_\lambda^\dagger$ and $X_\lambda^\dagger \omega(g) =  \omega_{\lambda}(g)X_\lambda^\dagger$.

We observe that filtered randomized benchmarking explicitly exposes the $m$th power of the Fourier transform of the implementation map $\phi$ and the measure $\nu$ restricted to the irreducible representation $\tau_\lambda$.
This observation already justifies the name \emph{filtered} \ac{RB}.
It is interesting to note that the general structure of a filtered \ac{RB} signal is even slightly simpler than the signal form of \ac{RB} protocols that include an inversion, c.f.~\cite[Eq.~(75)]{helsen_general_2022}:
The latter additionally involves an inverse Fourier transformation.

Without further assumption on in the implementation map $\phi$, 
the form of Eq.~\eqref{eq:Flambda-as-Fourier-transform} is far from describing a simple functional dependence. 
For sufficiently large $m$, we can, however, hope that the signal is mostly governed by the dominant eigenvalue of $\widehat{\phi\nu}[\omega_{\lambda}]$, giving rise to a simpler signal model that can be fitted.
This approximation constitutes the core of \ac{RB}. 
And the following theorem formulates a precise statement aligned with this expectation.

To show that the filtered \ac{RB} signal \eqref{eq:Flambda-as-Fourier-transform} indeed follows a `simple' exponential decay, we use that $\widehat{\phi\nu}[\omega_{\lambda}]$ is close to $\widehat{\omega\nu}[\omega_{\lambda}]$ if the implementation is sufficiently good.
Recall from Sec.~\ref{sec:rho-designs} that $\widehat{\omega\nu}[\omega_{\lambda}]$ has the interpretation as a moment operator for the probability measure $\nu$ and, thus, {we have the following spectral decomposition (cf.~Eq.~\eqref{eq:moment-op-spectral-decomposition}):
\begin{equation}
\label{eq:moment-op-FT-spectral-decomposition}
 \widehat{\omega\nu}[\tau_\lambda]
 \simeq
 \begin{bmatrix}
  \id & 0 \\
  0 & \Lambda_\lambda
 \end{bmatrix}\,.
\end{equation}
The subdominant eigenvalues are controlled by the spectral gap $\Delta_\lambda$ of the moment operator, i.e.~$\snorm{\Lambda_\lambda} \leq \Delta_\lambda$.
}
If the deviation of  $\widehat{\phi\nu}[\omega_{\lambda}]$ from $\widehat{\omega\nu}[\omega_{\lambda}]$ is small compared to $\Delta_\lambda$, the perturbation theory of invariant subspaces \cite{stewart_matrix_1990} ensures that $\widehat{\phi\nu}[\omega_{\lambda}]$ 
{can be block-diagonalized similar to the moment operator in Eq.~\eqref{eq:moment-op-FT-spectral-decomposition}.}
This gives rise to a simple explicit formula for the \ac{RB} signal.
Formally, we arrive at the following result:

\begin{theorem}[Signal guarantee for filtered \ac{RB} with random circuits]
\label{thm:rb-data-random-circuits}
 Fix a non-trivial irrep of real type $\lambda$ in $\omega$ and let $\omega_\lambda\subset\omega$ be its isotypic subrepresentation with multiplicity $n_\lambda$.
 Suppose the spectral gap of $\widehat{\omega\nu}[\omega_{\lambda}]$ is larger than $\Delta_\lambda>0$ and there is  $\delta_\lambda >0$ such that
 \begin{equation}\tag{$\mathbb A$}
 \label{eq:thm:initialbelief}
  \snormb{\widehat{\phi\nu}[\omega_{\lambda}] - \widehat{\omega\nu}[\omega_{\lambda}]} \leq \delta_\lambda < \frac{\Delta_\lambda}{4} \, .
 \end{equation}
 Then, the $\lambda$-filtered \ac{RB} signal given in Eq.~\eqref{eq:Flambda-as-Fourier-transform} obeys
 \begin{equation}\tag{$\mathbb S$}\label{eq:thm:signal-form}
   \FD_\lambda(m)  =  \tr\left(A_\lambda I_\lambda^m\right) + \tr\left( B_\lambda O_\lambda^m\right) ,
 \end{equation}
 where $A_\lambda, I_\lambda\in \R^{n_\lambda \times n_\lambda}$, and $B_\lambda, O_\lambda\in \R^{k_\lambda\times k_\lambda}$ for $k_\lambda\coloneqq d_\lambda d^2 - n_\lambda$.
Moreover, $I_\lambda$ and $O_\lambda$ do not depend on the initial state and measurement and we have
 \begin{align}
 \label{eq:bounds-on-perturbed-blocks}
  \spec_{\C}(I_\lambda) &\subset {D_1(0) \cap D_{2\delta}(1)} \, , &
  \snorm{O_\lambda} &\leq 1 - \Delta_\lambda + 2\delta_\lambda \, \,
 \end{align}
 {where $D_\varepsilon(z)$ is the disc of radius $\varepsilon$ centered at $z\in\C$.}
 The magnitude of the second matrix exponential decay is suppressed as
\begin{equation}\tag{$\mathbb B$}
\abs[\big]{ \tr\left( B_\lambda O_\lambda^m\right) }
 \leq
 c_\lambda\, \sqrt{\osandwich{\rho}{P_\lambda}{\rho}}\, g(\delta_\lambda/\Delta_\lambda) (1-\Delta_\lambda + 2\delta_\lambda)^m.
 \label{eq:thm:subdominant-decay-bound}
\end{equation}
Here, $c_\lambda$  is an irrep-dependent constant, defined as
\begin{align}
 c_\lambda
 &\coloneqq
 \begin{cases} 
  \sqrt{\tr S_\lambda}\, \snorm{S_\lambda^+} & \text{ if } [P_\lambda,M]=0, \\
  \min\{\sqrt{d_\lambda},\sqrt{d}\} \snorm{S_\lambda^+} & \text{ else}.
 \end{cases}
\end{align}
The function given by $g(x) \coloneqq (1-4x)x + \frac{1+x}{1-4x}$ is monotonically increasing and diverges for $x\rightarrow 1/4$.
\end{theorem}

As we see later in the proof of Thm.~\ref{thm:rb-data-random-circuits}, for Haar-random sampling $\nu=\mu$, the matrix $I_\lambda$ is in fact exactly the matrix $M_\lambda$ derived by \textcite{helsen_general_2022}, there appearing in the signal form of uniform standard \ac{RB} with general compact groups $G$.
In other words, \emph{filtered \ac{RB} and standard \ac{RB} yield the same decay rates}.

Note that an important special case of Theorem \ref{thm:rb-data-random-circuits} 
is given by $\Delta_\lambda=1$, 
i.e.\ when filtered \ac{RB} is performed with an \emph{exact $\omega_\lambda\otimes\omega$-design}.
The theorem only applies to non-trivial irreps of $\omega$.
The case of the trivial irrep can be derived analogously, but differs in the specific bounds.
We provide the statement for the trivial irrep in Sec.~\ref{sec:filtering-trivial-irrep} at the end of this section.
As mentioned in Sec.~\ref{sec:setting}, we assume, for the sake of presentation, that the irrep is of \emph{real type}.
However, our statement still holds with minor changes for irreps of complex or quaternionic type, as briefly outlined in Sec.~\ref{sec:complex-quaternionic-type}.
The differences to the real type is an effectively increased multiplicity 
for irreps of complex or quaternionic type.

Before presenting the {proof of Thm.~\ref{thm:rb-data-random-circuits} in Sec.~\ref{sec:proof-data-form}, it is instructive to take a closer look at their structure.
Theorem \ref{thm:rb-data-random-circuits}} consists of three central parts: First, the precise \emph{assumption} \eqref{eq:thm:initialbelief} on the quality of the implementation.
The assumption quantifies an initial belief that is sufficient to ensure the correct functioning of the \ac{RB} protocol.
Second, a statement for the \emph{expected signal form} \eqref{eq:thm:signal-form}. 
Here, the first summand is the dominant contribution that constitutes the model for fitting the signal. 
In particular, after performing the fit, the \ac{RB} protocol outputs the eigenvalues of $I_\lambda$ as the \ac{RB} decay parameters.
The second summand in Eq.~\eqref{eq:thm:signal-form} describes a sub-dominant contribution to the expected \ac{RB} signal that we want to be small in order to fit the dominant exponential decay.
To this end, the third part of the theorem, provides a \emph{bound on the sub-dominant contribution} \eqref{eq:thm:subdominant-decay-bound}.
Since by assumption $1- \Delta_\lambda + 2\delta_\lambda < 1$, the bound decays exponentially in $m$ and becomes small for sufficiently large $m$.
We derive more explicit sufficient conditions on the sequence length as corollaries of Theorem~\ref{thm:rb-data-random-circuits} in Section~\ref{sec:sequence-lengths}.

We will now discuss the parts in more detail.

\subsubsection{Assumptions on the quality of the implementation function}
\label{sec:implementation-quality}

The central assumption of Theorem \ref{thm:rb-data-random-circuits} reads
\noeqref{eq:thm:initialbelief_restatement}
\begin{equation}
\label{eq:thm:initialbelief_restatement}
\tag{\ref{eq:thm:initialbelief}}
\snorm{\widehat{\phi\nu}[\omega_{\lambda}] - \widehat{\omega\nu}[\omega_{\lambda}]} \leq \delta_\lambda < \frac{\Delta_\lambda}{4} \, .
\end{equation}
Although our analysis works in the full perturbative regime $\delta_\lambda < \Delta_\lambda/4$, there are certain quantities like $g(\delta_\lambda/\Delta_\lambda)$ that diverge for $\delta_\lambda/\Delta_\lambda\rightarrow 1/4$.
In practice, this means that the implementation error should be bounded away from $1/4$, for instance $\delta_\lambda/\Delta_\lambda\leq 1/5$ such that $g(1/5)\approx 6$ can be considered constant.

Intuitively, the assumption $\snorm{\widehat{\phi\nu}[\omega_{\lambda}] - \widehat{\omega\nu}[\omega_{\lambda}]} < \Delta_\lambda/4$ can be phrased as the assumption that the implementation function $\phi$ is sufficiently close to the reference representation $\omega$ \emph{on average w.r.t.~the measure $\nu$}.
Since we do typically not know the implementation function $\phi$, there is a priori no way of determining this error measure and verifying assumption \eqref{eq:thm:initialbelief_restatement}. 
Hence, the assumption that the implementation error $\snorm{ \widehat{\phi\nu}[\omega_{\lambda}] - \widehat{\omega\nu}[\omega_{\lambda}]}$ is small should be seen as an initial belief on the quality of the experiment.
An advantage of the approach taken here is that this initial belief only ever involves the quality of gates in $\supp(\nu)$.
Typically, these are gates native to the platform. 
Experimentally motivated noise models might then be used to approximate the implementation error $\snorm{ \widehat{\phi\nu}[\omega_{\lambda}] - \widehat{\omega\nu}[\omega_{\lambda}]}$, or trust can be build in independent experiments.

Since this implementation error does not bear an obvious operational meaning, we attempt to relate this quantity to more familiar ones in the following.
As a starting point, we may use the bound
 \begin{align}
  \snorm{ \widehat{\phi\nu}[\omega_{\lambda}] - \widehat{\omega\nu}[\omega_{\lambda}]}
  &\leq \int_G \snorm{ \bar\omega_\lambda(g)\otimes ( \phi(g) - \omega(g) )  } \\
  &= \int_G \snorm{\phi(g) - \omega(g)} \dd\nu(g), \label{eq:implementation-error-bound}
 \end{align}
where we first used Eq.~\eqref{eq:def-fourier-transform-dagger-cc} and the triangle inequality, and then that unitaries have unit spectral norm.
Eq.~\eqref{eq:implementation-error-bound} is probably a crude bound as averaging inside the norm might significantly reduce the error.
Moreover, we discard the irrep-specific component.
Nevertheless, the RHS of Eq.~\eqref{eq:def-fourier-transform-dagger-cc} has a clear meaning as the average error of the gates that are primitives in the experiment,
although the spectral distance of quantum channels lacks an operational interpretation.

Furthermore, recall that the \emph{average gate fidelity} between a quantum channel $\mathcal{E}$ and a unitary gate $U$ is defined as
\begin{equation}
 F_\mathrm{avg}(U,\mathcal E) \coloneqq \int \sandwich{\psi}{U^\dagger \mathcal{E}(\ketbra{\psi}{\psi}) U}{\psi} \dd\psi.
\end{equation}
Hence, the average gate \emph{in}fidelity between $\phi(g)$ and $\omega(g)$ is 
\begin{align}
\MoveEqLeft[2] 1 - F_\mathrm{avg}(\omega(g),\phi(g)) \\
&= \int \sandwich{\psi}{\left(\id - \omega(g)^\dagger\phi(g)\right)(\ketbra{\psi}{\psi})}{\psi} \dd\psi \\
&\leq \snorm{\id - \omega(g)^\dagger\phi(g)} 
= \snorm{\omega(g)-\phi(g)}.
\end{align}
In the last step, we used that the spectral norm is unitarily invariant so we can multiply with the unitary superoperator $\omega(g)$ from the left.
Hence, if we assume that the $\nu$-average of $\snorm{\phi(g) - \omega(g)}$ is small, this implies that the $\nu$-averaged infidelity between $\phi$ and $\omega$ is small, too.
On a superficial level, this is exactly the quantity which randomized benchmarking claims to measure.%
\footnote{The difficulties in this interpretation of \ac{RB} coming from the intrinsic gauge freedom of the protocol have been intensively discussed in the literature \cite{proctor2017WhatRandomizedBenchmarking,wallman2018randomized,Merkel18,helsen_general_2022}. We do not intend to contribute to this discussion in this work.}
Thus, one can understand the assumption that the implementation is reasonable good in spectral distance as a consistency condition for \ac{RB}.
That is, if the implementation is sufficiently good, we are in a regime where \ac{RB} can estimate how good it precisely is.

Finally, let us comment on the condition $\snorm{ \widehat{\phi\nu}[\omega_{\lambda}] - \widehat{\omega\nu}[\omega_{\lambda}]} \leq \Delta_\lambda/4$.
If $\nu$ is a poor approximation to a design, then the spectral gap is small and hence the implementation has to be rather good.
In contrast, if $\nu$ is an exact design, then $\Delta_\lambda = 1$, and more noise can be tolerated.
Hence, there is a trade-off between the quality of gates and the quality of random circuits.
In particular, if the gate errors become too large, the scrambling is no longer controlled by the random circuit.
More concretely, local random circuits have spectral gaps that scale as $O(1/n)$ (cf.~Sec.~\ref{sec:application-to-random-circuits}), thus \emph{the implementation of gates has to improve with the number of qubits}, at least for a general perturbative argument to hold.
A constant spectral gap $O(1)$ can be achieved by ``parallelizing'' the gates in the form of e.g.~a brickwork circuit. 
However, in this case, $\omega(g)$ describes an entire layers of parallel gates and we expect that the the implementation error generally scales with $n$.
More precisely, let us assume that the noise is local and we can write $\phi(g) = \omega(g) \otimes_{i=1}^n \mathcal{N}_i(g)$ where $\mathcal{N}_i(g)$ are single-qubit Pauli noise channels.
Using the telescopic identity $\otimes_{i=1}^n \mathcal{N}_i(g) - \id = \sum_{k=1}^n \otimes_{i=1}^{k-1} \mathcal{N}_i(g) \otimes (\mathcal{N}_k(g) - \id) \otimes \id^{\otimes n-k}$, we then find
\begin{align}
 \snorm{\phi(g) - \omega(g)}
 &=
 \snorm{ \omega(g) \left( \otimes_{i=1}^n \mathcal{N}_i(g) - \id \right)} \\
 &\leq
 \sum_{k=1}^n \snorm{\mathcal{N}_k(g) - \id} \,
\end{align}
where we used the triangle inequality and the fact that Pauli channels and the identity have unit spectral norm.
If the local noise is bounded as $\snorm{\mathcal{N}_k(g) - \id} \leq \varepsilon$ we thus obtain, using Eq.~\eqref{eq:implementation-error-bound}:
\begin{equation}
 \snorm{ \widehat{\phi\nu}[\omega_{\lambda}] - \widehat{\omega\nu}[\omega_{\lambda}]}
 \leq \varepsilon n \,.
\end{equation}
Hence, we can ensure that the implementation error is smaller than a constant spectral gap, if the single-qubit error rate $\varepsilon$ scales as $1/n$.
Similar findings have been reported by \textcite{Liu21BenchmarkingNear-term} and \textcite{Dalzell21RandomQuantumCircuits} who require gate error rates which are $O(1/n)$ and $O(1/n\log(n))$ respectively, {in related settings}.
Moreover, it was recently shown by \textcite{aharonov_polynomial-time_2022} that a constant gate error may render the circuit classically simulable.
Thus, it is to be expected that $n$-dependent gate erros are in fact necessary.

\subsubsection{The dominant signal}\label{sec:dominantSignal}
In the perturbative regime, Theorem~\ref{thm:rb-data-random-circuits} ensures that the \ac{RB} signal is the sum of two (matrix) exponential decays, Eq.~\eqref{eq:thm:signal-form}.
In particular, if the irrep $\tau_\lambda$ appears in $\omega$ without multiplicities the signal becomes a scalar decay governed by a single decay parameter $I_\lambda \in (1 - 2 \delta_\lambda, 1] \subset \RR$.
If $\tau_\lambda$ has multiplicity $n_\lambda$ the matrix $I_\lambda$ has the corresponding dimension. 
If $I_\lambda$ is diagonalizable (over $\CC$) the RB signal is a linear combination of exponentials in the (up-to $n_\lambda$) inequivalent and potentially complex eigenvalues of $I_\lambda$.
As a consequence the signal can decay and oscillate with the sequence length.
The decay parameter $I_\lambda$ does not dependent on \ac{SPAM} errors, these only affect the linear coefficients $A_\lambda$ and $B_\lambda$.
This behavior provides the desired \ac{SPAM} robustness of extracting the decay parameters.

Even though not included in the statement of the theorem,
the trace of the matrix coefficients $A_\lambda$ (the \ac{SPAM} constants) coincides with the filtered \ac{RB} signal which we would obtain for an ideal implementation map $\phi=\omega$ and the Haar measure $\nu=\mu$ but with the \emph{same \ac{SPAM} errors}.
From the expression \eqref{eq:tr-Alambda} for the trace of $A_\lambda$ given in Sec.~\ref{sec:proof-data-form} below, we obtain the 
\emph{signal under only \ac{SPAM} errors}:
\begin{align}
 \tr(A_\lambda)
 =
 \osandwich{\rho}{P_\lambda S^{+} \widehat{\omega}[\omega](\tilde M)}{\tilde\rho}  
 =
 \FD_\lambda(m)_\SPAM \, .
 \label{eq:Flambda-SPAM}
\end{align}

Furthermore, if $\lambda$ is multiplicity-free and aligned with $M$, then $\widehat{\omega}[\tau_\lambda](X_\lambda^\dagger \tilde M)$ is proportional to $X_\lambda^\dagger$ with proportionality factor given as $\tr(P_\lambda\tilde M)/d_\lambda$ (cf.~Sec.~\ref{sec:frame-operator}).
In this case, we have
\begin{align}
\FD_\lambda(m)_\SPAM
&=
d_\lambda^{-1}s_\lambda^{-1}\tr(P_\lambda\tilde M) \osandwich{\rho}{P_\lambda}{\tilde\rho} \\
&=
\frac{\tr(P_\lambda\tilde M)}{\tr(P_\lambda M)} \osandwich{\rho}{P_\lambda}{\tilde\rho} \, , \label{eq:FDspam_aligned}
\end{align}
where we used that $\tr(P_\lambda M) = d_\lambda s_\lambda$.
In the absence of \ac{SPAM} noise, we thus recover the ideal, noiseless signal \eqref{eq:ideal-signal-intro} from Sec.~\ref{sec:overview-protocol}:
\begin{equation}
    \FD_\lambda(m)_\ideal = \osandwich \rho {P_\lambda} \rho\,. 
\end{equation}

We can, thus, measure the deviation from the ideal signal due to \ac{SPAM} in terms of the two relative quantities 
that we call the \emph{SPAM visibilities}:
\begin{align}
 v_\mathrm{SP} &\coloneqq \frac{\left|\osandwich{\rho}{P_\lambda}{\tilde\rho}\right|}{\osandwich{\rho}{P_\lambda}{\rho}} , &
 v_\mathrm{M} &\coloneqq \frac{\abs{\tr(P_\lambda\tilde M)}}{\tr(P_\lambda M)}.
 \label{eq:visibilities}
\end{align}
In terms of the visibilities, we can rewrite the absolute value of the signal affected only by \ac{SPAM} for $\lambda$ multiplicity-free and aligned with $M$ as
$|\FD_\lambda(m)_\SPAM| = v_\mathrm{SP}v_\mathrm{M} \osandwich{\rho}{P_\lambda}{\rho}$.
SPAM errors can decrease the signal-to-noise ratio as well as the ratio between the dominant and sub-dominant signals.
As a result the number of samples and required sequence length to accurately estimate the dominant \ac{RB} signal and extract the decay parameter also depends on the strength of the \ac{SPAM} noise.
This situation should not come as a surprise as the ability to extract information depends crucially on the quality of state preparation and measurement.
We will make use of the visibilities to formulate explicit bounds in the following section, and will eventually assume that they are lower bounded by a constant.

Note that similar assumptions about the \ac{SPAM} constants were made in Refs.~\cite{HarHinFer19,Flammia2019EfficientEstimation}.
In contrast to \emph{stability} condition in Ref.~\cite{Flammia2019EfficientEstimation}, which requires that \ac{SPAM} constants are within additive error of their ideal value, the here introduced \emph{visibilities} capture relative deviations.

\paragraph{Examples.}
As an instructive example, we consider depolarizing \ac{SPAM} noise. 
Recall that $\tilde\rho = \ESP(\rho)$ and $\tilde M = M \EM$ with state preparation and measurement noise channels $\ESP$ and $\EM$, respectively.
Assuming that $\ESP=\EM = p\, \id + (1-p)\oketbra{\one}{\one}/d$ is a depolarizing channel, we have
\begin{align}
 v_\mathrm{SP} &= p + \frac{1-p}{d}\frac{\osandwich{\rho}{P_\lambda}{\one}}{\osandwich{\rho}{P_\lambda}{\rho}}               = p, \\
v_\mathrm{M} &= p + \frac{1-p}{d} \frac{\osandwich{\one}{P_\lambda M}{\one}}{ d_\lambda s_\lambda} = p,
\end{align}
where we use twice that $P_\lambda(\one) = 0$ since $\lambda$ is, by assumption, not the trivial irrep.
Hence, $\FD_\lambda(m)_\SPAM$ is suppressed by $p^2$ compared to the \ac{SPAM}-free situation.

In principle and without further assumptions, the state-preparation errors could increase the visibility and can change the sign of the filtered \ac{RB} signal. 
For example, consider the representation $b = 00$ of $\Cl{1}{2}^{\otimes 2}$.
Then, for $\rho = \ketbra {00}{00}$ and the maximally entangled Bell-state $\tilde\rho = \ketbra{\Psi_+}{\Psi_+}$, we have $\osandwich \rho {P_{00}} {\tilde \rho} = - \osandwich \rho {P_{00}} \rho$.

If $\HW{n}{p}<G$ and $\rho$ is a pure stabilizer state, e.g.\ $\rho = \ketbra 0 0$, we can establish that $v_\mathrm{SP} \leq 1$. 
To see this, recall that $w(a)$ denote the Weyl operators for $a \in \FF_p^n$.
Any stabilizer state can be written as $\rho = \frac1d \sum_{a \in L} \xi^{f(a)} w(a)$, where $L$ is a suitable subspace of $\FF_p^n$, $f:\, L\rightarrow \F_p$ is a suitable function on $L$, and $\xi = \exp(2\pi i/p)$ is a primitive $p$-th root of unity.
Furthermore, if $\HW{n}{p}<G$, $P_\lambda$ is diagonal in the Weyl basis and can be written as $P_\lambda = \frac1d \sum_{a \in \Omega} \oketbra {w(a)} {w(a)}$ for some set $\Omega \subset \FF_p^n$.
Using $\snorm{w(a)} \leq 1$ for all $a$, we then have $|\osandwich \rho {P_\lambda} {\tilde \rho}| = \frac1d \sum_{a \in L \cap \Omega} |\obraket{w(a)}{\tilde \rho}| \leq \frac1d \sum_{a \in L \cap \Omega} 1 = \osandwich \rho {P_\lambda} \rho$.
Hence, we have shown that $v_\mathrm{SP} \leq 1$.

Next, we show a similar statement for the measurement visibility $v_\mathrm{M}$:
Here, we only need that $[P_\lambda, M] = 0$, which is in particular the case if $\HW{n}{p}<G$.
Then, $P_\lambda M$ is a projector with range in the traceless subspace (since $\lambda$ is non-trivial by assumption).
Hence, $\tr(P_\lambda\tilde M) = \tr(P_\lambda M \EM)$ depends only on the unital and trace-preserving part of $\EM$ and we can thus replace $\EM$ with its projection onto unital and trace-preserving channels.
Since these channels have spectral norm one \cite[Thm.~4.27]{Wat18}, we finally find that $\abs{\tr(P_\lambda\tilde M)} \leq \trnorm{P_\lambda M} \snorm{\EM} = \tr(P_\lambda M)$, and thus $v_\mathrm{M} \leq 1$.

Although our assumptions do not explicitly exclude examples of `malicious noise', we generally expect physical noise processes in state preparation and measurement to be less targeted and unable to change the sign of the \ac{RB} signal (as this would lead to clearly observable negative decay curves).

\subsubsection{The bound on the sub-dominant signal}
\label{sec:bound-subdominant-signal}
For both gate-dependent noise and when using a non-uniform measure $\nu$, there exist a sub-dominant decay in the expected \ac{RB} signal, the second summand in Eq.~\eqref{eq:thm:signal-form}.
The third part of Theorem~\ref{thm:rb-data-random-circuits} provides %a 
the bound%
% , Eq.
~\eqref{eq:thm:subdominant-decay-bound} on sub-dominant decay.
The constant $c_\lambda$ can introduce a prefactor to the bound scaling polynomially in the dimension of the irrep $\lambda$.
Note that this prefactor stems from the inverse of the effective measurement frame in the filter function. It is a direct consequence of not implementing an inverse gate at the end of the sequences and can be seen as the price to pay for inversionless \ac{RB} compared to standard \ac{RB}.

To be more concrete, let us again assume that $\lambda$ is aligned with $M$.
For $\omega_\lambda$ multiplicity-free, $S_\lambda = s_\lambda \id_\lambda$, cp.~Eq.~\eqref{eq:frame-op-constant}.
If $s_\lambda \neq 0$ (otherwise we have $\FD_\lambda(m)=0$), we find the simple expressions $\snorm{S_\lambda^+} = s_\lambda^{-1}$ and $\tr S_\lambda = d_\lambda s_\lambda$.
Since $[P_\lambda, M]=0$, Theorem \ref{thm:rb-data-random-circuits} then states that $c_\lambda = \sqrt{d_\lambda/s_\lambda}$.
By Eqs.~\eqref{eq:lower-bound-Slambda} and \eqref{eq:upper-bound-Slambda}, $d_\lambda \leq d_\lambda/s_\lambda \leq d_\lambda^2$, hence the $\lambda$-dependent prefactor cannot exceed $d_\lambda$.
In the case of multiplicites, we have to adapt our argument slightly to use Eqs.~\eqref{eq:lower-bound-Slambda} and \eqref{eq:upper-bound-Slambda} for bounds on $c_\lambda = \sqrt{\tr(P_\lambda M)} \snorm{S_\lambda^+}$.
Recall  from Sec.~\ref{par:eigenvalues-frame-op} that $\tr(P_\lambda M) = \tr(S_\lambda) \geq d_\lambda \snorm{S_\lambda^+}^{-1}$, since $\snorm{S_\lambda^+}^{-1}$ is exactly the smallest non-zero eigenvalue of $S_\lambda$ and each eigenvalue occurs at least $d_\lambda$ times.
On the other hand, $\tr(P_\lambda M) \leq \tr(P_\lambda) = n_\lambda d_\lambda$ by Hölder's inequality.
Using the bounds \eqref{eq:lower-bound-Slambda} and \eqref{eq:upper-bound-Slambda} on $\snorm{S_\lambda^+}$, we then obtain the following inequalities:
\begin{equation}
  \sqrt{d_\lambda} \leq c_\lambda  \leq  d_\lambda \sqrt{n_\lambda d_\lambda} \, .
\end{equation}

For our examples from Sec.~\ref{sec:frame-operator}, namely unitary 2-groups, local products of unitary 2-groups, and the Heisenberg-Weyl group, all irreps are multiplicity-free and $\HW{n}{p}<G$.\footnote{Scalable unitary 2-groups are either dense in the unitary group or a suitable subgroup of the Clifford group containing $\HW{n}{p}$ \cite{guralnick_decompositions_2005,BannaiEtAl:2020:tgroups}, see also Ref.~\cite[Sec.~V]{haferkamp_efficient_2023} for a comprehensive summary.}
The latter fact implies that the projectors $P_\lambda$ are diagonal in the Weyl basis and in particular $[P_\lambda, M]=0$.
Thus, we can use $c_\lambda=\sqrt{d_\lambda / s_\lambda}$, and the $s_\lambda$ which have been computed in Sec.~\ref{sec:frame-operator}:
\begin{align}
 \sqrt{d_\Ad/s_\Ad} &= (d+1)\sqrt{d-1} \tag*{(unitary 2-groups)}, \\
 \sqrt{d_b/s_b} &= \left[(p+1)\sqrt{p-1}\right]^{n-|b|} \tag*{(local unitary 2-groups)}, \\
 \sqrt{d_{0,z}/s_{0,z}} &= 1 \tag*{(Heisenberg-Weyl group)}.
\end{align}

Besides $c_\lambda$, the additional factors $\sqrt{\osandwich{\rho}{P_\lambda}{\rho}}$ and $g(\delta_\lambda/\Delta_\lambda)$ appear in the bound \eqref{eq:thm:subdominant-decay-bound}.
Here, $\sqrt{\osandwich{\rho}{P_\lambda}{\rho}}\leq 1$ is bounded and can, in fact, be very small for small irreps.
If the implementation error $\delta_\lambda/\Delta_\lambda$ is bounded away from $1/4$, $g(\delta_\lambda/\Delta_\lambda)$ can be considered constant for all practical purposes.
We comment on this in more detail in Sec.~\ref{sec:sequence-lengths}.

\subsubsection{Proof of Theorem~\ref{thm:rb-data-random-circuits}}
\label{sec:proof-data-form}

At the core of our argument is a statement from matrix perturbation theory.
In Appendix~\ref{sec:perturbation-theory}, we collect relevant results from the perturbation theory of invariant subspace given in Ref.~\cite{stewart_matrix_1990}, and derive a corollary, Theorem~\ref{thm:perturbation-of-moment-operator}, that specifically applies to moment operators.

As a first step, we use the identity \eqref{eq:Flambda-as-Fourier-transform} for the expected \ac{RB} signal in terms of the Fourier transform of $\phi\nu$ restricted to the isotype $\omega_\lambda$:
\begin{equation}
 \FD_\lambda(m)
  = 
  \osandwichb{\rho}{X_\lambda S_\lambda^+\, \widehat{\phi\nu}[\omega_{\lambda}]^m  \big( X_\lambda^\dagger \tilde M \big) }{\tilde\rho} \, .
  \label{eq:Flambda-as-Fourier-transform-2}
\end{equation}
Before we proceed, we argue that the involved operators are, in fact, real since they act on the real vector space of Hermitian matrices.
Indeed, $\tau_\lambda$ is by assumption real and, thus, the isotypic component $V(\lambda)$ splits as $V(\lambda) = H(\lambda)\oplus i H(\lambda)$ where $H(\lambda)\subset\Herm(\mathcal H)=: H$, c.f.~Sec.~\ref{sec:setting}.
Since real representations have real characters, we find that the projector $P_\lambda$ is Hermiticity-preserving and we can choose a basis such that $X_\lambda$ and $X_\lambda^\dagger$ are too.
As $\tilde M$ is a quantum channel, we thus find that $X_\lambda^\dagger \tilde M$ is Hermiticity-preserving.
Next, it is immediate from its definition that $\widehat{\phi\nu}[\omega_{\lambda}]$ preserves the set of Hermiticity-preserving maps $V \rightarrow V(\lambda)$.
Recall that $S$ is a quantum channel, hence $S_\lambda$ has to be Hermiticity-preserving and thus $S_\lambda^+$ is Hermiticity-preserving, too (c.f.~Sec.~\ref{sec:operators-norms}).
Finally, this shows that all objects in Eq.~\eqref{eq:Flambda-as-Fourier-transform-2} can be treated as (super-)operators on the real vector space of Hermitian matrices, in particular they can be described by real matrices.
Thus, we only consider the restriction to $\Herm(\mathcal H)$ in the following.

We can write $\omega_\lambda = T^\dagger \left(\tau_\lambda \otimes \id_{n_\lambda}\right) T$ and $H(\lambda) = T^\dagger ( H_\lambda \otimes \R^{n_\lambda})$, for a suitable real orthogonal matrix $T$ and irreducible subspace $H_\lambda$.
As in Eq.~\eqref{eq:fourier-transform-reducible-rep}, we then find $\widehat{\phi\nu}[\omega_\lambda]^m(X_\lambda^\dagger \tilde M) = T^\dagger (\widehat{\phi\nu}[\tau_\lambda]^m \otimes \id_{n_\lambda})(T X_\lambda^\dagger \tilde M)$.
For the sake of notation, let us define the superoperators $\tilde M_\lambda \coloneqq T X_\lambda^\dagger \tilde M \in \Hom(H,H_\lambda\otimes \R^{n_\lambda})$ and $Q_\lambda^\dagger\coloneqq \oketbra{\tilde\rho}{\rho}X_\lambda S_\lambda^+ T^\dagger \in \Hom(H_\lambda\otimes \R^{n_\lambda},H) $.
Note that $\oketbra{\tilde M_\lambda}{Q_\lambda}$ is a linear operator on $\Hom(H,H_\lambda\otimes \R^{n_\lambda})$.
With this, we find
\begin{align}
  \FD_\lambda(m)%
  &= 
  \osandwichb{\rho}{X_\lambda S_\lambda^+\, \widehat{\phi\nu}[\omega_{\lambda}]^m  \big( X_\lambda^\dagger \tilde M \big) }{\tilde\rho}%
  \\
  &= \tr\left[ Q_\lambda^\dagger \left(\widehat{\phi\nu}[\tau_\lambda]^m \otimes \id_{n_\lambda}\right)(\tilde M_\lambda)  \right] \\
  &= \tr\left[ \left(\widehat{\phi\nu}[\tau_\lambda]^m \otimes \id_{n_\lambda}\right) \oketbra{\tilde M_\lambda}{Q_\lambda} \right]  .
  \label{eq:Flambda-matrix-inner-product}
\end{align}
We treat $\widehat{\phi\nu}[\tau_{\lambda}] =  \widehat{\omega\nu}[\tau_{\lambda}]  + E$ as a perturbation of the moment operator $\widehat{\omega\nu}[\tau_{\lambda}]$. 
Recall from Sec.~\ref{sec:rho-designs}, Eq.~\eqref{eq:moment-op-spectral-decomposition}, that $\widehat{\omega\nu}[\tau_{\lambda}]$ is block-diagonal where the upper block corresponds to the range of the projector $\widehat{\omega}[\tau_\lambda]$.
By assumption \eqref{eq:thm:initialbelief}, $\snorm{E}\leq \delta_\lambda < \Delta_\lambda/4$, hence we can invoke Thm.~\ref{thm:perturbation-of-moment-operator} to write
\begin{equation}
\label{eq:phihat-block-diagonalization}
\widehat{\phi\nu}[\tau_{\lambda}] 
= 
R_{\lambda,1} I_\lambda L_{\lambda,1}\ad + R_{\lambda,2} O_\lambda L_{\lambda,2}\ad,
\end{equation}
for suitable real operators $R_\lambda=[R_{\lambda,1},R_{\lambda,2}]$ and $L_\lambda=[L_{\lambda,1}, L_{\lambda,2}]$ with $L_\lambda^\dagger R_\lambda = \id$.
Moreover, $I_\lambda$ is a real linear operator on the $n_\lambda$-dimensional perturbed range of $\widehat{\omega}[\tau_{\lambda}] = R_{\lambda,1} L_{\lambda,1}^\dagger$.
Likewise, $O_\lambda$ is a real linear operator on the $(d_\lambda d^2 - n_\lambda)$-dimensional perturbation of the kernel.
{
From Thm.~\ref{thm:perturbation-of-moment-operator}, Eq.~\eqref{eq:spectral-resolution-moment-op}, it is immediate that $I_\lambda = X_{\lambda,1}^\dagger \widehat{\phi\nu}[\tau_{\lambda}] R_{\lambda,1}$ where $X_{\lambda,1}$ is a partial isometry such that $\widehat{\omega}[\tau_{\lambda}] = X_{\lambda,1} X_{\lambda,1}^\dagger$.
Hence, we conclude that if $\nu=\mu$ is the Haar measure, $I_\lambda$ is exactly the matrix $M_\lambda$ in Ref.~\cite[Thm.~8]{helsen_general_2022}.
}
We arrive at the following expression:
\begin{align}
  \FD_\lambda(m)
  &= 
  \tr\Big[ (L_{\lambda,1}^\dagger\otimes \id_{n_\lambda}) \oketbra{\tilde M_\lambda}{Q_\lambda}\, ( R_{\lambda,1}\otimes \id_{n_\lambda}) \\
  &\qquad \times ( I_\lambda^m \otimes \id_{n_\lambda} ) \Big] \\
  &\quad +  
  \tr\Big[ (L_{\lambda,2}^\dagger\otimes \id_{n_\lambda}) \oketbra{\tilde M_\lambda}{Q_\lambda}\, ( R_{\lambda,2}\otimes \id_{n_\lambda}) \\
  &\qquad \times ( O_\lambda^m \otimes \id_{n_\lambda}) \Big] \\
  &=
  \tr\left[A_\lambda I_\lambda^m\right]
  + 
  \tr\left[ B_\lambda O_\lambda^m \right] ,
  \label{eq:perturbation-2nd-term_moment}
\end{align}
with $A_\lambda \coloneqq L_{\lambda,1}^\dagger \tr_{n_\lambda}\left(\oketbra{\tilde M_\lambda}{Q_\lambda}\right) R_{\lambda,1} $ and $B_\lambda \coloneqq L_{\lambda,2}^\dagger \tr_{n_\lambda}\left(\oketbra{\tilde M_\lambda}{Q_\lambda}\right) R_{\lambda,2} $.
In particular, we have the following expression used in Sec.~\ref{sec:dominantSignal}:
\begin{align}
 \tr(A_\lambda)
 &=
 \tr\left[ \left( \widehat{\omega}[\tau_\lambda] \otimes \id_{n_\lambda} \right) \oketbra{\tilde M_\lambda}{Q_\lambda} \right] \\
 &=
 \osandwichb{\rho}{X_\lambda S_\lambda^+\, \widehat{\omega}[\omega_{\lambda}] \big( X_\lambda^\dagger \tilde M \big) }{\tilde\rho} \,.
 \label{eq:tr-Alambda}
\end{align}
Here, we used that $R_{\lambda,1} L_{\lambda,1}^\dagger = \widehat{\omega}[\tau_\lambda]$ and then traced the steps leading to Eq.~\eqref{eq:Flambda-matrix-inner-product} backwards.

The claimed spectral bound on $O_\lambda$ follows directly from Thm.~\ref{thm:perturbation-of-moment-operator}, as well as {$\snorm{I_\lambda - \id} < 2\delta_\lambda$.
The latter statement already shows that $\spec_{\C}(I_\lambda) \subset D_{2\delta}(1)$.
Note that $\widehat{\phi\nu}[\omega]$ maps quantum channels to completely positive, trace-non increasing maps and thus has $\diamond\rightarrow\diamond$ norm at most 1.
Hence, its eigenvalues have absolute value $\leq 1$ and the same holds for $\widehat{\phi\nu}[\tau_\lambda]$ as it corresponds to a block in the block diagonalization of $\widehat{\phi\nu}[\omega]$ by the irreps of $\omega$.
Since every eigenvalue of $I_\lambda$ is also an eigenvalue of $\widehat{\phi\nu}[\tau_\lambda]$, we conclude that $\spec_{\C}(I_\lambda) \subset D_1(0) \cap D_{2\delta}(1)$}.
This establishes the signal form \eqref{eq:thm:signal-form}.

Next, we can bound the subdominant decays as follows 
\begin{align}
 \MoveEqLeft[1]|\tr\left[ B_\lambda O_\lambda^m \right]|\\
 &\leq
 \trnorm{B_\lambda}\snorm{O_\lambda}^m \\
 &\leq 
 \trnorm[\big]{(L_{\lambda,2}^\dagger\otimes \id_{n_\lambda}) \oketbra{\tilde M_\lambda}{Q_\lambda}\, ( R_{\lambda,2}\otimes \id_{n_\lambda})}
 \snorm{O_\lambda}^m \\
 &= 
 \twonorm[\big]{(L_{\lambda,2}^\dagger\otimes \id_{n_\lambda})(\tilde M_\lambda)}
 \twonorm[\big]{(R_{\lambda,2}^\dagger\otimes \id_{n_\lambda})(Q_\lambda)}
 \snorm{O_\lambda}^m \\
 &\leq
 \snorm{L_{\lambda,2}}\snorm{R_{\lambda,2}} \twonorm{X_\lambda^\dagger \tilde M} 
 \twonorm{S_\lambda^+X_\lambda\oketbra{\rho}{\tilde\rho}}
 \snorm{O_\lambda}^m.
\end{align}
Here, we have used that the partial trace is a contraction w.r.t.~to trace norm and that we have $\twonorm{AB}\leq \twonorm{A}\snorm{B}$.
To proceed, recall that $\tilde M = M\EM$ and $M=\sum_i\oketbra{E_i}{E_i}$ is a projection i.e.~$M^2=M$.
We have $\twonorm{X_\lambda^\dagger \tilde M}^2 = \tr(P_\lambda M \EM \EM^\dagger M)$.
The superoperator $M \EM \EM^\dagger M$ is completely positive and self-adjoint, however, it is generally not trace-preserving. 
The range of $P_\lambda$ for $\lambda$ non-trivial has to lie within the traceless subspace of $\End(\mathcal H)$.
Thus, for the trace inner product of $P_\lambda$ and $M \EM \EM^\dagger M$, only the part of $M \EM \EM^\dagger M$ restricted to the traceless subspace plays a role.
In particular, we can without loss of generality replace $M \EM \EM^\dagger M$ by its projection onto unital and trace-preserving quantum channels, since it only changes $M \EM \EM^\dagger M$ outside of the traceless subspace.
However, $M$ is already unital and trace-preserving, thus the projection of $M \EM \EM^\dagger M$ is in fact the projection of $\EM \EM^\dagger$, conjugated by $M$.
Since unital and trace-preserving quantum channels have spectral norm 1, we have $\snorm{\EM\EM^\dagger} = 1$, cf.~Ref.~\cite[Thm.~4.27]{Wat18}, and thus we find using Hölder's inequality:
\begin{align}
   \twonorm{X_\lambda^\dagger \tilde M}^2
   &= 
   \tr(P_\lambda M \EM \EM^\dagger M)  \\
   &\leq 
   \trnorm{P_\lambda M} \snorm{\EM \EM^\dagger}  
   \leq
   \min\{d_\lambda,d\} .
   \label{eq:main-thm-proof-1}
\end{align}
In particular, if $[P_\lambda,M]=0$, we obtain the refined upper bound:
\begin{align}
 \twonorm{X_\lambda^\dagger \tilde M}^2
   &= 
   \tr(P_\lambda M \EM \EM^\dagger M)  \\
   &\leq 
   \trnorm{P_\lambda M} \snorm{\EM \EM^\dagger}  \\
   &=
   \tr(P_\lambda M) 
   = 
   \tr(S_\lambda).
   \label{eq:main-thm-proof-2}
\end{align}

Finally, the state-dependent part can be written as
\begin{align}
 \twonorm{S_\lambda^+ X_\lambda\oketbra{\rho}{\tilde\rho}}
 &\leq \snorm{S_\lambda^+} \sqrt{\osandwich{\rho}{P_\lambda}{\rho}} \, \twonorm{\tilde\rho} \\
 &\leq \snorm{S_\lambda^+} \sqrt{\osandwich{\rho}{P_\lambda}{\rho}}\, .
\end{align}
Combining the above bounds, we define
\begin{align}
 c_\lambda
 &\coloneqq
 \begin{cases} 
  \sqrt{\tr S_\lambda} \snorm{S_\lambda^+} & \text{ if } [P_\lambda,M]=0 \\
  \min\{\sqrt{d_\lambda},\sqrt{d}\} \snorm{S_\lambda^+} & \text{ else}
 \end{cases}.
\end{align}
To obtain the final bound, we use the following result from Thm.~\ref{thm:perturbation-of-moment-operator}: 
\begin{align}
 \snorm{L_{\lambda,2}} \snorm{R_{\lambda,2}} & \leq g(\delta_\lambda/\Delta_\lambda)  \, ,\\
 g(x) &\coloneqq (1-4x)x + \frac{1+x}{1-4x} \, .
\end{align}
Combining the above results, we then find
\begin{align}
 \MoveEqLeft[2]\abs[\big]{\FD_\lambda(m) - \tr\left[A_\lambda I_\lambda^m\right] }\\
 &\leq 
 c_\lambda \,
 \sqrt{\osandwich{\rho}{P_\lambda}{\rho}}
 \snorm{L_{\lambda,2}}\snorm{R_{\lambda,2}} 
 \snorm{O_\lambda}^m  \\
 &<
 c_\lambda\, \sqrt{\osandwich{\rho}{P_\lambda}{\rho}}\, g(\delta_\lambda/\Delta_\lambda) (1-\Delta_\lambda + 2\delta_\lambda)^m.
\end{align}

\subsubsection{Filtering onto trivial irrep: Measuring average trace-preservation}
\label{sec:filtering-trivial-irrep}

In the proof of Theorem~\ref{thm:rb-data-random-circuits}, we assumed that the irrep $\lambda$ is not trivial.
We treat this case separately in this section, as the arguments simplify considerably.
The filtered \ac{RB} protocol for the trivial irrep and with $G$ an exact unitary $1$-design has been proposed as a protocol to detect ``incoherent'' leakage errors in Ref.~\cite{WalBarEme16b}. See also Ref.~\cite{chasseur_complete_2015} in this context.
Going beyond simply aiming at completeness,
this section, thus provides guarantees for incoherent leakage benchmarking for non-uniform measures and gate-dependent noise.
We think that our framework can be extended to also capture more general (coherent) leakage errors, we however leave this for future work.

In the following, we assume that the only trivial irrep in $\omega$ is spanned by the identity matrix $\one$.
Using that $P_1 = \oketbra{\one}{\one}/d$ and $S_1 = P_1$, we find that
\begin{align}
 \MoveEqLeft[1] F_1(m) \\
 &= \frac{1}{d} \obraket{\rho}{\one} \osandwich{\one}{S_1^+ \widehat{\phi\nu}[\omega](\tilde M)}{\tilde \rho} \\
 &= \frac{1}{d} \sum_{i\in[d]} \int_G \osandwich{\tilde E_i}{\phi(g_m)\dots\phi(g_1)}{\tilde\rho} \dd\nu(g_1)\dots\dd\nu(g_m) \\
 &= \frac{1}{d} \int_G \osandwich{\tilde\one}{\phi(g_m)\dots\phi(g_1)}{\tilde\rho} \dd\nu(g_1)\dots\dd\nu(g_m)\, .
\end{align}
Here, we set $\tilde\one\coloneqq\EM^\dagger(\one)$ which is simply $\one$ if the measurement noise is trace-preserving.
In this case,  $F_1(m)$ can be interpreted as the average trace-preservation of a sequence of length $m$.
In particular, if $\phi$ (and $\EM$) is trace-preserving, $F_1(m) = 1/d$ for all $m$.

Note that the Fourier transform evaluated at the trivial irrep $\tau_1(g)\equiv 1$ can be identified with the integral over $\phi$:
\begin{equation}
\widehat{\phi\nu}[\tau_1]
=
\int_G 1(\argdot)\phi(g) \dd\nu(g) 
\simeq
\int_G \phi(g) \dd\nu(g) \,.
\end{equation}
Thus, we indeed recover a formula similar to the non-trivial case:
\begin{equation}
 F_1(m)
 =
 \frac{1}{d}\osandwich{\tilde\one}{\widehat{\phi\nu}[\tau_1]^m}{\tilde\rho}\, .
 \label{eq:rb-data-trivial}
\end{equation}
We can then show that $F_1(m)$ has the form of an exponential decay by treating $\widehat{\phi\nu}[\tau_1]$ as a perturbation of $\widehat{\omega\nu}[\tau_1]$.

\begin{theorem}[Filtering onto trivial irrep]
Suppose that the trivial irrep is multiplicity-free and $\widehat{\omega\nu}[\tau_1] \simeq \int_G \omega(g)\dd\nu(g)$ has a spectral gap $\Delta_1 > 0$.
If there is 
% a 
$\delta_1 > 0$ such that
\begin{equation}
 \snorm[\bigg]{ \int_G \big( \phi(g) - \omega(g) \big) \dd\nu(g) } \leq \delta_1 < \frac{\Delta_1}{4}
\end{equation}
then the filtered \ac{RB} signal \eqref{eq:Flambda-as-Fourier-transform} is given by
\begin{equation}
 F_1(m) = \frac{1}{d} \left( \frac{1}{d}  \tr\EM(\one) \, I_1^m + \tr(B_1 O_1^m) \right),
\end{equation}
where $\EM$ is the trace non-increasing measurement noise channel, $1-2\delta_1 < I_1 \leq 1$ and $B_1$ and $O_1$ are real operators on the traceless subspace of $\Herm(\C^d)$.
Moreover, $I_1$ and $O_1$ do not depend on the initial state and measurement, and we have $I_1=1$ if $\phi(g)$ is trace-preserving $\nu$-almost everywhere.
Finally, we have the bounds $\snorm{O_1} \leq 1 + \Delta_1 - 2\delta_1$ and
\begin{equation}
 \frac{1}{d} \abs[\big]{\tr(B_1 O_1^m)} \\
 \leq 
 c_1 \big(1 + \Delta_1 - 2\delta_1\big)^m \, ,
\end{equation}
where 
\begin{equation}
 c_1 = \left[
  \frac{1 + 2 \sqrt{d}}{\Delta_1/\delta_1 - 4}  
  + 2 \, \twonorm{\EM(\one)_0} \left( 1 +  \frac{1}{\Delta_1/\delta_1 - 4} \right)
  \right] \,,
\end{equation}
and $\EM(\one)_0 = \EM(\one) - \tr(\EM(\one)) \, \one/d$ is the traceless part of $\EM(\one)$.
\end{theorem}

\begin{proof}
We start from Eq.~\eqref{eq:rb-data-trivial} and apply perturbation theory to the block diagonalization of $\widehat{\omega\nu}[\tau_1]\simeq \int_G \omega(g)\dd\nu(g)$ as in the proof of Thm.~\ref{thm:rb-data-random-circuits}, c.f.~Sec.~\ref{sec:proof-data-form}.
Moreover, we can again restrict to the action on the real vector space of Hermitian matrices $\Herm(\C^d)$.
We then obtain from Thm.~\ref{thm:perturbation-of-moment-operator} in App.~\ref{sec:perturbation-theory} that
\begin{equation}
 \widehat{\phi\nu}[\tau_1] = R_{1} I_1 L_{1}^\dagger + R_{2} O_1 L_{2}^\dagger = I_1 P_1 + R_{2} O_1 L_{2}^\dagger \,.
\end{equation}
Here, we used that the first block is $1\times 1$, thus the first term becomes $I_1\, R_{1} L_{1}^\dagger$, and $R_{1} L_{1}^\dagger = P_1 = \oketbra{\one}{\one}/d$ is the projection onto the first block, i.e.~onto the trivial irrep of $\omega$.
The other block corresponds to the traceless subspace.
We then obtain
\begin{align}
 F_1(m) 
 &=
 \frac{1}{d}\osandwich{\tilde\one}{P_1}{\tilde\rho} I_1^m + \frac{1}{d} \osandwich{\tilde\one}{R_{2} O_1^m L_{2}^\dagger}{\tilde\rho} \\
 &=
 \frac{\tr\tilde\one}{d^2} I_1^m + \frac{1}{d} \osandwich{\tilde\one}{R_{2} O_1^m L_{2}^\dagger}{\tilde\rho} \,.
\end{align}

Let $E = \widehat{\phi\nu}[\tau_1] - \widehat{\omega\nu}[\tau_1]$ be the perturbation error, then we can use the formula $I_1 = 1 + E_{11} + E_{12}Q_1$ from Thm.~\ref{thm:perturbation-of-moment-operator}, where $E_{ij}=X_i^\dagger E X_j$ are the blocks of perturbation (here we rename the operators $P_i$ in Thm.~\ref{thm:perturbation-of-moment-operator} to $Q_i$ to avoid confusion).
Since the first block is one-dimensional, we can make the identification $X_1 \equiv \oket{\one}/\sqrt{d}$.
Note that if $\phi$ is trace-preserving $\nu$-almost everywhere, then
\begin{equation}
 X_1^\dagger E  = \frac{1}{\sqrt{d}} \int_G \obra{\one}(\phi(g)-\omega(g)) = \frac{1}{\sqrt{d}} \int_G \big[\obra{\one} - \obra{\one}\big] = 0.
\end{equation}
Thus, $E_{11} = 0$ and $E_{12}=0$ which shows that $I_1 = 1$ in this case.

Next, we use the formulae $R_{2} = X_1 Q_2 + X_2 ( Q_1Q_2 + \id_2)$ and $L_{2} = X_2 - X_1 Q_1^\dagger$.
Recall that $Q_1$ maps from the first block to the second block and vice versa for $Q_2$. 
Since the second block corresponds to the traceless subspace, we can make the further identifications $Q_1 =  \oket{q_1}$ and $Q_2 = \obra{q_2}$ for suitable traceless operators $q_1,q_2$ on $\Herm(\C^d)$. 
Note that we then have $\snorm{Q_i} = \twonorm{q_i}$ for $i=1,2$. 
Let $\tilde\rho_0 \simeq \tilde\rho - \one/d$ and $\tilde\one_0 \simeq \tilde\one - \tr(\tilde\one) \one/d $ be the traceless part of $\tilde\rho$ and $\tilde\one$, respectively.
We then find that
\begin{align}
 R_{2}^\dagger(\tilde\one) 
 &= (\id_2 + Q_2^\dagger Q_1^\dagger) X_2^\dagger(\tilde\one) + Q_2^\dagger X_1^\dagger(\tilde\one) \\
 &= \left( \obraket{q_1}{\tilde\one_0} + \frac{\tr\tilde\one}{\sqrt{d}} \right) q_2  + \tilde\one_0 \, ,\\
 L_{2}^\dagger(\tilde\rho)
 &= X_2^\dagger(\tilde\rho) - \frac{\tr\tilde\rho}{\sqrt{d}} \, q_1
 = \tilde\rho_0 - \frac{1}{\sqrt{d}} q_1\,.
\end{align}
Next, we use the bounds \eqref{eq:intermediateP1bound}, \eqref{eq:P2bound}, and \eqref{eq:P1P2bound} on $Q_1$ and $Q_2$, namely
\begin{align}
 \snorm{Q_1} &< 4\frac{\delta_1}{\Delta_1} < 1 \, , \qquad
 \snorm{Q_2} \leq \frac{2}{\Delta_1/\delta_1 - 4} \, , \\
 \snorm{Q_1}\snorm{Q_2} &\leq \frac{4 \delta_1/\Delta_1}{\Delta_1/\delta_1 - 4} \leq \frac{1}{\Delta_1/\delta_1 - 4} \, .
\end{align}
Moreover, we use $\tr\tilde\one = \tr(\EM(\one)) \leq \tr\one = d$ and $\twonorm{\tilde\rho_0}\leq 1$ to obtain
\begin{widetext}
 \begin{align}
 \MoveEqLeft[2]\osandwich{\tilde\one}{R_{2} O_1^m L_{2}^\dagger}{\tilde\rho}\\
 &=
 \left(\frac{\tr\tilde\one}{\sqrt{d}} + \obraket{\tilde\one_0}{q_1} \right) \osandwich{q_2}{O_1^m}{\tilde\rho_0 - q_1/\sqrt{d}}
 +
 \osandwich{\tilde\one_0}{O_1^m}{\tilde\rho_0 - q_1/\sqrt{d}} \\
 &\leq
 \left(\frac{\tr\tilde\one}{\sqrt{d}} + \twonorm{\tilde\one_0} \snorm{Q_1} \right) 
  \left( \snorm{Q_2} \twonorm{\tilde\rho_0} + \frac{1}{\sqrt{d}} \snorm{Q_1}\snorm{Q_2}  \right) \snorm{O_1}^m
 +
  \twonorm{\tilde\one_0} \left( \twonorm{\tilde\rho_0} + \frac{1}{\sqrt{d}} \snorm{Q_1} \right)\snorm{O_1}^m \\
 &\leq
  \left( \sqrt{d} \snorm{Q_2} + \snorm{Q_1}\snorm{Q_2}\right) \snorm{O_1}^m 
 +
  \twonorm{\tilde\one_0} \left( 1 + \frac{1}{\sqrt{d}}\snorm{Q_1}\right) 
  \bigg( 1 +  \snorm{Q_1}\snorm{Q_2} \bigg)  \snorm{O_1}^m  \\
 &\leq
  \left[
  \frac{1 + 2 \sqrt{d}}{\Delta_1/\delta_1 - 4}  
  + 2 \, \twonorm{\tilde\one_0} \left( 1 +  \frac{1}{\Delta_1/\delta_1 - 4} \right)
  \right]
  \snorm{O_1}^m  \, .
\end{align}
\end{widetext}
The remaining claims follow from Thm.~\ref{thm:perturbation-of-moment-operator} as in Thm.~\ref{thm:rb-data-random-circuits}.
\end{proof}

%%%%------
\subsubsection{Signal form for representations of complex and quaternionic type}
\label{sec:complex-quaternionic-type}

{

In Sec.~\ref{sec:setting}, we assumed that the representation $\omega$ decomposes into the same number of irreps when considered as a complex representation on $V=L(\mathcal H)$ or as a real representation on $H = \Herm(\mathcal H)$.
In other words, all irreps of $\omega$ are of \emph{real type}.
However, this assumption may not always be fulfilled.
For instance, consider the cyclic group $G=\langle R \rangle$ generated by $R:=e^{i\pi/4 X}$, the single-qubit $\pi/2$-rotation about the $X$ axis.
The irreps of $\omega$ correspond to the invariant subspaces of
\begin{equation}
 R(\cdot)R^\dagger \simeq
 \begin{bmatrix}
  1 & 0 \\
  0 & 1 \\
 \end{bmatrix}
 \oplus
 \begin{bmatrix}
  0 & -1 \\
  1 & 0
 \end{bmatrix}\, .
\end{equation}
Clearly, the second block is not diagonalizable over $\R$.
Hence, in addition to two trivial irreps, $\omega$ has a two-dimensional irrep over $\R$.
However, over the complex numbers, the latter irrep decomposes further into two one-dimensional irreps.
This changes the expected signal form compared to Thm.~\ref{thm:rb-data-random-circuits}, as we briefly discuss in the following.

In Sec.~\ref{sec:real-representations} of the preliminaries, we discussed that a real irrep $\tau$ of a group $G$ may become reducible after complexification.
The possible cases that can occur are related to the commutant of $\tau$ which, by Frobenius' theorem, is either isomorphic to $\R$, $\C$, or the quaternions $\mathbb H$.
For the statement of our results, we consider the first case which implies that the complexification of $\tau$ is again irreducible.
The example given in terms of $R$ above falls into the second case.

The results in Sec.~\ref{sec:data-guarantees} can be generalized to irreps of complex and quaternionic type by noticing that Thm.~\ref{thm:rb-data-random-circuits} relies on the block diagonalization of the moment operator $\widehat{\omega\nu}[\tau_\lambda]$ in Eq.~\eqref{eq:moment-op-FT-spectral-decomposition}:
\begin{equation}
 \widehat{\omega\nu}[\tau_\lambda]
 \simeq
 \begin{bmatrix}
  \id & 0 \\
  0 & \Lambda_\lambda
 \end{bmatrix}\,.
\end{equation}
Recall from Sec.~\ref{sec:proof-data-form} that we can consider $\widehat{\omega\nu}[\tau_\lambda]$ as a real operator and hence its matrix representation is real.
Here, the identity block acts on the the range of $\widehat{\omega}[\tau_\lambda]$ and $\Lambda_\lambda$ acts on its kernel.
Similar to Prop.~\ref{prop:fourier-transform}, one finds that the range is exactly $(\tau_\lambda)'\otimes \R^{n_\lambda}$, where $(\tau_\lambda)'$ is the commutant of $\tau_\lambda$ and $n_\lambda$ is the multiplicity of $\tau_\lambda$ in $\omega$.
Schur's lemma then results in different dimensions of the upper left block, namely $n_\lambda$, $2n_\lambda$, or $4n_\lambda$, depending on whether the type of $\tau_\lambda$ is real, complex, or quaternionic, respectively.
The remaining arguments in the proof of Thm.~\ref{thm:rb-data-random-circuits} are still valid, thus we can conclude that the expected signal form is
\begin{equation}\tag{$\mathbb S'$}\label{eq:thm:signal-form-complex-quaternionic}
   \FD_\lambda(m)  =  \tr\left(A_\lambda I_\lambda^m\right) + \tr\left( B_\lambda O_\lambda^m\right) ,
 \end{equation}
 where now $A_\lambda, I_\lambda\in \R^{t_\lambda n_\lambda \times t_\lambda n_\lambda}$ and $t_\lambda=1,2,4$ when $\tau_\lambda$ is of real, complex, or quaternionic type, respectively.
Hence, the effect of irreps of complex or quaternionic type is to \emph{effectively increase the multiplicity}.
As a consequence, we obtain a matrix exponential decay in the expected signal, even for multiplicity-free irreps such as the example $G=\langle R \rangle$ above.
}

%%% ---------------------------------------------
\subsection{Sampling complexity of filtered randomized benchmarking}
%%% ---------------------------------------------
\label{sec:sampling-complexity}

The expression for the filtered \ac{RB} signal $\FD_\lambda(m)$ from~Eq.~\eqref{eq:filtered-rb-data} has the form of an expectation value for the random variable $f_\lambda(i, g_1,\dots,g_m)$ where $(i,g_1,\dots,g_m) \sim p(i|g_1,\dots,g_n)\dd\nu(g_1)\dots\dd\nu(g_m)$.
In the following, we consider the unbiased estimator $\hat{\FD}_\lambda(m)$, given as the mean of $N$ iid samples $(i^{(l)},g_1^{(l)},\dots,g_n^{(l)})$:
\begin{align}
\label{eq:FD-estimator}
  \hat{\FD}_\lambda(m)%
  &= \frac{1}{N} \sum_{l=1}^N f_\lambda(i^{(l)}, g_1^{(l)},\dots,g_m^{(l)})\, .
\end{align}
In this section, we derive bounds on the number of samples $N$ needed to guarantee that $\hat{\FD}_\lambda(m)$ is, with high probability, close to $\FD_\lambda(m)$.
We base our analysis on the variance of $\hat{\FD}_\lambda(m)$ because 
the function $f_\lambda$ may take on values as large as the Hilbert space dimension $d$. %assume values is a false friend
We have $\Var[\hat{\FD}_\lambda(m)] = \Var[f_\lambda]/N$, thus the second moment of $f_\lambda$ is key for our sampling complexity bounds.
To this end, we show that the second moment of $f_\lambda$ is close to the second moment in the idealized situation where all gates are noiseless and sampled from the Haar measure on $G$, provided that the sequence length $m$ is sufficiently large.
Hence, the noisy implementation and the non-uniform sampling cannot disturb the efficiency of filtered randomized benchmarking.

The proof strategy is as follows:
First, we show that the second moment $\EE[f_\lambda^2]$ has a similar form as the first moment $\EE[f_\lambda] = \FD_\lambda(m)$ and thus admits an analogous perturbative expansion as in Thm.~\ref{thm:rb-data-random-circuits}.
Then, we proceed by deriving appropriate bounds on the subdominant terms, in analogy to Sec.~\ref{sec:subdominant-decay}.
Finally, we combine these results to relate $\EE[f_\lambda^2]$ to its value in the idealized situation, and derive additive and relative-precision guarantees for the estimator $\hat{\FD}_\lambda(m)$.
For additive precision, we find that the sampling complexity is essentially the same as in the idealized situation.
In contrast, relative precision requires that the number of sampling increases with the sequence length as $1/I_\lambda^{2m}$ where $I_\lambda$ is the decay parameter from Thm.~\ref{thm:rb-data-random-circuits}.
This is however unavoidable since the signal $\FD_\lambda(m)$ decays as $I_\lambda^m$.

The estimator $\hat{\FD}_\lambda(m)$ is sometimes called a \emph{single-shot estimator} as it requires that every circuit is measured exactly once.
Another frequently used data acquisition scheme that also yields an estimator 
for $\FD_\lambda(m)$ is the following: 
Sample $N_C$ different random circuits according to $\nu$, and then take $N_M$ samples from the outcome distribution for each circuit.
We discuss the resulting \emph{multi-shot estimator} in App.~\ref{sec:estimators}, and show that its variance involves an additional term compared to the one of $\hat{\FD}_\lambda(m)$, which is a fourth moment of $\nu$ w.r.t.~$\omega$.
As explained in App.~\ref{sec:estimators}, the sampling complexity of such a scheme is generally higher than for the single-shot estimator (see also Ref.~\cite{helsen_thrifty_2022} in this context).
The precise difference depends on the Hilbert space dimension and makes using the single-shot estimator particularly important for small dimensions.
When the Hilbert space dimension is large compared to the inverse desired precision, using less sequences and more shots per sequence yields essentially the same sampling complexity.
The techniques in this section can be readily applied to the multi-shot estimator and we expect a qualitatively similar statement to Thm.~\ref{thm:sampling-complexity-additive} in this case.

In analogy to Eq.~\eqref{eq:Flambda-as-Fourier-transform} for the first moment $\FD_\lambda(m)$, the perturbative expansion of the second moment $\EE[f_\lambda^2] $ is based on the following observation:
\begin{align}
 \EE[f_\lambda^2]
  &=
   \sum_{i \in [d]} \int_{G^{m}} f_\lambda(i, g_1,\dots,g_m)^2 \\
  &\qquad \times
   p(i| g_1, \ldots, g_m) \dd\nu(g_1)\cdots\dd\nu(g_m) \\
  &=
  \sum_{i \in [d]}\int 
  \left[\osandwichb{\rho}%
  {P_\lambda S^+ 
  \omega(g_1)^\dagger\cdots\omega(g_m)^\dagger}{E_i}\right]^2  \\
  &\qquad \times
   \osandwichb{\tilde E_i}{\dd(\phi\nu)(g_m)\dots\dd(\phi\nu)(g_1)}{\tilde\rho} \\
  &= 
   \obra[\big]{\rho^{\otimes 2}}(P_\lambda S^+)^{\otimes 2} \\
  &\qquad \times
   \left( \int \omega(g)^{\dagger\otimes 2} (\argdot) \dd(\phi\nu)(g)\right)^m \left( \tilde M_3 \right)
   \oket[\big]{\tilde\rho} \\
  &=
  \osandwichb{\rho^{\otimes 2}}%
  {(X_\lambda S_\lambda^+)^{\otimes 2}\,
  \widehat{\phi\nu}[\omega_{\lambda}^{\otimes 2}]^m \big( X_\lambda^{\dagger\otimes 2} \tilde M_3 \big) }%
  {\tilde\rho}\, ,
  \label{eq:second-moment-as-fourier-transform}
\end{align}
where we have defined 
\begin{equation}
  \tilde M_3 
  \coloneqq 
  \sum_{i\in [d]} \oketbra{E_i\otimes E_i}{\tilde E_i}\, .
\end{equation}

Note that $\omega_\lambda^{\otimes 2}$ is generally reducible and can thus be decomposed into isotypic representations $\omega^{(2)}_\sigma$.
Moreover, $S_\lambda^{\otimes 2}$ commutes with $\omega_\lambda^{\otimes 2}$ and is thus block-diagonal in this decomposition.
Finally, we can write the partial isometry $X_\lambda^{\otimes 2}$ as a concatenation of partial isometries on the isotypic components. 
In summary, we have:
\begin{align}
 \omega_\lambda^{\otimes 2} &= \bigoplus_{\sigma \in \Irr(\omega_\lambda^{\otimes 2})} \omega^{(2)}_\sigma, &
 (S_\lambda^+)^{\otimes 2} &= \bigoplus_{\sigma \in \Irr(\omega_\lambda^{\otimes 2})} T_\sigma^+, \\
 X_\lambda^{\otimes 2} &= \bigoplus_{\sigma \in \Irr(\omega_\lambda^{\otimes 2})} Y_\sigma \, .
\end{align}
Hence, the second moment can be written as 
\begin{multline}
   \EE[f_\lambda^2] \\
   = 
   \sum_{\sigma \in \Irr(\omega_\lambda^{\otimes 2})} 
   \osandwichb{\rho^{\otimes 2}}%
   {Y_\sigma T_\sigma^+\,
   \widehat{\phi\nu}[\omega^{(2)}_\sigma]^m \big( Y_\sigma^{\dagger} \tilde M_3 \big) }%
   {\tilde\rho}\, .
\label{eq:RB-data-second-moment-decomp}
\end{multline}
The expressions on the right-hand side have the same form as the filtered \ac{RB} signal itself.
Thus, we can argue as in the proof of Theorem \ref{thm:rb-data-random-circuits} to compute the RHS of Eq.~\eqref{eq:RB-data-second-moment-decomp}, provided that the appropriate assumptions are fulfilled for every $\sigma\in\Irr(\omega_\lambda^{\otimes 2})$.
As it turns out, the multiplicity of $\sigma$ in $\omega_\lambda^{\otimes 2}$ does not affect the form of Eq.~\eqref{eq:RB-data-second-moment-decomp}. 

\begin{theorem}[Data guarantees for second moment of filtered \ac{RB} with random circuits]
\label{thm:second-moment}
Fix a non-trivial irrep $\tau_\lambda$ appearing in $\omega$.
Suppose that for all $\sigma\in\Irr(\tau_\lambda^{\otimes 2})$ the spectral gap of $\widehat{\omega\nu}[\tau_{\sigma}]$ is lower bounded by $\Delta_\sigma > 0$, and there are $\delta_\sigma >0$ such that
 \begin{equation}
  \snorm{\widehat{\phi\nu}[\tau_\sigma] - \widehat{\omega\nu}[\tau_\sigma]} \leq \delta_\sigma < \frac{\Delta_\sigma}{4}.
 \end{equation}
 Then, the second moment of the $\lambda$-filtered \ac{RB} signal estimator obeys
 \begin{equation}
   \EE[f_\lambda^2]  = \sum_{\substack{\sigma\in\Irr(\omega) \\ \cap\Irr(\tau_\lambda^{\otimes 2})}} \tr\left( C_\sigma I_\sigma^m\right) +  \sum_{\sigma \in \Irr(\tau_\lambda^{\otimes 2})} \tr\left( D_\sigma O_\sigma^m\right)\, .
 \label{eq:data-form-second-moment}
 \end{equation}
 Here, $C_\sigma,I_\sigma\in \R^{n_\sigma \times n_\sigma }$, where $n_\sigma$ is the multiplicity of $\tau_\sigma$ in $\omega$.
The matrices $I_\sigma$ and $O_\sigma$ do not depend on the initial state and measurement.  For $\sigma\in\Irr(\omega) \cap\Irr(\tau_\lambda^{\otimes 2})$, $I_\sigma$ and $O_\sigma$ are the same as in Thm.~\ref{thm:rb-data-random-circuits}.
We have the bounds
 \begin{align}
  \spec_{\C}(I_\sigma) & \subset D_1(0)\cap D_{2\delta_\sigma}(1) \,, \\
  \snorm{O_\sigma} & \leq 
    \begin{cases}
        1 - \Delta_\sigma + 2\delta_\sigma, & \text{if } \sigma\in\Irr(\omega), \\
        1 - \Delta_\sigma + \delta_\sigma, & \text{else}.
    \end{cases}
 \end{align}
Define the quantities
\begin{align}
 \Delta_\lambda^{(3)} &\coloneqq  \min_{\sigma\in\Irr(\tau_\lambda^{\otimes 2})} \Delta_\sigma, &
 r_\lambda^{(3)} &\coloneqq  \max_{\sigma\in\Irr(\tau_\lambda^{\otimes 2})} \delta_\sigma/\Delta_\sigma, %&
\label{eq:sampling-complexity-max-bounds}
\end{align}
Then, we can bound the second sum of matrix exponentials as follows:
\begin{align}
 \abs[\bigg]{ \sum_{\sigma} \tr\left( D_\sigma O_\sigma^m\right) } 
 &\leq
 c_\lambda \snorm{S_\lambda^+} \osandwich{\rho}{P_\lambda}{\rho}\, g(r_\lambda^{(3)}) \\
 &\qquad \times \left(1 - \Delta_\lambda^{(3)} \left( 1 - 2 r_\lambda^{(3)} \right)\right)^m  \, ,
 \label{eq:sampling-complex-decay-bound}
\end{align}
where $c_\lambda$ is given in Thm.~\ref{thm:rb-data-random-circuits} and $g(x) \coloneqq  (1-4x)x + \frac{1+x}{1-4x}$.
\end{theorem}

The proof of the theorem is given in Sec.~\ref{sec:proof-thm-10}.
We proceed with a discussion of Thm.~\ref{thm:second-moment}. 

%--------------------------------------------------------
\subsubsection{Discussion of assumptions and dominant signal}
\label{sec:discussion-second-moment}
%--------------------------------------------------------

The similarities to Thm.~\ref{thm:rb-data-random-circuits} are imminent, hence we concentrate on the differences between the theorems for the first and second moment of $f_\lambda$.

\paragraph{Simplified assumptions.}
Instead of involving only a single irrep $\lambda$, Thm.~\ref{thm:second-moment} involves all irreps appearing in the tensor square $\tau_\lambda^{\otimes 2}$.
Note that we would find a similar situation in Thm.~\ref{thm:rb-data-random-circuits}, if we would not filter on irreps, but on reducible subrepresentations instead.
Theorem \ref{thm:second-moment} makes only irrep-specific assumptions as the perturbative expansion is done independently for every irrep $\sigma\in\Irr(\tau_\lambda^{\otimes 2})$.
In practice, it might be simpler to use $\sigma$-independent bounds.
To this end, the quantity $\Delta_\lambda^{(3)}$ introduced in Thm.~\ref{thm:second-moment} is helpful since it bounds the spectral gap of the third moment operator $\widehat{\omega\nu}[\tau_\lambda^{\otimes 2}]$:
\begin{align}
 1 - \Delta_\lambda^{(3)}
 &\geq \max_{\sigma \in \Irr(\tau_\lambda^{\otimes 2})} \snorm{ \widehat{\omega\nu}[\tau_\sigma] - \widehat{\omega}[\tau_\sigma] } \\
 &=
 \snorm{ \widehat{\omega\nu}[\tau_\lambda^{\otimes 2}] - \widehat{\omega}[\tau_\lambda^{\otimes 2}] } \, .
\end{align}
A sufficient condition for the assumptions of Thm.~\ref{thm:second-moment} is then $\delta_\lambda^{(3)}\leq\Delta_\lambda^{(3)}/4$, where
\begin{align}
 \delta_\lambda^{(3)} 
 &\geq
 \snorm{ \widehat{\phi\nu}[\tau_\lambda^{\otimes 2}] - \widehat{\omega\nu}[\tau_\lambda^{\otimes 2}] }\\
 &=
 \max_{\sigma \in \Irr(\tau_\lambda^{\otimes 2})} \snorm{ \widehat{\phi\nu}[\tau_\sigma] - \widehat{\omega\nu}[\tau_\sigma] } \, .
\end{align}
In this case, we have $r_\lambda^{(3)}\leq \delta_\lambda^{(3)} / \Delta_\lambda^{(3)} $ and $g(r_\lambda^{(3)})\leq g(\delta_\lambda^{(3)} / \Delta_\lambda^{(3)})$ since $g$ is monotonic.

\paragraph{Form of the dominant signal and the role of SPAM noise.}
In the following, we discuss the role of the matrix coefficients appearing in the dominant terms of the second moment \eqref{eq:data-form-second-moment}. 
Analogously to the \ac{SPAM} constants $\tr(A_\lambda)$ discussed in Sec.~\ref{sec:dominantSignal}, 
we derive in Sec.~\ref{sec:proof-thm-10}, Eq.~\eqref{eq:trace-Csigma} that 
\begin{align}
 \tr(C_\sigma) 
 &= \osandwichb{\rho^{\otimes 2}}%
   {Y_\sigma T_\sigma^+\,
   \widehat{\omega}[\omega^{(2)}_\sigma] \big( Y_\sigma^{\dagger} \tilde M_3 \big) }%
   {\tilde\rho} \, . 
\end{align}
Hence, $\tr(C_\sigma)$ is the $\sigma$-contribution to the \emph{second moment of the ideal, noiseless implementation $\phi=\omega$ with unitaries sampled from the Haar measure $\nu=\mu$, but subject to the same \ac{SPAM} noise}.
Consequently, the sum over $\sigma$ yields the total second moment:
\begin{align}
 \sum_{\sigma\in\Irr(\omega_\lambda^{\otimes 2})} \tr(C_\sigma) 
 &= 
 \osandwichb{\rho^{\otimes 2}}%
   {(X_\lambda S_\lambda^+)^{\otimes 2}\,
   \widehat{\omega}[\tau_\lambda^{\otimes 2}] \big( X_\lambda^{\dagger\otimes 2} \tilde M_3 \big) }%
   {\tilde\rho}\\
 &=: \EE[f_\lambda^2] _\SPAM\, .
\label{eq:sum-trace-Csigma}
\end{align}
Note that irreps $\sigma\in\Irr(\tau_\lambda^{\otimes 2})$ which are \emph{not} in $\Irr(\omega)$ do not contribute to the sum since $\widehat{\omega}[\omega^{(2)}_\sigma]=0$ in this case.

As in Sec.~\ref{sec:data-guarantees}, we would like to compare the second moment under only \ac{SPAM} noise $\EE[f_\lambda^2] _\SPAM$ to the ideal, noiseless second moment $\EE[f_\lambda^2] _\ideal$.
The discussion is made somewhat more complicated by the presence of multiplicites in $\tau_\lambda^{\otimes 2}$ even if all irreps in $\omega$ are multiplicity-free.
Consequently, the rank of the projector $\widehat{\omega}[\tau_\lambda^{\otimes 2}]$ is generally much larger than the rank of $\widehat{\omega}[\tau_\lambda]$ and given by the summed multiplicities of the irreps $\sigma\in\Irr(\tau_\lambda^{\otimes 2})\cap\Irr(\omega)$, c.f.~App.~\ref{sec:sampling_complexity_ideal_case}.

Similar to Sec.~\ref{sec:dominantSignal}, we conjecture that for any physically relevant setting, the effect of \ac{SPAM} noise is the reduction of the magnitude of the $\sigma$-contribution compared to the \ac{SPAM}-free case, $\abs{\tr(C_\sigma)} \leq \tr(C_\sigma)_\ideal$.
We are able to prove this under the assumption that $G$ contains the Heisenberg-Weyl group $\HW{n}{p}$.

\begin{proposition}
\label{prop:2nd-moment-vis-leq-one}
 Suppose that $\HW{n}{p} \subset G$, the measurement is in the computational basis, and $\rho$ is a computational basis state.
 Then we have 
 \begin{multline}
 \abs{\tr(C_\sigma)}
 \leq
 \tr(C_\sigma)_\ideal \\
 \coloneqq
 \osandwichb{\rho^{\otimes 2}}{P^{(2)}_\sigma (S^+)^{\otimes 2} \widehat{\omega}[\omega^{\otimes 2}]\big(M_3 \big) }{\rho} \, .
 \end{multline}
\end{proposition}

\begin{proof}
 We consider the matrix representation of $\widehat{\omega}[\omega^{\otimes 2}]$ in the orthogonal basis $X_{a,b,c}\coloneqq \oketbra{w(a)\otimes w(b)}{w(c)}$ of $\Hom(V,V\otimes V)$ for $a,b,c\in\F_p^{2n}$. 
 Note that $\twonorm{X_{a,b,c}}=d^{3/2}$. 
 Let $\widehat{\omega}[\omega^{\otimes 2}]|_{\HW{n}{p}}$ be the restriction of the Fourier transform to the Heisenberg-Weyl group.
 The range of $\widehat{\omega}[\omega^{\otimes 2}]|_{\HW{n}{p}}$ is exactly given by the subspace of $\HW{n}{p}$-equivariant maps, this is $X \mathcal{W}(a) = \mathcal{W}(a)^{\otimes 2} X$ for all $a\in\F_p^{2n}$, where $\mathcal{W}(a) = w(a)(\argdot)w(a)^\dagger$ is a unitary Weyl channel.
 It is straightforward to check that an orthonormal basis for this subspace is given by $X_{a,b,a+b}$ for $a,b\in\F_p^{2n}$.
 By the invariance of the Haar measure on $G$, $\widehat{\omega}[\omega^{\otimes 2}]$ is left and right invariant under the multiplication with $\widehat{\omega}[\omega^{\otimes 2}]|_{\HW{n}{p}}$, hence $\widehat{\omega}[\omega^{\otimes 2}]$ is diagonal in the $X_{a,b,a+b}$ basis.
 In particular, the diagonal entries are either 0 or 1 since $\widehat{\omega}[\omega^{\otimes 2}]$ is a projector.
 Hence, we find
 \begin{align}
  \MoveEqLeft[2]\widehat{\omega}[\omega^{\otimes 2}]\big(\tilde M_3 \big)\\
  &=
   \frac{1}{d^{3}}\sum_{a,b}
   \oket{X_{a,b,a+b}}
   \obraket{X_{a,b,a+b}}{\tilde M_3} \\
  &=
   \frac{1}{d^{3}} \sum_{a,b} \sum_{x\in\F_p^n} \oket{X_{a,b,a+b}} \\
  &\qquad \times
   \obraket{w(a)\otimes w(b)}{E_x \otimes E_x} \obraket{\tilde E_x}{w(a+b)} \\
  &=
   \frac{1}{d^{3}} \sum_{a,b}
   \delta_{a_x,0}\delta_{b_x,0} \oket{X_{a,b,a+b}} \\
  &\qquad \times
   \sum_{x\in\F_p^n} \xi^{-(a_z+b_z)\cdot x} \osandwich{E_x}{\EM}{Z(a_z+b_z)} \\
  &=
   \frac{1}{d^{3}}
   \sum_{z,z'} \oket{X_{z,z',z+z'}}\osandwich{Z(z+z')}{\EM}{Z(z+z')}  \, .
 \end{align}
 Here, the sums are over suitable index sets for $a,b$ and $z,z'$ that correspond to the support of $\widehat{\omega}[\omega^{\otimes 2}]$.
 In the second step, we used the definition of the Weyl operators, c.f.~Eq.~\eqref{eq:weyl-ops}, and the orthonormality of the computational basis to conclude that the contribution of Weyl operators with a non-vanishing $X$ component $a_x\neq 0$ is zero.
 Afterwards, we use of the definition $Z(z) = \sum_x \xi^{z\cdot x} \ketbra{x}{x}$.
 Note that we have $|\osandwich{Z(z+z')}{\EM}{Z(z+z')}| \leq d$ with equality in the noiseless case $\EM=\id$.
 Next, we consider the overlap of $X_{z,z',z+z'}$ with the remaining terms.
 To this end, we use $\rho=\ketbra{x}{x}$ for some $x$, and that $P^{(2)}_\sigma$ and $(S^+)^{\otimes 2}$ are positive semi-definite and diagonal in the Weyl basis.
 We find
 \begin{align}
  \MoveEqLeft[2]\abs[\big]{\osandwichb{\rho^{\otimes 2}}{P^{(2)}_\sigma (S^+)^{\otimes 2}}{Z(z)\otimes Z(z')} \obraket{Z(z+z')}{\tilde\rho}} \\
  &=
  \big(P^{(2)}_\sigma\big)_{z,z'} (S^+)_z (S^+)_{z'}  \\
  &\qquad \times \abs{\obraket{\rho^{\otimes 2}}{Z(z)\otimes Z(z')}}\, \abs{\obraket{Z(z+z')}{\tilde\rho}} \\
  &\leq
  \big(P^{(2)}_\sigma\big)_{z,z'} (S^+)_z (S^+)_{z'} \\
  &=
   \osandwichb{\rho^{\otimes 2}}{P^{(2)}_\sigma (S^+)^{\otimes 2}}{Z(z)\otimes Z(z')}\, \obraket{Z(z+z')}{\rho} \, .
 \end{align}
 Here, $(P^{(2)}_\sigma)_{z,z'}$ and $(S^+)_z$ are the diagonal entries of $P^{(2)}_\sigma$ and $S^+$.
 In the last step, we use that $\obraket{\rho^{\otimes 2}}{Z(z)\otimes Z(z')}\obraket{Z(z+z')}{\tilde\rho} = \xi^{(z+z')\cdot x} \xi^{-(z+z')\cdot x} = 1$.
 Finally, we obtain the desired result:
 \begin{align}
 \abs{\tr(C_\sigma)}
 &\leq
 \frac{1}{d^3} \sum_{z,z'}
  \big|\osandwichb{\rho^{\otimes 2}}{P^{(2)}_\sigma (S^+)^{\otimes 2}}{Z(z)\otimes Z(z')}\big| \\
 &\quad \times
  \big|\obraket{Z(z+z')}{\tilde\rho} \big|\, 
  \abs[\big]{\osandwich{Z(z+z')}{\EM}{Z(z+z')}}
  \\
 &\leq 
 \frac{1}{d^2} \sum_{z,z'}
 \osandwichb{\rho^{\otimes 2}}{P^{(2)}_\sigma (S^+)^{\otimes 2}}{Z(z)\otimes Z(z')} \\
 &\quad \times
  \obraket{Z(z+z')}{\rho} \\
 &=
 \tr(C_\sigma)_\ideal \, .
 \label{eq:tr-Csigma-ideal-weyl-basis}
 \end{align}
\end{proof}

Although, the contribution $\tr(C_\sigma)$ per irrep $\sigma$ can be split in a state preparation and a measurement term, we think that doing so does not bear the same level of insight as the treatise in Sec.~\ref{sec:dominantSignal}.
Thus, we simply define the \emph{second moment \ac{SPAM} visibility} as
\begin{equation}
 v_\SPAM^{(2)} \coloneqq \frac{\EE[f_\lambda^2] _\SPAM}{\EE[f_\lambda^2] _\ideal} \,.
 \label{eq:2nd-order-visibility}
\end{equation}
As a consequence of Prop.~\ref{prop:2nd-moment-vis-leq-one}, we have $v_\SPAM^{(2)} \leq 1$ for groups $G$ that contain a Heisenberg-Weyl group.

\paragraph{Examples.}
Let us again consider depolarizing \ac{SPAM} noise. 
That is, we assume $\tilde\rho = \ESP(\rho)$ and $\tilde M = M \EM$ with state preparation and measurement noise channels $\ESP=\EM = p\, \id + (1-p)\oketbra{\one}{\one}/d$.
Moreover, let us assume that $S$ is invertible.
Then, we find that 
\begin{align}
 \MoveEqLeft[1]\EE[f_\lambda^2] _\SPAM\\
 &=
 \osandwichb{\rho^{\otimes 2}}{P_\lambda^{\otimes 2} (S^{-1})^{\otimes 2} \widehat{\omega}[\omega^{\otimes 2}]\big(\tilde M_3 \big) }{\tilde\rho} \\
 &=
 \osandwichb{\rho^{\otimes 2}}{P_\lambda^{\otimes 2} (S^{-1})^{\otimes 2} \widehat{\omega}[\omega^{\otimes 2}]\big(M_3 \big) \EM\ESP}{\rho} \\
 &=
 p^2 \EE[f_\lambda^2] _\ideal \\
 &\quad
 +
 \frac{1-p^2}{d}
 \osandwichb{\rho^{\otimes 2}}{P_\lambda^{\otimes 2} (S^{-1})^{\otimes 2} \widehat{\omega}[\omega^{\otimes 2}]\big(M_3 \big)}{\one}\obraket{\one}{\rho} \\
 &=
 p^2 \EE[f_\lambda^2] _\ideal \\
 &\quad
 +
 \frac{1-p^2}{d}
 \sum_{i\in[d]} \int_G 
 \osandwichb{\rho}{P_\lambda S^{-1} \omega(g)}{E_i}^2
 \dd\mu(g) \,,
 \label{eq:2nd-moment-dep-noise}
\end{align}
where we used $\osandwich{E_i}{\omega(g)}{\one} = \obraket{E_i}{\one} = 1$ in the last step.
Using that the expressions are real-valued, we find
\begin{align}
 \MoveEqLeft[2]\sum_{i\in[d]} \int_G 
 \osandwichb{\rho}{P_\lambda S^{-1} \omega(g)}{E_i}^2
 \dd\mu(g) \\
 &=
 % \sum_{i\in[d]} \int_G 
 \osandwichb{\rho}{P_\lambda S^{-1} \omega(g)^\dagger%}{E_i} %\\
 % &\qquad\times
 % \osandwichb{E_i}{
 M
 \omega(g) S^{-1} P_\lambda}{\rho}
 \dd\mu(g) \\
 &=
 \osandwichb{\rho}{P_\lambda S^{-1} S S^{-1} P_\lambda}{\rho}  \\
 &=
 \osandwichb{\rho}{P_\lambda S^{-1}}{\rho} \,,
\end{align}
where we inserted the definition of $M$ and $S$, and used that $P_\lambda$ and $S$ commute.
Typically, the second term in Eq.~\eqref{eq:2nd-moment-dep-noise} is small. 
To see this, let us for simplicity assume that $\HW{n}{p} \subset G$, the measurement is in the computational basis, and $\rho$ is a computational basis state.
Then, $P_\lambda$ and $M$ are diagonal in the Weyl basis and commute.
Moreover, choose a decomposition of the $\lambda$-isotype such that $S_\lambda = \sum_i s_\lambda^{(i)} P_\lambda^{(i)}$  where $s_\lambda^{(i)} = \tr(P_\lambda^{(i)} M)/d_\lambda$, see Eq.~\eqref{eq:frame-op-constant}.
We find
\begin{align}
\frac{1}{d} \osandwichb{\rho}{P_\lambda S^{-1}}{\rho} 
&=
\frac{1}{d} \osandwichb{\rho}{P_\lambda M S^{-1}}{\rho} \\
&=
\frac{1}{d^2} \tr(P_\lambda M S^{-1}) \\
&=
\frac{1}{d^2} \sum_{i=1}^{n_\lambda} \frac{d_\lambda}{\tr(P_\lambda^{(i)} M)} \tr(P_\lambda^{(i)} M)\\
&=
\frac{n_\lambda d_\lambda}{d^2}\,,
\label{eq:tr-Plambda-M-Sinv}
\end{align}
where we used that $M(\rho)=\rho$ and $|\obraket{\rho}{Z(z)}|^2 = 1$.
Hence, the \ac{SPAM} visibility becomes
\begin{equation}
 v_\SPAM^{(2)} = p^2 + \frac{1-p^2}{\EE[f_\lambda^2] _\ideal} \frac{n_\lambda d_\lambda}{d^2} \, .
 \label{eq:2nd-moment-vis-example}
\end{equation}
Note that Prop.~\ref{prop:2nd-moment-vis-leq-one} implies the lower bound $\EE[f_\lambda^2] _\ideal \geq \frac{n_\lambda d_\lambda}{d^2}$.
We formulate this as a separate proposition.

\begin{proposition}[Second moment bounds]
\label{prop:2nd-moment-bounds}
 Suppose that $\HW{n}{p} \subset G$, the measurement is in the computational basis, and $\rho$ is a computational basis state.
 Then we have the following bounds on the ideal second moment.
 \begin{equation}
  \frac{k_\lambda d_\lambda}{d^2}
  \leq 
  \EE[f_\lambda^2] _\ideal 
  \leq
  \frac{(k_\lambda d_\lambda)^2}{d^2} \, .
 \end{equation}
 Here, $k_\lambda \leq n_\lambda$ is the number of distinct non-zero eigenvalues of $S_\lambda$.
 In particular if $S_\lambda$ is invertible, then $k_\lambda = n_\lambda$.
\end{proposition}

\begin{proof}
The first inequality follows from Prop.~\ref{prop:2nd-moment-vis-leq-one} and Eq.~\eqref{eq:2nd-moment-vis-example} after a straightforward modification to Eq.~\eqref{eq:tr-Plambda-M-Sinv} when $S_\lambda$ is not invertible.
Moreover, from Eqs.~\eqref{eq:tr-Csigma-ideal-weyl-basis} and \eqref{eq:tr-Plambda-M-Sinv} we find in a similar manner
\begin{align}
 \MoveEqLeft[1]\EE[f_\lambda^2] _\ideal\\
 &=
 \frac{1}{d^2} \sum_{z,z'}
 \osandwichb{\rho^{\otimes 2}}{(P_\lambda S^+)^{\otimes 2}}{Z(z)\otimes Z(z')} \obraket{Z(z+z')}{\rho} \\
 &=
 \frac{1}{d^2} \sum_{z,z'}
 \osandwichb{Z(z)\otimes Z(z')}{(P_\lambda M S^+)^{\otimes 2}}{Z(z)\otimes Z(z')} \\
 &\leq
 \frac{1}{d^2}
 \tr(P_\lambda M S^+)^2 
 =
 \frac{(k_\lambda d_\lambda)^2}{d^2}\,.
\end{align}
\end{proof}

For our typical examples from Sec.~\ref{sec:frame-operator}, we compute the exact second moments subject to \ac{SPAM} noise in App.~\ref{sec:sampling_complexity_ideal_case}.
In particular, we give the explicit \ac{SPAM} dependence for these examples.
Moreover, we find $\EE[f_\lambda^2] _\SPAM \leq \EE[f_\lambda^2] _\ideal$ in agreement with Prop.~\ref{prop:2nd-moment-vis-leq-one}, and that the lower bound in Prop.~\ref{prop:2nd-moment-bounds} for the ideal second moments is surprisingly tight.
The results are summarized in the following Prop.~\ref{prop:second-moments-ideal}.
Interestingly, the second moment of local unitary 3-groups is bounded by a constant if the local dimension $p$ is chosen as 2, i.e.~for qubits.

\begin{proposition}[Second moments of ideal implementation for typical examples]
\label{prop:second-moments-ideal}
For our typical examples from Sec.~\ref{sec:frame-operator}, namely unitary 3-groups, local unitary 3-groups, and the Heisenberg-Weyl group, we have $\EE[f_\lambda^2] _\SPAM\leq\EE[f_\lambda^2] _\ideal$ and the ideal second moments are given as
\begin{align}
 \EE[f_\Ad^2] _\ideal
 &\leq
 \begin{cases}
  1 - d^{-2} & d=2, \\
  3 - d^{-2} & d\geq 3,
 \end{cases}
 \tag*{(unitary 3-groups)}, \\
 \EE[f_b^2]_\ideal
 &\leq 
 \begin{cases}
  3^{n-|b|}/d^2 & p=2, \\
  3^{n-|b|} & p\geq 3,
 \end{cases}
 \tag*{(local unitary 3-groups)}, \\
 \EE[f_{0,z}^2]_\ideal
 &= 
 d^{-2}, 
 \tag*{(Heisenberg-Weyl groups)}.
\end{align}
\end{proposition}

\paragraph{Non-malicious SPAM noise.}
According to Thm.~\ref{thm:second-moment}, the matrix coefficients $C_\sigma$ are modulated by the matrices $I_\sigma$ in the presence of gate noise and non-uniform sampling.
A central problem in analyzing the resulting sampling complexity is that -- although the total second  moment is non-negative -- some of the contributions $\tr(C_\sigma)$ per irrep might be \emph{negative}.
The explicit examples in App.~\ref{sec:sampling_complexity_ideal_case} show that this can in fact happen (e.g.~if $G$ is a unitary 3-design), as a consequence of malicious (i.e.~fine-tuned and thus unrealistic) \ac{SPAM} noise.
Indeed, a negative contribution requires that either the measurement noise introduces permutations of the measurement outcomes (thus rendering the measurement useless) or that we accidentally prepare a state $\tilde\rho$ which has vanishing fidelity with $\rho$.

Proposition \ref{prop:2nd-moment-vis-leq-one} already shows that the SPAM-free contributions $\tr(C_\sigma)_\ideal$ are non-negative under the assumptions that $\HW{n}{p}\subset G$.
The following lemma shows the same statement in a more general setting. 

\begin{proposition}[Non-negativity of $\tr(C_\sigma)$ in SPAM-free case]
\label{prop:Csigma}
 Suppose that the measurement basis $E_i = \ketbra{i}{i}$ can be generated with gates from $G$ acting on $E_1$, and assume that the initial state is $\rho = E_1$.
 Then, for any irrep $\sigma\in\Irr(\omega_\lambda^{\otimes 2})$, we have 
 \begin{equation}
 \tr(C_\sigma)_{\ideal}
 =
 \osandwichb{\rho^{\otimes 2}}{P^{(2)}_\sigma (S^+)^{\otimes 2} \widehat{\omega}[\omega^{\otimes 2}]\big(M_3 \big) }{\rho}{}
 \geq 0.
 \end{equation}
 Here, $P^{(2)}_\sigma$ is the projector onto the $\sigma$-isotype of $\omega_\lambda^{\otimes 2}$.
\end{proposition}

\begin{proof}
 Since the measurement basis is generated by gates from $G$, we can write
 \begin{align}
  \MoveEqLeft[1]\osandwichb{\rho^{\otimes 2}}{P^{(2)}_\sigma(S^+)^{\otimes 2} \widehat{\omega}[\omega^{\otimes 2}]\big(M_3 \big) }{\rho}\\
  &=
  d
  \int_G 
  \osandwichb{\rho^{\otimes 2}}{P^{(2)}_\sigma (S^+)^{\otimes 2}\, \omega^{\otimes 2}(g)^\dagger}{E_1^{\otimes 2}} \\
  &\qquad\times
  \osandwich{E_1}{\omega(g)}{\rho} 
  \dd\mu(g) \nonumber \\
  &=
  d
  \int_G 
  \osandwichb{\rho^{\otimes 3}}{ \big(P^{(2)}_\sigma(S^+)^{\otimes 2}\otimes\id\big) \, \omega^{\otimes 3}(g)^\dagger}{E_1^{\otimes 3}}
  \dd\mu(g) \nonumber \\
  &=
  d\,
  \osandwichb{\rho^{\otimes 3}}{ \big(P^{(2)}_\sigma(S^+)^{\otimes 2}\otimes\id\big) P_1^{(3)}}{E_1^{\otimes 3}}\, ,
  \label{eq:lem-Csigma-1}
 \end{align}
 where $P_1^{(3)}$ is the projector onto the trivial isotype in $\omega^{\otimes 3}$.
 Recall that $S^+$ is in the commutant of $\omega$, and thus $(S^+)^{\otimes 2}$ commutes with $P_\sigma^{(2)}$ as it projects onto a subrepresentation of $\omega^{\otimes 2}$.
 Since both commute with $\omega^{\otimes 2}(g)$ for all $g\in G$, $(P^{(2)}_\sigma(S^+)^{\otimes 2}\otimes\id)$ commutes with $P_1^{(3)}$, hence the product $(P^{(2)}_\sigma(S^+)^{\otimes 2}\otimes\id) P_1^{(3)}$ is a positive semidefinite operator.
 By assumption, $\rho = E_1$, and thus Eq.~\eqref{eq:lem-Csigma-1} is non-negative as claimed.
\end{proof}

Thus, negative contributions $\tr(C_\sigma)$ can usually only occur as a consequence of malicious noise.
In particular, if the \ac{SPAM} noise is sufficiently small, it can be guaranteed that $\tr(C_\sigma)$ remains non-negative.
Although it is straightforward to derive sufficient bounds on the strength of the \ac{SPAM} noise for this purpose, this would not add much to the discussion. 
Instead, we add the non-negativity of $\tr(C_\sigma)$ as an assumption and define to this end:

\begin{definition}[Non-malicious \ac{SPAM} noise]
\label{def:non-malicious-SPAM}
 We call the \ac{SPAM} noise \emph{non-malicious} if $\tr(C_\sigma)\geq 0$ for all $\sigma\in\Irr(\tau_\lambda^{\otimes 2})$ .
\end{definition}

% -------------------------------------------
\subsubsection{Guarantees for sampling complexity}
\label{sec:guarantees-sampling-complexity}
% -------------------------------------------

Finally, we prove guarantees for the sampling complexity of filtered randomized benchmarking.
To this end, we assume that the subdominant contributions to $\FD_\lambda(m)$ and $\EE[f_\lambda^2] $ are bounded and combine them with Chebyshev's inequality.
Suitable error bounds that accomplish this are later derived in Sec.~\ref{sec:sequence-lengths}, and formulated as Lem.~\ref{lem:sequence-length}, \ref{lem:sequence-length-second-moment}, and \ref{lem:relative-sequence-length-second-moment} 

In the following, we assume for simplicity that all relevant irreps $\sigma\in\Irr(\omega)\cap\Irr(\tau_\lambda^{\otimes 2})$ are multiplicity-free in $\omega$.
Note that this is certainly fulfilled if \emph{all} irreps of $\omega$ are multiplicity-free -- which is the case for most relevant examples.

\begin{theorem}[Sampling complexity of filtered \ac{RB} -- additive precision]
\label{thm:sampling-complexity-additive}
Fix a non-trivial irrep $\lambda\in\Irr(\omega)$ such that $\lambda$ and all $\sigma\in\Irr(\omega)\cap\Irr(\tau_\lambda^{\otimes 2})$ are multiplicity-free in $\omega$.
Suppose that we have non-malicious \ac{SPAM} noise, and the assumptions of 
% Thm.~\ref{thm:rb-data-random-circuits} and
Thm.~\ref{thm:second-moment} are fulfilled.
Moreover, let the sequence length $m$ be sufficiently large such that the subdominant terms in 
$\EE[f_\lambda^2] $ (see Eq.~\eqref{eq:data-form-second-moment}) are bounded by an additive error $\alpha > 0$.
Then, the mean estimator $\hat{\FD}_\lambda(m)$ for $N$ samples, cf.~Eq.~\eqref{eq:FD-estimator}, is close to the expected value with high probability,
\begin{equation}
\Pr\left[ \left|\hat{\FD}_\lambda(m) - \FD_\lambda(m) \right| > \epsilon  \right] \leq \delta \, ,
\end{equation}
provided that 
\begin{align}
 N 
 &\geq 
 \frac{1}{\varepsilon^2 \delta}
 \left( \EE[f_\lambda^2] _\SPAM + \alpha \right) \, .
\end{align}
Here, $\EE[f_\lambda^2] _\SPAM$ is the second moment of the filter function for the ideal, noiseless implementation $\phi=\omega$ where unitaries are sampled from the Haar measure on $G$, but with the same \ac{SPAM} noise, c.f.~Eq.~\eqref{eq:sum-trace-Csigma}.
\end{theorem}

\begin{proof}
Chebyshev's inequality guarantees that $|\hat{\FD}_\lambda(m) - \FD_\lambda(m) | \leq \epsilon $ with probability at most $1-\delta$, provided that the number of samples $N$ fulfills
\begin{equation}
N 
\geq 
\frac{\Var[f_\lambda]}{\varepsilon^2 \delta} 
= 
\frac{1}{\varepsilon^2 \delta}
\left( \EE[f_\lambda^2] - \FD_\lambda(m)^2 \right)
\end{equation}
Discarding $\FD_\lambda(m)^2$ will only make the right hand side larger, thus we can concentrate on bounding the second moment.
Under the made assumptions, we find using Thm.~\ref{thm:second-moment} that 
\begin{align}
 \EE[f_\lambda^2] 
 &\leq
 \sum_\sigma
 \tr( C_\sigma )  I_\sigma^m + \alpha \\
 &\leq
 \sum_\sigma
 \tr( C_\sigma ) + \alpha
 =
 \EE[f_\lambda^2]_\SPAM + \alpha \,, 
 \end{align}
 where the sums are taken over $\sigma\in\Irr(\omega)\cap\Irr(\tau_\lambda^{\otimes 2})$.
 For the second and third step, we used that $I_\sigma \leq 1$ and $\tr(C_\sigma)\geq 0$ for all $\sigma$, as well as $\sum_\sigma \tr(C_\sigma) = \EE[f_\lambda^2] _\SPAM$ by Eq.~\eqref{eq:sum-trace-Csigma}.
\end{proof}

Note that the variance $\Var[f_\lambda]$, which we bounded in the proof of the sampling complexity theorem \ref{thm:sampling-complexity-additive}, is in fact decaying with $m$, and so is $N$.
However, the reason for this is the simple fact that the quantity to be estimated, $\FD_\lambda(m)$, is decaying with $m$, too.
Hence, we would have to reduce the error $\varepsilon$ with increasing $m$ to get meaningful estimates.
Therefore, we give a sampling complexity guarantee with relative error in the following.

\begin{theorem}[Sampling complexity of filtered \ac{RB} -- relative precision]
\label{thm:sampling-complexity-relative}
Fix a non-trivial irrep $\lambda\in\Irr(\omega)$ such that $\lambda$ and all $\sigma\in\Irr(\omega)\cap\Irr(\tau_\lambda^{\otimes 2})$ are multiplicity-free in $\omega$.
Suppose that we have non-malicious SPAM noise, and the assumptions of Thm.~\ref{thm:rb-data-random-circuits} and Thm.~\ref{thm:second-moment} are fulfilled.
Moreover, let the sequence length $m$ be sufficiently large such that the subdominant terms in $\FD_\lambda(m)$ and $\EE[f_\lambda^2] $ are bounded by relative errors $\gamma > 0$ and $\kappa > 0$, respectively. 
Then, the mean estimator $\hat{\FD}_\lambda(m)$ for $N$ samples, cf.~Eq.~\eqref{eq:FD-estimator}, is close to the expected value with high probability,
\begin{equation}
\Pr\left[ \left|\hat{\FD}_\lambda(m) - \FD_\lambda(m) \right| > \varepsilon  \FD_\lambda(m) \right] \leq \delta,
\end{equation}
provided that 
\begin{align}
 N 
 &\geq 
 \frac{1}{\varepsilon^2 \delta}
 \left(
 \frac{(1+\kappa)}{(1-\gamma)^2}
 \frac{\EE[f_\lambda^2] _\SPAM}{\FD_\lambda(m)_\SPAM^2}
 I_\lambda^{-2m}
 - 1 \right) \, .
\end{align}
Here, $\FD_\lambda(m)_\SPAM$ and $\EE[f_\lambda^2] _\SPAM$ are the first and second moment of the filter function, respectively, for the ideal, noiseless implementation $\phi=\omega$ where unitaries are sampled from the Haar measure on $G$, but with the same SPAM noise, c.f.~Eqs.~\eqref{eq:Flambda-SPAM} and \eqref{eq:sum-trace-Csigma}.
\end{theorem}

\begin{proof}
Chebyshev's inequality guarantees that $|\hat{\FD}_\lambda(m) - \FD_\lambda(m) | \leq \varepsilon \FD_\lambda(m) $ with probability at most $1-\delta$, provided that the number of samples $N$ fulfills
\begin{equation}
N 
\geq 
\frac{\Var[f_\lambda]}{\FD_\lambda(m)^2 \varepsilon^2 \delta} 
= 
\frac{1}{\varepsilon^2 \delta}
\left(  \frac{\EE[f_\lambda^2] }{\FD_\lambda(m)^2} - 1 \right)
\end{equation}
Theorems \ref{thm:rb-data-random-circuits} and \ref{thm:second-moment} then give
\begin{align}
 \frac{\EE[f_\lambda^2] }{\FD_\lambda(m)^2}
 &\leq
 \frac{(1+\kappa)\sum_{\sigma} \tr( C_\sigma )  I_\sigma^m  }{(1-\gamma)^2 \tr(A_\lambda)^2 I_\lambda^{2m}} \\
 &\leq
 % \frac{(1+\kappa)}{(1-\gamma)^2}
 % \frac{\EE[f_\lambda^2] _\SPAM}{\FD_\lambda(m)_\SPAM^2}  
 % \left(\frac{I_\mathrm{max}}{I_\lambda^2}\right)^m \,, 
 \frac{(1+\kappa)}{(1-\gamma)^2}
 \frac{1}{\FD_\lambda(m)_\SPAM^2}  
 \sum_{\sigma}
 \tr(C_\sigma)
 \left(\frac{I_\sigma}{I_\lambda^2}\right)^m , 
 \end{align}
 where the sums are taken over $\sigma\in\Irr(\omega)\cap\Irr(\tau_\lambda^{\otimes 2})$.
 The claim then follows as in the proof of Thm.~\ref{thm:sampling-complexity-additive}.
\end{proof}

% -------------------------------------
\subsubsection{Proof of Theorem~\ref{thm:second-moment}}
\label{sec:proof-thm-10}

Since $\omega_\lambda$ is a $\tau_\lambda$-isotype, $\Irr(\omega_\lambda^{\otimes 2}) = \Irr(\tau_\lambda^{\otimes 2})$.
For $\sigma \in \Irr(\tau_\sigma^{\otimes 2})$, let 
$m_\sigma$ denote the multiplicity of $\tau_\sigma$ in $\omega_\lambda^{\otimes 2}$.
Starting from the decomposition given in Eq.~\eqref{eq:RB-data-second-moment-decomp}, we treat the Fourier operators $\widehat{\phi\nu}[\omega^{(2)}_\sigma]\simeq\widehat{\phi\nu}[\tau_{\sigma}]\otimes\id_{m_\sigma}$ independently.
As in the proof of Thm.~\ref{thm:rb-data-random-circuits}, we write
\begin{multline}
  \osandwichb{\rho^{\otimes 2}}%
   {Y_\sigma T_\sigma^+\,
   \widehat{\phi\nu}[\omega^{(2)}_\sigma]^m \big( Y_\sigma^{\dagger} \tilde M_3 \big) }%
   {\tilde\rho} \\
  = \tr\left[ \left(\widehat{\phi\nu}[\tau_{\sigma}]\otimes\id_{m_\sigma}\right)^m \oketbra{\tilde M_{3,\sigma}}{Q_{3,\sigma}} \right] \,,
\end{multline}
for suitable superoperators $\tilde M_{3,\sigma}$ and $Q_{3,\sigma}$.
Moreover, we can again restrict all operators to their action on Hermitian matrices, and thus consider them as real operators.

First, let us consider the case when the irrep $\sigma\in\Irr(\tau_\lambda^{\otimes 2})$ is \emph{not} contained in $\omega$.
Then, the Haar moment operator $\widehat{\omega}[\tau_\sigma]$ is identically zero, so $\widehat{\omega\nu}[\tau_\sigma]$ cannot have an eigenvalue 1.
In this case, we may not invoke perturbation theory, but we can simply set 
\begin{align}
 O_\sigma &\coloneqq  \widehat{\phi\nu}[\tau_{\sigma}], & D_\sigma &\coloneqq  \tr_{m_\sigma}(\oketbra{\tilde M_{3,\sigma}}{Q_{3,\sigma}}).
\end{align}
Hence, we have $\snorm{O_\sigma} \leq \snorm{\widehat{\omega\nu}[\tau_\sigma]} + \delta_\sigma \leq 1 - \Delta_\sigma + \delta_\sigma$.
Analogous to the proof of Thm.~\ref{thm:rb-data-random-circuits}, we then find
\begin{equation}
 \abs[\big]{ \tr\left( D_\sigma O_\sigma^m\right) } 
 \leq  
 \twonorm{Y_\sigma^\dagger \tilde M_3}\twonormb{T_\sigma^+ Y_\sigma \oketbra{\rho^{\otimes 2}}{\tilde\rho}} \snorm{O_\sigma}^m  .
 \label{eq:thm-second-moment-0}
\end{equation}

If $\sigma\in\Irr(\tau_\lambda^{\otimes 2})\cap\Irr(\omega)$, we proceed analogously to the proof of Thm.~\ref{thm:rb-data-random-circuits} and apply perturbation theory to $\widehat{\phi\nu}[\tau_{\sigma}]$ with parameters $(\delta_\sigma,\Delta_\sigma)$.
This results in 
\begin{multline}
  \osandwichb{\rho^{\otimes 2}}%
   {Y_\sigma T_\sigma^+\,
   \widehat{\phi\nu}[\omega^{(2)}_\sigma]^m \big( Y_\sigma^{\dagger} \tilde M_3 \big) }%
   {\tilde\rho} \\
  = \tr\left[ C_\sigma I_\sigma^m \right] + \tr\left[ D_\sigma O_\sigma^m \right],
\end{multline}
where $I_\sigma \in \RR^{n_\sigma \times n_\sigma}$ with $n_\sigma$ the multiplicity of $\tau_\sigma$ in $\omega$ and $C_\sigma = L_{\lambda,1}^\dagger \tr_{m_\sigma}(\oketbra{\tilde M_{3,\sigma}}{Q_{3,\sigma}}) R_{\lambda,1}$ and $D_\sigma = L_{\lambda,2}^\dagger \tr_{_\sigma}(\oketbra{\tilde M_{3,\sigma}}{Q_{3,\sigma}}) R_{\lambda,2}$.
Using $R_{\lambda,1}L_{\lambda,1}^\dagger = \widehat{\omega}[\tau_\sigma]$, we find in particular that
\begin{align}
 \tr(C_\sigma) 
 &= \tr( \oketbra{\tilde M_{3,\sigma}}{Q_{3,\sigma}}\, \widehat{\omega}[\tau_\sigma]\otimes\id_{m_\sigma} ) \\
 &= \osandwichb{\rho^{\otimes 2}}%
   {Y_\sigma T_\sigma^+\,
   \widehat{\omega}[\omega^{(2)}_\sigma] \big( Y_\sigma^{\dagger} \tilde M_3 \big) }%
   {\tilde\rho} \, .
 \label{eq:trace-Csigma}
\end{align}
Moreover, we have the following bound:
\begin{multline}
 \abs[\big]{ \tr\left( D_\sigma O_\sigma^m\right) } \\
 \leq  
 g(\delta_\sigma/\Delta_\sigma) \,
 \twonorm{Y_\sigma^\dagger \tilde M_3}\twonorm{T_\sigma^+ Y_\sigma \oketbra{\rho^{\otimes 2}}{\tilde\rho}} \snorm{O_\sigma}^m \, .
 \label{eq:thm-second-moment-1}
\end{multline}

Finally, we derive the claimed bound on the sum $\sum_{\sigma\in\Irr(\tau_\lambda^{\otimes 2})} \tr\left( D_\sigma O_\sigma^m\right)$.
To this end, we use Eqs.~\eqref{eq:thm-second-moment-0} and \eqref{eq:thm-second-moment-1} for $\sigma \notin \Irr(\omega)$ and $\sigma \in \Irr(\omega)$, respectively.
The bounds are almost identical, except for the appearance of the factor $g(\delta_\sigma/\Delta_\sigma)$ in Eq.~\eqref{eq:thm-second-moment-1}.
However, $g(x) \geq 1$ for all $x\geq 0$, hence the prefactor is uniformly bounded by
\begin{align}
 \max_{\sigma\in\Irr(\tau_\lambda^{\otimes 2})\cap\Irr(\omega)}  g(\delta_\sigma/\Delta_\sigma) 
 &= g\Big( \max_{\sigma\in\Irr(\tau_\lambda^{\otimes 2})\cap\Irr(\omega)}  \delta_\sigma/\Delta_\sigma\Big) \\
 &\leq g\Big( \max_{\sigma\in\Irr(\tau_\lambda^{\otimes 2})}  \delta_\sigma/\Delta_\sigma\Big) \\
 &=  g(r_\lambda^{(3)})\,,
\end{align}
where we additionally used the monotonicity of $g$.
Furthermore, we find the uniform bounds:
\begin{align}
 \twonorm{ T_\sigma^+ Y_\sigma \oketbra{\rho^{\otimes 2}}{\tilde\rho} } 
 &\leq 
 \snorm{T_\sigma^+} \sqrt{\osandwich{\rho^{\otimes 2}}{P_\sigma}{\rho^{\otimes 2}}} \\
 &\leq
 \snorm{S_\lambda^+}^2 \osandwich{\rho}{P_\lambda}{\rho}\,, \\
 \max_{\sigma\in\Irr(\tau_\lambda^{\otimes 2})} \snorm{O_\sigma}
 &\leq
 1 - \min_{\sigma\in\Irr(\tau_\lambda^{\otimes 2})} \Delta_\sigma \left( 1 - 2 \frac{\delta_\sigma}{\Delta_\sigma} \right) \\
 &\leq 
 1 - \Delta_\lambda^{(3)} \left( 1 - 2 r_\lambda^{(3)} \right).
\end{align}
Next, we compute, using the concavity of the square root:
\begin{align}
\MoveEqLeft[1]\bigg(\sum_{\sigma\in\Irr(\tau_\lambda^{\otimes 2})} \twonorm{Y_\sigma^\dagger \tilde M_3} \bigg)^2\\
&= 
\bigg(\sum_{\sigma\in\Irr(\tau_\lambda^{\otimes 2})} \sqrt{\tr\left( P_\sigma M_3 \EM \EM^\dagger M_3^\dagger \right)} \bigg)^2 \\
&\leq 
\sum_{\sigma\in\Irr(\tau_\lambda^{\otimes 2})}  \sum_{i,j\in[d]} \osandwich{E_i}{\EM\EM^\dagger}{E_j} \osandwich{E_j\otimes E_j}{P_\sigma}{E_i\otimes E_i} \\
&=
\sum_{i,j\in[d]} \osandwich{E_i}{\EM\EM^\dagger}{E_j} \osandwich{E_j}{P_\lambda}{E_i}^2 \\
&\leq
\twonorm{X_\lambda^\dagger \tilde M}^2 \, .
\end{align}
Here, we have used 
$\sum_\sigma P_\sigma = P_\lambda^{\otimes 2}$, 
$\osandwich{E_i}{\EM\EM^\dagger}{E_j}\geq 0$ since $\EM$ and $\EM^\dagger$ are completely positive, 
$\osandwich{E_i}{P_\lambda}{E_j}^2\leq \osandwich{E_i}{P_\lambda}{E_j}$, and inserted the expansion as in Eq.~\eqref{eq:main-thm-proof-1}.
Hence, we can use the appropriate bounds from the proof of Thm.~\ref{thm:rb-data-random-circuits} for this sum.

Using the definition of $c_\lambda$ from Thm.~\ref{thm:rb-data-random-circuits} and combining the above bounds, we then find the claimed bound:
\begin{align}
 \abs[\bigg]{ \sum_{\sigma \in \Irr(\tau_\lambda^{\otimes 2})} \tr\left( D_\sigma O_\sigma^m\right) }
 &\leq
  c_\lambda \snorm{S_\lambda^+} \osandwich{\rho}{P_\lambda}{\rho} \, g(r_\lambda^{(3)}) \\
 &\qquad\times
  \left(1 - \Delta_\lambda^{(3)} \left( 1 - 2 r_\lambda^{(3)} \right)\right)^m  .
\end{align}

% -------------------------------------
\subsubsection{Sampling complexity of filtering onto trivial irrep}

Here, things are slightly simpler:
The filter function takes only the values 0 and 1, and hence Hoeffding's inequality guarantees sample efficiency.

% -------------------------------------
\subsection{Sufficient sequence lengths}
\label{sec:sequence-lengths}
% -------------------------------------

One of the main implications of Theorem~\ref{thm:rb-data-random-circuits} is 
that we can control the ratio between the dominant signal and the subdominant signal in \eqref{eq:thm:signal-form}.
Since the subdominant signal is suppressed with increasing sequence length $m$ according to \eqref{eq:thm:subdominant-decay-bound}, it is sufficient to choose $m$ large enough to be able to accurately extract the decay parameters.
In the following, we derive and discuss a corresponding lower bound on the sequence length.
Let us stress that using too short sequences may involve the danger of \emph{overestimating decay rates} and thus -- if interpreted as average gate fidelities -- of reporting \emph{too large gate fidelities}.
Similar concerns have already been raised in the non-uniform \ac{RB} literature \cite{boone2019randomized,Proctor2019DirectRandomized}.

We focus here on generally applicable bounds.
In particular, we do not make use of any details about the noise in the implementation map, the measure, or even the representation theory of the group $G$.
We expect that for specific settings, a more refined analysis using more assumptions on the noise and a specific measure and group yields improved the bounds.

Furthermore, the required relative suppression might be relaxed using a more sophisticated data processing.
The suitability and limitations of methods like ESPRIT \cite{roy_estimation_1986} for randomized benchmarking has been previously discussed in Ref.~\cite{helsen_general_2022}.
Note however that the performance of these methods also crucially depend on the separation of poles on the real axis and the total number of poles.
We leave the study of their applicability to future work.

\subsubsection{Exposing the dominant signal}
\label{sec:subdominant-decay}

If the desired suppression of the subdominant signal is $\alpha$, then the bound \eqref{eq:thm:subdominant-decay-bound} in Thm.~\ref{thm:rb-data-random-circuits} yields the following lower bound on the sequence length $m$:
\begin{equation}
  m \geq 
  \frac{ \log c_\lambda + \frac12  \log\,\osandwich{\rho}{P_\lambda}{\rho} + \log g\big(\frac{\delta_\lambda}{\Delta_\lambda}\big) + \log \frac1\alpha}{\log\frac{1}{1-\Delta_\lambda + 2 \delta_\lambda}}.
  \label{eq:sequence-length-bound}
\end{equation}
The right hand side of Eq.~\eqref{eq:sequence-length-bound} depends heavily on the group $G$ and the irrep $\tau_\lambda$ on which we filter, on the spectral gap $\Delta_\lambda$, as well as the SPAM noise and implementation error $\delta$.

The bound \eqref{eq:sequence-length-bound} only guarantees the suppression by an \emph{additive error} $\alpha$ which is fine as long as $\FD_\lambda(m) = O(1)$.
However, this is not the case for two reasons.
First, $\FD_\lambda(m)$ can be (exponentially) small for small irreps, as it ideally measures the overlap of the initial state with the irrep.
Second, $\FD_\lambda(m)$ decays with $m$ and thus the error should decrease with $m$, too.
Hence, a sensible solution is to require that the subdominant terms are instead suppressed by a \emph{relative error} $\gamma$.
The derivation of an analogous bound to \eqref{eq:sequence-length-bound} is postponed to the proof of the following Lemma.

\begin{lemma}[Sequence length bounds for extraction of decay parameters]
\label{lem:sequence-length}
For any non-trivial irrep $\lambda \in \Irr(\omega)$, the subdominant decay in Thm.~\ref{thm:rb-data-random-circuits} is bounded by $\alpha > 0$ provided that
\begin{multline}
 m \geq 
 \frac{\Delta_\lambda^{-1}}{1 - 2 \frac{\delta_\lambda}{\Delta_\lambda}}
 \Big( \log c_\lambda + \tfrac12  \log\, \osandwich{\rho}{P_\lambda}{\rho}\\
  + \log g\big(\tfrac{\delta_\lambda}{\Delta_\lambda}\big) + \log \tfrac1\alpha \Big).
 \label{eq:sequence-length-bound-2}
\end{multline}
Moreover, if $\lambda$ is multiplicity-free, the subdominant decay is smaller than $\gamma\, |\tr(A_\lambda I_\lambda^m)|$ for $\gamma > 0$ provided that 
\begin{multline}
 m \geq 
 \frac{\Delta_\lambda^{-1}}{1 - 4 \frac{\delta_\lambda}{\Delta_\lambda}}
 \Big( \log c_\lambda + \tfrac12  \log\,\osandwich{\rho}{P_\lambda}{\rho} + \log g\big(\tfrac{\delta_\lambda}{\Delta_\lambda}\big) \\
  + \log\big(|\FD_\lambda(m)_\SPAM|^{-1}\big)  + \log \tfrac1\gamma \Big) \, .
  \label{eq:relative-sequence-length-bound}
\end{multline}
\end{lemma}
Here, $\FD_\lambda(m)_\SPAM$ denotes again the filtered \ac{RB} signal that one would obtain for perfect unitaries sampled from the Haar measure on $G$, but subject to SPAM noise.

\begin{proof}
The first result follows directly by using the bound $\log(1+x) \leq x$ for $x > -1$ in Eq.~\eqref{eq:sequence-length-bound}.
Next, let us assume that $\lambda$ is multiplicity-free.
Then, we can rewrite the first term in $\FD_\lambda(m)$ using Thm.~\ref{thm:rb-data-random-circuits} and Eq.~\eqref{eq:Flambda-SPAM} as follows:
\begin{align}
 \abs{\tr(A_\lambda I_\lambda^m)}
 &= |\FD_\lambda(m)_\SPAM| \, I_\lambda^m  \\
 &>
 |\FD_\lambda(m)_\SPAM| \, (1 - 2\delta_\lambda)^m.
\end{align}
We then find the following bound for the relative suppression:
\begin{multline}
 \frac{\abs[\big]{\FD_\lambda(m) - \tr\left( A_\lambda I_\lambda^m\right) }}{\abs{\tr\left( A_\lambda I_\lambda^m\right)}}
 <
 \frac{c_\lambda \sqrt{\osandwich{\rho}{P_\lambda}{\rho}}  \, g\big(\tfrac{\delta_\lambda}{\Delta_\lambda}\big)}{|\FD_\lambda(m)_\SPAM|} \\
 \times
 \left(\frac{1-\Delta_\lambda+2\delta_\lambda}{1 - 2\delta_\lambda}\right)^m .
 \label{eq:lemma-seq-length-relative-suppression}
\end{multline}
We use the following inequality which is valid for all $\delta_\lambda/\Delta_\lambda \in [0,1/4] $ and $\Delta\in [0,1)$:
\begin{align}
 \log\left(\frac{1-\Delta_\lambda+2\delta_\lambda}{1-2\delta_\lambda}\right) 
 &\leq 
 \log(1-\Delta_\lambda) \left(1 - 4\frac{\delta_\lambda}{\Delta_\lambda}\right) \\
 &\leq
  - \left(1 - 4\frac{\delta_\lambda}{\Delta_\lambda}\right)\Delta_\lambda\, .
\label{eq:log-bound-relative-error}
\end{align}
The first inequality follows since the left hand side is strictly monotonically increasing in $\delta_\lambda$ for any $\Delta_\lambda\in [0,1)$.
Thus, it is a convex function which is upper bounded in the interval $[0,\Delta_\lambda/4]$ by a straight line.
Requiring that Eq.~\eqref{eq:lemma-seq-length-relative-suppression} is less than $\gamma>0$, and using Eq.~\eqref{eq:log-bound-relative-error} yields the claimed bound on $m$.
\end{proof}

We proceed by discussing the individual contributions to the bounds in Lem.~\ref{lem:sequence-length}.
In general, we conjecture that $\snorm{S_\lambda^+} = O(\poly(d_\lambda))$ such that we have $\log c_\lambda = O(\log(d_\lambda))$.
Hence, we expect that the sequence length scales at least with $\log(d_\lambda)$ in the worst case.
By again specializing $\lambda$ to be multiplicity-free and aligned with $M$, we can make this more precise:
By Eqs.~\eqref{eq:lower-bound-Slambda} and \eqref{eq:upper-bound-Slambda}, we have $d_\lambda \leq d_\lambda/s_\lambda \leq d_\lambda^2$, hence we find $\log(c_\lambda)=\frac12\log(d_\lambda/s_\lambda)=\Theta(\log d_\lambda)$.
Furthermore, we have $\delta_\lambda/\Delta_\lambda < 1/4$ in the perturbative regime, and hence the prefactor in Eq.~\eqref{eq:sequence-length-bound-2} cannot exceed $2/\Delta_\lambda$.
As discussed in Sec.~\ref{sec:implementation-quality}, $g(x)$ diverges for $x\rightarrow\infty$, and we thus have to assume that the implementation error is actually bounded away from $1/4$, say $\delta_\lambda/\Delta_\lambda\leq 1/5$ such that $\log g(1/5)\approx 1.798\leq 1.8$.
Under this assumption, we can also bound the otherwise diverging prefactor in Eq.~\eqref{eq:relative-sequence-length-bound} as $5/\Delta_\lambda$.
Finally, the contribution by $\osandwich{\rho}{P_\lambda}{\rho}\leq 1$ is negative and may counter the effect of $\log(d_\lambda)$ in certain regimes.
Under the above assumptions, we can then bring Eq.~\eqref{eq:sequence-length-bound-2} into the simplified form
\begin{equation}
 m 
 \geq \frac{2}{\Delta_\lambda} \left( \log d_\lambda + \tfrac12  \log\,\osandwich{\rho}{P_\lambda}{\rho} + \log \tfrac1\alpha + 1.8 \right). 
 \label{eq:sequence-length-bound-2-simplified}
\end{equation}

The same simplification applies to the relative error bound \eqref{eq:relative-sequence-length-bound}.
In this case, we have to  additionally bound the term depending on $\FD_\lambda(m)_\SPAM$.
Assuming that $[P_\lambda,M]=0$,
we can rewrite $\FD_\lambda(m)_\SPAM = v_\mathrm{SP}v_\mathrm{M} \osandwich{\rho}{P_\lambda}{\rho}$ in terms of the SPAM visibilities introduced in Section~\ref{sec:dominantSignal}.
We arrive at the bound (assuming $\delta_\lambda/\Delta_\lambda\leq 1/5$):
\begin{multline}
 m 
 \geq \frac{5}{\Delta_\lambda} \Big( \log d_\lambda + \tfrac12  \log(\osandwich{\rho}{P_\lambda}{\rho}^{-1}) \\
 + \log\tfrac1\gamma + \log \tfrac{1}{v_\mathrm{SP}v_\mathrm{M}} + 1.8 \Big). 
 \label{eq:relative-sequence-length-bound-simplified}
\end{multline}
This bound is almost identical to the additive error bound \eqref{eq:sequence-length-bound-2-simplified}, however, it also depends on the amount of \ac{SPAM} noise relative to the ideal coefficients.
This is unavoidable since the \ac{SPAM} noise typically decreases the strength of the signal as discussed in Section \ref{sec:dominantSignal}.
For example, for depolarizing \ac{SPAM} noise of strength $1-p$ we found that $\FD_\lambda(m)_\SPAM$ is suppressed by $p^2$.
Accordingly, we find a contribution of $2 \log(1/p)$ to the sequence length \eqref{eq:relative-sequence-length-bound-simplified} in this situation.

Finally, an important contribution to both bounds \eqref{eq:sequence-length-bound-2-simplified} and \eqref{eq:relative-sequence-length-bound-simplified} is the inverse spectral gap.
For fast-scrambling random circuits, like brickwork circuits, we have $\Delta^{-1} = O(1)$, i.e.~the spectral gap is independent of the dimension $d=2^n$.
In this setting, we also have $\log(d_\Ad)=O(n)$, and thus the sequence length of filtered \ac{RB} requires \emph{linear circuit depth}.
We comment on this in more detail for various random circuits in Sec.~\ref{sec:application-to-random-circuits}, and give precise scalings with small constants.

However, the scaling depends on both $d_\lambda$ and $\Delta_\lambda^{-1}$, allowing for scenarios with even shorter circuit depth for smaller irrep dimensions.
For instance, we have $d_\lambda=1$ for the Heisenberg-Weyl group $\HW{n}{p}$, and many irreps of the local Clifford group have sub-exponential dimension.
Since these are both local groups, it is also straightforward to obtain large spectral gaps.

% -------------------------------------------
\subsubsection{Obtaining sampling complexity bounds}
\label{sec:subdominant-decays-second-moment}
% -------------------------------------------

As in Sec.~\ref{sec:subdominant-decay}, we study the suppression of the subdominant decays in Thm.~\ref{thm:second-moment}.
Our goal is to derive sufficient conditions on the sequence length $m$ under which the assumptions of our sampling complexity theorems \ref{thm:sampling-complexity-additive} and \ref{thm:sampling-complexity-relative} are fulfilled.
To this end, the following lemma gives a sufficient condition for the sequence length $m$ in terms of the spectral gap $\Delta_\lambda^{(3)}$, and irrep-specific quantities.

\begin{lemma}
\label{lem:sequence-length-second-moment}
The sum of subdominant decays in Thm.~\ref{thm:second-moment} is less than $\beta > 0$ provided that
\begin{align}
 m \geq 
 \frac{(\Delta_\lambda^{(3)})^{-1}}{1 - 2 r_\lambda^{(3)}} 
 \Big( \log(c_\lambda \snorm{S_\lambda^+}) + \log \,\osandwich{\rho}{P_\lambda}{\rho} \\
  + \log g(r_\lambda^{(3)}) + \log \tfrac{1}{\beta}  \Big).
 \label{eq:sequence-length-bound-sampling-complexity}
\end{align}
\end{lemma}

\begin{proof}
The statement follows in complete analogy to the proof of Lem.~\ref{lem:sequence-length}.
\end{proof}

As for Lemma \ref{lem:sequence-length}, we expect that the term $\log(c_\lambda \snorm{S_\lambda^+})$ generally scales as $O(\log(d_\lambda)$, however now with a larger prefactor.
If $\lambda \in \Irr(\omega)$ is multiplicity-free and $[P_\lambda,M]=0$, we can use $s_\lambda^{-1}\leq d_\lambda$, c.f.~Eq.~\eqref{eq:upper-bound-Slambda}, to obtain $\log(c_\lambda \snorm{S_\lambda^+})=\frac12 \log(d_\lambda/s_\lambda^3) \leq 2 \log(d_\lambda)$.
As in Sec.~\ref{sec:subdominant-decay}, we want to assume that the implementation error is bounded away from 1/4, say $\delta_\sigma/\Delta_\sigma \leq 1/5$ for all $\sigma$, such that we have $r_\lambda^{(3)}\leq 1/5$ and $g(r_\lambda^{(3)}) \leq 1.8$.
Then, we can bring Eq.~\eqref{eq:sequence-length-bound-sampling-complexity} into the more appealing form
\begin{align}
 m \geq \frac{2}{\Delta_\lambda^{(3)}} \Big( 2\log d_\lambda + \log \,\osandwich{\rho}{P_\lambda}{\rho} + \log\tfrac{1}{\beta} + 1.8 \Big).
\end{align}

Next, we want to derive a sequence length bound involving relative errors.
Compared to Lem.~\ref{lem:sequence-length}, the proof is slightly more delicate since we compare the subdominant decays with the \emph{sum} of dominant ones $\sum_\sigma \tr(C_\sigma I_\sigma^m)$.
Finding a good lower bound for this sum is made difficult by the fact that some of the individual terms may a priori be negative, c.f.~the discussion in Sec.~\ref{sec:discussion-second-moment}.
We do not think that this happens in any practically relevant scenario, but include this as an assumption in the following lemma.

\begin{lemma}
\label{lem:relative-sequence-length-second-moment}
Assume non-malicious \ac{SPAM} noise and fix a non-trivial irrep $\lambda\in\Irr(\omega)$ such that all $\sigma\in\Irr(\omega)\cap\Irr(\tau_\lambda^{\otimes 2})$ are multiplicity-free in $\omega$.
Then, the sum of subdominant decays in Thm.~\ref{thm:second-moment} is suppressed by a relative error $\kappa > 0$ compared to the dominant decays, provided that
\begin{multline}
 m \geq 
 \frac{(\Delta_\lambda^{(3)})^{-1}}{1 - 4\frac{\delta_\lambda^{(3)}}{\Delta_\lambda^{(3)}}}
 \Big( \log \tfrac{c_\lambda}{s_\lambda} +  \log \,\osandwich{\rho}{P_\lambda}{\rho} \\
 + \log g\big(\tfrac{\delta_\lambda^{(3)}}{\Delta_\lambda^{(3)}}\big) + \log(1/\EE[f_\lambda^2] _\SPAM) + \log\tfrac{1}{\kappa} \Big) \, .
  \label{eq:relative-sequence-length-bound-second-moment}
\end{multline}
Here, $\EE[f_\lambda^2] _\SPAM$ is the second moment of the filter function that one would obtain for the ideal, noiseless implementation $\phi=\omega$ where unitaries are sampled from the Haar measure on $G$, but subject to \ac{SPAM} noise, c.f.~Eq.~\eqref{eq:sum-trace-Csigma}.
\end{lemma}

\begin{proof}
We use $r_\lambda^{(3)}\leq\delta_\lambda^{(3)}/\Delta_\lambda^{(3)}$ where $\delta_\lambda^{(3)} \coloneqq  \max_\sigma \delta_\sigma$, c.f.~Sec.~\ref{sec:discussion-second-moment}, and $I_\sigma \geq 1-2\delta_\sigma \geq 1-2\delta_\lambda^{(3)}$.
We then have
\begin{multline}
 \frac{\left|\sum_{\sigma\in\Irr(\tau_\lambda^{\otimes 2})} \tr(D_\sigma O_\sigma^m)\right|}{\sum_{\sigma\in\Irr(\tau_\lambda^{\otimes 2})} \tr(C_\sigma)I_\sigma^m}
 \leq
 \frac{c_\lambda/s_\lambda \, \osandwich{\rho}{P_\lambda}{\rho} \, g(\delta_\lambda^{(3)}/\Delta_\lambda^{(3)})}{\EE[f_\lambda^2] _\SPAM} \\
 \times
 \left(
 \frac{1-\Delta_\lambda^{(3)} + 2\delta_\lambda^{(3)}}{1-2\delta_\lambda^{(3)}}
 \right)^m \, .
 \label{eq:relative-error-bound-second-moment}
\end{multline}
The statement then follows as in the proof of Lem.~\ref{lem:sequence-length}.
\end{proof}

As for Lem.~\ref{lem:sequence-length-second-moment}, assuming that $\delta_\lambda^{(3)}/\Delta_\lambda^{(3)}\leq 1/5$ and $[P_\lambda,M]=0$, it is sufficient to fulfill the simplified bound
\begin{align}
 m 
 &\geq 
 \frac{5}{\Delta_\lambda^{(3)}} 
 \Big( 2 \log d_\lambda +  \log\,\osandwich{\rho}{P_\lambda}{\rho} \\
 &\qquad
  + \log(1/\EE[f_\lambda^2] _\SPAM) + \log \tfrac{1}{\kappa} + 1.8 \Big) \\
 &=
 \frac{5}{\Delta_\lambda^{(3)}} 
 \Big( 2 \log d_\lambda +  \log\,\osandwich{\rho}{P_\lambda}{\rho} + \log(1/\EE[f_\lambda^2]_\ideal)  \\
 &\qquad
  + \log( 1 / v_\SPAM^{(2)} ) + \log \tfrac{1}{\kappa} + 1.8 \Big) . \ \
 \label{eq:2nd-moment-relative-sequence-length-bound-simplified}
\end{align}
Hence, we have a similar situation as in Eq.~\eqref{eq:relative-sequence-length-bound-simplified}.
In Sec.~\ref{sec:discussion-second-moment}, we showed that for weak depolarizing \ac{SPAM} noise with strength $1-p$,  $v_\SPAM^{(2)} \approx p^2$ and thus find the same \ac{SPAM} noise dependence as in the previous Sec.~\ref{sec:subdominant-decay}.
We have a more detailed comparison between Eq.~\eqref{eq:relative-sequence-length-bound-simplified} and Eq.~\eqref{eq:2nd-moment-relative-sequence-length-bound-simplified} in the following.

To fulfill the assumptions of our sampling complexity Theorem \ref{thm:sampling-complexity-relative}, the sequence length $m$ has to be sufficiently large such that the subdominant terms in the first and second moment are suppressed by relative errors $\gamma$ and $\kappa$.
By Lemmas \ref{lem:sequence-length} and \ref{lem:relative-sequence-length-second-moment}, it is sufficient to choose $m$ as follows
\begin{align}
 m 
 &\geq 
 \frac{5}{\Delta_\lambda} \Big( \log d_\lambda + \tfrac12  \log\,\osandwich{\rho}{P_\lambda}{\rho} \\
 &\qquad
  + \log(1/\FD_\lambda(m)_\SPAM)  + \log \tfrac{1}{\gamma} + 1.8 \Big)\, , \label{eq:seq-length-comparison-1} \\
 m 
 &\geq 
 \frac{5}{\Delta_\lambda^{(3)}} \Big( 2 \log d_\lambda +  \log\,\osandwich{\rho}{P_\lambda}{\rho} \\
 &\qquad
 + \log(1/\EE[f_\lambda^2] _\SPAM) + \log \tfrac{1}{\gamma} + 1.8 \Big)\, . \label{eq:seq-length-comparison-2}
\end{align}
Note that in contrast to Sec.~\ref{sec:subdominant-decay}, where to goal was to reliably find the dominant decay parameter, we do not require $\kappa$ to be small.
In principle, it is enough if $\kappa = O(1)$, for instance $\kappa=1$ would be sufficient.
In practice, $\kappa$ enters linearly in the required number of samples by Thm.~\ref{thm:sampling-complexity-relative}, but only logarithmically in Eq.~\eqref{eq:seq-length-comparison-2}.
Hence, one would try to choose it as small as possible, finding a compromise between number of samples and sequence length.

In general, it is not simple to answer which of the lower bounds \eqref{eq:seq-length-comparison-1} and \eqref{eq:seq-length-comparison-2} is larger.
Typically, we expect that the second bound implies the first one for the following reasons:
First, we typically have $\Delta_\lambda^{(3)}\leq \Delta_\lambda$ with equality in many practically relevant cases, see Sec.~\ref{sec:application-to-random-circuits}.
Second, the first two terms in Eq.~\eqref{eq:seq-length-comparison-2} are twice as large as in Eq.~\eqref{eq:seq-length-comparison-1}.
However, the comparison between $\FD_\lambda(m)_\SPAM$ and $\EE[f_\lambda^2] _\SPAM$ is not as simple.
In Prop.~\ref{prop:second-moments-ideal}, we saw examples in which either the first or second moment are larger. 

\subsubsection{Example}
\label{sec:seq-length-examples}

To be able to make a more concrete statement, we consider the case that $G$ is a unitary 3-design and $d=2^n$ (i.e.~a $n$-qubit system).
Recall that in this case, we have $\FD_\lambda(m)_\SPAM = v_\mathrm{M} v_\mathrm{SP} \osandwich{\rho}{P_\lambda}{\rho}$.
If relevant, we assume that the \ac{SPAM} visibilities are such that $v_\SPAM^{(2)} \approx v_\mathrm{M} v_\mathrm{SP}$ (which is e.g.~the case for depolarizing \ac{SPAM} noise).
We use $c_\Ad = (d+1)\sqrt{d-1}$ and $c_\Ad/s_\Ad = (d+1)^2\sqrt{d-1}$, c.f.~Secs.~\ref{sec:frame-operator} and \ref{sec:bound-subdominant-signal}, as well as $\osandwich{\rho}{P_\Ad}{\rho} = (d-1)/d$.
To achieve additive error suppression, Lems.~\ref{lem:sequence-length} and \ref{lem:sequence-length-second-moment} require us to compute the following expressions:
\begin{align}
\log(c_\Ad) + \frac12 \log{\osandwich{\rho}{P_\Ad}{\rho}}
&=
\log \frac{d^2-1}{\sqrt{d}}  \, , \label{eq:seq-length-irrep-terms-3-design} \\
\log(c_\Ad/s_\Ad) +  \log{\osandwich{\rho}{P_\Ad}{\rho}}
 &=
 \log \tfrac{\sqrt{d-1}(d^2-1)(d+1)}{d} \, .
\end{align}
Thus, the second expression is always larger than the first, and assuming that the error in Lems.~\ref{lem:sequence-length} and \ref{lem:sequence-length-second-moment} are chosen as $\alpha \geq \beta$, we find that the following bound on $m$ is sufficient for the assumptions of Thm.~\ref{thm:sampling-complexity-additive}:
\begin{equation}
 m 
 \geq 
 \frac{2}{\Delta_\Ad^{(3)}} \big( 1.75n + \log \tfrac{1}{\alpha} + 1.8 \big)\, , \quad
 \text{(additive error)}
 \label{eq:seq-length-3-design}
\end{equation}
where we assumed that $d=2^n$ and used that $\log \frac{\sqrt{d-1}(d^2-1)(d+1)}{d}\leq 1.75n$ for all $n\in\N$.

For relative error suppression by Lems.~\ref{lem:sequence-length} and \ref{lem:relative-sequence-length-second-moment}, we use the exact expression for the second moment, computed in App.~\ref{sec:sampling-complexity-3-designs}, Eq.~\eqref{eq:2nd-moment-3-design}.
After a short calculation, we find the expressions
\begin{align}
\MoveEqLeft[3]\log(c_\Ad) + \frac12 \log{\osandwich{\rho}{P_\Ad}{\rho}} + \log(1/\FD_\Ad(m)_\ideal)\\
&=
\log\left( \sqrt{d}(d+1) \right) , \\
\MoveEqLeft[3] \log(c_\Ad/s_\Ad) +  \log{\osandwich{\rho}{P_\Ad}{\rho}} + \log(1/\EE[f_\Ad^2] _\ideal) \\
 &=
 \log \frac{\sqrt{d-1}d(d+1)(d+2)}{3d-2} \, .
\end{align}
Again, the second expression is always larger than the first.
Assuming $\gamma\leq\kappa$, we hence find that it is sufficient for $m$ to fulfill the bound
\begin{multline}
 m 
 \geq 
 \frac{5}{\Delta_\Ad^{(3)}} \left( 1.75n + \log(1/v_\mathrm{M} v_\mathrm{SP}) + \log \tfrac{1}{\gamma} + 1.8 \right)\, . \\
 \tag*{(relative error)}
\end{multline}
Here, we again used that for $d=2^n$ and $n\geq 2$, $\log \frac{\sqrt{d-1}d(d+1)(d+2)}{3d-2}\leq 1.75n$.

% -------------------------------------
\subsection{Application to common random circuits}
\label{sec:application-to-random-circuits}
% -------------------------------------

The signal guarantees for filtered \ac{RB} presented in Thms.~\ref{thm:rb-data-random-circuits} and \ref{thm:second-moment} prominently involve the spectral gap of the used random circuit, and so do the sequence lengths bounds derived in Sec.~\ref{sec:sequence-lengths}.
In this section, we want to illustrate these statements by applying our results to common choices of groups $G$ and random circuits $\nu$, resulting in concrete lower bounds on the sequence length for filtered \ac{RB} in practically relevant cases.
As a byproduct, this section may serve as a guideline for the scenarios which are not explicitly covered in this paper.

Concretely, we consider the unitary group $G=\U(2^n)$ and the Clifford group $G=\Cl{n}{2}$.
For both groups, we can make use of previously derived results for unitary 3-designs.
In particular, we have the sequence length bound \eqref{eq:seq-length-3-design} which we here state in a slightly different form:
\begin{equation} \label{eq:sequence-length-applications}
\begin{aligned}
 m &\geq \frac{2}{\Delta_3(\nu)} \big( 1.75 n + \log\tfrac{1}{\alpha} + 1.8 \big) \, .
\end{aligned}
\end{equation}
Here, we use that the $t$-design spectral gaps $\Delta_t(\nu)$ (i.e.~w.r.t.~$\omega^{\otimes t}$) for $t=2,3$ bound the relevant gaps in Lem.~\ref{lem:sequence-length} and \ref{lem:sequence-length-second-moment}, respectively, and that the gaps are monotonic in $t$.
For sequence lengths $m$ larger than the bound \eqref{eq:sequence-length-applications}, Theorem \ref{thm:rb-data-random-circuits} then guarantees that the expected filtered \ac{RB} signal $\FD_\Ad(m)$ is well-described by a single exponential decay of the form $A_\Ad I_\Ad^m$, up to an additive error $\alpha$. 
Moreover, Theorem \ref{thm:sampling-complexity-additive} gives a 
lower bound on the number of samples $N$ sufficient to estimate $\FD_\Ad(m)$ within additive error $\alpha$ and with probability $1-\delta$:
\begin{equation}
 N \geq 
 \frac{1}{\varepsilon^2 \delta}
 \big(3 + \alpha\big) \, ,
\end{equation}
Here, we used the bound \eqref{eq:2nd-moment-3-design-bound} on the second moment of unitary 3-designs given in App.~\ref{sec:sampling-complexity-3-designs}.

In the following, we discuss the dependence of the above bounds on the spectral gaps $\Delta_t(\nu)$ for various random circuits and $t=2,3$.
Studies of random circuits usually give spectral gaps that scale as $1/n$ or better in the number of qubits $n$ \cite{harrow_random_2009,brown_convergence_2010,haferkamp_efficient_2023,haferkamp_improved_2021,HarMeh18}. 

However, non-asymptotic results with explicit and \emph{small} constants for low designs orders are not so easy to obtain.
In the following, we summarize literature results for so-called \emph{local random circuits} and \emph{brickwork circuits} (to be defined shortly), and complement them with own numerical studies for small numbers of qubits.
For the design orders we are interested in, $t=2$ and $3$, the spectral gaps of these random circuits are the same when defined with gates from the unitary group or the Clifford group.
We then discuss how the results can be adapted when Haar-random unitaries are replaced by Clifford generators including the case of different probability weights.

\subsubsection{A collection of spectral gap bounds}

Let us begin by defining local random circuits and brickwork circuits.

\begin{definition}
\label{def:LRC-BW}
 Let $\mu_2$ denote the Haar measure on $\U(d^2)$.
 \begin{enumerate}
  \item Let $\mathcal G = ([n],E)$ be a graph on $n$ vertices. 
  A \emph{\ac{LRC} on $n$ qudits with connectivity graph $\mathcal{G}$} is a probability measure $\mu_\LRC$ on $\U(d^n)$ given by drawing an edge $e\in E$ uniformly at random, and apply a Haar-random unitary $U\sim \mu_2$ to the two qudits connected by $e$.
  A \emph{\ac{NN} \acl{LRC} with open/periodic boundary conditions} is a \ac{LRC} where $E = \{ (i,i+1) \, | \, i\in[n-1] \}$ and $E = \{ (i,i+1) \, | \, i\in[n-1] \} \cup \{ (n,1) \}$, respectively.
  \item Let $\mu_\mathrm{even}$ and $\mu_\mathrm{odd}$ be the probability measures on $\U(d^n)$ that apply independent Haar-random unitaries $U\sim \mu_2$ in parallel on the qudit pairs $(i,i+1)$ where $i\in[n-1]$ is even or odd, respectively. 
  A \ac{BW} circuit on $n$ qudits is given by the probability measure $\mu_\BW = \mu_\mathrm{even}\ast\mu_\mathrm{odd}$.
 \end{enumerate}
\end{definition}

As noted by \textcite{znidaric_exact_2008}, the exact $t=2$ spectral gap for \ac{NN} local random circuits follows from the equivalence to an integrable spin chain model and evaluates to 
\begin{align}
\label{eq:exact-LRC-gaps}
 \Delta_2(\mu_\LRC)
 &=
 \begin{cases}
  \frac{1}{n-1}(1-\frac45\cos(\pi/n)) & \text{(OBC)} \\
  \frac{2}{n}(1-\frac45\cos(\pi/n)) & \text{(PBC)}
 \end{cases}\,.
\end{align}
Asymptotically, we thus have $\Delta_2(\mu_\LRC) \sim \frac{1}{5n}$ and $\Delta_2(\mu_\LRC) \sim \frac{2}{5n}$ for open and periodic boundary conditions, respectively.
The spectral gaps for \ac{NN} local random circuits were also computed numerically for $t=2$ and $n\leq21$ qubits, as well as for $t=3$ and $n\leq 11$ by \textcite{cwiklinski_local_2013}.
For \acp{LRC} on a complete graph (i.e.~with all-to-all connectivity), \textcite{brown_convergence_2010} have established the scaling $\Delta_t = \frac{5}{6n} + O(\frac{1}{n^2})$ with a $t$-independent leading coefficient.
Finally, \textcite[Sec.~IV]{haferkamp_improved_2021} used Knabe bounds to promote exact values of the spectral gap for $t=2,\dots,5$ from small $n$ to lower bounds for all $n$.
We summarize the relevant results on spectral gap bounds for $t=2,3$ in Tab.~\ref{tab:spectral_gaps}.

\begin{table}[h]
\centering
\begin{tabular}{lcccc}
\toprule
 random circuit & $t=2$ & $t=3$ \\
 \midrule
 \ac{NN} local random circuit with PBC & $5n$ & $5n$  \\
 \ac{NN} local random circuit with OBC & $5n$ & $5n$ \\
 local random circuit on complete graph & $\sim 6n/5$ & $\sim 6n/5$ \\
 brickwork (odd number of layers) & $25/9$ & 42 \\
 brickwork (even number of layers) & $50/9$ & 42 \\
 \bottomrule
\end{tabular} 
\caption{
  Upper bounds on the inverse spectral gap $\Delta^{-1}_t$ of the $t$-th moment operator for certain families of random quantum circuits on $n$ qubits. 
  Here, \ac{NN} stands for nearest neighbor, and PBC and OBC mean periodic and open boundary conditions, respectively. 
  The results for \ac{NN} local random circuits follow from the exact results for $t=2$ in Ref.~\cite{znidaric_exact_2008}, c.f.~Eq.~\eqref{eq:exact-LRC-gaps}, and the $t=3$ bounds in Ref.~\cite{haferkamp_improved_2021}.
  The asymptotic scaling $\sim 6n/5$ for local random circuits on complete graphs is shown in Ref.~\cite{brown_convergence_2010}.
  The brickwork result for $t=2$ is deduced from the frame potential calculation in Ref.~\cite{Hun19}, while the result for $t=3$ follows by applying the detectability lemma to the $t=3$ spectral gap of local random circuits.
}
\label{tab:spectral_gaps}
\end{table}

There is also evidence to believe that the spectral gaps of LRC are independent of $t$, as long as $t$ is small enough (including $t=2,\dots,4$) \cite{haferkamp_improved_2021,hunter_hones_private}.
In this case, not only the $t=2$ but also the $t=3$ bounds for \ac{LRC} presented in Tab.~\ref{tab:spectral_gaps} are asymptotically tight.
For $n \leq 10$, we have numerically computed the respective spectral gaps and present them in Fig.~\ref{fig:LRC-spectrum}; finding a good agreement with Ref.~\cite{haferkamp_improved_2021}.
In particular, the $t=2$ and $t=3$ spectral gaps are identical up to numerical precision.
We can observe that the spectral gap for small $n$ deviates notably from the lower bound $1/5n$.
This can be used to reduce the sequence lengths of filtered \ac{RB} experiments on small systems by a factor of 1.3 to 5.

\begin{figure}[h]
\centering
\includegraphics{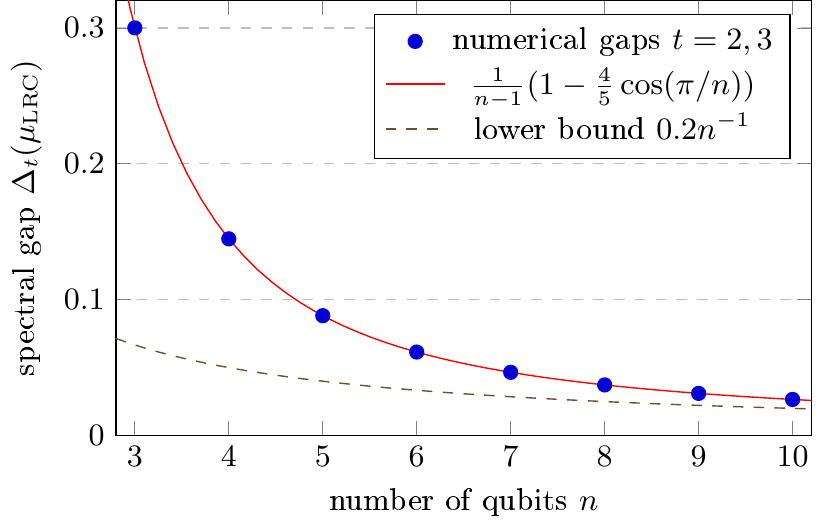}
\caption{%
  Spectral gap of nearest-neighbor local random circuits with open boundary conditions.
  The numerical results coincide for $t=2$ and $t=3$, and agree with the analytical results for $t=2$, c.f.~Eq.~\eqref{eq:exact-LRC-gaps}.
  The results are in good agreement with Refs.~\cite{haferkamp_improved_2021} and asymptotically approach the lower bound $1/5n$.
}
\label{fig:LRC-spectrum}
\end{figure}

The spectral gap bound of \ac{NN} \acp{LRC} can be used to bound the spectral gaps of the corresponding brickwork circuit by using the \emph{detectability lemma} \cite{aharonov_detectability_2008,anshu_simple_2016}:
\begin{equation}
 1 - \Delta_t(\mu_\BW) \leq \left( \frac{n \Delta_t(\mu_{\LRC})}{4} + 1\right)^{-\frac12}.
\end{equation}
While this technique generally results in the correct scaling in $n$, the constants are often sub-optimal.
For instance, inserting the bound $\Delta_2(\mu_{\LRC}) \geq \frac{1}{5n}$ results in $\Delta_2(\mu_\BW) \geq 0.024$.
For $t=2$, \textcite{haferkamp_improved_2021} instead rely on the explicit computation of the brickwork frame potential by \textcite{Hun19} to bound the circuit depth.
The latter result yields a bound on the difference of the brickwork moment operator $\MO_2(\mu_\BW)$ to the Haar projector in the Schatten 2-norm (also called \emph{Frobenius norm}) as follows
\begin{equation}
\label{eq:brickwork-frame-potential-bound}
 \twonormb{ \MO_2(\mu_\BW)^m - \MO_2(\mu) }^2 \leq 2 \left[ 1 + \left(\frac45\right)^{2(2m-1)} \right]^{\lfloor\frac{n}{2}\rfloor-1} - 2\, .
\end{equation}
Although the spectral norm is upper bounded by the 2-norm, $\snorm{\argdot}\leq\twonorm{\argdot}$, 
there is an obstacle in directly deriving a bound on the spectral gap $\Delta(\mu_\BW)$ from this result:
The moment operator $\MO_2(\mu_\BW)$ is not normal and thus there might be a strict inequality in 
\begin{align}
 \snormb{ \MO_2(\mu_\BW)^m - \MO_2(\mu) } 
 &\leq  \snormb{ \MO_2(\mu_\BW) - \MO_2(\mu) }^m \\
 &= \left( 1 - \Delta(\mu_\BW) \right)^m.
\end{align}
However, using the symmetrization trick, c.f.~Sec.~\ref{sec:rho-designs}, we can nevertheless deduce a decent spectral gap bound from Eq.~\eqref{eq:brickwork-frame-potential-bound}.
To the best of our knowledge, this is the best bound on the spectral gap of brickwork circuits and it has not been reported in the literature so far (the possibility has occurred to experts though \cite{jonas}) and one order of magnitude larger than the bound from the detectability lemma.

To this end, recall that $\mu_\BW = \mu_\mathrm{even}\ast\mu_\mathrm{odd}$ where the two layers act in parallel on qubit pairs $(i,i+1)$ where $i$ is even or odd, respectively.
These two layers are each symmetric and $\mu_\mathrm{even}^{\ast 2}=\mu_\mathrm{even}$, $\mu_\mathrm{odd}^{\ast 2}=\mu_\mathrm{odd}$.
The `inverted measure' \eqref{eq:inversted_measure} is then given as $\tilde \mu_\BW = \mu_\mathrm{odd}\ast\mu_\mathrm{even}$ such that $\eta_\BW\coloneqq \tilde \mu_\BW \ast \mu_\BW = \mu_\mathrm{odd}\ast\mu_\mathrm{even}\ast\mu_\mathrm{odd}$.
Powers then have the form $\eta_\BW^{\ast m} = \mu_\mathrm{odd}\ast (\mu_\mathrm{even}\ast\mu_\mathrm{odd})^{\ast m}$, i.e.~they correspond to a brickwork circuit which starts and ends with an odd layer and thus involves an odd number of layers in total.
Since the moment operator of the symmetrized measure $\eta_\BW$ is self-adjoint, we then find by using the appropriate bound in Ref.~\cite{Hun19}:
\begin{align}
 (1- \Delta(\eta_\BW))^m
 &=
 \snormb{ \MO_2(\eta_\BW) - \MO_2(\mu) }^m \\
 &=
 \snormb{ \MO_2(\eta_\BW)^m - \MO_2(\mu) } \\
 &\leq
 \twonormb{ \MO_2(\eta_\BW)^m - \MO_2(\mu) } \\
 &\leq
 \left( 2 \left[ 1 + \left(\frac45\right)^{4m} \right]^{\lfloor\frac{n}{2}\rfloor-1} - 2 \right)^{\frac12}.
\end{align}
Taking the $m$-th root on both sides, we can observe that the left hand side does not depend on $m$ anymore. 
Hence, we can take the limit $m\rightarrow\infty$ of the resulting right hand side to obtain an upper bound on $1- \Delta(\eta_\BW)$.
To compute the limit, note that we have the following lower and upper bound for sufficiently large $m$:
\begin{equation}
\label{eq:limit}
\begin{split}
  2 \left(\left\lfloor\frac{n}{2}\right\rfloor-1\right)  \left(\frac45\right)^{4m} 
  &\leq
  2 \left[ 1 + \left(\frac45\right)^{4m} \right]^{\lfloor\frac{n}{2}\rfloor-1} - 2 \\
  &\leq
  4 \left(\left\lfloor\frac{n}{2}\right\rfloor-1\right)  \left(\frac45\right)^{4m}.
\end{split}
\end{equation}
The lower bound in Eq.~\eqref{eq:limit} follows from Bernoulli's inequality $(1+x)^r \geq 1 + rx$ for $r\in\N_0$ and $x\geq -1$.
The upper bound follows from $(1+x)^r \leq e^{rx} \leq 1 + 2 rx$ which holds for any $x\geq 0$ and $r\geq 0$ such that $0\leq rx \leq 1$.
The latter condition is certainly fulfilled for large enough $m$.
Taking the $2m$-th root of Eq.~\eqref{eq:limit} and the limit $m\rightarrow\infty$, we see that the lower and upper bound both converge to $16/25$, hence we arrive at the result
\begin{equation}
 1- \Delta(\eta_\BW) \leq \frac{16}{25} \qquad  \Rightarrow \qquad  \Delta(\eta_\BW) \geq \frac{9}{25}.
\end{equation}
By Eq.~\eqref{eq:symmetrized_spectral_gap}, this implies the following bound on the spectral gap of $\mu_\BW$ (i.e.~brickwork circuits with an \emph{even} number of layers):
\begin{equation}
 \Delta(\mu_\BW) \geq \frac{9}{50}\, .
\end{equation}
However, we expect this lower bound to be quite loose and attribute this to the proof technique.
Realistically, we do not expect a large difference between the spectral gap of the ``symmetric'' brickwork circuit $\eta_\BW$ and the non-symmetric one $\mu_\BW$.
The obtained bounds can also be found in Tab.~\ref{tab:spectral_gaps}.

\subsubsection{Spectral gaps for random circuits composed of Clifford generators}

We are in particular interested in random circuits where the individual components are not drawn from the local Haar measure $\mu_2$, but instead from a set of local generators according to the measure $\nu$.
For local random circuits, we can lift results on the local spectral gaps of $\nu$ to global spectral gaps using the ``local-to-global'' lemma in Ref.~\cite{OnoBueKli17}.

\begin{lemma}[{\cite[Lem.~16]{OnoBueKli17}}]
\label{lem:local-to-global}
 Let $\nu$ be a symmetric probability measure on $\U(d^2)$ and let $([n],E)$ be a graph on $n$ vertices.
 Let $\nu_e$ be the measure $\nu$ with support on the factors corresponding to $e\in E$, and set $\nu_\LRC = \sum_{e\in E} p_e \nu_e$ for some probabilites $p_e$.
 Let $\mu_2$ be the Haar measure on $\U(d^2)$, $\mu_\LRC$ as in Def.~\ref{def:LRC-BW}, and let $\mu$ be the Haar measure on $\U(d^n)$.
 Then, we have
 \begin{align}
  \MoveEqLeft[2]\snorm{ \MO_{\rho_n}(\nu_\LRC) - \MO_{\rho_n}(\mu) } \\
  &\leq
  1 - \left( 1 - \snorm{ \MO_{\rho_2}(\nu) - \MO_{\rho_2}(\mu_2)  } \right) \\
  &\quad\times
  \left( 1 - \snorm{  \MO_{\rho_n}(\mu_\LRC) - \MO_{\rho_n}(\mu)} \right),
 \end{align}
 for any representation $\rho_n$ of $\U(d^n)$ that factorizes for $\U(d)^{\times n}$ as $\rho_n(g_1,\dots,g_n) = \rho(g_1)\otimes\dots\otimes\rho(g_n)$.
\end{lemma}

We have numerically computed the spectral gap of the $t$-th moment operator $\MO_t(\nu_\mathcal{G})$ for a probability measure $\nu_\mathcal{G}$ which draws from the set of two-qubit Clifford generators 
\begin{multline}
 \mathcal{G} \coloneqq  \big\{ P U \; | \; P\in\{\one,X,Y,Z\}^{\otimes 2}, \\
 \; U \in \{\one,S,H\}^{\otimes 2}\cup\{CX\} \big\} \,.
 \label{eq:cliff_generator_set}
\end{multline}
To this end, we have varied the probability $p$ of choosing $CX$ as the Clifford component $U$ in the generator set $\mathcal{G}$, and found that the value $p=0.35$ maximizes the spectral gap, resulting in $\Delta_t^{-1}(\nu_{\mathcal{G}}) = 10.99$ for both $t=2,3$.
Using Lemma \ref{lem:local-to-global} and Table \ref{tab:spectral_gaps}, we then find the following bound on the inverse spectral gap of the \ac{LRC} $\nu_{\LRC,\mathcal{G}}$ composed of gates from $\mathcal{G}$:
\begin{align}
 \Delta_t^{-1}(\nu_{\LRC,\mathcal{G}})
 &\leq
 \Delta_t^{-1}(\nu_{\LRC}) 
 \Delta_t^{-1}(\nu_{\mathcal{G}}) \\
 &\leq
 \begin{cases}
   55n & \text{NN, open/periodic BC}, \\
   14n & \text{all-to-all connectivity},
 \end{cases}
 \label{eq:spectral_gap_LRC_generators}
\end{align}
where $t=2,3$.
Recall that the obtained upper bound yields a bound on the sequence length $m$, i.e.~the number of generators to be applied, by Eq.~\eqref{eq:sequence-length-applications}.

We can now compare this result to the circuit depth that one would obtain for a \ac{LRC} with Haar-random 2-qubit Clifford unitaries.
Note that an arbitrary 2-qubit Clifford unitary can be implemented using at most 3 $CX$ gates in depth $\leq 9$ \cite{bravyi_6-qubit_2022,bravyi_hadamard-free_2021}.
Hence, the required depth would be at most $9m$ where $m$ is the sequence length for the exact \ac{LRC} given by Eq.~\eqref{eq:sequence-length-applications}.
Since this bound differs from the one using generators only by a prefactor, it is sufficient to compare $9 \times 5n = 45 n$ (\ac{NN}) and $9 \times 5n/6 = 15n/2$ (all-to-all) with the respective bounds \eqref{eq:spectral_gap_LRC_generators} obtained before.
We see that the exact \ac{LRC} implementation would require a similar circuit depth than our generator-based approach.
However, the latter scheme only requires $p=0.35$ $CX$ gates on average while random 2-qubit Clifford unitaries require an average number of 1.5 $CX$ gates \cite{bravyi_6-qubit_2022}.
Moreover, we expect the bound \eqref{eq:spectral_gap_LRC_generators} to be rather loose due to the use of Lem.~\ref{lem:local-to-global}, such that the generator-based local random circuits should perform equally well or better in practice.

\begin{figure}[h]
\centering
\includegraphics{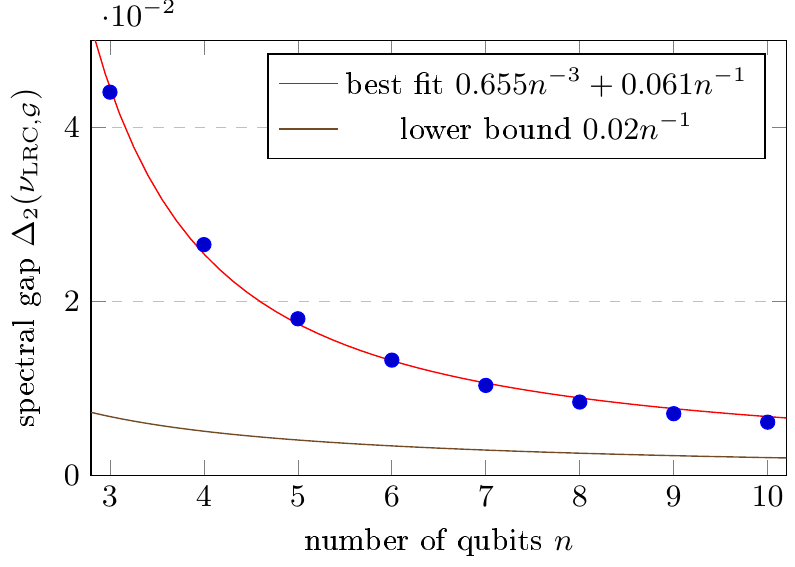}
\caption{%
 Numerically computed spectral gap $\Delta_2(\nu_{\LRC,\mathcal{G}},p=0.35)$ of local random circuits with local gates drawn from the set of Clifford generators $\mathcal{G}$ and $CX$ probability $p=0.35$.
 The fit suggests that the derived bound \eqref{eq:spectral_gap_LRC_generators} can be reduced by a factor of 3.
}
\label{fig:LRC-spectrum-generators}
\end{figure}

For a more precise comparison, we have exactly computed the spectral gap $\Delta_2(\nu_{\LRC,\mathcal{G}})$ for \ac{NN} open boundary conditions.
The results for $n\leq 10$ qubits are shown in Fig.~\ref{fig:LRC-spectrum-generators}.
Our numerical results suggest that the bound \eqref{eq:spectral_gap_LRC_generators} can be improved for \ac{NN} connectivity by a factor of 3, resulting in $\Delta_t^{-1}(\nu_{\LRC,\mathcal{G}})\leq 16.5n$.

Finally, the detectability lemma \cite{aharonov_detectability_2008} in its generalized version \cite{anshu_simple_2016} gives a bound on the spectral gap of brickwork circuits in terms of the one of a \ac{NN} local random circuit $\mu_{\LRC}$ as follows (see e.g.~Ref.~\cite[Eq.~(33)]{haferkamp_improved_2021})
\begin{equation}
 \snorm{\MO_t(\mu_\BW) - \MO_t(\mu)} \leq \left( \frac{n \Delta_t(\mu_{\LRC})}{4} + 1\right)^{-\frac12}.
\end{equation}
This bound still holds true if the Haar-random 2-qubit unitaries in both random circuits are replaced by gates from $\mathcal{G}$ drawn according to $\nu_{\mathcal G}$.
Using the conjectured bound $\Delta_t^{-1}(\nu_{\LRC,\mathcal{G}})\leq 16.5n$, we then find
\begin{equation}
 \Delta_t(\nu_{\BW,\mathcal{G}})^{-1} 
 \leq 134
 \qquad (t = 2,3).
\end{equation}
We suspect that this bound can be significantly improved.

% -------------------------------------
\subsection{Discussion of sequence length bounds}
\label{sec:random-walk-numerics}
% -------------------------------------

One might wonder whether the upper bounds derived in Thms.~\ref{thm:rb-data-random-circuits} and \ref{thm:second-moment} are tight, and if not, how loose they are.
A potential leeway in these bounds directly results in sub-optimal lower bounds on the sufficient sequence lengths in Sec.~\ref{sec:sequence-lengths}.
{
As we argue in the following, there is no reason to believe that the bounds in Thms.~\ref{thm:rb-data-random-circuits} and \ref{thm:second-moment} are not tight, nevertheless the sequence length bounds are generally \emph{not} tight.
The latter statement can be illustrated at the example of a unitary 2-group $G$ in the absence of any noise.
In this case, filtered \ac{RB} computes (c.f.~Sec.~\ref{sec:overview-protocol}):
\begin{equation}
 \FD_{\Ad}(m)
  =: (d+1) Z_\nu(m) - \frac{d+1}{d} \,,
\end{equation}
where
\begin{equation}
Z_\nu(m) \coloneqq \int_G \sum_{x\in\F_2^{n}} p_\ideal(x|g)^2 \dd\nu^{\ast m}(g)
\end{equation}
is the \emph{collision probability} of the $m$-layer random circuit.
The convergence of $Z_\nu(m)$ has been studied in the context of \emph{anti-concentration} of outcome distributions and the quantum advantage of random circuit sampling.
For specific random circuits, it was shown that $Z_\nu(m)$ is close the Haar-random value $2/(d+1)$, up to a relative error, if the random circuit has logarithmic depth \cite{BarakEtAl:2020:Spoofing,dalzell_random_2022}.
This implies that the subdominant contributions to $\FD_{\Ad}(m)$ are suppressed by an \emph{additive} error already at logarithmic instead of linear depth.
We discuss this issue in some detail in the following and, to this end, numerically compute $\FD_\Ad(m)$ for a noiseless random circuit.
}

A more detailed numerical analysis for linear \ac{XEB} with different random Clifford circuits was recently presented by \textcite{chen_linear_2022}, for up to 1225 qubits and different noise models.
Since this is a special case of our filtered \ac{RB} protocol, it nicely shows a typical behavior of the filtered \ac{RB} signal.
The findings in Ref.~\cite{chen_linear_2022} for these examples are in general agreement with the following discussion (even though much more detailed).

\begin{figure*}
\centering
\includegraphics{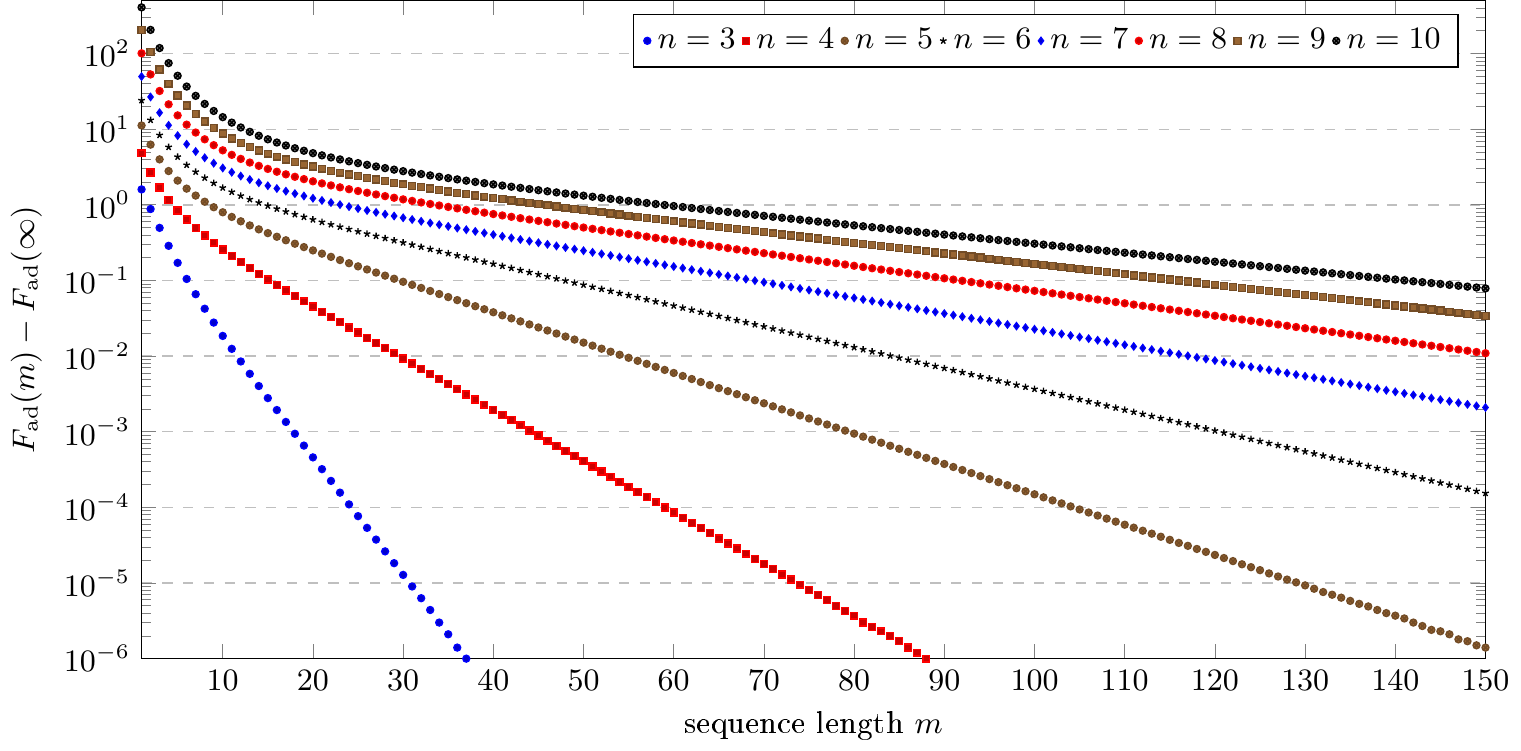}
\caption{
  Decay of the contribution to the filtered \ac{RB} signal $\FD_\Ad(m)$ stemming from the use of a random circuit instead of an exact unitary 2-design.
  The data has been computed numerically for $n=3,\dots,10$ qubits using a local random circuit with Haar-random 2-local unitaries in the noise-free setting.
  The decay is quickly dominated by a single exponential decay given by the second largest eigenvalue $(1-\Delta)$ of the second moment operator, and thus in good agreement with the bound of Thm.~\ref{thm:rb-data-random-circuits}.
  The fast decay in the beginning can be contributed to a superposition of exponential decays associated with smaller eigenvalues.
  To be able to observe the additional \ac{RB} decay in the presence of noise, the \ac{RB} decay must be slower than the `mixing decay' shown here.
  This is exactly quantified by the condition \eqref{eq:thm:initialbelief} in Thm.~\ref{thm:rb-data-random-circuits}.
}
\label{fig:LRC_decays}
\end{figure*}

{Here, we consider a nearest-neighbor local random circuit with open boundary conditions as defined in Def.~\ref{def:LRC-BW}.}
Formally, this is described by the following probability measure on $G=\U(2^n)$:
\begin{equation}
\label{eq:LRC-measure}
 \mu_\LRC = \frac{1}{n-1}\sum_{i=1}^{n-1} \mu_{(i,i+1)},
\end{equation}
where $\mu_{(i,i+1)}$ is the local Haar measure on the pair $(i,i+1)$.
We have numerically computed the noiseless filtered \ac{RB} signal $\FD_\Ad(m)$ for $n=3,\dots,10$ qubits and sequence lengths $m=1,\dots,150$.
To analyze the subdominant contributions to the signal coming from the use of the random circuit in the sense of Thm.~\ref{thm:rb-data-random-circuits}, we have subtracted the asymptotic value $1 - \frac{1}{d}$ from the data.
The difference is shown in Fig.~\ref{fig:LRC_decays}.
Since we deal with the noise-free case, the signal is given as a linear combination of exponential decays with decay rates corresponding to the eigenvalues of the second moment operator $\MO_2(\mu_\LRC) = \widehat{\omega\nu}_\LRC[\omega]$.
Consequently, we can observe two regimes: For small $m$, all eigenvalues decay quickly except for the second largest one which then dominates the signal for larger values of $m$.

Indeed, we find that the latter regime is well-approximated by a single exponential decay.
Corresponding fits give decay rates that are in very good agreement with the theoretical expectation $1-\Delta_2$ from Thm.~\ref{thm:rb-data-random-circuits}. 
To this end, we have compared the fit parameters with the analytical expression for the spectral gaps $\Delta_2$ of the second moment operator $\MO_2(\nu)$, c.f.~Eq.~\eqref{eq:exact-LRC-gaps}.
Next, we want to compare the magnitude of the subdominant contribution as a function of $m$ with the prediction by Lem.~\ref{lem:sequence-length}.
In the absence of noise, we have $\delta_\Ad=0$ and hence we find that we should take $m\geq \Delta_2^{-1}( \log((d^2-1)/\sqrt{d}) + \log\tfrac{1}{\alpha})$ to achieve an additive error $\alpha>0$.
Numerically, we find that the subdominant contribution already falls below a fixed additive error for values of $m$ that are up to two times smaller than our bound.
Since asymptotically $\Delta_2^{-1} \sim 5n$ (see Sec.~\ref{sec:application-to-random-circuits}) and $\log((d^2-1)/\sqrt{d})  \approx 1.05n$ , Eq.~\eqref{eq:seq-length-3-design} predicts a quadratic dependence on $n=\log_2(d)$.
Although we have not gathered enough data points, our data seem to be better compatible with a $n\log(n)$ dependence, and this is indeed what we expect from the anti-concentration results \cite{dalzell_random_2022}.
Nevertheless, the above results suggest that the concrete numerical values resulting from our bound \eqref{eq:seq-length-3-design} are not unrealistic and may still serve as a decent estimate for the required sequence lengths of filtered \ac{RB}, at least for a small to moderate number of qubits $n$.

As the subdominant contributions decay exactly with $1-\Delta_2$, the reason for the sub-optimality of our bound has to lie in the prefactors.
Indeed, our bound \eqref{eq:thm:subdominant-decay-bound} is valid for all $m$, and thus includes the fast scrambling regime for small $m$.
On the logarithmic scale of Fig.~\ref{fig:LRC_decays}, the bound thus corresponds to a shift of the exponential decays in the vertical direction.
Nevertheless, the prefactor of the bound \eqref{eq:thm:subdominant-decay-bound}, $c_\Ad \sqrt{\osandwich{\rho}{P_\Ad}{\rho}} = (d^2-1)/\sqrt{d}$, is still a bit too large compared to $\FD_\Ad(1)$, the maximum value of the signal.
Indeed, assuming the initial state is $\rho=\ketbra{0}{0}^{\otimes n}$, and that we measure in the computational basis, we find using $M= M_\loc^{\otimes n}$ that
\begin{equation}
\widehat{\omega\mu}_{(i,i+1)}[\omega](M) = M_\loc^{\otimes i-1}\otimes S_\loc \otimes M_\loc^{\otimes n - i - 1},
\end{equation}
where $S_\loc = \widehat{\omega}[\omega](M_\loc^{\otimes 2})$ is the frame operator on 2 qubits, cp.~Eq.~\eqref{eq:frame-operator-as-dep-channel}.
We arrive at the following expression:
\begin{align}
 \FD_\Ad(1) 
 &= \frac{d+1}{n-1} \sum_{i=1}^{n-1} \osandwich{\rho}{P_\Ad \widehat{\omega\mu}_{(i,i+1)}[\omega](M)}{\rho} \\
 &= (d+1) \bigg[ \tr\big(\ketbra{0}{0} M_\loc(\ketbra{0}{0})\big)^{n-2} \\
 &\qquad \times
  \tr\big(\ketbra{0}{0}^{\otimes 2} S_\loc(\ketbra{0}{0}^{\otimes 2})\big) - \frac{1}{d} \bigg] \\
 &= (d+1) \left( \frac{2}{5} - \frac{1}{d} \right).
\end{align}
Thus, we have $\FD_\Ad(1) - \osandwich{\rho}{P_\Ad}{\rho} = (d+1)\frac{2}{5} - 2$.

In conclusion, a uniform bound of the form \eqref{eq:thm:subdominant-decay-bound} for all $m$ still requires an exponentially large prefactor, resulting in a quadratic lower bound on the sequence length.
Improving on our sequence length bounds thus requires to directly estimate the extent of the fast scrambling regime for small $m$.
To this end, we think that more information about the spectrum of the second moment operator $\MO_2(\mu_\LRC)$ beyond its spectral gap is needed.
The presence of sufficiently large noise will however perturb the moment operator in a way that mixes eigenspaces, and thus complicates an analytical treatment.
Here, new and more explicit techniques beyond matrix perturbation theory may be needed to derive tight bounds.
However, it is unclear whether the noise-free scaling persists under the general noise assumptions used in Thm.~\ref{thm:rb-data-random-circuits}, or whether further assumptions on the noise are needed.
For instance, previous works in this direction explicitly assume local noise \cite{Liu21BenchmarkingNear-term,Dalzell21RandomQuantumCircuits}.
Finally, we think that these techniques would also be highly relevant in studying the properties of noisy random circuits in general.

%%% ---------------------------------------------
\subsection{Towards better filter functions}
%%% ---------------------------------------------
\label{sec:other-filters}

In Section \ref{sec:sequence-lengths}, we have argued that filtered randomized benchmarking 
as presented here 
needs sufficiently long random sequences in order to suppress the subdominant decay terms.
The main source of these terms is the usage of random circuits (with non-uniform measures) instead of the Haar-random unitaries and the decay is, in fact, showing their convergence towards the Haar measure.
In particular, this decay also occurs in the noise-free case, c.f.~Sec.~\ref{sec:random-walk-numerics}.

Instead of increasing the sequence length to achieve sufficient convergence of the random circuit, one might hope that a smart change in the filter function allows to filter on the relevant decays directly.
A similar consideration motivated the heuristic estimator for linear \ac{XEB} proposed in Ref.~\cite{RinottShohamKalai:2022}.
In this section, we propose two novel choices of filter functions and provide evidence that they indeed accomplish this goal.

A priori, it is not clear how these new schemes behave in the presence of general gate-dependent noise.
We found that the analysis for the new filter functions requires some non-trivial extensions to our perturbative approach, as we require a more detailed control on the perturbation of eigenspaces of the moment operator.
We leave a detailed analysis and comparison for future work.

We think that the framework of filtered randomized benchmarking and the following filter functions are of interest for the theory of linear \acf{XEB}.
First, a perturbative analysis has the potential to go beyond the usual white noise assumption \cite{Boixo2018CharacterizingQuantum,Dalzell21RandomQuantumCircuits}.
Second, the proposed filter functions might help to resolve some difficulties in finding appropriate estimators for the fidelity that can be equipped with rigorous guarantees when using linear \ac{XEB} with random circuits \cite{Liu21BenchmarkingNear-term}.

\subsubsection{Filtering using the exact frame operator}
\label{sec:filter-with-exact-frame-operator}

As sketched in Sec.~\ref{sec:overview-protocol}, filtered randomized benchmarking follows a similar idea as shadow tomography.
For ideal gates and ideal state preparation and measurement, $\hat{\FD}_\lambda(m)$ estimates for any sequence length $m$ the expression
\begin{align}
 \MoveEqLeft[1]\FD_\lambda(m) \\
 &= 
 \sum_{i \in [d]} \int_{G^{m}} f_\lambda(i , g_1\cdots g_m) \\
 &\qquad\times
  p(i| g_1, \ldots, g_m) \dd\nu(g_1)\cdots\dd\nu(g_m) \\
 &= \sum_{i \in [d]} \int_{G} \osandwich{\rho}{P_\lambda S^{-1}\omega(g)^\dagger}{E_i} \osandwich{E_i}{\omega(g)}{\rho} \dd\nu^{*m}(g) \\
 &= \osandwich{\rho}{P_\lambda S^{-1} S_{\nu^{*m}}}{\rho} \, .
\end{align}
Here, we defined the the frame operator associated to a sequence of length $m$ as
\begin{equation}
 S_{\nu^{*m}} = \int_G \omega(g)^\dagger M \omega(g) \dd\nu^{*m}(g)\, .
\end{equation}
If $\nu$ is not the Haar measure (or an appropriate exact design) then $S^{-1}$ does not cancel $S_{\nu^{*m}}$.
Since $\nu^{*m}$ eventually converges to the Haar measure, $S^{-1}$ becomes a good approximation to the exact inverse frame operator $S_{\nu^{*m}}^{-1}$ with increasing $m$.
The discrepancy between the two frame operators is visible as the additional decay in Thm.~\ref{thm:rb-data-random-circuits}.
Clearly, the solution would be to use the correct frame operator, i.e.~to redefine the filter function as 
\begin{equation}
  f_{\nu,\lambda}(i,g_1,\dots,g_m) \coloneqq \osandwich{E_i}{\omega(g_1\cdots g_m)S_{\nu^{*m}}^{-1} P_\lambda}{\rho}\, .
\label{eq:filter-with-correct-frame-op}
\end{equation}

However, the evaluation of this filter function 
is, in practice, more challenging than its Haar-random alternative.
The analytical and numerical evaluation of the frame operator and the dual frame is an active research topic in shadow tomography \cite{hu_classical_2021,bu_classical_2022,akhtar_scalable_2022,bertoni_shallow_2022,arienzo_closed-form_2022}.
The inversion of the frame operator is simplified for measures $\nu$ which are invariant under left multiplication with Pauli operators \cite{bu_classical_2022}, or consist only of Clifford unitaries.\footnote{$M$ is a Pauli channel and conjugation by Clifford channels maps Pauli channels to Pauli channels. The frame operator is thus a convex combination of Pauli channels, hence a Pauli channel itself and in particular diagonal in the Pauli basis.}
For those, the frame operator is diagonal in the Pauli basis, hence inversion is straightforward.
Nevertheless, the computational cost of the classical post-processing increases significantly for the filter function~\eqref{eq:filter-with-correct-frame-op}. 

\subsubsection{The trace filter}
\label{sec:trace-filter}

Our second proposal is based on the observation that every moment operator is block-diagonal, c.f.~Eq.~\eqref{eq:moment-op-spectral-decomposition}.
In particular, the exact frame operator $S_\nu = \widehat{\omega\nu}[\omega](M)$ can be decomposed as
\begin{equation}
 S_\nu = \widehat{\omega}[\omega](M) + (\widehat{\omega\nu}[\omega]-\widehat{\omega}[\omega])(M) = S + T_\nu,
\end{equation}
where $S=\widehat{\omega}[\omega](M)$ is the Haar-random frame operator which lies in the commutant of $\omega$ and $T_\nu$ is orthogonal to $S$, $\tr(ST_\nu)=0$.
With this notation, the filtered \ac{RB} signal \eqref{eq:filtered-rb-data} considered in the previous sections becomes (in the noise-free case)
\begin{align}
 \FD_\lambda(m) 
 &= 
 \osandwich{\rho}{P_\lambda S^{+} S_\nu}{\rho} \\
 &= 
 \osandwich{\rho}{P_\lambda S^{+} S}{\rho} + \osandwich{\rho}{P_\lambda S^{+} T_\nu}{\rho} .
\end{align}
Next, we show that -- under certain assumptions on the group $G$ and the measure $\nu$ -- we can change the filter function such that the second term vanishes identically.
The idea is to make the \ac{RB} signal take the form of a \emph{trace inner-product} of super-operators instead of a `matrix element' defined by $\rho$.
This allows us to project exactly on the commutant in the post-processing of $\omega$ in the data.
We, thereby, filter not only to an irrep but to the specific dominant subspace.
In this way, $\FD_\lambda(m)$ is not affected by the ``non-uniformness'' of our measure $\nu$ and we keep the simple structure of the inverse frame operator $S^{-1}$ in contrast to Sec.~\ref{sec:filter-with-exact-frame-operator}.

We exemplify this idea for unitary $2$-groups, i.e.\ essentially the multi-qubit Clifford group and specific subgroups thereof.
Let $\omega(g) = U_g(\argdot)U_g^\dagger$ and $G\subset\U(d)$ be a unitary 2-group. 
In this way, there is only a single non-trivial irrep to consider, namely the adjoint one $\lambda=\Ad$, and $P_\Ad = \id - \oketbra{\one}{\one}/d$.
Moreover, we assume that the measure $\nu$ is \emph{right-invariant} under the local Clifford group $\ClG{1}^{\times n}$, in particular we need $G \supset \ClG{1}^{\times n}$.
This is not much of a restriction, since we can always perform a layer of Haar-random single-qubit Clifford unitaries (or Haar-random single qubit unitaries) in the beginning of our random circuit at negligible cost.
As a consequence of the invariance assumption, the frame operator $S_\nu$ is not only diagonal in the Pauli basis, but its diagonal elements depend only on the support of a Pauli operator.
(Here, by `support' we mean on which qubits the operator acts non-trivially.)
To see this, let $w(u)$ and $w(v)$ be two Pauli operators with equal support, then we can find a local Clifford unitary $C=C_1\otimes\dots\otimes C_n$ such that $w(u) = C w(v) C^\dagger$.
Hence:
\begin{align}
 \osandwich{w(u)}{S_\nu}{w(u)}
 &=
 \int_{G} \osandwich{w(u)}{\omega(g)^\dagger M \omega(g)}{w(u)} \dd\nu(g) \\
 &=
 \int_{G} \osandwich{w(v)}{\omega(gC)^\dagger M \omega(gC)}{w(v)} \dd\nu(g) \\
 &=
 \osandwich{w(v)}{S_\nu}{w(v)}\, ,
\end{align}
where we have used the right invariance of $\nu$ in the last step.
Note that for any superoperator $\mathcal X$ with this property, we have
\begin{align}
 \tr(\mathcal X)
 &=
 d^{-1} \sum_{u\in\F_2^{2n}} \osandwich{w(u)}{\mathcal X}{w(u)} \\
 &=
 d^{-1} \sum_{z\in\F_2^n} 3^{|z|} \osandwich{Z(z)}{\mathcal X}{Z(z)}\, .
\end{align}
This is because every $Z(z)$-operator has support on $|z|$ many qubits and thus can be mapped to $3^{|z|}$ many different Pauli operators with identical support under local Cliffords.

Finally, we define the \emph{trace filter} as
\begin{equation}
\begin{split}
  f_\mathrm{tr}(i,g_1,\dots,g_m) &\coloneqq \osandwich{E_i}{\omega(g_1\cdots g_m)S^{-1}}{\xi} - \tr(\xi)\, , \\
  \xi &\coloneqq \frac{1}{d^2} \sum_{z\in\F_2^n} 3^{|z|} Z(z)\, .
\end{split}
\label{eq:trace-filter}
\end{equation}
Suppose we prepare the system in the $\ket 0$ state such that $\rho=d^{-1}\sum_z Z(z)$.
Using the diagonality of the frame operators, we then find that an ideal implementation yields the \ac{RB} signal
\begin{align}
 \FD_\mathrm{tr}(m) 
 &= 
 \frac{1}{d^2} \sum_{z\in\F_2^n} 3^{|z|} \osandwich{Z(z)}{S^{-1} S_\nu}{\rho} - \tr(\xi) \\
 &= 
 \frac{1}{d^3} \sum_{z\in\F_2^n} 3^{|z|} \osandwich{Z(z)}{S^{-1} S_\nu}{Z(z)} - \frac{1}{d} \\
 &=\frac1{d^2} \left[\tr(S^{-1}S_\nu)\right] - \frac1d\\
 &=
 \frac{1}{d^2} \left[ \tr(S^{-1}S) + \tr(S^{-1} T_\nu) \right] - \frac{1}{d}\,. %\\
 %&= \osandwich{\rho}{P_\Ad}{\rho}.
\end{align}
Now, $S$ is in the commutant of $\omega$, hence so is $S^{-1}$, and thus $\tr(S^{-1} T_\nu)=0$.
In conclusion, $\FD_\mathrm{tr}(m) = \osandwich{\rho}{P_\Ad}{\rho}$ and is not decaying with the sequence length. 
Using the trace filter \eqref{eq:filter-function} in post-processing, thus, yields an estimator for the state fidelity.

Finally, we note there are also different ways to construct similar trace filters.
If one does not want to assume local Clifford invariance, but the frame operators are still diagonal, the scheme can be adapted as follows.
Instead of $\rho=\ketbra{0}{0}$, prepare tensor powers of the `facet' magic state $\ketbra{F}{F} = ( \one + (X+Y+Z)/\sqrt{3} )/2$ which have the form
\begin{equation}
 \ketbra{F}{F}^{\otimes n} = \frac{1}{d} \sum_{u\in\F_2^{2n}} 3^{-\mathrm{wt}(u)/2} w(u)\, ,
\end{equation}
where $\mathrm{wt}(u)$ denotes the Pauli weight of the operator $w(u)$.
The $\xi$ operator can then be changed as follows
\begin{equation}
 \xi = \frac{1}{d^2} \sum_{u\in\F_2^{2n}} 3^{\mathrm{wt}(u)/2} w(u) \,.
\end{equation}
It is easy to see that this provides the right result.

%%% ---------------------------------------------
\subsection{Detailed comparison to related works}
\label{sec:related-works}
%%% ---------------------------------------------

Here, we want to compare our results, and in particular the assumptions needed for our signal guarantees, with the previously existing work by \textcite{helsen_general_2022,kong_framework_2021,chen_randomized_2022,chen_linear_2022}.

The main differences are as follows:
Our results hold for \ac{RB} experiments with gates from an arbitrary compact group $G$ sampled according to an arbitrary (well-behaved) probability measure.
The work of \textcite{helsen_general_2022} deals with arbitrary finite groups $G$, and their approach is readily generalized to compact groups with equivalent guarantees \cite{kong_framework_2021}.
However, for non-uniform sampling with measure $\nu$, their results require that $\nu^{*m}$ converges quickly to the Haar measure in total variation distance, which is quite a strong assumption.
Moreover, most results in Refs.~\cite{helsen_general_2022,kong_framework_2021} hold for a \ac{RB} protocol with inversion gate.
However, Ref.~\cite{helsen_general_2022} introduces the idea of filtered \ac{RB}, but only for Haar-random sampling and without a proper analysis of its sampling complexity.
The works by \textcite{chen_randomized_2022,chen_linear_2022} improve on the latter aspects in the sense that no inversion gate is needed and general probability measures $\nu$ are admitted under the assumption that $\nu^{*m}$ converges to a unitary 2-design.
Explicit examples for such random circuits with Clifford gates are discussed in detail in Ref.~\cite{chen_linear_2022}.

To make this more precise, recall the central assumption \eqref{eq:thm:initialbelief} of our main theorems \ref{thm:rb-data-random-circuits} and \ref{thm:second-moment}:
\begin{equation}
\snorm{\widehat{\phi\nu}[\omega_{\lambda}] - \widehat{\omega\nu}[\omega_{\lambda}]} \leq \delta_\lambda < \frac{\Delta_\lambda}{4} \, .
\end{equation}
This can be phrased as the assumption that the implementation function $\phi$ is sufficiently close to the reference representation $\omega$ \emph{on average w.r.t.~the measure $\nu$}.
In particular, only the quality of \emph{gates in the support of $\nu$} matter.
This assumption is not only less stringent in several aspects compared to Refs.~\cite{helsen_general_2022,kong_framework_2021,chen_randomized_2022,chen_linear_2022}, but our analysis also leads to more profound and tight results, as we will now explain.

\begin{enumerate}
 \item\label{item:per_irrep} As a consequence of filtering onto the irrep $\lambda$, we only require that the implementation error is small ``per irrep'' and compared to the irrep-specific spectral gap $\Delta_\lambda$ of $\widehat{\omega\nu}[\omega_\lambda]$.
 Refs.~\cite{helsen_general_2022,kong_framework_2021,chen_randomized_2022,chen_linear_2022} require conditions on global quantities. 

 \item We require that the deviation of the Fourier transforms of $\phi$ and $\omega$ is small.
 These are quantities which are already averaged over the group.
 In contrast, both Refs.~\cite{helsen_general_2022,kong_framework_2021,chen_randomized_2022,chen_linear_2022} formally require that the error $\norm{\phi(g)-\omega(g)}$, averaged over the group, is small.
 Closer inspection shows that the proofs in Refs.~\cite{chen_randomized_2022,chen_linear_2022} can be adapted to allow averaging inside the norm. 
 This is not the case for the approach in Refs.~\cite{helsen_general_2022,kong_framework_2021}.
 In general, we have the bound (cf.~Eq.~\eqref{eq:implementation-error-bound} below):
 \begin{equation}
  \snorm{ \widehat{\phi\nu}[\omega_{\lambda}] - \widehat{\omega\nu}[\omega_{\lambda}]}
  \leq \int_G \snorm{\phi(g) - \omega(g)} \dd\nu(g).
 \end{equation}
 We expect that the triangle inequality is quite loose in realistic situations and thus ``averaging inside the norm'' should lead to smaller quantities.

 Note that we are using the spectral norm to measure the implementation error while Refs.~\cite{chen_randomized_2022,chen_linear_2022} use the diamond-to-diamond norm for this purpose.\footnote{Note that the diamond-to-diamond norm can only be used for the `global' Fourier transform $\widehat{\phi\nu}[\omega]$ and not on a per-irrep basis. This is because the irreps generally fail to be operator algebras and thus cannot be endowed with Schatten norms.}
 As spectral and diamond-to-diamond norm are not ordered, it is not straightforward to compare the two approaches.

 \item We measure errors in \emph{spectral norm} for reasons which we explain shortly.
 Other norm choices include the diamond-to-diamond norm for the Fourier transforms, and the diamond norm for its vectorized version \eqref{eq:def-fourier-transform-dagger-cc}, as well as mixed approaches.
 The results in Refs.~\cite{chen_randomized_2022,chen_linear_2022} hold for an implementation error that is measured in diamond-to-diamond norm, while $\nu$ is assumed to be an approximate unitary 2-design in either diamond-to-diamond norm or the spectral norm, with approximation error less than 1.

 In general, norm choices based on the diamond norm imply much stronger notions of approximation for random circuits than for the spectral norm.
 In Ref.~\cite{chen_linear_2022}, it is shown that in order to constitute an approximate unitary 2-design in diamond-to-diamond norm, a random circuit has to have at least logarithmic depth in the number of qudits $n$.
 Moreover, in the same paper, it is also argued that the sequence length $m$ has the be linear in $n$.
 This directly implies that the results based on approximation in diamond-to-diamond norm require circuits of depth at least $\Omega(n\log n)$.
 In contrast, approximation in spectral norm allows for constant-depth random circuits (e.g.~brickwork circuits), while the sequence length $m$ is still linear in $n$.
 Although not explicitly stated in Refs.~\cite{chen_randomized_2022,chen_linear_2022}, the results for the spectral norm in these works thus imply that circuits of depth $O(n)$ are in fact already enough.
 In practice, the omitted constants and the regime of $n$ eventually decide which scaling is favorable.

 We took the generally larger convergence rates in diamond norm as an indication to use the spectral norm instead.
 Moreover, there is no conceptional difficulty in applying the spectral norm to Fourier transforms evaluated at subrepresentations, c.f.~point~\ref{item:per_irrep} above.
 Thus, this norm choice seems to be better suited for our filtering approach which is arguably most relevant for groups with more than one non-trivial irrep.
 Since we are then able to make the analysis on a per-irrep basis, we can show that it is sufficient to choose the sequence length as $m=O(\log(d_\lambda))$, c.f.~Sec.~\ref{sec:sequence-lengths} (assuming a constant spectral gap).
 In particular, we recover the results in Refs.~\cite{chen_randomized_2022,chen_linear_2022} for unitary 2-designs, c.f.~Sec.~\ref{sec:application-to-random-circuits}, but obtain more fine-grained and tighter bounds for general compact groups $G$.

 Finally, we also found that choosing the spectral norm naturally involves the spectral gap.
 Since the spectral gap is a well-studied quantity in the theory of random circuits, there exist plenty of literature results and techniques to bound the spectral gap for the random circuit of interest.
 In practice, it is to be expected that finding convergence rates in diamond-to-diamond norm is much harder.
 Indeed, direct results are rare \cite{harrow_random_2009,DanCleEme09}, indicating that leveraging spectral gap bounds via norm inequalities at the cost of additional dimensional factors might be the only course of action (as it is often the case for the much better studied diamond norm approximations of the vectorized channel twirl \cite{brandao_local_2016,HarMeh18,haferkamp_improved_2021,haferkamp_random_2022} .

 \item Consequently, our assumptions on the probability measure $\nu$ are minimal.
 There are no assumptions on the spectral gap $\Delta_\lambda$ except that it is larger than zero.
 Its value, however, limits the amount of imperfection in $\phi$ that can be tolerated, as well the required sequence length of the random circuits.
 An similar limitation can be found for the spectral norm results in Ref.~\cite{chen_randomized_2022}, as well as a qualitatively similar statement for the diamond-to-diamond norm (although the implementation error is measured in a different norm, thus the results are not strictly comparable).
 Finally, Ref.~\cite{helsen_general_2022} requires that the probability measure $\nu$ is close to the Haar measure in \emph{total variation distance}. 
 This is of course much more stringent than our assumption on the existence of a spectral gap.
 
 \item We take averages with respect to the relevant measure $\nu$. 
 In particular, it is completely irrelevant whether the implementation is good outside the support of $\nu$.
 This is a desired property, since only gates in $\supp(\nu)$ are ever applied in the \ac{RB} experiment and hence all other gates should not play a role in the analysis.
 The work of \textcite{chen_randomized_2022,chen_linear_2022} shares this property, while \textcite{helsen_general_2022} use Haar averages, even for non-uniform sampling.

 \item To the best of our knowledge, our work gives the first sampling complexity bounds for filtered \ac{RB} for arbitrary groups and random circuits (including linear \ac{XEB}) that are close to the noise-free situation and in this sense, optimal. 
\end{enumerate}

Note that the requirement that the spectral gap $\Delta_\lambda$ is non-zero implies that the measure $\nu$ defines an approximate $\tau_\lambda\otimes\omega$-design.
The size of the spectral gap determines the convergence rate of the design and hence strongly influences our sequence length bound, c.f.~Lem.~\ref{lem:sequence-length}.
For this, it is sufficient that $\nu$ defines an approximate $\omega^{\otimes 2}$-design \emph{for the group $G$ which we actually use for benchmarking}.
This is the proper generalization of the \emph{unitary 2-design assumption} made in standard \ac{RB} literature, when the benchmarking group is not the full unitary group or Clifford group.
Likewise, our sampling complexity guarantees require that $\nu$ is an approximate $\tau_\lambda^{\otimes 2}\otimes\omega$-design (or $\omega^{\otimes 3}$-design), replacing the common unitary 3-design assumption.

%%% ---------------------------------------------
\section{Conclusion}
%%% ---------------------------------------------
\label{sec:conclusion}

\emph{Filtered \acf{RB}} is the collection of experimental data obtained by applying random sequences of gates, followed by a suitable post-processing.
Summarized by the motto `measure first -- analyze later', filtered \ac{RB} is part of a modern class of characterization protocols such as \emph{shadow tomography} \cite{Aaronson2018ShadowTomography,huang_predicting_2020,elben_randomized_2022}, randomized \emph{gate set tomography} \cite{Gu2021RandomizedLinearGate,Brieger21CompressiveGateSet}, and other random sequence protocols \cite{helsen_estimating_2021}.
These protocols essentially share their first stage -- the data acquisition from random sequences of gates -- and only differ in the post-processing of this data.
The advantage of such protocols is that different conclusions can be drawn from the same data.

Compared to standard \ac{RB}, filtered \ac{RB} has the advantage to avoid the application of the final \emph{inversion gate}.
This avoids the problem that the inverse of a unitary can have large circuit depth even if the original unitary does not.
Moreover, it relaxes the requirements on the used gate set as it is not needed to efficiently compute the inversion gate.
Another advantage is only apparent when the protocol is used for groups $G$ with more than one non-trivial irrep (i.e.~groups which are not unitary 2-designs) -- in this case, standard \ac{RB} produces a linear combination of exponential decays in one-to-one correspondence with the irreps of $G$.
In practice, fitting these decays is often not possible and there is no way of attributing them to individual irreps.
Filtered \ac{RB} allows to address specific irreps and produces a single exponential decay that is straightforward to fit.\footnote{Assuming that the irrep is multiplicity-free and a random circuit with sufficiently large spectral gap is used.}

In this work, we have developed a general theory of filtered \ac{RB} with random circuits and arbitrary gate-dependent (Markovian and time-stationary) noise.
Our theory is based on harmonic analysis on compact groups and neatly combines representation theory with the theory of random circuits.
As such, it can be seen as a mathematically elegant advancement of Fourier-based approaches to \ac{RB} \cite{Merkel18,helsen_general_2022}, which does not require to implement a Fourier transform `by hand' but instead effectively performs it in the post-processing.
{We hope that our theory clarifies to which extent a group structure is needed for \ac{RB} and how one can go `beyond groups' \cite{chen_randomized_2022}.}

Concretely, we have shown that if the implementation error of the used gates is small enough compared to the spectral gap of the random circuit, then the noise cannot close the spectral gap of the relevant moment operator.
We argued that for local noise, this is the case if gate errors scale as $O(1/n)$.
As a consequence, the filtered \ac{RB} signal has two well-defined contributions.
The dominant one has the form of a matrix exponential decay and quantifies the average performance of the gate implementation.
This decay is superimposed by a additional, subdominant decays reflecting the convergence of the random circuit to a 2-design for the group $G$.
Importantly, if the implementation error is too large, the spectral gap may close and control over the contributions to the signal is lost.
In this regime, the signal does generally not reflect the average circuit performance.

We have derived sufficient conditions on the depth of the random circuit which guarantee that the subdominant contribution to the signal is negligible and the relevant decays can be extracted.
For random circuits which mix sufficiently fast, a \emph{circuit depth which is at most linear in the number of qudits} is sufficient.
{Although one may hope that this scaling can be improved to logarithmic depth, it is not clear whether this can be done under the general assumptions in this paper, or whether further assumptions on the physical noise (e.g.~locality) is needed.}

Additionally, we have shown that the use of random circuits instead of uniformly drawn unitaries from $G$ does not change the sampling complexity of filtered \ac{RB}. 
In particular, filtered \ac{RB} is \emph{sampling-efficient} if it is sampling-efficient when uniformly distributed unitaries are used.
To this end, we have computed the sampling complexity of ideal filtered \ac{RB} for unitary 3-designs, local unitary 3-designs, and the Pauli group, and found that it is indeed sampling-efficient in these important cases.

To illustrate our general results, we have applied them to commonly used groups and random circuits and have derived concrete, small constants for sufficient sequence lengths.
These explicit computation may also serve as a guideline when applying our general results to other groups and random circuits.

Finally, we have discussed other choices of filter functions which should result in a further reduction of the necessary circuit depths for filtered \ac{RB}.
Moreover, we think that these proposals are highly relevant for linear \acf{XEB} and the related random circuit sampling benchmark \cite{Boixo2018CharacterizingQuantum,Liu21BenchmarkingNear-term}.
However, a rigorous analysis for these alternative filter functions requires new techniques beyond the ones used in this paper and we leave such a study for future work. 
Ref.~\cite{helsen_estimating_2021} showed that 
by using more general filter functions one can perform  other robust gate-set characterization tasks, 
including ``filtered'' versions of interleaved randomized benchmarking, RB tomography, or Pauli channel estimation.
Our techniques can be applied to analyse variants with non-uniform measures of these protocols. 

We note that our analysis can be extended to incorporate non-Markovian noise, we however leave such a treatise for future work.

%%% ---------------------------------------------
\section{Acknowledgements}
%%% ---------------------------------------------

We thank Jonas Haferkamp and Nicholas Hunter-Jones for helpful discussions about random circuits and spectral gaps throughout the project, and for pointing out relevant literature and methods for Sec.~\ref{sec:application-to-random-circuits}. 
We thank Yelyzaveta Vodovozova, Jadwiga Wilkens, and the anonymous referees of the QIP and QCTiP conferences for helpful comments on the draft.
Furthermore, we would like to thank Jianxin Chen, Cupjin Huang, Linghang Kong, and in particular Dawei Ding for reaching out to us, helping us to improve our summary of Refs.~\cite{chen_randomized_2022,chen_linear_2022}, and clarifying its relation to our current work. 
IR would like to thank Jonas Helsen, Emilio Onorati, Albert Werner and Jens Eisert for countless discussions on randomized benchmarking over the recent years. 
MH and MK are funded by the Deutsche Forschungsgemeinschaft (DFG, German Research Foundation) within the Emmy Noether program (grant number 441423094) and 
by the German Federal Ministry of Education and Research (BMBF) within the funding program ``quantum technologies -- from basic research to market'' via the joint project MIQRO (grant number 13N15522).

%%% =============================================
\section{Acronyms}

\input{myacronyms}

\bibliography{generatorRBnew,mk}

%%% =============================================

\appendix

%%% =============================================
\section{Matrix perturbation theory}
%%% =============================================
\label{sec:perturbation-theory}

In this section, we review some results in matrix perturbation theory by \textcite{stewart_matrix_1990} and derive the corollaries needed to prove our main results.
To this end, we introduce some definitions.
First, we call a family of norms $\norm{\argdot}$ 
on matrices $\CC^{n\times m}$ \emph{matrix norms}, 
if it is submultiplicative w.r.t.\ to matrix multiplication; i.e.\ if $M\in\C^{n\times m}$ and $N\in\C^{m\times k}$, then $\norm{M N}\leq \norm{M}\norm{N}$. 
The \emph{separation function} between two square matrices $A\in\C^{n\times n}$ and $B\in\C^{m\times m}$ is given by
\begin{equation}
\label{eq:def-separation-function}
\sep(A,B) \coloneqq \inf_{\norm{P}=1} \norm{AP - PB},
\end{equation}
where the infimum is taken over all matrices $P\in\C^{n\times m}$.
The separation function is stable and continuous in the sense that \cite[Theorem~2.5, p.~234]{stewart_matrix_1990}
\begin{equation}
\label{eq:perturbation-separation-function}
\begin{aligned}
\sep(A,B) - \norm{E} - \norm{F} 
&\leq \sep(A+E,B+F) \\ 
&\leq \sep(A,B) + \norm{E} + \norm{F}.
\end{aligned}
\end{equation}
Following Ref.~\cite{stewart_matrix_1990}, we say that $A\in \CC^{n\times n}$ has a \emph{spectral resolution} if there is a block-diagonal decomposition of the form
\begin{equation}
\label{eq:spectral-resolution}
  [Y_1\, Y_2]\ad A\, [X_1\, X_2] = \begin{bmatrix} A_1 & 0\\ 0 & A_2
  \end{bmatrix},
\end{equation} 
for matrices $X_i$, $Y_j$ with $Y_i^\dagger X_j = \delta_{i,j}\id $, and $A_i = Y_i^\dagger A X_i$.
In this context, we write $E_{i,j} \coloneqq Y_i\ad E X_j$ for any matrix $E$. 

The separation function can be used to quantify the effect of small additive perturbations on a spectral resolution as follows.

\begin{theorem}[%
{\cite[Theorem~2.8, p.~238]{stewart_matrix_1990}}%
]\label{thm:left-invariant}
Let $A$ be a matrix with spectral resolution 
\begin{equation}
  [Y_1\, Y_2]\ad A\, [X_1\, X_2] = \begin{bmatrix} A_1 & 0 \\ 0 & A_2
  \end{bmatrix} \, ,
\end{equation}
$\norm{\argdot}$ be a matrix norm and $E$ some other matrix (perturbation). 
Suppose we have
\begin{align}
   \delta \coloneqq \sep(A_1, A_2) - \norm{E_{11}} - \norm{E_{22}} &> 0, \\ \text{and}\qquad \frac{\norm{E_{21}}\norm{E_{12}}}{\delta^2} &< \frac14,
\end{align} 
then there exist a unique matrix $P$ fulfilling
\begin{equation}
 P (A_1 + E_{11}) - (A_2 + E_{22}) P = E_{21} - P E_{12} P \, ,
 \label{eq:P-identities}
\end{equation}
and $\norm{P} < 2\, \frac{\norm{E_{21}}}\delta$,
such that the matrices
\begin{align}
  \tilde X_1 &\coloneqq X_1 + X_2 P, &
  \tilde A_1 &\coloneqq A_1 + E_{11} + E_{12} P, \\
  \tilde Y_2 &\coloneqq Y_2 - Y_1 P\ad, &
  \tilde A_2 &\coloneqq A_2 + E_{22} - P E_{12},
\end{align}
fulfill
\begin{align}\label{eq:subspace_cond}
  (A + E) \tilde X_1 &= \tilde X_1 \tilde A_1, &
  \tilde Y_2\ad (A + E) &= \tilde A_2 \tilde Y_2\ad\, .
\end{align}
\end{theorem}

\begin{remark}[Perturbation of real matrices]
\label{rem:pert-theory-real-matrices}
Note that Theorem \ref{thm:left-invariant} can also be used for real matrices:
Suppose $A$ is a real matrix with spectral resolution given by real matrices $X_i$, $Y_j$ with $Y_i^T X_j = \delta_{i,j}\id$, and $E$ is a real perturbation.
Let $P$ be the matrix produced by Thm.~\ref{thm:left-invariant}.
Then, complex conjugation of the defining Eq.~\eqref{eq:P-identities} shows that $\bar P$ is also a solution to Eq.~\eqref{eq:P-identities}.
Since the solution is unique, we have $\bar P = P$ and hence $P$ is real, as well as the matrices $\tilde X_i$, $\tilde Y_i$, and $\tilde A_i$.
\end{remark}

For the following discussion, we denote $\tilde Y_1 \coloneqq Y_1$ and $\tilde X_2 \coloneqq X_2$. 
Then also $\tilde X$ and $\tilde Y$ satisfy the orthogonality relations 
$\tilde Y_i^\dagger \tilde X_j = \delta_{i,j}\id$. 
In block matrix notation we then have
\begin{align}\label{eq:invar_subspaces}
[\tilde X_1\, \tilde X_2] 
&\coloneqq 
[X_1\,X_2] 
\begin{bmatrix}
\id & 0\\ P & \id 
\end{bmatrix}
\,, &
[\tilde Y_1\, \tilde Y_2] 
&\coloneqq
[Y_1\, Y_2] 
\begin{bmatrix}
\id & -P\ad\\ 0 & \id 
\end{bmatrix}. 
\end{align}
Moreover, the subspace conditions~\eqref{eq:subspace_cond} 
% of the theorem 
translate to 
\begin{align}
  \MoveEqLeft[2]
  [\tilde Y_1\, \tilde Y_2]\ad (A + E) [\tilde X_1 \,\tilde X_2] \\
  &= %
  \begin{bmatrix}
    \tilde Y_1\ad (A + E)\tilde X_1 & %
      \tilde Y_1\ad (A + E)\tilde X_2 \\
    \tilde Y_2\ad (A + E)\tilde X_1 & %
      \tilde Y_2\ad (A + E)\tilde X_2 \\
  \end{bmatrix}
  \\
  &= 
  \begin{bmatrix}
    \tilde Y_1\ad \tilde X_1 \tilde A_1& %
       Y_1\ad (A + E) X_2 \\
    \tilde Y_2\ad \tilde X_1 \tilde A_1& %
      \tilde A_2 \tilde Y_2\ad \tilde X_2 \\
  \end{bmatrix}
  = \begin{bmatrix}
    \tilde A_1& %
      E_{12} \\
    0 & %
      \tilde A_2 
  \end{bmatrix}. 
\end{align}

Of course an analogous  perturbation theorem holds when changing the roles of perturbed right- and left-invariant subspaces.  
To keep track of the notation let us write this down explicitly. 

\begin{corollary}[Flipped version]\label{cor:flipped_version}
Let $A$ be a matrix with spectral resolution 
\begin{equation}
  [Y_1\,Y_2]\ad A [X_1\, X_2] = \begin{bmatrix} A_1 & 0 \\ 0 & A_2 \end{bmatrix}, 
\end{equation}
$\norm{\argdot}$ be a matrix norm and $E$ some other matrix (perturbation). 
Suppose we have
\begin{align}
  \delta \coloneqq  \sep(A_1, A_2) - \norm{E_{11}} - \norm{E_{22}} &> 0, \\
  \text{and} \qquad \frac{\norm{E_{21}} \norm{E_{12}}}{\delta^2} &< \frac 14,
\end{align}
then there exist a unique matrix $P$ with $\norm P < 2\, \frac{\norm{E_{12}}} \delta$ such that the matrices
\begin{equation}
\label{eq:invar_subspaces_flipped}
\begin{aligned}
[\tilde X_1 \,\tilde X_2] 
&\coloneqq 
[X_1\, X_2] 
\begin{bmatrix}
\id & P\\ 0 & \id 
\end{bmatrix}, &
\tilde A_1 &\coloneqq A_1 + E_{11} - P E_{21} \, ,\\
[\tilde Y_1\, \tilde Y_2] 
&\coloneqq
[Y_1\, Y_2] 
\begin{bmatrix}
\id & 0\\ -P\ad & \id 
\end{bmatrix}, &
\tilde A_2 &\coloneqq A_2 + E_{22} + E_{21} P\, .
\end{aligned}
\end{equation}
fulfill $\tilde Y_i\ad \tilde X_j = \delta_{ij}\id$ and give rise to the following decomposition of the perturbed operator $A+E$:
\begin{equation}\label{eq:A+E-flipped_version}
  [\tilde Y_1\, \tilde Y_2]\ad (A+E) [\tilde X_1\, \tilde X_2]
  =
  \begin{bmatrix}
    \tilde A_1 & 0 \\
    E_{21} & \tilde A_2    
  \end{bmatrix}. 
  \end{equation} 
\end{corollary}

\begin{proof}
We apply Theorem \ref{thm:left-invariant} to $A\ad$ and use the corresponding norm $\norm{(\argdot)\ad}$. 
More explicitly, we take adjoints of the equations in the theorem and make the following replacements: 
$A \to A\ad$, $X\to Y$, $Y\to X$, similarly for the quantities with a tilde, $E\to E\ad$, $P\to -P\ad$ and $\norm{\argdot} \to \norm{(\argdot)\ad}$.
This yields the claimed bounds. 
Moreover, 
\begin{align}
  \tilde Y_1\ad (A + E) &= \tilde A_1 \tilde Y_1\ad, &
  (A + E) \tilde X_2 &= \tilde X_2\tilde A_2\,, 
\end{align}
where 
\begin{align}
  \tilde Y_1 &\coloneqq Y_1 - Y_2 P\ad\, , &
  \tilde A_1 &\coloneqq A_1 + E_{11} - P E_{21} \, , \label{eq:tA1_flipped} \\ 
  \tilde X_2 &\coloneqq X_2 + X_1 P \, ,&  
  \tilde A_2 &\coloneqq A_2 + E_{22} + E_{21} P\, . \label{eq:tA2_flipped}
\end{align}
Denoting $\tilde X_1 \coloneqq X_1$ and $\tilde Y_2\coloneqq Y_2$ we have again $\tilde Y_i\ad \tilde X_j = \delta_{ij} \id$. 
Moreover, this yields the result in block matrix notation as stated in the corollary. 
\end{proof}

Concatenating both versions of the perturbation theorem, gives us an expression for the entire spectral resolution of the perturbation (see also the analogous derivation in Ref.~\cite{helsen_general_2022}).

\begin{theorem}[%
Two-sided version%
]%
\label{thm:two-sided-perturbation-thm}
Let a matrix $A$ have a spectral resolution 
\begin{equation}
  [Y_1\, Y_2]\ad A\, [X_1\, X_2] = \begin{bmatrix} A_1 & 0 \\0 & A_2
  \end{bmatrix} ,
\end{equation}
$\norm{\argdot}$ be a matrix norm and $E$ some other matrix (perturbation). 
If
\begin{align}\label{eq:twosided:Ebounds}
  \delta \coloneqq \sep(A_1, A_2) - \norm{E_{11}} - \norm{E_{22}} &> 0, \\
  \text{and} \qquad
  \frac{\norm{E_{21}}\norm{E_{12}}}{\delta^2} &< \frac14
\end{align}
then there exist unique matrices $P_1, P_2$ satisfying
\begin{equation}
\label{eq:bound_P2}
  \norm{P_1} < 2 \,
    \frac{%
      \norm{E_{21}}
    }{%
    \delta%
    }, 
  \quad%
  \norm{P_2} < 2 \,
    \frac{%
      \norm{E_{12}}
    }{%
      \delta \left( 1 - 
        4 \frac{\norm{E_{21}}\norm{E_{12}}}{\delta^2} \right)
    } \,,
\end{equation}
such that the matrices
\begin{align}
  [\tilde X_1\, \tilde X_2] 
  &= 
  [ X_1\, X_2] \begin{bmatrix} \id & 0\\ P_1 & \id \end{bmatrix} 
              \begin{bmatrix} \id & P_2 \\0 & \id \end{bmatrix}, \\
  [\tilde Y_1\, \tilde Y_2]\ad 
  &= 
  \begin{bmatrix} \id & -P_2 \\0 & \id \end{bmatrix} \begin{bmatrix} \id &0 \\ -P_1 & \id \end{bmatrix}  [Y_1\,  Y_2]\ad, \\
  \tilde A_1 &= A_1 + E_{11} + E_{12} P_1, \\
  \tilde A_2 &= A_2 + E_{22} - P_1 E_{12},
\end{align}
fulfil $\tilde Y_i\ad \tilde X_j = \delta_{ij} \id$ and give rise to the spectral resolution
\begin{equation}\label{eq:A+Etwo-sided}
  [\tilde Y_1\, \tilde Y_2]\ad ( A+ E) [\tilde X_1\, \tilde X_2] = 
  \begin{bmatrix}
    \tilde A_1& 0 \\
    0 & \tilde A_2 \\
  \end{bmatrix}.
\end{equation}

\end{theorem}

\begin{proof}
We look at the output of the standard version of the perturbation statement, Theorem~\ref{thm:left-invariant}. 
There we have 
\begin{equation}
[\tilde Y_1\, \tilde Y_2]\ad ( A+ E) [\tilde X_1\, \tilde X_2] = \begin{bmatrix}
    \tilde A_1& %
     E_{12} \\
    0 & %
      \tilde A_2 \\
  \end{bmatrix}\,,
\end{equation}
with the existence of a matrix $P_1$ describing the perturbation between the invariant subspaces. 
In order to deal with the off-diagonal element $E_{12}$, let $\tilde A = \begin{bmatrix}\tilde A_1 &0 \\0 & \tilde A_2 \end{bmatrix}$ and consider 
$F = \begin{bmatrix} 0&  E_{12} \\ 0& 0 \end{bmatrix}$ as its perturbation, i.e., 
$F_{11} = F_{22} = F_{21} = 0$ and $F_{12} = E_{12}$. 

We can apply Corollary~\ref{cor:flipped_version} (flipped version) if 
\begin{equation}
  \delta' \coloneqq \sep(\tilde A_1, \tilde A_2)
\end{equation}
is large enough. 
Using the stability, cp.~Eq.~\eqref{eq:perturbation-separation-function}, of the separation function, i.e.\ 
\begin{equation}
  \abs{\sep(A_1 + B_1, A_2 + B_2) - \sep(A_1, A_2)} \leq \norm{B_1} + \norm{B_2}\, ,
\end{equation}
and the expressions for $\tilde A_1$ and $\tilde A_2$ from the first perturbation step,  
we find that 
\begin{equation}
\begin{split}
  \delta' &\geq \sep(A_1, A_2) - \norm{E_{11} + E_{12}P_1} - \norm{E_{22} - P_1E_{12}} \\ 
  &\geq \sep(A_1, A_2) - \norm{E_{11}} - \norm{E_{22}} - 2\norm{P_1}\norm{E_{12}} \\
  &\geq \delta - 4 \frac{\norm{E_{21}}\norm{E_{12}}}\delta\,,
\end{split}
\end{equation}
where in the last step we have used the norm-bound for $P_1$ and the definition of $\delta$ from the first step of the perturbation. 
Thus, $\delta' > 0$ is ensured if
\begin{equation}
  \frac14  > \frac{\norm{E_{21}}\norm{E_{12}}}{\delta^2}, 
\end{equation}
which is the condition from the first perturbation step. 
The other condition for applying Corollary~\ref{cor:flipped_version} (flipped version) is trivially fulfilled since $\norm{F_{21}} = 0$. 
Thus, we established the existence of $P_2$ with 
\begin{equation}
  \norm{P_2} \leq 2\, \frac{ \norm{E_{12}}}{\delta'} 
\end{equation}
giving rise to the stated condition. 

The final expressions for $\tilde A_i$ and the decomposition \eqref{eq:A+Etwo-sided} follow from the fact that 
$F = \begin{bmatrix} 0&  E_{12} \\ 0& 0 \end{bmatrix}$ 
does not contribute in \eqref{eq:tA1_flipped}, \eqref{eq:tA2_flipped} and \eqref{eq:A+E-flipped_version}. 
The statements for the perturbed invariant eigenspaces follows by combining the individual ones~\eqref{eq:invar_subspaces} and~\eqref{eq:invar_subspaces_flipped}.
The orthogonality relation follows by straightforward inspection. 
\end{proof}

It will also be useful to have a specialized perturbation result for \emph{moment operators} of random quantum circuits, where the perturbation is controlled in spectral norm.

\begin{theorem}[Perturbation of a moment operator]
\label{thm:perturbation-of-moment-operator}
Consider a moment operator, i.e., an operator $A$ with spectral decomposition
\begin{equation}
 A = [X_1\, X_2] \begin{bmatrix} \id & 0 \\ 0 & \Lambda \end{bmatrix} [X_1\, X_2]^\dagger ,
\end{equation}
where $[X_1\, X_2]$ is unitary.
Suppose that $\snorm{\Lambda} \leq 1 - \Delta$ for some $\Delta\in (0,1]$.
Let $E$ be a (potentially non-Hermitian) perturbation bounded as $\snorm{E} < \Delta/4$ with blocks $E_{ij} = X_i\ad E X_j$ w.r.t.~the spectral decomposition of $A$. 
Then, there exist unique matrices $P_1$ and $P_2$ with $\snorm{P_1} < 4 \snorm{E}/\Delta$ and 
$\snorm{P_2} < \frac{2\snorm E}{\Delta - 4\snorm E}$ 
such that the matrices
\begin{align}
  [R_{1}\, R_{2}] &= 
  [X_1\, X_2] %
  \begin{bmatrix} 
    \id & 0\\ P_1 & \id  
  \end{bmatrix}%
  \begin{bmatrix} \id & P_2 \\ 0& \id \end{bmatrix}, \\
  [L_{1}\, L_{2}]\ad &=  %
  \begin{bmatrix} \id & -P_2 \\ 0& \id \end{bmatrix}\begin{bmatrix} 
    \id & 0 \\ -P_1 & \id  
  \end{bmatrix} [X_1\, X_2]\ad \,,
  \label{eq:perturbed-basis} \\
  I &= \id + E_{11} + E_{12} P_1, \\
  O &= \Lambda + E_{22} - P_1 E_{12}\,, \label{eq:perturbed-projector}
\end{align}
fulfill $R_i\ad L_j = \delta_{ij} \id_i$ and $R_i L_j\ad = \delta_{ij} \Pi_i$, where $\id_i = X_i^\dagger X_i$ and $\Pi_i = X_i X_i^\dagger$ are the identity on the $i$-th eigenspace and the projector onto the $i$-th eigenspace, respectively.
We have the spectral resolution
\begin{equation}
\label{eq:spectral-resolution-moment-op}
\begin{aligned}
  [L_1\, L_2]\ad ( A+ E) [R_1\, R_2] &= 
  \begin{bmatrix}
    I & 0 \\
    0 & O \\
  \end{bmatrix} \,, \\
   A + E &= R_{1} I L_{1}\ad + R_{2} O L_{2}\ad \, .
\end{aligned}
\end{equation}
Moreover, the following bounds hold:
\begin{gather}
\begin{aligned}
  % \snorm{I} & \geq 1 - 2\snorm{E}, &
  \snorm{I-\id} &< 2\snorm{E}, &
  \snorm{O-\Lambda} &< 2 \snorm E, \\
  \sep(I,O) &\geq \Delta - 4\snorm{E} > 0, &&
\end{aligned} \\
  \snorm{L_{2}}\snorm{R_{2}} \leq \frac{\snorm{E}}{\Delta}\left(1 - \frac{\snorm{E}}{\Delta}\right) + \frac{1+\snorm{E}/\Delta}{1-4\snorm{E}/\Delta}.
  % \snorm{L_{\lambda,2}} &\leq 1 + 4\, \frac{\snorm{E}}{\Delta}, & \snorm{R_{\lambda,2}} & \leq \frac{\Delta-\snorm{E}}{\Delta-4\snorm{E}}.
\end{gather}
Finally, if $A$ and $E$ are real matrices, and $[X_1\, X_2]$ is a real orthogonal matrix, then $P_1, P_2$ are real matrices and so are $[R_{1}\, R_{2}]$, $[L_{1}\, L_{2}]$, $I$, and $O$.
\end{theorem}

Note that for $\snorm E \to \Delta/4$ the upper bound on $\snorm{P_2}$, and consequently $\snorm{L_{2}}\snorm{R_{2}}$, diverges, i.e., the perturbation argument does not work in this limit. 

\begin{proof}
The proof follows by checking the conditions of Theorem~\ref{thm:two-sided-perturbation-thm}. 
First, we need to bound the separation function between $\id$ and $\Lambda$.
Using $\sep(\id,0)=1$ and the stability of the separation function \eqref{eq:perturbation-separation-function}, we have
\begin{equation}
  \sep(\id, 0 + \Lambda) \geq \sep(\id,0) - \snorm{\Lambda} \geq  1 - (1-\Delta) = \Delta.
\end{equation}
Since isometries have unit spectral norm, we find
\begin{equation}\label{eq:Ebounds}
  \snorm{E_{11}},\snorm{E_{21}}, \snorm{E_{12}},\snorm{E_{22}} \leq \snorm{E}.
\end{equation}
By assumption, we then have
\begin{equation}\label{eq:deltaE_bound}
  \delta \coloneqq \sep(\id, \Lambda) - \snorm{E_{11}} - \snorm{E_{22}} \geq \Delta - 2 \snorm E > \frac {\Delta} 2 > 0 
\end{equation}
and 
\begin{equation}
  \frac{\snorm{E_{12}}\snorm{E_{21}}}{\delta^2} \leq 
  \left(\frac{\snorm E}{\Delta - 2 \snorm E} \right)^2 \leq \frac 14\,. 
\end{equation}
Thus, Theorem~\ref{thm:two-sided-perturbation-thm} implies the existence of matrices $P_1$ and $P_2$ such that Eqs.~\eqref{eq:perturbed-projector} and \eqref{eq:perturbed-basis} hold. 

Next, we establish the claimed norm bounds on $P_1$ and $P_2$.
Since $x-1 > \frac{x}{2}$ for any $x > 2$, we have with $x = \Delta/(2\snorm{E}) > 2$
\begin{equation}
\label{eq:intermediateP1bound}
\begin{aligned}
  \snorm{P_1} &\leq 2 \,\frac{\snorm{E_{21}}}{\delta}  
  &\leq 
  \frac{2\snorm{E}}{\Delta-2\snorm E}
  = \frac{1}{\Delta/(2\snorm{E}) - 1} \\
  &< 4 \,\frac{\snorm E}{\Delta}\, .
\end{aligned}
\end{equation}
Using Eq.~\eqref{eq:Ebounds} and repeatedly Eq.~\eqref{eq:deltaE_bound}, the bound \eqref{eq:bound_P2} on $\snorm{P_2}$ becomes
\begin{align}\label{eq:intermediateP2bound}
 \snorm{P_2} 
 &\leq 
 \frac{2\snorm{E}}{\delta - 4\snorm{E}^2/\delta}\\
 &\leq
 \frac{2\snorm{E}}{\Delta - 2 \snorm{E} - 4\snorm{E}^2/(\Delta-2\snorm{E})} \\
 &= \frac{2\snorm{E}(\Delta-2\snorm{E})}{(\Delta - 2 \snorm{E})^2 - 4\snorm{E}^2} \\
 &= \frac{2\snorm{E}(\Delta-2\snorm{E})}{\Delta^2-4\snorm{E}\Delta} 
 \leq \frac{2\snorm E}{\Delta - 4\snorm E}.\label{eq:P2bound}
\end{align}

We continue by deriving the remaining bounds. 
We have
\begin{align}
  \snorm{I-\id} 
  &\leq \snorm{E_{11}} + \snorm{E_{12}P_1} \\
  &< \snorm{E} + 4\frac{\snorm{E}^2}{\Delta} 
  < 2 \snorm{E}, \\
  \snorm{O-\Lambda} 
  &\leq \snorm {E_{22}} + \snorm {P_1 E_{12}} \\
  &<  \snorm E + 4 \frac{\snorm{E}^2}{\Delta}
  < 2 \snorm E. 
\end{align}
Then, by the stability of the separation function, 
\begin{equation}
 \sep(I,O) \geq \sep(\id,\Lambda) - 4 \snorm{E} \geq \Delta - 4 \snorm{E}.
\end{equation}
Using the bounds on $P_1$ and $P_2$, we can bound the basis change operators $L_{2} = X_2-X_1 P_1\ad$ and $R_{2} = X_1 P_2 + X_2P_1P_2 + X_2$ as
\begin{align}
  \snorm{R_{2}} &\leq 1 + \snorm{P_1}\snorm{P_2} + \snorm{P_2} \, &
  \snorm{L_{2}} &\leq 1+ \snorm{P_1}.
\end{align}
We use the first inequalities in \eqref{eq:intermediateP1bound} and \eqref{eq:intermediateP2bound} on $P_1$ and $P_2$ and then Eq.~\eqref{eq:deltaE_bound} to obtain
\begin{align}
  \snorm{P_1}\snorm{P_2}
  &\leq 
  \frac{4 \snorm{E}^2}{\delta^2-4\snorm{E}^2} \\
  &\leq 
  \frac{4 \snorm{E}^2}{(\Delta - 2\snorm{E})^2-4\snorm{E}^2}\\
  &=
  \frac{4 \snorm{E}^2}{\Delta^2 - 4 \Delta \snorm{E}}\, . 
  \label{eq:P1P2bound}
\end{align}
Together with the bound \eqref{eq:P2bound} on $P_2$ this yields 
\begin{align}
\MoveEqLeft[1]
 \snorm{L_{2}}\snorm{R_{2}}\\
 &\leq 1 + \left(2+\snorm{P_1}\right) \snorm{P_1}\snorm{P_2} + \snorm{P_1} + \snorm{P_2} \\
 &\leq 1 + \left(2 + 4\frac{\snorm{E}}{\Delta}\right) \frac{4 \snorm{E}^2}{\Delta^2 - 4 \Delta \snorm{E}} \\
 &\qquad
   +4\frac{\snorm{E}}{\Delta} + \frac{2\snorm{E}}{\Delta-4\snorm{E}} \\
 &= \left(1-4\frac{\snorm{E}}{\Delta}\right)\frac{\snorm{E}}{\Delta} + \frac{1+\snorm{E}/\Delta}{1-4\snorm{E}/\Delta}.
\end{align}
The final statement for the perturbation of real matrices follows readily from Rem.~\ref{rem:pert-theory-real-matrices}. 
\end{proof}

%%% =============================================
\section{Different estimators and their variances}
%%% =============================================
\label{sec:estimators}

For our analysis, we use a mean estimator, defined in Eq.~\eqref{eq:filtered-rb-data-estimator} as
\begin{align}
\label{eq:filtered-rb-estimator-app}
  \hat{\FD}_\lambda(m)%
  &= \frac{1}{N} \sum_{l=1}^N f_\lambda(i^{(l)}, g_1^{(l)} \cdots g_m^{(l)}) \, .
\end{align}
Importantly, we assume that samples are taking iid from the joint distribution $\dd p(i,g_1,\dots,g_m) = p(i|g_1,\dots,g_m)\dd\nu(g_1)\dots\dd\nu(g_m)$, i.e.~a random circuit is sampled from $\nu^{\times m}$, and then measured once to obtain a sample from the outcome distribution.

More generally, one may sample $N_C$ random circuits, and then measure each circuit $N_M$ times.
Such an approach can have experimental advantages depending on the platform, and the degree and flexibility of control automation.
In principle, running the same circuit many times can be optimized for high sampling rates on a hardware level, while changing the circuit in each run may pose an additional classical control overhead.
However, many modern quantum devices can be automatically programmed, thus performing a different circuit in each run is not necessarily more resource-intensive.
Nevertheless, using a smaller number of random circuits may crucially reduce the runtime and memory consumption required for post-processing in large-scale benchmarking approaches.

That being said, the statistics also behave differently.
If each circuit is sampled $N_M$ times, we obtain an estimator of the form
\begin{align}
\label{eq:filtered-rb-estimator-app-2}
  \hat{\mathcal{F}}_\lambda(m)%
  &= \frac{1}{N_C N_M} \sum_{l=1}^{N_C} \sum_{k=1}^{N_M} f_\lambda(i^{(k,l)}, g_1^{(l)} \cdots g_m^{(l)}) \, .
\end{align}
Note that the samples $i^{(k,l)}$ for the same $l$ are conditioned on the same circuit, and thus the variance of $\hat{\mathcal{F}}_\lambda(m)$ is not simply $\Var[\FD_\lambda]/N_C N_M$.
Instead, we have to use the law of total variance, which allows us to use conditional expected values and variances:
\begin{widetext}
\begin{align}
\Var[\hat{\mathcal{F}}_\lambda(m)]
&=
\frac{1}{N_C}
\bigg(
 \frac{1}{N_M} \EE_{g_j} \Var[f_\lambda(i,g_j)\,|\,g_j]  
+ 
\Var_{g_j} \EE_i[f_\lambda(i,g_j)\,|\,g_j] \bigg) \label{eq:var-filtered-rb-estimator-app-3} \\
&=
\frac{1}{N_C}
\bigg(
 \frac{1}{N_M} 
 \EE_{g_j} \left( \EE_i[f_\lambda(i,g_j)^2\,|\,g_j] - \EE_i[f_\lambda(i,g_j)\,|\,g_j]^2 \right) 
 + 
 \EE_{g_j}\left( \EE_i[f_\lambda(i,g_j)\,|\,g_j]^2 \right) - \left( \EE_{g_j} \EE_i[f_\lambda(i,g_j)\,|\,g_j] \right)^2
\bigg) \\
&=
\frac{1}{N_C}
\bigg(
 \frac{1}{N_M} \EE_{g_j}\EE_i[f_\lambda(i,g_j)^2\,|\,g_j] 
 + 
 \left(1-\frac{1}{N_M}\right) \EE_{g_j}\EE_i[f_\lambda(i,g_j)\,|\,g_j]^2 
 -
 \left( \EE_{g_j} \EE_i[f_\lambda(i,g_j)\,|\,g_j] \right)^2
\bigg) \, . \label{eq:var-filtered-rb-estimator-app-2}
\end{align}
\end{widetext}

Hence, $\Var[\hat{\mathcal{F}}_\lambda(m)]$ consists of three terms:
The first and second moment of the filter function, $\EE[\FD_\lambda]$ and $\EE[f_\lambda^2]$, respectively, which give the variance of the estimator \eqref{eq:filtered-rb-estimator-app} and are thus discussed in the main text.
The middle term in Eq.~\eqref{eq:var-filtered-rb-estimator-app-2} gives an additional contribution if $N_M > 1$, i.e.~if more than one sample is taken per circuit.
In terms of the random circuit ensemble, it corresponds to a \emph{fourth moment}, while the other two terms correspond to second and third moments.

In terms of sample complexity, it does not seem advantageous to choose $N_M >  1$.
For instance, if we sample Haar-randomly from the unitary group, then one can check that 
\begin{align}
 \EE_{g_j} \Var[f_\lambda(i,g_j)\,|\,g_j]
 &=
 2\, \frac{d^2-1}{(d+2)(d+3)} \,, \\
 \Var_{g_j} \EE_i[f_\lambda(i,g_j)\,|\,g_j] 
 &=
 4\, \frac{d-1}{(d+2)(d+3)} \,, 
 \\ 
 \Var[\FD_\lambda^2]
 &=
 2\,\frac{d-1}{d+2} \, .
\end{align}
(Some of these expressions are computed in App.~\ref{sec:sampling-complexity-3-designs}.)
Hence, comparing $\Var[\hat{\mathcal{F}}_\lambda(m)]$ with $\Var[\hat{\FD}_\lambda(m)] = \Var[\FD_\lambda^2]/N$ for $N=N_C N_M$, one can readily verify that $\Var[\hat{\mathcal{F}}_\lambda(m)]$ is larger for any value of $N_C$ and $N_M$.
More precisely, the difference is of order $O((1/N_C - 1/N)/d)$.
Hence if $d$ is small (say less than 10 qubits) and the total number of samples $N$ is fixed, choosing $N_M = 1$ is clearly optimal.
If instead $N_C$ is fixed, then increasing $N_M$ far beyond $d$ does not improve the accuracy of the estimator.
Moreover, if $d$ is very large, than the second term in Eq.~\eqref{eq:var-filtered-rb-estimator-app-3} is negligible and thus $\Var[\hat{\mathcal{F}}_\lambda(m)] = O(1/N) = \Var[\hat{\FD}_\lambda(m)]$.
In this regime, the difference between the two estimators in terms of sampling complexity can be neglected.

In principle, the arguments in Sec.~\ref{sec:sampling-complexity} can be adapted to also treat the middle term in Eq.~\eqref{eq:var-filtered-rb-estimator-app-2} perturbatively, as we have done it in Thm.~\ref{thm:second-moment}.
However, in this case many more irreps, and thus terms, appear, which makes the analysis more difficult.
We think that a sampling complexity theorem for the estimator $\Var[\hat{\mathcal{F}}_\lambda(m)]$ similar to Thm.~\ref{thm:sampling-complexity-additive} can be formulated, but we expect the guarantees to be worse than for $\Var[\hat{\FD}_\lambda(m)]$.

%%% =============================================
\section{Sampling complexity for ideal implementations}
%%% =============================================

\label{sec:sampling_complexity_ideal_case}

In Sec.~\ref{sec:sampling-complexity}, Thm.~\ref{thm:sampling-complexity-relative}, we have established bounds on the sampling complexity of filtered \ac{RB}.
These involve the second moments of the estimator if an ideal implementation is used for which $\nu=\mu$ is the Haar measure and $\phi=\omega$ is the reference representation.
This ideal second moment is given by Eq.~\eqref{eq:sum-trace-Csigma} as follows
\begin{multline}
  \EE[f_\lambda^2]_\SPAM \\
  =
  \osandwichb{\rho^{\otimes 2}}%
  {(X_\lambda S_\lambda^+)^{\otimes 2}\,
  \widehat{\omega}[\omega_{\lambda}^{\otimes 2}] \big( X_\lambda^{\dagger\otimes 2} \tilde M_3 \big) }%
  {\tilde\rho}.
  \label{eq:rb-data-second-moment}
\end{multline}
Here,
\begin{equation}{}
  \tilde M_3 
  \coloneqq 
  \sum_{i\in [d]} \oketbra*{E_i\otimes E_i}{\tilde E_i}
  =
  \sum_{i\in [d]} \oketbra*{E_i\otimes E_i}{E_i}\EM ,
\end{equation}
is an operator associated with our measurement basis $E_i=\ketbra{i}{i}$ and the measurement noise $\EM$.
The second moment involves a projector of the form
\begin{equation}
\label{eq:second-moment-fourier-operator}
 \widehat{\omega}[\omega_{\lambda}^{\otimes 2}] = \int_G \omega_\lambda(g)^\dagger\otimes\omega_\lambda(g)^\dagger (\argdot) \,\omega(g) \dd\mu(g).
\end{equation}
Note that the above statements also hold if $\nu$ is an appropriate design, namely a $\bar\omega_\lambda\otimes\bar\omega_\lambda\otimes\omega$-design.
Generally, this is a weaker assumption than assuming a $\omega^{\otimes 3}$-design.

Analogous to the proof of Prop.~\ref{prop:fourier-transform}, only the irreps $\tau_{\sigma}\subset\omega$ which appear in $\tau_\lambda\otimes\tau_\lambda$ contribute to this projection.
Compared to Prop.~\ref{prop:fourier-transform}, we here use a different formulation which is more practical in the concrete cases below.
Let $\Sigma\coloneqq \Irr(\omega_\lambda\otimes\omega_\lambda) \cap\Irr(\omega)$ be the set of common irreps and let $n_\sigma$ and $m_\sigma$ be the multiplicities of $\tau_\sigma$ in $\omega$ and $\omega_\lambda\otimes\omega_\lambda$, respectively.
Let us label the copies as $\tau_\sigma^{(i)}$.
We find
\begin{align}
 \widehat{\omega}[\omega_{\lambda}^{\otimes 2}]  
 &= 
 \bigoplus_{\sigma\in\Sigma} 
 \bigoplus_{i=1}^{m_\sigma}
 \bigoplus_{j=1}^{n_\sigma}
 \int_G \tau_{\sigma}^{(i)}(g)^\dagger (\argdot) \,\tau_{\sigma}^{(j)}(g)  \dd\mu(g) \\
 &=
 \bigoplus_{\sigma}
 \bigoplus_{i=1}^{m_\sigma}
 \bigoplus_{j=1}^{n_\sigma}
 \oketbra{I_\sigma^{(i,j)}}{I_\sigma^{(i,j)}}, \label{eq:second-moment-contributions}
\end{align}
where $I_\sigma^{(i,j)}$ is a $G$-equivariant isometry between the irreps $\tau_{\sigma}^{(j)}$ and $\tau_{\sigma}^{(i)}$, a so-called \emph{intertwiner}.
In the following, we construct $G$-equivariant isomorphisms $I_\sigma^{(i,j)}$ which only after normalization with $\twonorm{I_\sigma^{(i,j)}}$ become an isometry.
We still refer to such maps as \emph{intertwiners}.

Unfortunately, not much more can be said about the common irreps without fixing a specific group and representation.
In general, $\tau_\lambda\otimes\tau_\lambda$ contains a trivial irrep if and only if $\tau_\lambda$ is a real representation.
For our applications, this is the case, hence the rank of $\widehat{\omega}[\omega_{\lambda}^{\otimes 2}]$ is at least the product of the multiplicity $n_\lambda$ of $\tau_\lambda$ in $\omega$ with the multiplicity of the trivial irrep.

In the following, we consider some concrete cases and assume that $G < \U(d)$.
The reference representation is taken as $\omega(g)\coloneqq U_g (\argdot) U_g^\dagger$ where $g\mapsto U_g$ is the defining representation of $\U(d)$.

The following can alternatively be deduced with appropriate Haar integration formulas, and has already been partially calculated elsewhere, for instance in Ref.~\cite{huang_predicting_2020}.
Here, we give a self-contained derivation in consistent notation which might be useful to readers not familiar with other works and can serve as a basis for analogous calculations for other groups.

\subsection{Second moment for unitary 3-designs}
\label{sec:sampling-complexity-3-designs}

If $G$ is a unitary 3-design, it is in particular a 2-design and hence $\omega$ decomposes into the trivial irrep $\tau_1$ and the adjoint irrep $\tau_\Ad$.
We fix $\omega_\lambda \equiv \tau_\Ad$.
The 3-design property implies that the integral in Eq.~\eqref{eq:second-moment-fourier-operator} is the same as over $\U(d)$.
For $\U(d)$, the representation $\tau_\Ad\otimes \tau_\Ad$ contains the trivial irrep with multiplicity one.
A general argument implies that $\tau_\Ad\otimes \tau_\Ad$ has to contain the adjoint irrep, too.
To find its multiplicity, we compute the character inner product.
In the following, $\chi$ and $\chi_\Ad$ denote the characters of $\omega$ and $\tau_\Ad$, respectively.
\begin{align}
 m_\Ad 
 \coloneqq \langle \chi_\Ad^2, \chi_\Ad \rangle 
 &= \langle (\chi-1)^2, \chi_\Ad \rangle = \langle \chi^2, \chi_\Ad \rangle - 2 \\
 &= \langle \chi^3, 1 \rangle - \langle \chi^2, 1 \rangle - 2 
 = 
 \begin{cases}
  1 & \text{if } d = 2, \\
  2 & \text{if } d \geq 3.  
 \end{cases} \label{eq:mult-adjoint-irrep}
\end{align}
Here, we multiply used $\chi = \chi_\Ad+1$ and the following combinatorical identity \cite{diaconis_eigenvalues_1994,rains_increasing_1998,scott_optimizing_2008}:
\begin{equation}
 \langle \chi^t,1\rangle 
 = \int_{U(d)} \abs{\tr U_g}^{2t} \dd\mu(g) 
 =
 \begin{cases}
  \frac{(2t)!}{t!(t+1)!} & \text{if } d = 2, \\
  t!  & \text{if } d \geq t.
 \end{cases}
\end{equation}
In summary, we find
\begin{equation}
 \rank \widehat{\omega}[\tau_\Ad^{\otimes 2}] =
 \begin{cases}
  2 & \text{if } d = 2, \\
  3  & \text{if } d \geq 3.
 \end{cases}
\end{equation}
Next, we derive the precise contributions of these irreps to $\widehat{\omega}[\tau_\Ad^{\otimes 2}]$ by Eq.~\eqref{eq:second-moment-contributions} and argue that the overlaps of the remaining terms in the second moment \eqref{eq:rb-data-second-moment} with these subspaces are small.

\paragraph{Contribution from the trival irreps.} 
Note that $\tau_\Ad\otimes \tau_\Ad$ is the restriction of $\omega\otimes\omega$ to the tensor square $M_0(d)^{\otimes 2}$ of complex traceless $d\times d$ matrices.
The trivial subrepresentation of $\omega\otimes\omega$ is spanned by the identity matrix $\one\otimes\one$ and the flip $F$ which can be written as $F  = F_0 + \frac{\one\otimes\one}{d}$ for a matrix $F_0 \in M_0(d)^{\otimes 2}$.
Hence, the trivial irrep of $\tau_\Ad\otimes \tau_\Ad$ is spanned by $F_0$ and thus equivalent to the trivial irrep of $\omega$, spanned by $\one$, under the intertwiner
\begin{equation}
 I_1 \coloneqq  \oketbra{F_0}{\one}.
\end{equation}
Indeed, $I_1/\twonorm{I_1}$ is an $G$-equivariant isometry.
The contribution of the trivial irreps in $\tau_\Ad\otimes \tau_\Ad$ and $\omega$ is then given as the orthogonal projection onto $I_1$:
\begin{align}
 \Pi_1(X) 
 &\coloneqq  \int_G \tau_\Ad(g)^\dagger \otimes \tau_\Ad(g)^\dagger (\argdot) \tau_1(g) \dd\mu(g) \\
 &= \frac{\obraket{I_1}{X}}{\obraket{I_1}{I_1}} \, I_1, \\
 \obraket{I_1}{I_1} 
 &= d(d^2-1).
\end{align}
Here, the normalization follows from $\twonorm{\one}^2 = d$ and $\twonorm{F_0}^2 = \twonorm{F}^2 - d^{-2} \twonorm{\one\otimes\one} = d^2-1$.

The overlap of $(X_\Ad^\dagger)^{\otimes 2} \tilde M_3$ with $I_1$ can be computed using the swap trick as follows
\begin{align}
\MoveEqLeft[1]
 \obraket{I_1}{(X_\Ad^\dagger)^{\otimes 2} \tilde M_3}\\
 &=
 \sum_{i=1}^d \obraket{\tilde E_i}{\one} \obraket{(X_\Ad^\dagger E_i)^{\otimes 2}}{F_0} \\
 &= \sum_{i=1}^d \tr(\tilde E_i) \tr\left[ \left(E_i - \frac{\one}{d}\right)^{\otimes 2} \left( F - \frac{\one\otimes\one}{d} \right) \right] \\
 &= \sum_{i=1}^d \tr(\tilde E_i) \left[ \tr\left(E_i - \frac{\one}{d}\right)^{2} - \frac{1}{d} \left(\tr\left(E_i -\frac{\one}{d} \right)\right)^2  \right] \\
 &= \sum_{i=1}^d \tr(\tilde E_i) \left( 1- \frac{1}{d} \right) \\
 &= \left( 1- \frac{1}{d} \right) \tr \EM^\dagger\left( \sum_i \ketbra{i}{i} \right)  \\
 &= ( d - 1 ) \frac{\tr\tilde\one}{d} \leq d-1 \, .
\end{align}
In the last step, we use $\sum_i E_i = \sum_i \ketbra{i}{i} = \one$ and the definition $\tilde\one\coloneqq\EM^\dagger(\one)$.
The final bound follows since the measurement noise $\EM$ is trace non-increasing and thus $\tr\tilde\one \leq \tr\one = d$.
In particular, if $\EM$ is trace-preserving, we have equality.
Hence, we have the following for the projection of $(X_\Ad^\dagger)^{\otimes 2} \tilde M_3$ onto the trival contribution. 
Recall that we have found $S_\Ad^+ = (d+1)\id_\Ad$ in Sec.~\ref{sec:frame-operator}.
\begin{align}
\left(S_\Ad^+\right)^{\otimes 2} \Pi_1((X_\Ad^\dagger)^{\otimes 2} \tilde M_3) 
&=
\frac{\tr\tilde\one}{d}\frac{(d-1)(d+1)^2}{d(d^2-1)} \oketbra{F_0}{\one} \\
&=
\frac{\tr\tilde\one}{d}\frac{d+1}{d} \oketbra{F_0}{\one}. 
\end{align}
Hence, the total contribution by the trivial irrep is
\begin{align}
\MoveEqLeft[2]
\osandwich{\rho^{\otimes 2}}
{X_\Ad^{\otimes 2}\left(S_\Ad^+\right)^{\otimes 2} \Pi_1((X_\Ad^\dagger)^{\otimes 2} \tilde M_3)}
{\tilde\rho} \\
&=
 \frac{\tr\tilde\one}{d}\frac{d+1}{d} \osandwich{\rho^{\otimes 2}}{X_\Ad^{\otimes 2}}{F_0} \obraket{\one}{\tilde\rho} \\
 &=
 \frac{\tr\tilde\one}{d}\frac{d+1}{d} \tr\left( \left(\rho-\frac{\one}{d}\right)^{\otimes 2} \left( F - \frac{\one\otimes\one}{d}\right) \right) \\
 &= 
 \frac{\tr\tilde\one}{d} \frac{d^2-1}{d^2}
 \leq
 \frac{d^2-1}{d^2}
 \label{eq:sampling-complexity-3-designs-trivial}
\end{align}

\paragraph{Contribution from the adjoint irreps.} 
The two adjoint irreps in $\tau_\Ad\otimes\tau_\Ad$ come from the identification of certain product operators $A\otimes B\in M_0(d)\otimes M_0(d)$ with their product $AB\in M_0(d)$ or $BA\in M_0(d)$ .
Clearly, under such maps, the representation $\tau_\Ad\otimes\tau_\Ad$ corresponds to $\tau_\Ad$.

To makes this precise and explicit, we describe the irreps in the orthogonal basis of Weyl operators.
In contrast to Sec.~\ref{sec:clifford-group}, we here use the `modular definition' of Weyl operators which is available in any dimension $d$.\footnote{However, the resulting groups do not have such nice properties as in the `finite field definition' in prime-power dimensions. Nevertheless, we here only need an orthogonal basis for the traceless subspace.}
To this end, we label the measurement basis as $\ket{x}$ where $x\in\Z_d$.
The Weyl operators are then defined as $w(x,z)\ket{y}\coloneqq \xi^{zy}\ket{y+x}$ where $z,x,y\in\Z_d$, $\xi$ is a primitive $d$-th root of unity, and all operations are performed in the ring $\Z_d$.
It is convenient to group the arguments as $a=(z,x)\in\Z_d^2$.
The Weyl operators are unitary, traceless if $a\neq 0$, and $w(a)w(b)\propto w(a+b)$.

Consider the following linear maps $M_0(d)\rightarrow M_0(d)\otimes M_0(d) $ defined in the Weyl basis:
\begin{align}
 I_\Ad^{(1)} &\coloneqq  \sum_{\substack{a,b\in\Z_d^2\setminus 0 \\ a+b \neq 0}} \oketbra{w(a)\otimes w(b)}{w(a)w(b)}, \\
 I_\Ad^{(2)} &\coloneqq  \sum_{\substack{a,b\in\Z_d^2\setminus 0 \\ a+b \neq 0}} \oketbra{w(a)\otimes w(b)}{w(b)w(a)}.
\end{align}
These maps identify a traceless Weyl operator $w(c)\in M_0(d)$ with the uniform superposition of all $w(a)\otimes w(b) \in M_0(d)\otimes M_0(d)$ such that $a+b = c$.
Since the sum is commutative, there are two choices of attributing the summands to tensor factors, leading to the two maps $I_\Ad^{(1)}$ and $I_\Ad^{(2)}$.
It is straightforward to check (e.g.~by computing $I_\Ad^{(i)\dagger}I_\Ad^{(i)}$) that both maps are isometries, up to normalization by $\twonorm{I_\Ad^{(i)}}$.
The two linear maps are linearly independent, except for the case $d=2$: 
There, the traceless Pauli operators $X,Y,Z$ mutually anti-commute and hence 
\begin{equation}
 I_\Ad^{(1)} = - I_\Ad^{(2)} \qquad \text{(for $d=2$)}.
\end{equation}
We claim that the images of $I_\Ad^{(1)}$ and $I_\Ad^{(2)}$ correspond to a single ($d=2$) or two copies ($d\geq 3$) of $\tau_\Ad$ in $\tau_\Ad\otimes\tau_\Ad$.
In the latter case, the two copies can be mapped onto each other by permuting the tensor factors.
Note that this is in perfect alignment with the previously computed multiplicity \eqref{eq:mult-adjoint-irrep}.

The claim can be verified by a simple, but somewhat lengthy calculation which shows that $I_\Ad^{(1)}$ and $I_\Ad^{(2)}$ are indeed fixed by the representation $\tau_\Ad^\dagger \otimes \tau_\Ad^\dagger (\argdot) \tau_\Ad$.
For $d=2$, we can take either of the two intertwiners, say $I_\Ad^{(1)}$, as a basis for the range of the projector
\begin{equation}
 \Pi_\Ad 
 \coloneqq  
 \int_G \tau_\Ad(g)^\dagger \otimes \tau_\Ad(g)^\dagger (\argdot) \tau_\Ad(g) \dd\mu(g)  \,.
\end{equation}
For $d\geq 3$, the two intertwiners are linearly independent, but do not form an orthogonal basis, as one can readily compute:
\begin{align}
 \obraket{ I_\Ad^{(1)} }{ I_\Ad^{(2)} } 
 &= d^2 \sum_{\substack{a,b\in\Z_d^2\setminus 0 \\ a+b \neq 0}} \obraket{w(b)w(a)}{w(a)w(b)} \\
 &= d^3 \sum_{a\neq 0} \sum_{\substack{b\neq 0 \\ b \neq -a}} \xi^{[a,b]} \\
 &= d^3 \sum_{a\neq 0} ( d^2 \delta_{a,0} - 2 ) \\
 &= - 2 d^3 (d^2-1).
\end{align}
Here, $\xi = e^{2\pi i/d}$ and $[a,b] = a_1b_2 - a_2 b_1$ is the symplectic form measuring whether $w(a)$ and $w(b)$ commute.
In the second to last step we used that character orthogonality implies $\sum_{b\in\Z_d^2} \xi^{[a,b]} = d^2 \delta_{a,0}$.
Likewise, we find 
\begin{align}
 \obraket{ I_\Ad^{(i)} }{ I_\Ad^{(i)} } 
 &= d^2 \sum_{\substack{a,b\in\Z_d^2\setminus 0 \\ a+b \neq 0}} \obraket{w(a)w(b)}{w(a)w(b)} \\
 &= d^3 (d^2-1)(d^2-2).
\end{align}
Hence, for $d\geq 3$, we use Gram-Schmidt orthogonalization on the intertwiners to express the projector: 
\begin{align}
 \tilde I_\Ad^{(2)} 
 &\coloneqq  I_\Ad^{(2)} - \frac{\obraket{ I_\Ad^{(2)} }{ I_\Ad^{(1)} } }{d^3 (d^2-1)(d^2-2)} I_\Ad^{(1)} 
 = I_\Ad^{(2)} + \frac{2}{d^2-2} I_\Ad^{(1)} \\
 \MoveEqLeft[2]
 \obraket{\tilde I_\Ad^{(2)} }{\tilde I_\Ad^{(2)} } \\
 &= \obraket{ I_\Ad^{(2)} }{ I_\Ad^{(2)} } + \frac{4}{d^2-2}\obraket{ I_\Ad^{(2)} }{ I_\Ad^{(1)} } \\
 &\qquad
  + \frac{4}{(d^2-2)^2}\obraket{ I_\Ad^{(1)} }{ I_\Ad^{(1)} } \\
 &= d^3 (d^2-1)(d^2-2) - \frac{8 d^3 (d^2-1)}{d^2-2} + \frac{4 d^3 (d^2-1)}{d^2-2} \\
 &= d^3 (d^2-1)(d^2-2) \left( 1 - \frac{4}{(d^2-2)^2} \right) \\
 &= d^3 (d^2-1)(d^2-2) \frac{ d^2 (d^2 - 4 ) }{(d^2-2)^2}.
\end{align}
For our purposes, it is enough to compute the projection of superoperators $L:\,M_0(d)\rightarrow M_0(d)\otimes M_0(d) $ which are invariant under the permutation of tensor factors of $M_0(d)\otimes M_0(d)$, this is $\pi L = L$ for $\pi \in S_2$.
Since $I_\Ad^{(2)} = \pi I_\Ad^{(1)}$ for $\pi\neq\id$, it is then immediate that $\obraket{I_\Ad^{(2)}}{L}=\obraket{I_\Ad^{(1)}}{L}$ and thus we have the following expansion for $d\geq 3$:
\begin{align}
 \Pi_\Ad(L) 
 &=
 \frac{\obraket{I_\Ad^{(1)}}{L}}{\obraket{ I_\Ad^{(1)} }{ I_\Ad^{(1)}}} I_\Ad^{(1)} 
 +
 \frac{\obraket{\tilde I_\Ad^{(2)}}{L}}{\obraket{\tilde I_\Ad^{(2)} }{\tilde I_\Ad^{(2)}}} \tilde I_\Ad^{(2)} \\
 &= \frac{\obraket{I_\Ad^{(1)}}{L}}{d^3 (d^2-1)(d^2-2)} \bigg( I_\Ad^{(1)} + \\
 &\quad
  +  \frac{(d^2-2)^2}{ d^2 (d^2 - 4 ) } \left( 1 +  \frac{2}{d^2-2} \right) \tilde I_\Ad^{(2)}\bigg) \\
 &= \frac{\obraket{I_\Ad^{(1)}}{L}}{d^3 (d^2-1)(d^2-2)} \bigg( I_\Ad^{(1)} + \\
 &\quad
  + \frac{d^2-2}{ d^2 - 4 } \left( I_\Ad^{(2)} + \frac{2}{d^2-2} I_\Ad^{(1)}\right) \bigg) \\
 &= \frac{\obraket{I_\Ad^{(1)}}{L}}{d^3 (d^2-1)(d^2-2)} \left( \frac{d^2-2}{d^2-4} I_\Ad^{(1)} +  \frac{d^2-2}{ d^2 - 4 } I_\Ad^{(2)}  \right) \\
 &= \frac{\obraket{I_\Ad^{(1)}}{L}}{d^3 (d^2-1)(d^2-4)}  \left( I_\Ad^{(1)} +  I_\Ad^{(2)}  \right).
 \label{eq:3-design-ad-projection}
\end{align}
We find the following for the overlap of $(X_\Ad^\dagger)^{\otimes 2} \tilde M_3$ with $I_\Ad^{(1)}$:
\begin{align}
\MoveEqLeft[1]
 \obraket{I_\Ad^{(1)}}{(X_\Ad^\dagger)^{\otimes 2} \tilde M_3} \\
 &= \sum_{\substack{a,b\in\Z_d^2\setminus 0 \\ a+b \neq 0}} \sum_{y\in\Z_d} 
   \obraket{w(a)\otimes w(b)}{(X_\Ad^\dagger E_y)^{\otimes 2}} \obraket{\tilde E_y}{w(a)w(b)} \\
 &= \sum_{\substack{a,b\in\Z_d^2\setminus 0 \\ a+b \neq 0}} \sum_{y\in\Z_d}
  \tr(w(a)^\dagger E_y)\tr(w(b)^\dagger E_y) \obraket{\tilde E_y}{w(a)w(b)}\\
 &= \sum_{\substack{z,z'\in\Z_d\setminus 0 \\ z+z' \neq 0}} \sum_{y\in\Z_d}
  \omega^{-(z+z')y} \obraket{\tilde E_y}{Z(z+z')} \\
 &= \sum_{\substack{z,z'\in\Z_d\setminus 0 \\ z+z' \neq 0}} 
 \Bigg( \EM^\dagger\bigg( \sum_{y\in\Z_d} \omega^{(z+z')y} \ketbra{y}{y} \bigg) \bigg| Z(z+z') \Bigg) \\
 &= \sum_{\substack{z,z'\in\Z_d\setminus 0 \\ z+z' \neq 0}} \osandwich{Z(z+z')}{\EM}{Z(z+z')} \\
 &= (d-2) \sum_{z\in\Z_d\setminus 0} \osandwich{Z(z)}{\EM}{Z(z)} \\
 &= d(d-2)\,\tr(P_\Ad  \tilde M),
\label{eq:3-design-ad-projection-overlap}
\end{align}
where we have used that $M = d^{-1} \sum_{z\in\Z_d}\oketbra{Z(z)}{Z(z)}$ and $\tilde M = M \EM$.
The expression $\tr(P_\Ad \tilde M)$ is discussed in the main text. In particular for $\lambda$-multiplicity free and aligned with M, we found  the upper bound $\tr(P_\Ad \tilde M) \leq \tr(P_\Ad M) = d-1$.

Using Eq.~\eqref{eq:3-design-ad-projection} with $L=(X_\Ad^\dagger)^{\otimes 2} \tilde M_3$, we find the following expression for the projection onto the adjoint contribution for $d\geq 3$:
\begin{equation}
 \Pi_\Ad\left( (X_\Ad^\dagger)^{\otimes 2} \tilde M_3 \right)
 = \frac{\tr(P_\Ad \tilde M)}{d^2 (d^2-1)(d+2)}  \left( I_\Ad^{(1)} +  I_\Ad^{(2)}  \right).
 \label{eq:3-design-Pi-ad-on-M3}
\end{equation}
Recall that for $d=2$, the projection is instead given by projecting onto $I_\Ad^{(1)}$ only.
However, since $\obraket{I_\Ad^{(1)}}{L}=\obraket{I_\Ad^{(2)}}{L}=-\obraket{I_\Ad^{(1)}}{L}$ for any symmetric $L$, the overlap has to vanish.
Applied to $L=(X_\Ad^\dagger)^{\otimes 2} \tilde M_3$ this immediately shows that the adjoint contribution to the second moment is zero.

Finally, multiplying Eq.~\eqref{eq:3-design-Pi-ad-on-M3} with $S_\Ad^+ = (d+1)\id_\Ad$ does only change the prefactor such that we get the contribution of the adjoint irreps to the second moment by contracting Eq.~\eqref{eq:3-design-Pi-ad-on-M3} with $\rho^{\otimes 2}$ and $\tilde\rho$.
Note that the left hand side is again symmetric and thus the contributions of $I_\Ad^{(1)}$ and $I_\Ad^{(2)}$ are identical.
What is left is to compute the overlap of $I_\Ad^{(1)}$ with the states.
Here, we again assume that $\rho$ is pure, and let $\rho_0 = X_\Ad^\dagger(\rho) \simeq \rho - \one/d$ be its traceless part.
From the definition of $I_\Ad^{(1)}$, we have $\obra{\rho_0^{\otimes 2}}I_\Ad^{(1)} = d^2 \obra{\rho_0^2}X_\Ad$ and the purity of $\rho$ implies $X_\Ad^\dagger(\rho_0^2) = (d-2)/d \, \rho_0$.
Hence, we find
\begin{equation}
 \osandwich{\rho^{\otimes 2}}{X_\Ad^{\otimes 2}  I_\Ad^{(1)} }{\tilde\rho}
 =
 d^2 \osandwich{\rho_0^{2}}{X_\Ad}{\tilde\rho}
 =
 d(d-2) \osandwich{\rho}{P_\Ad}{\tilde\rho}.
\end{equation}
Combining the above results, we obtain
\begin{align}
\MoveEqLeft[2]
\osandwich{\rho^{\otimes 2}}{X_\Ad^{\otimes 2} \left(S_\Ad^+\right)^{\otimes 2} \Pi_\Ad\left( (X_\Ad^\dagger)^{\otimes 2} \tilde M_3 \right) }{\tilde\rho} \\
 &= \frac{2\tr(P_\Ad \tilde M)(d+1)}{d^2 (d-1)(d+2)} \osandwich{\rho^{\otimes 2}}{X_\Ad^{\otimes 2}  I_\Ad^{(1)} }{\tilde\rho} \\
 &= \frac{2\tr(P_\Ad \tilde M)(d+1)(d-2)}{d (d-1)(d+2)} \osandwich{\rho}{P_\Ad}{\tilde\rho} \, .
 \label{eq:sampling-complexity-3-designs-adjoint}
\end{align}

\paragraph{Summing the contributions.}
We have found that the second moment for unitary 3-designs with SPAM noise can be expressed as
\begin{align}
\MoveEqLeft[2]
  \EE[f_\Ad^2]_\SPAM \\
  &=
  \frac{\tr\tilde\one}{d} \frac{d^2-1}{d^2}
  +
  \frac{2\tr(P_\Ad \tilde M)(d+1)(d-2)}{d (d-1)(d+2)} \osandwich{\rho}{P_\Ad}{\tilde\rho} \label{eq:2nd-moment-3-design-SPAM} \\
  &\leq 
  1 - \frac{1}{d^2}
  +
  \frac{2(d+1)(d-1)(d-2)}{d^2(d+2)} \label{eq:2nd-moment-3-design} \\
  &\leq
  \begin{cases}
    \frac{3}{4} & \text{if } d=2, \\
    3 - \frac{1}{d^2} & \text{if } d\geq 3.
  \end{cases} \label{eq:2nd-moment-3-design-bound}
\end{align}
In the second-to-last step we used $\tr\tilde\one\leq d$, $\osandwich{\rho}{P_\Ad}{\tilde\rho} \leq 1 - 1/d$, and the bound $\tr(P_\Ad \tilde M) \leq d-1$ from before.
Note that the right hand side of Eq.~\eqref{eq:2nd-moment-3-design} is exactly the expression for $\EE[f_\Ad^2] _\ideal$, i.e.~the second moment in the absence of SPAM noise.

Subtracting the square of the first moment, $(d-1)^2/d^2$, we thus find the following exact expression for the variance, in the absence of SPAM noise:
\begin{equation}
 \Var[f_\Ad] 
 =
 2\,
 \frac{d-1}{d+2} \, .
\end{equation}

\paragraph{Analysis of SPAM visibilities}
From Eq.~\eqref{eq:2nd-moment-3-design}, we can observe that the noise-free contributions coming from the trivial and the adjoint irreps in $\tau_\Ad^{\otimes 2}$ are modulated by the SPAM visibilities for the first moments.
In particular, we have
\begin{align}
 \frac{\tr\tilde\one}{d} &= \frac{\tr(P_1 \tilde M)}{d} = v_{\mathrm{M},1} \,, &
 \frac{\tr(P_\Ad\tilde M)}{d-1} &= v_{\mathrm{M},\Ad} \,, \\
 \osandwich{\rho}{P_\Ad}{\tilde\rho} &= v_{\mathrm{SP},\Ad} \frac{d-1}{d} \, .
\end{align}
We have the tight bounds
\begin{align}
0 &\leq v_{\mathrm{M},1} \leq 1 \, , &
- \frac{1}{d-1} &\leq  v_{\mathrm{SP},\Ad} \leq 1 \, , \\
- \frac{1}{d-1} &\leq v_{\mathrm{M},\Ad} \leq 1 \, .
\end{align}
The first two lower bounds are saturated for $\EM$ being the `discard' operation and $\tilde\rho$ being orthogonal to $\rho$, respectively.
Moreover, we have
\begin{equation}
 \tr(P_\Ad \tilde M) = \frac{1}{d} \sum_{z\neq 0} \osandwich{Z(z)}{\EM}{Z(z)}.
\end{equation}
Clearly, we can replace $\EM$ with its projection onto Weyl channels.
If $\EM = w(a)(\argdot)w(a)^\dagger$, then we find 
\begin{align}
 \frac{1}{d} \sum_{z\neq 0} \osandwich{Z(z)}{\EM}{Z(z)}
 &=
 \sum_{z\neq 0} \xi^{z\cdot a_x}
 =
 \begin{cases}
  d - 1, & \text{if } a_x = 0, \\
  -1, & \text{else}.
 \end{cases}
\end{align}
Hence, the lowest value can be achieved if $\EM$ is a convex combination of $X$-type Weyl operators, which shows the last lower bound.

We see that the contribution from the trivial irrep in Eq.~\eqref{eq:2nd-moment-3-design} cannot be negative, but the one from the adjoint irrep can, with lower bound $-\frac{2(d+1)(d-2)}{d^2(d+2)} \geq -\frac{2}{d}$.

\subsection{Second moment for local unitary 3-designs}
\label{sec:sampling-complexity-local-3-designs}

Next, we assume that $G$ is a local unitary 3-design, \ie~$G$ factorizes as $G = G_\loc^{\times m}$ where each copy of $G_\loc$ is a unitary 3-design acting on $\C^q$, and $d=q^m$ (e.g.~the single-qubit Clifford group $G_\loc=\Cl{1}{2}$).\footnote{The following arguments hold with minor adaptations if $G = G_1\times\dots\times G_m$ where each $G_i$ is a unitary 3-group, acting on Hilbert spaces of not necessarily equal dimensions.}
The irreps of $G$ are then simply the tensor products of the irreps of $G_\loc$ and hence we can label them by a binary vector $b\in\{0,1\}^m$ where $b_i=1$ and $b_i=0$ correspond to the trivial and adjoint irrep on the $i$-th factor, respectively.

By Eq.~\eqref{eq:second-moment-fourier-operator}, we have to decompose $\tau_b\otimes\tau_b$ into the relevant irreps $\tau_a$ appearing in $\omega$.
To this end, it is convenient to write $\tau_b\otimes \tau_b \simeq \bigotimes_{i=1}^m \tau_{b_i}\otimes\tau_{b_i}$.
Then, $\tau_{b_i}\otimes\tau_{b_i} = 1$ if $b_i=1$ and otherwise contains exactly one copy of the trivial irrep and $m_\Ad$ copies of the adjoint irrep, where $m_\Ad=1$ if $q=2$ and $m_\Ad=2$ else (c.f.~Eq.~\eqref{eq:mult-adjoint-irrep}).
Hence, the relevant irreps in $\tau_b\otimes\tau_b$ are labelled by $a \in \{0,1\}^m$ where $a_i = 1$ if $b_i=1$ and otherwise arbitrary.
If $a\circ b$ denotes the bitwise (Hadamard) product, then we can formulate this condition as $|a \circ b|=|b|$.
The multiplicity of the irrep $\tau_a$ is given as $m_a \coloneqq  m_\Ad^{|\bar a|}$.
Thus, we arrive at the decomposition
\begin{equation}
 \tau_b\otimes \tau_b 
 \simeq 
 \bigotimes_{i=1}^m \tau_{b_i}\otimes\tau_{b_i}
 \simeq
 \bigoplus_{a:\, |a \circ b| = |b|} \tau_a^{\oplus m_a }\; \oplus \; \text{irrelevant irreps}\, .
 \label{eq:tau-b-square-decomp}
\end{equation}

We treat each of the $2^{|\bar b|}$ irreps $\tau_a$ individually.
The rank of the projector $\widehat{\omega}[\tau_a^{\oplus m_a }]$ is $m_a = m_\Ad^{|\bar a|}$ and we can construct a basis for its range using tensor products of the intertwiners in Sec.~\ref{sec:sampling-complexity-3-designs}.
To this end, we also need to define a (local) intertwiner between $\tau_1$ and $\tau_1\otimes\tau_1$ which we can take as $J_1\coloneqq \oketbra{\one\otimes\one}{\one}$.

For now, let us assume that the first $|b|$ bits of $b$ are set and the remaining ones are zero, and the same holds for $a$, \ie~its first $|a|\geq |b|$ bits are set and otherwise zero.
Similar to Sec.~\ref{sec:sampling-complexity-3-designs}, we want to assume that $L:\,M(q^m)\rightarrow M(q^m)\otimes M(q^m)\simeq \bigotimes_{i=1}^m M(q)\otimes M(q) $ is symmetric in the sense that it is left-invariant under permutations $\bigotimes_{i=1}^m \pi_i$ for $\pi_i\in S_2$ that permute the factors on every qudit.
In this case, we can write
\begin{multline}
 \widehat{\omega}[\tau_a^{\oplus m_a }] (L)
 = \\
 \frac{\obraket{J_1^{\otimes |b|}\otimes I_1^{\otimes |a|-|b|} \otimes I_\Ad^{(1)\otimes |\bar a|}}{L}}{q^{3|b|}(q(q^2-1))^{|a|-|b|}(q^3 (q^2-1)(q^2-2m_\Ad))^{|\bar a|}} \\
 \times J_1^{\otimes |b|} \otimes I_1^{\otimes |a|-|b|}  \otimes  I_\Ad^{\otimes |\bar a|} \, ,
 \label{eq:local-3-design-L-projection}
\end{multline}
where we define $I_\Ad \coloneqq I_\Ad^{(1)}$ if $q=2$ and $I_\Ad \coloneqq I_\Ad^{(1)} +  I_\Ad^{(2)}$ for $q\geq 3$.

It seems reasonable to assume that both the measurement basis as well as the initial state share the locality structure of $G$, however the state preparation and measurement noise might fail to do so.
For a local measurement basis $\ket{x} = \bigotimes_{i=1}^m \ket{x_i}$, the measurement operator $M_3$ becomes
\begin{align}
 M_3 
 &= \sum_{x\in\Z_q^m} \oketbra{E_x\otimes E_x}{E_x} \\
 &\simeq \bigotimes_{i=1}^m \sum_{x_i\in\Z_q} \oketbra{E_{x_i}\otimes E_{x_i}}{E_{x_i}} =: M_{3,\loc}^{\otimes m} \, .
\end{align}
Retracing the steps from Sec.~\ref{sec:sampling-complexity-3-designs} carefully, we find that 
\begin{align}
 J_1^\dagger M_{3,\loc} &= \oketbra{\one}{\one\otimes\one}M_{3,\loc} =  \oketbra{\one}{\one}, \\
 I_1^\dagger M_{3,\loc} &= \oketbra{\one}{F_0}M_{3,\loc} = \frac{q-1}{q} \oketbra{\one}{\one}, \\
 I_\Ad^{(1)\dagger} M_{3,\loc} &= (q-2) \sum_{z\in\Z_q\setminus 0} \oketbra{Z(z)}{Z(z)} =: q(q-2) \tau.
\end{align}
Here we have set $\tau \coloneqq  P_\Ad M_\loc$.
Taking $L=M_3\EM$ in Eq.~\eqref{eq:local-3-design-L-projection} requires us to evaluate the following inner product:
\begin{align}
\MoveEqLeft[1]
 \obraket{J_1^{\otimes |b|}\otimes I_1^{\otimes |a|-|b|} \otimes I_\Ad^{(1)\otimes |\bar a|}}{M_3\EM}\\
 &=
 \left(\frac{q-1}{q}\right)^{|a|-|b|} \left(q(q-2)\right)^{|\bar a|} \tr\left( \oketbra{\one}{\one}^{\otimes |a|} \otimes  \tau^{\otimes |\bar a|} \EM \right) \\
 &\leq
 \left(\frac{q-1}{q}\right)^{|a|-|b|} \left(q(q-2)\right)^{|\bar a|} \tr\left( \oketbra{\one}{\one}^{\otimes |a|} \otimes  \tau^{\otimes |\bar a|} \right) \\
 &=
 \left(\frac{q-1}{q}\right)^{|a|-|b|} \left(q(q-2)\right)^{|\bar a|} q^{|a|} (q-1)^{|\bar a|} \\
 &=
 q^{|\bar a|+|b|} (q-1)^{|\bar b|} (q-2)^{|\bar a|} \\
 &=
 q^{m+|\bar a|-|\bar b|} (q-1)^{|\bar b|} (q-2)^{|\bar a|}
 \label{eq:local-3-design-measurement-contraction}
\end{align}
Here, the inequality follows as in Eq.~\eqref{eq:main-thm-proof-2}, i.e.~by the observation that we can replace $\EM$ by its unital and trace-preserving part with unit spectral norm and apply Hölder's inequality.
Note that we have equality in the absence of SPAM noise.

Furthermore, let $\rho$ be a pure product state, w.l.o.g.~$\rho = \rho_\loc^{\otimes m}$.
Then, we have to contract it with the operators in Eq.~\eqref{eq:local-3-design-L-projection}.
The necessarily computations have already been performed in Sec.~\ref{sec:sampling-complexity-3-designs}, in particular Eq.~\eqref{eq:sampling-complexity-3-designs-trivial} and \eqref{eq:sampling-complexity-3-designs-adjoint}.
\begin{align}
\MoveEqLeft[2]
 \osandwichb{\rho^{\otimes 2}}%
  {J_1^{\otimes |b|} \otimes I_1^{\otimes |a|-|b|} \otimes I_\Ad^{\otimes |\bar a|}  }%
  {\tilde\rho}\\
 &= \left(\frac{q-1}{q}\right)^{|a|-|b|}  (q-2)^{|\bar a|} m_\Ad^{|\bar a|} \\
 &\qquad
  \times \obraket[\big]{\one^{\otimes |a|} \otimes \left(q \rho_\loc - \one\right)^{\otimes |\bar a|}  }{\tilde\rho} \\
 &\leq \left(\frac{q-1}{q}\right)^{|a|-|b|}  (q-2)^{|\bar a|}  (q-1)^{|\bar a|} m_\Ad^{|\bar a|} \\
 &= q^{|b|-|a|} (q-1)^{|\bar b|}(q-2)^{|\bar a|} m_\Ad^{|\bar a|}\\
 &= q^{|\bar a|-|\bar b|} (q-1)^{|\bar b|}(q-2)^{|\bar a|} m_\Ad^{|\bar a|},
 \label{eq:local-3-design-state-contraction}
\end{align}
where the upper bound follows from Hölder's inequality.
Again, we have equality in the absence of SPAM noise.

Recall that we have up to now assumed that only the very first bits of $b$ and $a$ are set and the others are zero.
The result for an arbitrary bitstring $b\in\{0,1\}^m$ and $a\in\{0,1\}^m$ such that $|a\circ b|=|b|$ can be obtained by applying a suitable permutation to the basis of the projection in Eq.~\eqref{eq:local-3-design-L-projection}.
However, it is straightforward to check that the established upper bounds \eqref{eq:local-3-design-measurement-contraction} and \eqref{eq:local-3-design-state-contraction} still hold, even when such a permutation is applied.

Finally, using $S_b^+=(q+1)^{n-|b|}\id_b$ from Sec.~\ref{sec:frame-operator}, we obtain the following upper bound for the second moment of local unitary 3-designs for $q\geq 3$:
\begin{widetext}
\begin{align}
\EE[f_b^2]_\SPAM
 &=
 \osandwichb{\rho^{\otimes 2}}%
  {(X_b S_b^+)^{\otimes 2}\,
  \widehat{\omega}[\tau_b \otimes \tau_b] \big( X_b^{\dagger\otimes 2} \tilde M_3 \big) }%
  {\tilde\rho} \\
 &\leq 
  (q+1)^{2|\bar b|} 
  \sum_{a:\,|a\circ b|=|b|}
  \frac{q^{m+2(|\bar a|-|\bar b|)} (q-1)^{2|\bar b|} (q-2)^{2|\bar a|}2^{|\bar a|}}%
  {q^{3(m-|\bar b|)}(q(q^2-1))^{|\bar b|-|\bar a|}(q^3 (q^2-1)(q^2-4))^{|\bar a|}}
  \\
 &=
 \frac{(q^2-1)^{|\bar b|} }{q^{2m}} 
 \sum_{a:\,|a\circ b|=|b|}
 \left(\frac{2(q-2)}{q+2}\right)^{|\bar a|}\\
 &=
 \frac{(q^2-1)^{|\bar b|} }{q^{2m}} 
 \sum_{k=0}^{|\bar b|}
 \binom{|\bar b|}{k}
 \left(\frac{2(q-2)}{q+2}\right)^{k} \\
 &=
 \frac{(q^2-1)^{|\bar b|} }{q^{2m}} 
 \left(\frac{3q-2}{q+2}\right)^{|\bar b|} \label{eq:2nd-moment-local-unitary-3-designs} \\
 &\leq
 \frac{(3q^2)^{|\bar b|}}{q^{2m}} = 3^m \left(\frac{3}{q^2}\right)^{|b|} \leq 3^{m-|b|} \, .
\end{align}
\end{widetext}
The last inequality follows since we always have $3/q^2 \leq 1/3$ for $q\geq 3$.
Hence, the second moment is only reasonably bounded if all but logarithmically many bits in $b$ are set (i.e.~we only have a logarithmic number of adjoint irreps).
As argued earlier, the inequalities leading to the expression in Eq.~\eqref{eq:2nd-moment-local-unitary-3-designs} are tight in the absence of SPAM noise, and hence this is the result for $\EE[f_b^2]_\ideal$.

For $q=2$, we can deduce similarly to Sec.~\ref{sec:sampling-complexity-3-designs} that all contributions containing an adjoint irrep have to vanish.
Concretely, the overlap in Eq.~\eqref{eq:local-3-design-measurement-contraction} is zero whenever $|\bar a| \neq 0$.
Hence, the only irrep in Eq.~\eqref{eq:tau-b-square-decomp} which contributes to the second moment is the one for which $a = (1,\dots,1)$ is the all-ones vector.
We then find 
\begin{align}
\EE[f_b^2]_\SPAM
 &\leq
  \frac{q^m(q^2-1)^{2|\bar b|}}%
  {q^{3m - |\bar b|}(q(q^2-1))^{|\bar b|}} \\
  &=
  \frac{(q^2-1)^{|\bar b|}}%
  {q^{2m}} 
  =
  \left(\frac34\right)^{|\bar b|} \left(\frac14\right)^{|b|}
  \leq
  1.
 \label{eq:2nd-moment-local-unitary-3-designs-qubits}
\end{align}
Again, note that we have equality in the first inequality in the absence of SPAM noise.

\subsection{Second moment for the Heisenberg-Weyl/Pauli group}
\label{sec:sampling-complexity-paulis}

As another example, we consider the Heisenberg-Weyl group $G=\HW{n}{p}$ as defined in Sec.~\ref{sec:clifford-group}, and write the Weyl operators as $w(v)$ with $v\in \F_p^{2n}$.
The Heisenberg-Weyl group acts naturally on $(\C^p)^{\otimes n}$ and its conjugation representation decomposes into one-dimensional irreps since
\begin{equation}
 w(v) w(u) w(v)^\dagger = \xi^{[v,u]} w(u),
\end{equation}
where $\xi$ is a primitive $p$-th root of unity and $[v,u]$ is the standard symplectic product on $\F_p^{2n}$.
Hence, let us label the irreps of $\HW{n}{p}$ by $\tau_u$.

Clearly, $\tau_u\otimes\tau_u \simeq \tau_{2u}$ and hence the projector of the second moment becomes
\begin{equation}
 \widehat{\omega}[\tau_u\otimes\tau_u ] = \frac{\oketbra{I_u}{I_u}}{p^{3n}}, \qquad I_u \coloneqq  \oketbra{w(u)\otimes w(u)}{w(2u)}.
\end{equation}
Thus, we find for $u=(z,x)$, similar to the calculation in Eq.~\eqref{eq:3-design-ad-projection-overlap}:
\begin{align}
 \obraket{I_u}{(X_u^\dagger)^{\otimes 2}\tilde M_3} 
 &= \sum_{y\in\F_p^n} \obraket{w(u)\otimes w(u)}{E_y\otimes E_y} \obraket{\tilde E_y}{w(2u)} \\
 &= \delta_{x,0} \sum_{y\in\F_p^n} \xi^{2z\cdot y} \obraket{\EM^\dagger(\ketbra{y}{y})}{Z(2z)} \\
 &= \delta_{x,0}\, \osandwich{Z(2z)}{\EM}{Z(2z)} \leq \delta_{x,0}\, p^n.
\end{align}
Note that for qubits, $p=2$, we have $2z=0$ and hence the last inequality is an equality if $\EM$ is trace-preserving.

Hence, we find the following bound for the second moment
\begin{align}
\EE[f_b^2]_\SPAM
 &=
 \osandwichb{\rho^{\otimes 2}}%
  {(X_u S_u^+)^{\otimes 2}\,
  \widehat{\omega}[\tau_u \otimes \tau_u] \big( X_u^{\dagger\otimes 2} \tilde M_3 \big) }%
  {\tilde\rho} \\
 &\leq \frac{\delta_{x,0}}{p^{2n}} \obraket{\rho}{Z(z)}^2 \obraket{Z(z)}{\tilde\rho} \leq \frac{\delta_{x,0}}{p^{2n}},
\end{align}
where we have used Hölder's inequality twice in the last step.
Note that all inequalities are saturated in the absence of SPAM noise.

\end{document}

%% file: mymath.tex
\newtheorem{theorem}{Theorem}
\newtheorem*{theorem*}{Theorem}
\newtheorem{corollary}[theorem]{Corollary}
\newtheorem{lemma}[theorem]{Lemma}
\newtheorem{proposition}[theorem]{Proposition}

\newtheorem{definition}[theorem]{Definition}
\theoremstyle{definition}
\newtheorem{remark}{Remark}
\newcommand{\rmd}{\ensuremath\mathrm{d}}

\newcommand{\dd}{\,\rmd}

\DeclareMathOperator{\tr}{tr}
\renewcommand{\Re}{\operatorname{Re}}

\DeclareMathOperator{\ran}{ran}
\DeclareMathOperator{\rank}{rank}

\DeclareMathOperator{\supp}{supp}

\newcommand{\fro}{\mathrm{F}}

\DeclareMathOperator{\spec}{spec}

\DeclareMathOperator{\poly}{poly}

\newcommand{\id}{\mathrm{id}}
\DeclareMathOperator{\U}{U}%unitary group U(n)
\DeclareMathOperator{\GL}{GL}%general linear group
\DeclareMathOperator{\Var}{Var}

\DeclareMathOperator{\sep}{sep} % separation function in matrix perturbation theory
\newcommand{\loc}{\mathrm{loc}}
\newcommand{\SPAM}{\mathrm{SPAM}}
\newcommand{\ideal}{\mathrm{ideal}}

\DeclareMathOperator{\Herm}{Herm} % Hermitian operators
\DeclareMathOperator{\End}{End}
\DeclareMathOperator{\Hom}{Hom}

\newcommand{\Ad}{\mathrm{ad}}
\DeclareMathOperator{\Irr}{Irr}
\DeclareMathOperator{\Borel}{Borel}

\NewDocumentCommand\Cl{mm}{
    \ensuremath{\mathrm{Cl}_{#1}\IfNoValueTF{#2}{}{(#2)}}%
}
\NewDocumentCommand\HW{mm}{
    \ensuremath{\mathrm{HW}_{#1}\IfNoValueTF{#2}{}{(#2)}}%
}

\newcommand{\ClG}[1]{\ensuremath{\mathrm{Cl}_{#1}}}

\newcommand{\CC}{\mathbb{C}}
\newcommand{\RR}{\mathbb{R}}
\newcommand{\QQ}{\mathbb{Q}}

\newcommand{\1}{\mathbbm{1}}
\newcommand{\EE}{\mathbb{E}}
\newcommand{\PP}{\mathbb{P}}
\newcommand{\FF}{\mathbb{F}}

\newcommand{\N}{\mathbb{N}}
\newcommand{\Z}{\mathbb{Z}}
\newcommand{\R}{\mathbb{R}}
\newcommand{\C}{\mathbb{C}}
\newcommand{\F}{\mathbb{F}}

\newcommand{\one}{\mathbbm{1}}

\newcommand{\argdot}{{\,\cdot\,}}

\newcommand{\myleft}{\mathopen{}\mathclose\bgroup\left}
\newcommand{\myright}{\aftergroup\egroup\right}
\newcommand{\ad}{^\dagger}

\DeclarePairedDelimiterX{\abs}[1]{\lvert}{\rvert}{%
  \ifblank{#1}{\,\cdot\,}{#1}
}   % absolute value

\DeclarePairedDelimiterX\norm[1]\lVert\rVert{%
  \ifblank{#1}{\,\cdot\,}{#1}
}   % norm

\DeclarePairedDelimiterX{\iiiNorm}[1]{\lvert}{\rvert}{%
  \delimsize\lvert\delimsize\lvert#1\delimsize\rvert\delimsize\rvert%
}

\DeclarePairedDelimiterXPP\snorm[1]{}\lVert\rVert{_\infty}{\ifblank{#1}{\,\cdot\,}{#1}}   %spectral norm  =  (2->2)-norm
\newcommand{\snormb}[1]{\snorm[\big]{#1}}

\DeclarePairedDelimiterXPP\twonorm[1]{}\lVert\rVert{_2}{\ifblank{#1}{\,\cdot\,}{#1}}   % 2-norm
\newcommand{\twonormb}[1]{\twonorm[\big]{#1}}

\DeclarePairedDelimiterXPP\trnorm[1]{}\lVert\rVert{_1}{\ifblank{#1}{\,\cdot\,}{#1}}   % trace norm

\DeclarePairedDelimiterXPP\fnorm[1]{}\lVert\rVert{_{\fro}}{\ifblank{#1}{\,\cdot\,}{#1}}   % Fro-norm

\DeclarePairedDelimiterXPP\dnorm[1]{}\lVert\rVert{_\diamond}{\ifblank{#1}{\,\cdot\,}{#1}}   % diamond norm

\DeclarePairedDelimiterXPP\cbnorm[1]{}\lVert\rVert{_\mathrm{cb}}{\ifblank{#1}{\,\cdot\,}{#1}}   % CB-norm
\DeclarePairedDelimiterXPP\onenorm[1]{}\lVert\rVert{_{1\rightarrow 1}}{\ifblank{#1}{\,\cdot\,}{#1}}   % (1->1)-norm
\DeclarePairedDelimiterXPP\ddnorm[1]{}\lVert\rVert{_{\diamond\rightarrow \diamond}}{\ifblank{#1}{\,\cdot\,}{#1}}   % (\diamond->\diamond)-norm
\DeclarePairedDelimiterXPP\ssnorm[1]{}\lVert\rVert{_{\infty\rightarrow\infty}}{\ifblank{#1}{\,\cdot\,}{#1}}   % (\infty->\infty)-norm

\DeclarePairedDelimiterX\Set[1]\{\}{%
  
  #1
}

\DeclarePairedDelimiterX\innerp[2]{\langle}{\rangle}{%
  \ifblank{#1}{\,\cdot\,}{#1} , \ifblank{#2}{\,\cdot\,}{#2}%
}

\DeclarePairedDelimiter{\bra}{\langle}{\vert}
\DeclarePairedDelimiter{\ket}{\vert}{\rangle}

\DeclarePairedDelimiterX\braket[2]{\langle}{\rangle}%
  {#1\kern0.15ex\delimsize\vert\kern0.15ex\mathopen{}#2}

\DeclarePairedDelimiterX\ketbra[2]{\vert}{\vert}%
  {#1\kern0.15ex\delimsize\rangle\delimsize\langle\kern0.15ex\mathopen{}#2}

\DeclarePairedDelimiterX\sandwich[3]{\langle}{\rangle}%
  {#1\,\delimsize\vert\kern0.15ex\mathopen{}#2\kern0.15ex\delimsize\vert\kern0.15ex\mathopen{}#3}

\DeclarePairedDelimiter{\obra}{(}{\vert}
\DeclarePairedDelimiter{\oket}{\vert}{)}

\DeclarePairedDelimiterX\obraket[2]{(}{)}%
  {#1\kern0.15ex\delimsize\vert\kern0.15ex\mathopen{}#2}

\DeclarePairedDelimiterX\oketbra[2]{\vert}{\vert}%
  {#1\kern0.15ex\delimsize)\delimsize(\kern0.15ex\mathopen{}#2}

\DeclarePairedDelimiterX\osandwich[3]{(}{)}%
  {#1\,\delimsize\vert\kern0.15ex\mathopen{}#2\kern0.15ex\delimsize\vert\kern0.15ex\mathopen{}#3}
\newcommand{\osandwichb}[3]{\osandwich[\big]{#1}{#2}{#3}}

\renewcommand{\Pr}{\operatorname{\PP}}

\newcommand{\ie}{i.\,e.}

 \newcommand{\FD}{F}
\newcommand{\MO}{\mathsf{M}}
\newcommand{\MA}{\mathcal{M}}
\newcommand{\LRC}{\mathrm{LRC}}
\newcommand{\BW}{\mathrm{BW}}

\newcommand{\EM}{\mathcal{E}_{\mathrm{M}}}
\newcommand{\ESP}{\mathcal{E}_{\mathrm{SP}}}

%% file: mk_macros.tex
  \providecommand{\pra}{Phys.\ Rev.\ A}

  \providecommand{\prl}{Phys.\ Rev.\ Lett.}

\makeatletter
\AtBeginDocument{%
  \renewcommand*{\AC@hyperlink}[2]{%
    \begingroup
      \hypersetup{hidelinks}%
      \hyperlink{#1}{#2}%
    \endgroup
  }%
}
\makeatother

%% file: myacronyms.tex
\begin{acronym}[POVM]\itemsep.5\baselineskip
\acro{ACES}{averaged circuit eigenvalue sampling}
\acro{AGF}{average gate fidelity}

\acro{BOG}{binned outcome generation}
\acro{BW}{brickwork}

\acro{CP}{completely positive}
\acro{CPT}{completely positive and trace preserving}

\acro{DFE}{direct fidelity estimation} 

\acro{FT}{Fourier transform}

\acro{GST}{gate set tomography}
\acro{GTM}{gate-independent, time-stationary, Markovian}

\acro{HOG}{heavy outcome generation}

\acro{irrep}{irreducible representation}

\acro{LRC}{local random circuit}

\acro{MUBs}{mutually unbiased bases} 
\acro{MW}{micro wave}

\acro{NISQ}{noisy and intermediate scale quantum}
\acro{NN}{nearest-neighbor}

\acro{OVM}{operator-valued measure}

\acro{POVM}{positive operator-valued measure}
\acro{PVM}{projector-valued measure}

\acro{QAOA}{quantum approximate optimization algorithm}

\acro{RB}{randomized benchmarking}

\acro{SFE}{shadow fidelity estimation}
\acro{SIC}{symmetric, informationally complete}
\acro{SPAM}{state preparation and measurement}

\acro{QPT}{quantum process tomography}

\acro{rf}{radio frequency}

\acro{TT}{tensor train}
\acro{TV}{total variation}

\acro{VQE}{variational quantum eigensolver}

\acro{XEB}{cross-entropy benchmarking}

\end{acronym}

%% file: generatorRB.bbl
\begin{thebibliography}{112}%
\makeatletter
\providecommand \@ifxundefined [1]{%
 \@ifx{#1\undefined}
}%
\providecommand \@ifnum [1]{%
 \ifnum #1\expandafter \@firstoftwo
 \else \expandafter \@secondoftwo
 \fi
}%
\providecommand \@ifx [1]{%
 \ifx #1\expandafter \@firstoftwo
 \else \expandafter \@secondoftwo
 \fi
}%
\providecommand \natexlab [1]{#1}%
\providecommand \enquote  [1]{``#1''}%
\providecommand \bibnamefont  [1]{#1}%
\providecommand \bibfnamefont [1]{#1}%
\providecommand \citenamefont [1]{#1}%
\providecommand \href@noop [0]{\@secondoftwo}%
\providecommand \href [0]{\begingroup \@sanitize@url \@href}%
\providecommand \@href[1]{\@@startlink{#1}\@@href}%
\providecommand \@@href[1]{\endgroup#1\@@endlink}%
\providecommand \@sanitize@url [0]{\catcode `\\12\catcode `\$12\catcode
  `\&12\catcode `\#12\catcode `\^12\catcode `\_12\catcode `\%12\relax}%
\providecommand \@@startlink[1]{}%
\providecommand \@@endlink[0]{}%
\providecommand \url  [0]{\begingroup\@sanitize@url \@url }%
\providecommand \@url [1]{\endgroup\@href {#1}{\urlprefix }}%
\providecommand \urlprefix  [0]{URL }%
\providecommand \Eprint [0]{\href }%
\providecommand \doibase [0]{https://doi.org/}%
\providecommand \selectlanguage [0]{\@gobble}%
\providecommand \bibinfo  [0]{\@secondoftwo}%
\providecommand \bibfield  [0]{\@secondoftwo}%
\providecommand \translation [1]{[#1]}%
\providecommand \BibitemOpen [0]{}%
\providecommand \bibitemStop [0]{}%
\providecommand \bibitemNoStop [0]{.\EOS\space}%
\providecommand \EOS [0]{\spacefactor3000\relax}%
\providecommand \BibitemShut  [1]{\csname bibitem#1\endcsname}%
\let\auto@bib@innerbib\@empty
%</preamble>
\bibitem [{\citenamefont {{Eisert}}\ \emph {et~al.}(2020)\citenamefont
  {{Eisert}}, \citenamefont {{Hangleiter}}, \citenamefont {{Walk}},
  \citenamefont {{Roth}}, \citenamefont {{Markham}}, \citenamefont {{Parekh}},
  \citenamefont {{Chabaud}},\ and\ \citenamefont
  {{Kashefi}}}]{Eisert2020QuantumCertificationAnd}%
  \BibitemOpen
  \bibfield  {author} {\bibinfo {author} {\bibfnamefont {J.}~\bibnamefont
  {{Eisert}}}, \bibinfo {author} {\bibfnamefont {D.}~\bibnamefont
  {{Hangleiter}}}, \bibinfo {author} {\bibfnamefont {N.}~\bibnamefont
  {{Walk}}}, \bibinfo {author} {\bibfnamefont {I.}~\bibnamefont {{Roth}}},
  \bibinfo {author} {\bibfnamefont {D.}~\bibnamefont {{Markham}}}, \bibinfo
  {author} {\bibfnamefont {R.}~\bibnamefont {{Parekh}}}, \bibinfo {author}
  {\bibfnamefont {U.}~\bibnamefont {{Chabaud}}},\ and\ \bibinfo {author}
  {\bibfnamefont {E.}~\bibnamefont {{Kashefi}}},\ }\bibfield  {title} {\bibinfo
  {title} {Quantum certification and benchmarking},\ }\href
  {https://doi.org/10.1038/s42254-020-0186-4} {\bibfield  {journal} {\bibinfo
  {journal} {Nat. Rev. Phys.}\ }\textbf {\bibinfo {volume} {2}},\ \bibinfo
  {pages} {382} (\bibinfo {year} {2020})},\ \Eprint
  {https://arxiv.org/abs/1910.06343} {arXiv:1910.06343 [quant-ph]} \BibitemShut
  {NoStop}%
\bibitem [{\citenamefont {{Kliesch}}\ and\ \citenamefont
  {{Roth}}(2021)}]{Kliesch2020TheoryOfQuantum}%
  \BibitemOpen
  \bibfield  {author} {\bibinfo {author} {\bibfnamefont {M.}~\bibnamefont
  {{Kliesch}}}\ and\ \bibinfo {author} {\bibfnamefont {I.}~\bibnamefont
  {{Roth}}},\ }\bibfield  {title} {\bibinfo {title} {Theory of quantum system
  certification},\ }\href {https://doi.org/10.1103/PRXQuantum.2.010201}
  {\bibfield  {journal} {\bibinfo  {journal} {PRX Quantum}\ }\textbf {\bibinfo
  {volume} {2}},\ \bibinfo {pages} {010201} (\bibinfo {year} {2021})},\
  \bibinfo {note} {tutorial},\ \Eprint {https://arxiv.org/abs/2010.05925}
  {arXiv:2010.05925 [quant-ph]} \BibitemShut {NoStop}%
\bibitem [{\citenamefont {Emerson}\ \emph {et~al.}(2005)\citenamefont
  {Emerson}, \citenamefont {Alicki},\ and\ \citenamefont
  {\.{Z}yczkowski}}]{EmeAliZyc05}%
  \BibitemOpen
  \bibfield  {author} {\bibinfo {author} {\bibfnamefont {J.}~\bibnamefont
  {Emerson}}, \bibinfo {author} {\bibfnamefont {R.}~\bibnamefont {Alicki}},\
  and\ \bibinfo {author} {\bibfnamefont {K.}~\bibnamefont {\.{Z}yczkowski}},\
  }\bibfield  {title} {\bibinfo {title} {Scalable noise estimation with random
  unitary operators},\ }\href {https://doi.org/10.1088/1464-4266/7/10/021}
  {\bibfield  {journal} {\bibinfo  {journal} {J. Opt. B}\ }\textbf {\bibinfo
  {volume} {7}},\ \bibinfo {pages} {S347} (\bibinfo {year} {2005})},\ \Eprint
  {https://arxiv.org/abs/arXiv:quant-ph/0503243} {arXiv:quant-ph/0503243}
  \BibitemShut {NoStop}%
\bibitem [{\citenamefont {{L{\'e}vi}}\ \emph {et~al.}(2007)\citenamefont
  {{L{\'e}vi}}, \citenamefont {{L{\'o}pez}}, \citenamefont {{Emerson}},\ and\
  \citenamefont {{Cory}}}]{Levi07EfficientErrorCharacterization}%
  \BibitemOpen
  \bibfield  {author} {\bibinfo {author} {\bibfnamefont {B.}~\bibnamefont
  {{L{\'e}vi}}}, \bibinfo {author} {\bibfnamefont {C.~C.}\ \bibnamefont
  {{L{\'o}pez}}}, \bibinfo {author} {\bibfnamefont {J.}~\bibnamefont
  {{Emerson}}},\ and\ \bibinfo {author} {\bibfnamefont {D.~G.}\ \bibnamefont
  {{Cory}}},\ }\bibfield  {title} {\bibinfo {title} {Efficient error
  characterization in quantum information processing},\ }\href
  {https://doi.org/10.1103/PhysRevA.75.022314} {\bibfield  {journal} {\bibinfo
  {journal} {\pra}\ }\textbf {\bibinfo {volume} {75}},\ \bibinfo {eid} {022314}
  (\bibinfo {year} {2007})},\ \Eprint {https://arxiv.org/abs/quant-ph/0608246}
  {arXiv:quant-ph/0608246 [quant-ph]} \BibitemShut {NoStop}%
\bibitem [{\citenamefont {Dankert}\ \emph {et~al.}(2009)\citenamefont
  {Dankert}, \citenamefont {Cleve}, \citenamefont {Emerson},\ and\
  \citenamefont {Livine}}]{DanCleEme09}%
  \BibitemOpen
  \bibfield  {author} {\bibinfo {author} {\bibfnamefont {C.}~\bibnamefont
  {Dankert}}, \bibinfo {author} {\bibfnamefont {R.}~\bibnamefont {Cleve}},
  \bibinfo {author} {\bibfnamefont {J.}~\bibnamefont {Emerson}},\ and\ \bibinfo
  {author} {\bibfnamefont {E.}~\bibnamefont {Livine}},\ }\bibfield  {title}
  {\bibinfo {title} {Exact and approximate unitary 2-designs and their
  application to fidelity estimation},\ }\href
  {https://doi.org/10.1103/PhysRevA.80.012304} {\bibfield  {journal} {\bibinfo
  {journal} {Phys. Rev. A}\ }\textbf {\bibinfo {volume} {80}},\ \bibinfo
  {pages} {012304} (\bibinfo {year} {2009})},\ \Eprint
  {https://arxiv.org/abs/quant-ph/0606161} {arXiv:quant-ph/0606161 [quant-ph]}
  \BibitemShut {NoStop}%
\bibitem [{\citenamefont {Emerson}\ \emph {et~al.}(2007)\citenamefont
  {Emerson}, \citenamefont {Silva}, \citenamefont {Moussa}, \citenamefont
  {Ryan}, \citenamefont {Laforest}, \citenamefont {Baugh}, \citenamefont
  {Cory},\ and\ \citenamefont {Laflamme}}]{EmeSilMou07}%
  \BibitemOpen
  \bibfield  {author} {\bibinfo {author} {\bibfnamefont {J.}~\bibnamefont
  {Emerson}}, \bibinfo {author} {\bibfnamefont {M.}~\bibnamefont {Silva}},
  \bibinfo {author} {\bibfnamefont {O.}~\bibnamefont {Moussa}}, \bibinfo
  {author} {\bibfnamefont {C.}~\bibnamefont {Ryan}}, \bibinfo {author}
  {\bibfnamefont {M.}~\bibnamefont {Laforest}}, \bibinfo {author}
  {\bibfnamefont {J.}~\bibnamefont {Baugh}}, \bibinfo {author} {\bibfnamefont
  {D.~G.}\ \bibnamefont {Cory}},\ and\ \bibinfo {author} {\bibfnamefont
  {R.}~\bibnamefont {Laflamme}},\ }\bibfield  {title} {\bibinfo {title}
  {Symmetrized characterization of noisy quantum processes},\ }\href
  {https://doi.org/10.1126/science.1145699} {\bibfield  {journal} {\bibinfo
  {journal} {Science}\ }\textbf {\bibinfo {volume} {317}},\ \bibinfo {pages}
  {1893} (\bibinfo {year} {2007})},\ \Eprint {https://arxiv.org/abs/0707.0685}
  {arXiv:0707.0685 [quant-ph]} \BibitemShut {NoStop}%
\bibitem [{\citenamefont {{Knill}}\ \emph {et~al.}(2008)\citenamefont
  {{Knill}}, \citenamefont {{Leibfried}}, \citenamefont {{Reichle}},
  \citenamefont {{Britton}}, \citenamefont {{Blakestad}}, \citenamefont
  {{Jost}}, \citenamefont {{Langer}}, \citenamefont {{Ozeri}}, \citenamefont
  {{Seidelin}},\ and\ \citenamefont {{Wineland}}}]{KniLeiRei08}%
  \BibitemOpen
  \bibfield  {author} {\bibinfo {author} {\bibfnamefont {E.}~\bibnamefont
  {{Knill}}}, \bibinfo {author} {\bibfnamefont {D.}~\bibnamefont
  {{Leibfried}}}, \bibinfo {author} {\bibfnamefont {R.}~\bibnamefont
  {{Reichle}}}, \bibinfo {author} {\bibfnamefont {J.}~\bibnamefont
  {{Britton}}}, \bibinfo {author} {\bibfnamefont {R.~B.}\ \bibnamefont
  {{Blakestad}}}, \bibinfo {author} {\bibfnamefont {J.~D.}\ \bibnamefont
  {{Jost}}}, \bibinfo {author} {\bibfnamefont {C.}~\bibnamefont {{Langer}}},
  \bibinfo {author} {\bibfnamefont {R.}~\bibnamefont {{Ozeri}}}, \bibinfo
  {author} {\bibfnamefont {S.}~\bibnamefont {{Seidelin}}},\ and\ \bibinfo
  {author} {\bibfnamefont {D.~J.}\ \bibnamefont {{Wineland}}},\ }\bibfield
  {title} {\bibinfo {title} {Randomized benchmarking of quantum gates},\ }\href
  {https://doi.org/10.1103/PhysRevA.77.012307} {\bibfield  {journal} {\bibinfo
  {journal} {Phys. Rev. A}\ }\textbf {\bibinfo {volume} {77}},\ \bibinfo {eid}
  {012307} (\bibinfo {year} {2008})},\ \Eprint
  {https://arxiv.org/abs/0707.0963} {arXiv:0707.0963 [quant-ph]} \BibitemShut
  {NoStop}%
\bibitem [{\citenamefont {Magesan}\ \emph {et~al.}(2012)\citenamefont
  {Magesan}, \citenamefont {Gambetta},\ and\ \citenamefont
  {Emerson}}]{Magesan2012}%
  \BibitemOpen
  \bibfield  {author} {\bibinfo {author} {\bibfnamefont {E.}~\bibnamefont
  {Magesan}}, \bibinfo {author} {\bibfnamefont {J.~M.}\ \bibnamefont
  {Gambetta}},\ and\ \bibinfo {author} {\bibfnamefont {J.}~\bibnamefont
  {Emerson}},\ }\bibfield  {title} {\bibinfo {title} {Characterizing quantum
  gates via randomized benchmarking},\ }\href
  {https://doi.org/10.1103/PhysRevA.85.042311} {\bibfield  {journal} {\bibinfo
  {journal} {Phys. Rev. A}\ }\textbf {\bibinfo {volume} {85}},\ \bibinfo
  {pages} {042311} (\bibinfo {year} {2012})},\ \Eprint
  {https://arxiv.org/abs/1109.6887} {arXiv:1109.6887} \BibitemShut {NoStop}%
\bibitem [{\citenamefont {Helsen}\ \emph
  {et~al.}(2022{\natexlab{a}})\citenamefont {Helsen}, \citenamefont {Roth},
  \citenamefont {Onorati}, \citenamefont {Werner},\ and\ \citenamefont
  {Eisert}}]{helsen_general_2022}%
  \BibitemOpen
  \bibfield  {author} {\bibinfo {author} {\bibfnamefont {J.}~\bibnamefont
  {Helsen}}, \bibinfo {author} {\bibfnamefont {I.}~\bibnamefont {Roth}},
  \bibinfo {author} {\bibfnamefont {E.}~\bibnamefont {Onorati}}, \bibinfo
  {author} {\bibfnamefont {A.~H.}\ \bibnamefont {Werner}},\ and\ \bibinfo
  {author} {\bibfnamefont {J.}~\bibnamefont {Eisert}},\ }\bibfield  {title}
  {\bibinfo {title} {A general framework for randomized benchmarking},\ }\href
  {https://doi.org/10.1103/PRXQuantum.3.020357} {\bibfield  {journal} {\bibinfo
   {journal} {{PRX} Quantum}\ }\textbf {\bibinfo {volume} {3}},\ \bibinfo
  {pages} {020357} (\bibinfo {year} {2022}{\natexlab{a}})},\ \Eprint
  {https://arxiv.org/abs/2010.07974} {arXiv:2010.07974 [quant-ph]} \BibitemShut
  {NoStop}%
\bibitem [{\citenamefont {{Arute}}\ \emph {et~al.}(2019)\citenamefont
  {{Arute}}, \citenamefont {{Arya}}, \citenamefont {{Babbush}}, \citenamefont
  {{Bacon}}, \citenamefont {{Bardin}}, \citenamefont {{Barends}}, \citenamefont
  {{Biswas}}, \citenamefont {{Boixo}}, \citenamefont {{Brandao}}, \citenamefont
  {{Buell}}, \citenamefont {{Burkett}}, \citenamefont {{Chen}}, \citenamefont
  {{Chen}}, \citenamefont {{Chiaro}}, \citenamefont {{Collins}}, \citenamefont
  {{Courtney}}, \citenamefont {{Dunsworth}}, \citenamefont {{Farhi}},
  \citenamefont {{Foxen}}, \citenamefont {{Fowler}}, \citenamefont {{Gidney}},
  \citenamefont {{Giustina}}, \citenamefont {{Graff}}, \citenamefont
  {{Guerin}}, \citenamefont {{Habegger}}, \citenamefont {{Harrigan}},
  \citenamefont {{Hartmann}}, \citenamefont {{Ho}}, \citenamefont {{Hoffmann}},
  \citenamefont {{Huang}}, \citenamefont {{Humble}}, \citenamefont {{Isakov}},
  \citenamefont {{Jeffrey}}, \citenamefont {{Jiang}}, \citenamefont {{Kafri}},
  \citenamefont {{Kechedzhi}}, \citenamefont {{Kelly}}, \citenamefont
  {{Klimov}}, \citenamefont {{Knysh}}, \citenamefont {{Korotkov}},
  \citenamefont {{Kostritsa}}, \citenamefont {{Land huis}}, \citenamefont
  {{Lindmark}}, \citenamefont {{Lucero}}, \citenamefont {{Lyakh}},
  \citenamefont {{Mandr{\`a}}}, \citenamefont {{McClean}}, \citenamefont
  {{McEwen}}, \citenamefont {{Megrant}}, \citenamefont {{Mi}}, \citenamefont
  {{Michielsen}}, \citenamefont {{Mohseni}}, \citenamefont {{Mutus}},
  \citenamefont {{Naaman}}, \citenamefont {{Neeley}}, \citenamefont {{Neill}},
  \citenamefont {{Niu}}, \citenamefont {{Ostby}}, \citenamefont {{Petukhov}},
  \citenamefont {{Platt}}, \citenamefont {{Quintana}}, \citenamefont
  {{Rieffel}}, \citenamefont {{Roushan}}, \citenamefont {{Rubin}},
  \citenamefont {{Sank}}, \citenamefont {{Satzinger}}, \citenamefont
  {{Smelyanskiy}}, \citenamefont {{Sung}}, \citenamefont {{Trevithick}},
  \citenamefont {{Vainsencher}}, \citenamefont {{Villalonga}}, \citenamefont
  {{White}}, \citenamefont {{Yao}}, \citenamefont {{Yeh}}, \citenamefont
  {{Zalcman}}, \citenamefont {{Neven}},\ and\ \citenamefont
  {{Martinis}}}]{Arute2019QuantumSupremacy}%
  \BibitemOpen
  \bibfield  {author} {\bibinfo {author} {\bibfnamefont {F.}~\bibnamefont
  {{Arute}}}, \bibinfo {author} {\bibfnamefont {K.}~\bibnamefont {{Arya}}},
  \bibinfo {author} {\bibfnamefont {R.}~\bibnamefont {{Babbush}}}, \bibinfo
  {author} {\bibfnamefont {D.}~\bibnamefont {{Bacon}}}, \bibinfo {author}
  {\bibfnamefont {J.~C.}\ \bibnamefont {{Bardin}}}, \bibinfo {author}
  {\bibfnamefont {R.}~\bibnamefont {{Barends}}}, \bibinfo {author}
  {\bibfnamefont {R.}~\bibnamefont {{Biswas}}}, \bibinfo {author}
  {\bibfnamefont {S.}~\bibnamefont {{Boixo}}}, \bibinfo {author} {\bibfnamefont
  {F.~G.~S.~L.}\ \bibnamefont {{Brandao}}}, \bibinfo {author} {\bibfnamefont
  {D.~A.}\ \bibnamefont {{Buell}}}, \bibinfo {author} {\bibfnamefont
  {B.}~\bibnamefont {{Burkett}}}, \bibinfo {author} {\bibfnamefont
  {Y.}~\bibnamefont {{Chen}}}, \bibinfo {author} {\bibfnamefont
  {Z.}~\bibnamefont {{Chen}}}, \bibinfo {author} {\bibfnamefont
  {B.}~\bibnamefont {{Chiaro}}}, \bibinfo {author} {\bibfnamefont
  {R.}~\bibnamefont {{Collins}}}, \bibinfo {author} {\bibfnamefont
  {W.}~\bibnamefont {{Courtney}}}, \bibinfo {author} {\bibfnamefont
  {A.}~\bibnamefont {{Dunsworth}}}, \bibinfo {author} {\bibfnamefont
  {E.}~\bibnamefont {{Farhi}}}, \bibinfo {author} {\bibfnamefont
  {B.}~\bibnamefont {{Foxen}}}, \bibinfo {author} {\bibfnamefont
  {A.}~\bibnamefont {{Fowler}}}, \bibinfo {author} {\bibfnamefont
  {C.}~\bibnamefont {{Gidney}}}, \bibinfo {author} {\bibfnamefont
  {M.}~\bibnamefont {{Giustina}}}, \bibinfo {author} {\bibfnamefont
  {R.}~\bibnamefont {{Graff}}}, \bibinfo {author} {\bibfnamefont
  {K.}~\bibnamefont {{Guerin}}}, \bibinfo {author} {\bibfnamefont
  {S.}~\bibnamefont {{Habegger}}}, \bibinfo {author} {\bibfnamefont {M.~P.}\
  \bibnamefont {{Harrigan}}}, \bibinfo {author} {\bibfnamefont {M.~J.}\
  \bibnamefont {{Hartmann}}}, \bibinfo {author} {\bibfnamefont
  {A.}~\bibnamefont {{Ho}}}, \bibinfo {author} {\bibfnamefont {M.}~\bibnamefont
  {{Hoffmann}}}, \bibinfo {author} {\bibfnamefont {T.}~\bibnamefont {{Huang}}},
  \bibinfo {author} {\bibfnamefont {T.~S.}\ \bibnamefont {{Humble}}}, \bibinfo
  {author} {\bibfnamefont {S.~V.}\ \bibnamefont {{Isakov}}}, \bibinfo {author}
  {\bibfnamefont {E.}~\bibnamefont {{Jeffrey}}}, \bibinfo {author}
  {\bibfnamefont {Z.}~\bibnamefont {{Jiang}}}, \bibinfo {author} {\bibfnamefont
  {D.}~\bibnamefont {{Kafri}}}, \bibinfo {author} {\bibfnamefont
  {K.}~\bibnamefont {{Kechedzhi}}}, \bibinfo {author} {\bibfnamefont
  {J.}~\bibnamefont {{Kelly}}}, \bibinfo {author} {\bibfnamefont {P.~V.}\
  \bibnamefont {{Klimov}}}, \bibinfo {author} {\bibfnamefont {S.}~\bibnamefont
  {{Knysh}}}, \bibinfo {author} {\bibfnamefont {A.}~\bibnamefont {{Korotkov}}},
  \bibinfo {author} {\bibfnamefont {F.}~\bibnamefont {{Kostritsa}}}, \bibinfo
  {author} {\bibfnamefont {D.}~\bibnamefont {{Land huis}}}, \bibinfo {author}
  {\bibfnamefont {M.}~\bibnamefont {{Lindmark}}}, \bibinfo {author}
  {\bibfnamefont {E.}~\bibnamefont {{Lucero}}}, \bibinfo {author}
  {\bibfnamefont {D.}~\bibnamefont {{Lyakh}}}, \bibinfo {author} {\bibfnamefont
  {S.}~\bibnamefont {{Mandr{\`a}}}}, \bibinfo {author} {\bibfnamefont {J.~R.}\
  \bibnamefont {{McClean}}}, \bibinfo {author} {\bibfnamefont {M.}~\bibnamefont
  {{McEwen}}}, \bibinfo {author} {\bibfnamefont {A.}~\bibnamefont {{Megrant}}},
  \bibinfo {author} {\bibfnamefont {X.}~\bibnamefont {{Mi}}}, \bibinfo {author}
  {\bibfnamefont {K.}~\bibnamefont {{Michielsen}}}, \bibinfo {author}
  {\bibfnamefont {M.}~\bibnamefont {{Mohseni}}}, \bibinfo {author}
  {\bibfnamefont {J.}~\bibnamefont {{Mutus}}}, \bibinfo {author} {\bibfnamefont
  {O.}~\bibnamefont {{Naaman}}}, \bibinfo {author} {\bibfnamefont
  {M.}~\bibnamefont {{Neeley}}}, \bibinfo {author} {\bibfnamefont
  {C.}~\bibnamefont {{Neill}}}, \bibinfo {author} {\bibfnamefont {M.~Y.}\
  \bibnamefont {{Niu}}}, \bibinfo {author} {\bibfnamefont {E.}~\bibnamefont
  {{Ostby}}}, \bibinfo {author} {\bibfnamefont {A.}~\bibnamefont {{Petukhov}}},
  \bibinfo {author} {\bibfnamefont {J.~C.}\ \bibnamefont {{Platt}}}, \bibinfo
  {author} {\bibfnamefont {C.}~\bibnamefont {{Quintana}}}, \bibinfo {author}
  {\bibfnamefont {E.~G.}\ \bibnamefont {{Rieffel}}}, \bibinfo {author}
  {\bibfnamefont {P.}~\bibnamefont {{Roushan}}}, \bibinfo {author}
  {\bibfnamefont {N.~C.}\ \bibnamefont {{Rubin}}}, \bibinfo {author}
  {\bibfnamefont {D.}~\bibnamefont {{Sank}}}, \bibinfo {author} {\bibfnamefont
  {K.~J.}\ \bibnamefont {{Satzinger}}}, \bibinfo {author} {\bibfnamefont
  {V.}~\bibnamefont {{Smelyanskiy}}}, \bibinfo {author} {\bibfnamefont {K.~J.}\
  \bibnamefont {{Sung}}}, \bibinfo {author} {\bibfnamefont {M.~D.}\
  \bibnamefont {{Trevithick}}}, \bibinfo {author} {\bibfnamefont
  {A.}~\bibnamefont {{Vainsencher}}}, \bibinfo {author} {\bibfnamefont
  {B.}~\bibnamefont {{Villalonga}}}, \bibinfo {author} {\bibfnamefont
  {T.}~\bibnamefont {{White}}}, \bibinfo {author} {\bibfnamefont {Z.~J.}\
  \bibnamefont {{Yao}}}, \bibinfo {author} {\bibfnamefont {P.}~\bibnamefont
  {{Yeh}}}, \bibinfo {author} {\bibfnamefont {A.}~\bibnamefont {{Zalcman}}},
  \bibinfo {author} {\bibfnamefont {H.}~\bibnamefont {{Neven}}},\ and\ \bibinfo
  {author} {\bibfnamefont {J.~M.}\ \bibnamefont {{Martinis}}},\ }\bibfield
  {title} {\bibinfo {title} {Quantum supremacy using a programmable
  superconducting processor},\ }\href
  {https://doi.org/10.1038/s41586-019-1666-5} {\bibfield  {journal} {\bibinfo
  {journal} {Nature}\ }\textbf {\bibinfo {volume} {574}},\ \bibinfo {pages}
  {505} (\bibinfo {year} {2019})},\ \Eprint {https://arxiv.org/abs/1910.11333}
  {arXiv:1910.11333 [quant-ph]} \BibitemShut {NoStop}%
\bibitem [{\citenamefont {Magesan}\ \emph {et~al.}(2011)\citenamefont
  {Magesan}, \citenamefont {Gambetta},\ and\ \citenamefont
  {Emerson}}]{MagGamEme11}%
  \BibitemOpen
  \bibfield  {author} {\bibinfo {author} {\bibfnamefont {E.}~\bibnamefont
  {Magesan}}, \bibinfo {author} {\bibfnamefont {J.~M.}\ \bibnamefont
  {Gambetta}},\ and\ \bibinfo {author} {\bibfnamefont {J.}~\bibnamefont
  {Emerson}},\ }\bibfield  {title} {\bibinfo {title} {Scalable and robust
  randomized benchmarking of quantum processes},\ }\href
  {https://doi.org/10.1103/PhysRevLett.106.180504} {\bibfield  {journal}
  {\bibinfo  {journal} {\prl}\ }\textbf {\bibinfo {volume} {106}},\ \bibinfo
  {pages} {180504} (\bibinfo {year} {2011})},\ \Eprint
  {https://arxiv.org/abs/1009.3639} {arXiv:1009.3639 [quant-ph]} \BibitemShut
  {NoStop}%
\bibitem [{\citenamefont {{Proctor}}\ \emph {et~al.}(2017)\citenamefont
  {{Proctor}}, \citenamefont {{Rudinger}}, \citenamefont {{Young}},
  \citenamefont {{Sarovar}},\ and\ \citenamefont
  {{Blume-Kohout}}}]{proctor2017WhatRandomizedBenchmarking}%
  \BibitemOpen
  \bibfield  {author} {\bibinfo {author} {\bibfnamefont {T.}~\bibnamefont
  {{Proctor}}}, \bibinfo {author} {\bibfnamefont {K.}~\bibnamefont
  {{Rudinger}}}, \bibinfo {author} {\bibfnamefont {K.}~\bibnamefont {{Young}}},
  \bibinfo {author} {\bibfnamefont {M.}~\bibnamefont {{Sarovar}}},\ and\
  \bibinfo {author} {\bibfnamefont {R.}~\bibnamefont {{Blume-Kohout}}},\
  }\bibfield  {title} {\bibinfo {title} {What randomized benchmarking actually
  measures},\ }\href {https://doi.org/10.1103/PhysRevLett.119.130502}
  {\bibfield  {journal} {\bibinfo  {journal} {\prl}\ }\textbf {\bibinfo
  {volume} {119}},\ \bibinfo {eid} {130502} (\bibinfo {year} {2017})},\ \Eprint
  {https://arxiv.org/abs/1702.01853} {arXiv:1702.01853 [quant-ph]} \BibitemShut
  {NoStop}%
\bibitem [{\citenamefont {Wallman}(2018)}]{wallman2018randomized}%
  \BibitemOpen
  \bibfield  {author} {\bibinfo {author} {\bibfnamefont {J.~J.}\ \bibnamefont
  {Wallman}},\ }\bibfield  {title} {\bibinfo {title} {Randomized benchmarking
  with gate-dependent noise},\ }\href
  {https://doi.org/10.22331/q-2018-01-29-47} {\bibfield  {journal} {\bibinfo
  {journal} {Quantum}\ }\textbf {\bibinfo {volume} {2}},\ \bibinfo {pages} {47}
  (\bibinfo {year} {2018})},\ \Eprint {https://arxiv.org/abs/1703.09835}
  {arXiv:1703.09835 [quant-ph]} \BibitemShut {NoStop}%
\bibitem [{\citenamefont {{Merkel}}\ \emph {et~al.}(2021)\citenamefont
  {{Merkel}}, \citenamefont {{Pritchett}},\ and\ \citenamefont
  {{Fong}}}]{Merkel18}%
  \BibitemOpen
  \bibfield  {author} {\bibinfo {author} {\bibfnamefont {S.~T.}\ \bibnamefont
  {{Merkel}}}, \bibinfo {author} {\bibfnamefont {E.~J.}\ \bibnamefont
  {{Pritchett}}},\ and\ \bibinfo {author} {\bibfnamefont {B.~H.}\ \bibnamefont
  {{Fong}}},\ }\bibfield  {title} {\bibinfo {title} {Randomized benchmarking as
  convolution: {Fourier} analysis of gate dependent errors},\ }\href
  {https://doi.org/10.22331/q-2021-11-16-581} {\bibfield  {journal} {\bibinfo
  {journal} {{Quantum}}\ }\textbf {\bibinfo {volume} {5}},\ \bibinfo {pages}
  {581} (\bibinfo {year} {2021})},\ \Eprint {https://arxiv.org/abs/1804.05951}
  {arXiv:1804.05951 [quant-ph]} \BibitemShut {NoStop}%
\bibitem [{\citenamefont {Kong}(2021)}]{kong_framework_2021}%
  \BibitemOpen
  \bibfield  {author} {\bibinfo {author} {\bibfnamefont {L.}~\bibnamefont
  {Kong}},\ }\href {http://arxiv.org/abs/2111.10357} {\bibinfo {title} {A
  framework for randomized benchmarking over compact groups}} (\bibinfo {year}
  {2021}),\ \Eprint {https://arxiv.org/abs/2111.10357} {arXiv:2111.10357
  [quant-ph]} \BibitemShut {NoStop}%
\bibitem [{\citenamefont {{Fran{\c{c}}a}}\ and\ \citenamefont
  {{Hashagen}}(2018)}]{Franca2018ApproximateRB}%
  \BibitemOpen
  \bibfield  {author} {\bibinfo {author} {\bibfnamefont {D.~S.}\ \bibnamefont
  {{Fran{\c{c}}a}}}\ and\ \bibinfo {author} {\bibfnamefont {A.~K.}\
  \bibnamefont {{Hashagen}}},\ }\bibfield  {title} {\bibinfo {title}
  {Approximate randomized benchmarking for finite groups},\ }\href
  {https://doi.org/10.1088/1751-8121/aad6fa} {\bibfield  {journal} {\bibinfo
  {journal} {J. Phys. A}\ }\textbf {\bibinfo {volume} {51}},\ \bibinfo {eid}
  {395302} (\bibinfo {year} {2018})},\ \Eprint
  {https://arxiv.org/abs/1803.03621} {arXiv:1803.03621 [quant-ph]} \BibitemShut
  {NoStop}%
\bibitem [{\citenamefont {{Proctor}}\ \emph {et~al.}(2019)\citenamefont
  {{Proctor}}, \citenamefont {{Carignan-Dugas}}, \citenamefont {{Rudinger}},
  \citenamefont {{Nielsen}}, \citenamefont {{Blume-Kohout}},\ and\
  \citenamefont {{Young}}}]{Proctor2019DirectRandomized}%
  \BibitemOpen
  \bibfield  {author} {\bibinfo {author} {\bibfnamefont {T.~J.}\ \bibnamefont
  {{Proctor}}}, \bibinfo {author} {\bibfnamefont {A.}~\bibnamefont
  {{Carignan-Dugas}}}, \bibinfo {author} {\bibfnamefont {K.}~\bibnamefont
  {{Rudinger}}}, \bibinfo {author} {\bibfnamefont {E.}~\bibnamefont
  {{Nielsen}}}, \bibinfo {author} {\bibfnamefont {R.}~\bibnamefont
  {{Blume-Kohout}}},\ and\ \bibinfo {author} {\bibfnamefont {K.}~\bibnamefont
  {{Young}}},\ }\bibfield  {title} {\bibinfo {title} {Direct randomized
  benchmarking for multiqubit devices},\ }\href
  {https://doi.org/10.1103/PhysRevLett.123.030503} {\bibfield  {journal}
  {\bibinfo  {journal} {\prl}\ }\textbf {\bibinfo {volume} {123}},\ \bibinfo
  {eid} {030503} (\bibinfo {year} {2019})},\ \Eprint
  {https://arxiv.org/abs/1807.07975} {arXiv:1807.07975 [quant-ph]} \BibitemShut
  {NoStop}%
\bibitem [{\citenamefont {Chasseur}\ and\ \citenamefont
  {Wilhelm}(2015)}]{chasseur_complete_2015}%
  \BibitemOpen
  \bibfield  {author} {\bibinfo {author} {\bibfnamefont {T.}~\bibnamefont
  {Chasseur}}\ and\ \bibinfo {author} {\bibfnamefont {F.~K.}\ \bibnamefont
  {Wilhelm}},\ }\bibfield  {title} {\bibinfo {title} {A complete randomized
  benchmarking protocol accounting for leakage errors},\ }\href
  {https://doi.org/10.1103/PhysRevA.92.042333} {\bibfield  {journal} {\bibinfo
  {journal} {Phys. Rev. A}\ }\textbf {\bibinfo {volume} {92}},\ \bibinfo
  {pages} {042333} (\bibinfo {year} {2015})},\ \Eprint
  {https://arxiv.org/abs/1505.00580 [quant-ph]} {arXiv:1505.00580 [quant-ph]}
  \BibitemShut {NoStop}%
\bibitem [{\citenamefont {Chasseur}\ \emph {et~al.}(2017)\citenamefont
  {Chasseur}, \citenamefont {Reich}, \citenamefont {Koch},\ and\ \citenamefont
  {Wilhelm}}]{chasseur_hybrid_2017}%
  \BibitemOpen
  \bibfield  {author} {\bibinfo {author} {\bibfnamefont {T.}~\bibnamefont
  {Chasseur}}, \bibinfo {author} {\bibfnamefont {D.~M.}\ \bibnamefont {Reich}},
  \bibinfo {author} {\bibfnamefont {C.~P.}\ \bibnamefont {Koch}},\ and\
  \bibinfo {author} {\bibfnamefont {F.~K.}\ \bibnamefont {Wilhelm}},\
  }\bibfield  {title} {\bibinfo {title} {Hybrid benchmarking of arbitrary
  quantum gates},\ }\href {https://doi.org/10.1103/PhysRevA.95.062335}
  {\bibfield  {journal} {\bibinfo  {journal} {Phys. Rev. A}\ }\textbf {\bibinfo
  {volume} {95}},\ \bibinfo {pages} {062335} (\bibinfo {year} {2017})},\
  \Eprint {https://arxiv.org/abs/1606.03927 [quant-ph]} {arXiv:1606.03927
  [quant-ph]} \BibitemShut {NoStop}%
\bibitem [{\citenamefont {Hangleiter}\ and\ \citenamefont
  {Eisert}(2023)}]{hangleiter_computational_2023}%
  \BibitemOpen
  \bibfield  {author} {\bibinfo {author} {\bibfnamefont {D.}~\bibnamefont
  {Hangleiter}}\ and\ \bibinfo {author} {\bibfnamefont {J.}~\bibnamefont
  {Eisert}},\ }\href {https://doi.org/10.48550/arXiv.2206.04079} {\bibinfo
  {title} {Computational advantage of quantum random sampling}} (\bibinfo
  {year} {2023}),\ \bibinfo {note} {accepted in Rev.\ Mod.\ Phys.},\ \Eprint
  {https://arxiv.org/abs/arXiv:2206.04079} {arXiv:arXiv:2206.04079 [quant-ph]}
  \BibitemShut {NoStop}%
\bibitem [{\citenamefont {{Boone}}\ \emph {et~al.}(2019)\citenamefont
  {{Boone}}, \citenamefont {{Carignan-Dugas}}, \citenamefont {{Wallman}},\ and\
  \citenamefont {{Emerson}}}]{boone2019randomized}%
  \BibitemOpen
  \bibfield  {author} {\bibinfo {author} {\bibfnamefont {K.}~\bibnamefont
  {{Boone}}}, \bibinfo {author} {\bibfnamefont {A.}~\bibnamefont
  {{Carignan-Dugas}}}, \bibinfo {author} {\bibfnamefont {J.~J.}\ \bibnamefont
  {{Wallman}}},\ and\ \bibinfo {author} {\bibfnamefont {J.}~\bibnamefont
  {{Emerson}}},\ }\bibfield  {title} {\bibinfo {title} {Randomized benchmarking
  under different gate sets},\ }\href
  {https://doi.org/10.1103/PhysRevA.99.032329} {\bibfield  {journal} {\bibinfo
  {journal} {\pra}\ }\textbf {\bibinfo {volume} {99}},\ \bibinfo {eid} {032329}
  (\bibinfo {year} {2019})},\ \Eprint {https://arxiv.org/abs/1811.01920}
  {arXiv:1811.01920 [quant-ph]} \BibitemShut {NoStop}%
\bibitem [{\citenamefont {Rinott}\ \emph {et~al.}(2022)\citenamefont {Rinott},
  \citenamefont {Shoham},\ and\ \citenamefont
  {Kalai}}]{RinottShohamKalai:2022}%
  \BibitemOpen
  \bibfield  {author} {\bibinfo {author} {\bibfnamefont {Y.}~\bibnamefont
  {Rinott}}, \bibinfo {author} {\bibfnamefont {T.}~\bibnamefont {Shoham}},\
  and\ \bibinfo {author} {\bibfnamefont {G.}~\bibnamefont {Kalai}},\ }\bibfield
   {title} {\bibinfo {title} {Statistical aspects of the quantum supremacy
  demonstration},\ }\href@noop {} {\bibfield  {journal} {\bibinfo  {journal}
  {Statistical Science}\ }\textbf {\bibinfo {volume} {37}},\ \bibinfo {pages}
  {322} (\bibinfo {year} {2022})},\ \Eprint {https://arxiv.org/abs/2008.05177}
  {arXiv:2008.05177 [quant-ph]} \BibitemShut {NoStop}%
\bibitem [{\citenamefont {Barak}\ \emph {et~al.}(2020)\citenamefont {Barak},
  \citenamefont {Chou},\ and\ \citenamefont {Gao}}]{BarakEtAl:2020:Spoofing}%
  \BibitemOpen
  \bibfield  {author} {\bibinfo {author} {\bibfnamefont {B.}~\bibnamefont
  {Barak}}, \bibinfo {author} {\bibfnamefont {C.-N.}\ \bibnamefont {Chou}},\
  and\ \bibinfo {author} {\bibfnamefont {X.}~\bibnamefont {Gao}},\ }\href@noop
  {} {\bibinfo {title} {Spoofing linear cross-entropy benchmarking in shallow
  quantum circuits}} (\bibinfo {year} {2020}),\ \Eprint
  {https://arxiv.org/abs/2005.02421} {arXiv:2005.02421 [quant-ph]} \BibitemShut
  {NoStop}%
\bibitem [{\citenamefont {Chen}\ \emph
  {et~al.}(2022{\natexlab{a}})\citenamefont {Chen}, \citenamefont {Ding},\ and\
  \citenamefont {Huang}}]{chen_randomized_2022}%
  \BibitemOpen
  \bibfield  {author} {\bibinfo {author} {\bibfnamefont {J.}~\bibnamefont
  {Chen}}, \bibinfo {author} {\bibfnamefont {D.}~\bibnamefont {Ding}},\ and\
  \bibinfo {author} {\bibfnamefont {C.}~\bibnamefont {Huang}},\ }\bibfield
  {title} {\bibinfo {title} {Randomized benchmarking beyond groups},\ }\href
  {https://doi.org/10.1103/PRXQuantum.3.030320} {\bibfield  {journal} {\bibinfo
   {journal} {{PRX} Quantum}\ }\textbf {\bibinfo {volume} {3}},\ \bibinfo
  {pages} {030320} (\bibinfo {year} {2022}{\natexlab{a}})},\ \Eprint
  {https://arxiv.org/abs/2203.12703} {arXiv:2203.12703 [quant-ph]} \BibitemShut
  {NoStop}%
\bibitem [{\citenamefont {Chen}\ \emph
  {et~al.}(2022{\natexlab{b}})\citenamefont {Chen}, \citenamefont {Ding},
  \citenamefont {Huang},\ and\ \citenamefont {Kong}}]{chen_linear_2022}%
  \BibitemOpen
  \bibfield  {author} {\bibinfo {author} {\bibfnamefont {J.}~\bibnamefont
  {Chen}}, \bibinfo {author} {\bibfnamefont {D.}~\bibnamefont {Ding}}, \bibinfo
  {author} {\bibfnamefont {C.}~\bibnamefont {Huang}},\ and\ \bibinfo {author}
  {\bibfnamefont {L.}~\bibnamefont {Kong}},\ }\href
  {http://arxiv.org/abs/2206.08293} {\bibinfo {title} {Linear cross entropy
  benchmarking with {Clifford} circuits}} (\bibinfo {year}
  {2022}{\natexlab{b}}),\ \Eprint {https://arxiv.org/abs/2206.08293}
  {arXiv:2206.08293} \BibitemShut {NoStop}%
\bibitem [{\citenamefont {{Liu}}\ \emph {et~al.}(2021)\citenamefont {{Liu}},
  \citenamefont {{Otten}}, \citenamefont {{Bassirianjahromi}}, \citenamefont
  {{Jiang}},\ and\ \citenamefont {{Fefferman}}}]{Liu21BenchmarkingNear-term}%
  \BibitemOpen
  \bibfield  {author} {\bibinfo {author} {\bibfnamefont {Y.}~\bibnamefont
  {{Liu}}}, \bibinfo {author} {\bibfnamefont {M.}~\bibnamefont {{Otten}}},
  \bibinfo {author} {\bibfnamefont {R.}~\bibnamefont {{Bassirianjahromi}}},
  \bibinfo {author} {\bibfnamefont {L.}~\bibnamefont {{Jiang}}},\ and\ \bibinfo
  {author} {\bibfnamefont {B.}~\bibnamefont {{Fefferman}}},\ }\href@noop {}
  {\bibinfo {title} {Benchmarking near-term quantum computers via random
  circuit sampling}} (\bibinfo {year} {2021}),\ \Eprint
  {https://arxiv.org/abs/2105.05232} {arXiv:2105.05232 [quant-ph]} \BibitemShut
  {NoStop}%
\bibitem [{\citenamefont {{Dalzell}}\ \emph {et~al.}(2021)\citenamefont
  {{Dalzell}}, \citenamefont {{Hunter-Jones}},\ and\ \citenamefont
  {{Brand{\~a}o}}}]{Dalzell21RandomQuantumCircuits}%
  \BibitemOpen
  \bibfield  {author} {\bibinfo {author} {\bibfnamefont {A.~M.}\ \bibnamefont
  {{Dalzell}}}, \bibinfo {author} {\bibfnamefont {N.}~\bibnamefont
  {{Hunter-Jones}}},\ and\ \bibinfo {author} {\bibfnamefont {F.~G.~S.~L.}\
  \bibnamefont {{Brand{\~a}o}}},\ }\href@noop {} {\bibinfo {title} {Random
  quantum circuits transform local noise into global white noise}} (\bibinfo
  {year} {2021}),\ \Eprint {https://arxiv.org/abs/2111.14907} {arXiv:2111.14907
  [quant-ph]} \BibitemShut {NoStop}%
\bibitem [{\citenamefont {{Hunter-Jones}}(2019)}]{Hun19}%
  \BibitemOpen
  \bibfield  {author} {\bibinfo {author} {\bibfnamefont {N.}~\bibnamefont
  {{Hunter-Jones}}},\ }\href@noop {} {\bibinfo {title} {Unitary designs from
  statistical mechanics in random quantum circuits}} (\bibinfo {year} {2019}),\
  \Eprint {https://arxiv.org/abs/1905.12053} {arXiv:1905.12053 [quant-ph]}
  \BibitemShut {NoStop}%
\bibitem [{\citenamefont {{Helsen}}\ \emph {et~al.}(2019)\citenamefont
  {{Helsen}}, \citenamefont {{Xue}}, \citenamefont {{Vandersypen}},\ and\
  \citenamefont {{Wehner}}}]{HelsenEtAl:2019:character}%
  \BibitemOpen
  \bibfield  {author} {\bibinfo {author} {\bibfnamefont {J.}~\bibnamefont
  {{Helsen}}}, \bibinfo {author} {\bibfnamefont {X.}~\bibnamefont {{Xue}}},
  \bibinfo {author} {\bibfnamefont {L.~M.~K.}\ \bibnamefont {{Vandersypen}}},\
  and\ \bibinfo {author} {\bibfnamefont {S.}~\bibnamefont {{Wehner}}},\
  }\bibfield  {title} {\bibinfo {title} {A new class of efficient randomized
  benchmarking protocols},\ }\href {https://doi.org/10.1038/s41534-019-0182-7}
  {\bibfield  {journal} {\bibinfo  {journal} {npj Quant. Inf.}\ }\textbf
  {\bibinfo {volume} {5}},\ \bibinfo {eid} {71} (\bibinfo {year} {2019})},\
  \Eprint {https://arxiv.org/abs/1806.02048} {arXiv:1806.02048 [quant-ph]}
  \BibitemShut {NoStop}%
\bibitem [{\citenamefont {Helsen}\ \emph
  {et~al.}(2022{\natexlab{b}})\citenamefont {Helsen}, \citenamefont {Nezami},
  \citenamefont {Reagor},\ and\ \citenamefont
  {Walter}}]{Helsen20MatchgateBenchmarking}%
  \BibitemOpen
  \bibfield  {author} {\bibinfo {author} {\bibfnamefont {J.}~\bibnamefont
  {Helsen}}, \bibinfo {author} {\bibfnamefont {S.}~\bibnamefont {Nezami}},
  \bibinfo {author} {\bibfnamefont {M.}~\bibnamefont {Reagor}},\ and\ \bibinfo
  {author} {\bibfnamefont {M.}~\bibnamefont {Walter}},\ }\bibfield  {title}
  {\bibinfo {title} {Matchgate benchmarking: {S}calable benchmarking of a
  continuous family of many-qubit gates},\ }\href
  {https://doi.org/10.22331/q-2022-02-21-657} {\bibfield  {journal} {\bibinfo
  {journal} {{Quantum}}\ }\textbf {\bibinfo {volume} {6}},\ \bibinfo {pages}
  {657} (\bibinfo {year} {2022}{\natexlab{b}})},\ \Eprint
  {https://arxiv.org/abs/2011.13048} {arXiv:2011.13048 [quant-ph]} \BibitemShut
  {NoStop}%
\bibitem [{\citenamefont {Flammia}\ and\ \citenamefont
  {Wallman}(2020)}]{Flammia2019EfficientEstimation}%
  \BibitemOpen
  \bibfield  {author} {\bibinfo {author} {\bibfnamefont {S.~T.}\ \bibnamefont
  {Flammia}}\ and\ \bibinfo {author} {\bibfnamefont {J.~J.}\ \bibnamefont
  {Wallman}},\ }\bibfield  {title} {\bibinfo {title} {Efficient estimation of
  {Pauli} channels},\ }\href {https://doi.org/10.1145/3408039} {\bibfield
  {journal} {\bibinfo  {journal} {ACM Transactions on Quantum Computing}\
  }\textbf {\bibinfo {volume} {1}},\ \bibinfo {pages} {1} (\bibinfo {year}
  {2020})},\ \Eprint {https://arxiv.org/abs/1907.12976} {arXiv:1907.12976
  [quant-ph]} \BibitemShut {NoStop}%
\bibitem [{\citenamefont {{Gambetta}}\ \emph {et~al.}(2012)\citenamefont
  {{Gambetta}}, \citenamefont {{C{\'o}rcoles}}, \citenamefont {{Merkel}},
  \citenamefont {{Johnson}}, \citenamefont {{Smolin}}, \citenamefont {{Chow}},
  \citenamefont {{Ryan}}, \citenamefont {{Rigetti}}, \citenamefont {{Poletto}},
  \citenamefont {{Ohki}}, \citenamefont {{Ketchen}},\ and\ \citenamefont
  {{Steffen}}}]{Gambetta2012CharacterizationAddressability}%
  \BibitemOpen
  \bibfield  {author} {\bibinfo {author} {\bibfnamefont {J.~M.}\ \bibnamefont
  {{Gambetta}}}, \bibinfo {author} {\bibfnamefont {A.~D.}\ \bibnamefont
  {{C{\'o}rcoles}}}, \bibinfo {author} {\bibfnamefont {S.~T.}\ \bibnamefont
  {{Merkel}}}, \bibinfo {author} {\bibfnamefont {B.~R.}\ \bibnamefont
  {{Johnson}}}, \bibinfo {author} {\bibfnamefont {J.~A.}\ \bibnamefont
  {{Smolin}}}, \bibinfo {author} {\bibfnamefont {J.~M.}\ \bibnamefont
  {{Chow}}}, \bibinfo {author} {\bibfnamefont {C.~A.}\ \bibnamefont {{Ryan}}},
  \bibinfo {author} {\bibfnamefont {C.}~\bibnamefont {{Rigetti}}}, \bibinfo
  {author} {\bibfnamefont {S.}~\bibnamefont {{Poletto}}}, \bibinfo {author}
  {\bibfnamefont {T.~A.}\ \bibnamefont {{Ohki}}}, \bibinfo {author}
  {\bibfnamefont {M.~B.}\ \bibnamefont {{Ketchen}}},\ and\ \bibinfo {author}
  {\bibfnamefont {M.}~\bibnamefont {{Steffen}}},\ }\bibfield  {title} {\bibinfo
  {title} {Characterization of addressability by simultaneous randomized
  benchmarking},\ }\href {https://doi.org/10.1103/PhysRevLett.109.240504}
  {\bibfield  {journal} {\bibinfo  {journal} {\prl}\ }\textbf {\bibinfo
  {volume} {109}},\ \bibinfo {eid} {240504} (\bibinfo {year} {2012})},\ \Eprint
  {https://arxiv.org/abs/1204.6308} {arXiv:1204.6308 [quant-ph]} \BibitemShut
  {NoStop}%
\bibitem [{\citenamefont {{McKay}}\ \emph {et~al.}(2020)\citenamefont
  {{McKay}}, \citenamefont {{Cross}}, \citenamefont {{Wood}},\ and\
  \citenamefont {{Gambetta}}}]{McKay2020CorrelatedRB}%
  \BibitemOpen
  \bibfield  {author} {\bibinfo {author} {\bibfnamefont {D.~C.}\ \bibnamefont
  {{McKay}}}, \bibinfo {author} {\bibfnamefont {A.~W.}\ \bibnamefont
  {{Cross}}}, \bibinfo {author} {\bibfnamefont {C.~J.}\ \bibnamefont
  {{Wood}}},\ and\ \bibinfo {author} {\bibfnamefont {J.~M.}\ \bibnamefont
  {{Gambetta}}},\ }\href@noop {} {\bibinfo {title} {Correlated randomized
  benchmarking}} (\bibinfo {year} {2020}),\ \Eprint
  {https://arxiv.org/abs/2003.02354} {arXiv:2003.02354 [quant-ph]} \BibitemShut
  {NoStop}%
\bibitem [{\citenamefont {{Huang}}\ \emph {et~al.}(2020)\citenamefont
  {{Huang}}, \citenamefont {{Kueng}},\ and\ \citenamefont
  {{Preskill}}}]{Huang2020Predicting}%
  \BibitemOpen
  \bibfield  {author} {\bibinfo {author} {\bibfnamefont {H.-Y.}\ \bibnamefont
  {{Huang}}}, \bibinfo {author} {\bibfnamefont {R.}~\bibnamefont {{Kueng}}},\
  and\ \bibinfo {author} {\bibfnamefont {J.}~\bibnamefont {{Preskill}}},\
  }\bibfield  {title} {\bibinfo {title} {Predicting many properties of a
  quantum system from very few measurements},\ }\href
  {https://doi.org/10.1038/s41567-020-0932-7} {\bibfield  {journal} {\bibinfo
  {journal} {Nature Physics}\ }\textbf {\bibinfo {volume} {16}},\ \bibinfo
  {pages} {1050–1057} (\bibinfo {year} {2020})},\ \Eprint
  {https://arxiv.org/abs/2002.08953} {arXiv:2002.08953 [quant-ph]} \BibitemShut
  {NoStop}%
\bibitem [{\citenamefont {Harrow}\ and\ \citenamefont
  {Low}(2009)}]{harrow_random_2009}%
  \BibitemOpen
  \bibfield  {author} {\bibinfo {author} {\bibfnamefont {A.~W.}\ \bibnamefont
  {Harrow}}\ and\ \bibinfo {author} {\bibfnamefont {R.~A.}\ \bibnamefont
  {Low}},\ }\bibfield  {title} {\bibinfo {title} {Random {quantum} {circuits}
  are approximate 2-designs},\ }\href
  {https://doi.org/10.1007/s00220-009-0873-6} {\bibfield  {journal} {\bibinfo
  {journal} {Commun. Math. Phys.}\ }\textbf {\bibinfo {volume} {291}},\
  \bibinfo {pages} {257} (\bibinfo {year} {2009})},\ \Eprint
  {https://arxiv.org/abs/0802.1919} {arXiv:0802.1919 [quant-ph]} \BibitemShut
  {NoStop}%
\bibitem [{\citenamefont {Brown}\ and\ \citenamefont
  {Viola}(2010)}]{brown_convergence_2010}%
  \BibitemOpen
  \bibfield  {author} {\bibinfo {author} {\bibfnamefont {W.~G.}\ \bibnamefont
  {Brown}}\ and\ \bibinfo {author} {\bibfnamefont {L.}~\bibnamefont {Viola}},\
  }\bibfield  {title} {\bibinfo {title} {Convergence rates for arbitrary
  statistical moments of random quantum circuits},\ }\href
  {https://doi.org/10.1103/PhysRevLett.104.250501} {\bibfield  {journal}
  {\bibinfo  {journal} {Phys. Rev. Lett.}\ }\textbf {\bibinfo {volume} {104}},\
  \bibinfo {pages} {250501} (\bibinfo {year} {2010})},\ \Eprint
  {https://arxiv.org/abs/0910.0913} {arXiv:0910.0913} \BibitemShut {NoStop}%
\bibitem [{\citenamefont {Haferkamp}\ \emph {et~al.}(2023)\citenamefont
  {Haferkamp}, \citenamefont {Montealegre-Mora}, \citenamefont {Heinrich},
  \citenamefont {Eisert}, \citenamefont {Gross},\ and\ \citenamefont
  {Roth}}]{haferkamp_efficient_2023}%
  \BibitemOpen
  \bibfield  {author} {\bibinfo {author} {\bibfnamefont {J.}~\bibnamefont
  {Haferkamp}}, \bibinfo {author} {\bibfnamefont {F.}~\bibnamefont
  {Montealegre-Mora}}, \bibinfo {author} {\bibfnamefont {M.}~\bibnamefont
  {Heinrich}}, \bibinfo {author} {\bibfnamefont {J.}~\bibnamefont {Eisert}},
  \bibinfo {author} {\bibfnamefont {D.}~\bibnamefont {Gross}},\ and\ \bibinfo
  {author} {\bibfnamefont {I.}~\bibnamefont {Roth}},\ }\bibfield  {title}
  {\bibinfo {title} {Efficient unitary designs with a system-size independent
  number of non-{Clifford} gates},\ }\href
  {https://doi.org/10.1007/s00220-022-04507-6} {\bibfield  {journal} {\bibinfo
  {journal} {Commun. Math. Phys.}\ }\textbf {\bibinfo {volume} {397}},\
  \bibinfo {pages} {995} (\bibinfo {year} {2023})},\ \Eprint
  {https://arxiv.org/abs/2002.09524} {arXiv:2002.09524 [quant-ph]} \BibitemShut
  {NoStop}%
\bibitem [{\citenamefont {{Haferkamp}}\ and\ \citenamefont
  {{Hunter-Jones}}(2021)}]{haferkamp_improved_2021}%
  \BibitemOpen
  \bibfield  {author} {\bibinfo {author} {\bibfnamefont {J.}~\bibnamefont
  {{Haferkamp}}}\ and\ \bibinfo {author} {\bibfnamefont {N.}~\bibnamefont
  {{Hunter-Jones}}},\ }\bibfield  {title} {\bibinfo {title} {Improved spectral
  gaps for random quantum circuits: Large local dimensions and all-to-all
  interactions},\ }\href {https://doi.org/10.1103/PhysRevA.104.022417}
  {\bibfield  {journal} {\bibinfo  {journal} {\pra}\ }\textbf {\bibinfo
  {volume} {104}},\ \bibinfo {eid} {022417} (\bibinfo {year} {2021})},\ \Eprint
  {https://arxiv.org/abs/2012.05259} {arXiv:2012.05259 [quant-ph]} \BibitemShut
  {NoStop}%
\bibitem [{\citenamefont {{Harrow}}\ and\ \citenamefont
  {{Mehraban}}(2018)}]{HarMeh18}%
  \BibitemOpen
  \bibfield  {author} {\bibinfo {author} {\bibfnamefont {A.}~\bibnamefont
  {{Harrow}}}\ and\ \bibinfo {author} {\bibfnamefont {S.}~\bibnamefont
  {{Mehraban}}},\ }\href@noop {} {\bibinfo {title} {Approximate unitary
  $t$-designs by short random quantum circuits using nearest-neighbor and
  long-range gates}} (\bibinfo {year} {2018}),\ \Eprint
  {https://arxiv.org/abs/1809.06957} {arXiv:1809.06957 [quant-ph]} \BibitemShut
  {NoStop}%
\bibitem [{\citenamefont {{Derbyshire}}\ \emph {et~al.}(2020)\citenamefont
  {{Derbyshire}}, \citenamefont {{Malo}}, \citenamefont {{Daley}},
  \citenamefont {{Kashefi}},\ and\ \citenamefont
  {{Wallden}}}]{Derbyshire2019RBanalogue}%
  \BibitemOpen
  \bibfield  {author} {\bibinfo {author} {\bibfnamefont {E.}~\bibnamefont
  {{Derbyshire}}}, \bibinfo {author} {\bibfnamefont {J.~Y.}\ \bibnamefont
  {{Malo}}}, \bibinfo {author} {\bibfnamefont {A.~J.}\ \bibnamefont {{Daley}}},
  \bibinfo {author} {\bibfnamefont {E.}~\bibnamefont {{Kashefi}}},\ and\
  \bibinfo {author} {\bibfnamefont {P.}~\bibnamefont {{Wallden}}},\ }\bibfield
  {title} {\bibinfo {title} {Randomized benchmarking in the analogue setting},\
  }\href {https://doi.org/10.1088/2058-9565/ab7eec} {\bibfield  {journal}
  {\bibinfo  {journal} {Quantum Sci. Technol.}\ }\textbf {\bibinfo {volume}
  {5}},\ \bibinfo {eid} {034001} (\bibinfo {year} {2020})},\ \Eprint
  {https://arxiv.org/abs/1909.01295} {arXiv:1909.01295 [quant-ph]} \BibitemShut
  {NoStop}%
\bibitem [{\citenamefont {{Shaffer}}\ \emph {et~al.}(2021)\citenamefont
  {{Shaffer}}, \citenamefont {{Megidish}}, \citenamefont {{Broz}},
  \citenamefont {{Chen}},\ and\ \citenamefont
  {{H{\"a}ffner}}}]{shaffer2021practical}%
  \BibitemOpen
  \bibfield  {author} {\bibinfo {author} {\bibfnamefont {R.}~\bibnamefont
  {{Shaffer}}}, \bibinfo {author} {\bibfnamefont {E.}~\bibnamefont
  {{Megidish}}}, \bibinfo {author} {\bibfnamefont {J.}~\bibnamefont {{Broz}}},
  \bibinfo {author} {\bibfnamefont {W.-T.}\ \bibnamefont {{Chen}}},\ and\
  \bibinfo {author} {\bibfnamefont {H.}~\bibnamefont {{H{\"a}ffner}}},\
  }\bibfield  {title} {\bibinfo {title} {Practical verification protocols for
  analog quantum simulators},\ }\href
  {https://doi.org/10.1038/s41534-021-00380-8} {\bibfield  {journal} {\bibinfo
  {journal} {npj Quant. Inf.}\ }\textbf {\bibinfo {volume} {7}},\ \bibinfo
  {eid} {46} (\bibinfo {year} {2021})},\ \Eprint
  {https://arxiv.org/abs/2003.04500} {arXiv:2003.04500 [quant-ph]} \BibitemShut
  {NoStop}%
\bibitem [{\citenamefont {{Wallman}}\ and\ \citenamefont
  {{Emerson}}(2016)}]{Wallman16NoiseTailoringFor}%
  \BibitemOpen
  \bibfield  {author} {\bibinfo {author} {\bibfnamefont {J.~J.}\ \bibnamefont
  {{Wallman}}}\ and\ \bibinfo {author} {\bibfnamefont {J.}~\bibnamefont
  {{Emerson}}},\ }\bibfield  {title} {\bibinfo {title} {Noise tailoring for
  scalable quantum computation via randomized compiling},\ }\href
  {https://doi.org/10.1103/PhysRevA.94.052325} {\bibfield  {journal} {\bibinfo
  {journal} {\pra}\ }\textbf {\bibinfo {volume} {94}},\ \bibinfo {eid} {052325}
  (\bibinfo {year} {2016})},\ \Eprint {https://arxiv.org/abs/1512.01098}
  {arXiv:1512.01098 [quant-ph]} \BibitemShut {NoStop}%
\bibitem [{\citenamefont {Elben}\ \emph {et~al.}(2022)\citenamefont {Elben},
  \citenamefont {Flammia}, \citenamefont {Huang}, \citenamefont {Kueng},
  \citenamefont {Preskill}, \citenamefont {Vermersch},\ and\ \citenamefont
  {Zoller}}]{elben_randomized_2022}%
  \BibitemOpen
  \bibfield  {author} {\bibinfo {author} {\bibfnamefont {A.}~\bibnamefont
  {Elben}}, \bibinfo {author} {\bibfnamefont {S.~T.}\ \bibnamefont {Flammia}},
  \bibinfo {author} {\bibfnamefont {H.-Y.}\ \bibnamefont {Huang}}, \bibinfo
  {author} {\bibfnamefont {R.}~\bibnamefont {Kueng}}, \bibinfo {author}
  {\bibfnamefont {J.}~\bibnamefont {Preskill}}, \bibinfo {author}
  {\bibfnamefont {B.}~\bibnamefont {Vermersch}},\ and\ \bibinfo {author}
  {\bibfnamefont {P.}~\bibnamefont {Zoller}},\ }\bibfield  {title} {\bibinfo
  {title} {The randomized measurement toolbox},\ }\bibfield  {journal}
  {\bibinfo  {journal} {Nat. Rev. Phys.}\ }\href
  {https://doi.org/10.1038/s42254-022-00535-2} {10.1038/s42254-022-00535-2}
  (\bibinfo {year} {2022}),\ \Eprint {https://arxiv.org/abs/2203.11374}
  {arXiv:2203.11374 [quant-ph]} \BibitemShut {NoStop}%
\bibitem [{\citenamefont {{Temme}}\ \emph {et~al.}(2017)\citenamefont
  {{Temme}}, \citenamefont {{Bravyi}},\ and\ \citenamefont
  {{Gambetta}}}]{Temme_2017}%
  \BibitemOpen
  \bibfield  {author} {\bibinfo {author} {\bibfnamefont {K.}~\bibnamefont
  {{Temme}}}, \bibinfo {author} {\bibfnamefont {S.}~\bibnamefont {{Bravyi}}},\
  and\ \bibinfo {author} {\bibfnamefont {J.~M.}\ \bibnamefont {{Gambetta}}},\
  }\bibfield  {title} {\bibinfo {title} {Error mitigation for short-depth
  quantum circuits},\ }\href {https://doi.org/10.1103/PhysRevLett.119.180509}
  {\bibfield  {journal} {\bibinfo  {journal} {\prl}\ }\textbf {\bibinfo
  {volume} {119}},\ \bibinfo {eid} {180509} (\bibinfo {year} {2017})},\ \Eprint
  {https://arxiv.org/abs/1612.02058} {arXiv:1612.02058 [quant-ph]} \BibitemShut
  {NoStop}%
\bibitem [{\citenamefont {Huang}\ \emph {et~al.}(2020)\citenamefont {Huang},
  \citenamefont {Kueng},\ and\ \citenamefont
  {Preskill}}]{huang_predicting_2020}%
  \BibitemOpen
  \bibfield  {author} {\bibinfo {author} {\bibfnamefont {H.-Y.}\ \bibnamefont
  {Huang}}, \bibinfo {author} {\bibfnamefont {R.}~\bibnamefont {Kueng}},\ and\
  \bibinfo {author} {\bibfnamefont {J.}~\bibnamefont {Preskill}},\ }\bibfield
  {title} {\bibinfo {title} {Predicting many properties of a quantum system
  from very few measurements},\ }\href
  {https://doi.org/10.1038/s41567-020-0932-7} {\bibfield  {journal} {\bibinfo
  {journal} {Nat. Phys.}\ }\textbf {\bibinfo {volume} {16}},\ \bibinfo {pages}
  {1050} (\bibinfo {year} {2020})},\ \Eprint {https://arxiv.org/abs/2002.08953}
  {arXiv:2002.08953} \BibitemShut {NoStop}%
\bibitem [{\citenamefont {{Knill}}(2005)}]{knill05}%
  \BibitemOpen
  \bibfield  {author} {\bibinfo {author} {\bibfnamefont {E.}~\bibnamefont
  {{Knill}}},\ }\bibfield  {title} {\bibinfo {title} {Quantum computing with
  realistically noisy devices},\ }\href {https://doi.org/10.1038/nature03350}
  {\bibfield  {journal} {\bibinfo  {journal} {\nat}\ }\textbf {\bibinfo
  {volume} {434}},\ \bibinfo {pages} {39} (\bibinfo {year} {2005})},\ \Eprint
  {https://arxiv.org/abs/quant-ph/0410199} {arXiv:quant-ph/0410199 [quant-ph]}
  \BibitemShut {NoStop}%
\bibitem [{\citenamefont {{Ware}}\ \emph {et~al.}(2021)\citenamefont {{Ware}},
  \citenamefont {{Ribeill}}, \citenamefont {{Rist{\`e}}}, \citenamefont
  {{Ryan}}, \citenamefont {{Johnson}},\ and\ \citenamefont {{da
  Silva}}}]{Ware21ExperimentalPauli-frame}%
  \BibitemOpen
  \bibfield  {author} {\bibinfo {author} {\bibfnamefont {M.}~\bibnamefont
  {{Ware}}}, \bibinfo {author} {\bibfnamefont {G.}~\bibnamefont {{Ribeill}}},
  \bibinfo {author} {\bibfnamefont {D.}~\bibnamefont {{Rist{\`e}}}}, \bibinfo
  {author} {\bibfnamefont {C.~A.}\ \bibnamefont {{Ryan}}}, \bibinfo {author}
  {\bibfnamefont {B.}~\bibnamefont {{Johnson}}},\ and\ \bibinfo {author}
  {\bibfnamefont {M.~P.}\ \bibnamefont {{da Silva}}},\ }\bibfield  {title}
  {\bibinfo {title} {Experimental {Pauli}-frame randomization on a
  superconducting qubit},\ }\href {https://doi.org/10.1103/PhysRevA.103.042604}
  {\bibfield  {journal} {\bibinfo  {journal} {\pra}\ }\textbf {\bibinfo
  {volume} {103}},\ \bibinfo {eid} {042604} (\bibinfo {year} {2021})},\ \Eprint
  {https://arxiv.org/abs/1803.01818} {arXiv:1803.01818 [quant-ph]} \BibitemShut
  {NoStop}%
\bibitem [{\citenamefont {Polloreno}\ \emph {et~al.}(2023)\citenamefont
  {Polloreno}, \citenamefont {Carignan-Dugas}, \citenamefont {Hines},
  \citenamefont {Blume-Kohout}, \citenamefont {Young},\ and\ \citenamefont
  {Proctor}}]{polloreno_theory_2023}%
  \BibitemOpen
  \bibfield  {author} {\bibinfo {author} {\bibfnamefont {A.~M.}\ \bibnamefont
  {Polloreno}}, \bibinfo {author} {\bibfnamefont {A.}~\bibnamefont
  {Carignan-Dugas}}, \bibinfo {author} {\bibfnamefont {J.}~\bibnamefont
  {Hines}}, \bibinfo {author} {\bibfnamefont {R.}~\bibnamefont {Blume-Kohout}},
  \bibinfo {author} {\bibfnamefont {K.}~\bibnamefont {Young}},\ and\ \bibinfo
  {author} {\bibfnamefont {T.}~\bibnamefont {Proctor}},\ }\href
  {https://doi.org/10.48550/arXiv.2302.13853} {\bibinfo {title} {A theory of
  direct randomized benchmarking}} (\bibinfo {year} {2023}),\ \Eprint
  {https://arxiv.org/abs/2302.13853 [quant-ph]} {arXiv:2302.13853 [quant-ph]}
  \BibitemShut {NoStop}%
\bibitem [{\citenamefont {{Stilck Fran{\c{c}}a}}\ \emph
  {et~al.}(2021)\citenamefont {{Stilck Fran{\c{c}}a}}, \citenamefont
  {{Strelchuk}},\ and\ \citenamefont
  {{Studzi{\'n}ski}}}]{StilckFranca2020EfficientBenchmarkingAnd}%
  \BibitemOpen
  \bibfield  {author} {\bibinfo {author} {\bibfnamefont {D.}~\bibnamefont
  {{Stilck Fran{\c{c}}a}}}, \bibinfo {author} {\bibfnamefont {S.}~\bibnamefont
  {{Strelchuk}}},\ and\ \bibinfo {author} {\bibfnamefont {M.}~\bibnamefont
  {{Studzi{\'n}ski}}},\ }\bibfield  {title} {\bibinfo {title} {Efficient
  benchmarking and classical simulation of quantum processes in the {Weyl}
  basis},\ }\href {https://doi.org/10.1103/PhysRevLett.126.210502} {\bibfield
  {journal} {\bibinfo  {journal} {\prl}\ }\textbf {\bibinfo {volume} {126}},\
  \bibinfo {eid} {210502} (\bibinfo {year} {2021})},\ \Eprint
  {https://arxiv.org/abs/2008.12250} {arXiv:2008.12250 [quant-ph]} \BibitemShut
  {NoStop}%
\bibitem [{\citenamefont {{Erhard}}\ \emph {et~al.}(2019)\citenamefont
  {{Erhard}}, \citenamefont {{Wallman}}, \citenamefont {{Postler}},
  \citenamefont {{Meth}}, \citenamefont {{Stricker}}, \citenamefont
  {{Martinez}}, \citenamefont {{Schindler}}, \citenamefont {{Monz}},
  \citenamefont {{Emerson}},\ and\ \citenamefont
  {{Blatt}}}]{Erhard2019CharacterizingLarge-scale}%
  \BibitemOpen
  \bibfield  {author} {\bibinfo {author} {\bibfnamefont {A.}~\bibnamefont
  {{Erhard}}}, \bibinfo {author} {\bibfnamefont {J.~J.}\ \bibnamefont
  {{Wallman}}}, \bibinfo {author} {\bibfnamefont {L.}~\bibnamefont
  {{Postler}}}, \bibinfo {author} {\bibfnamefont {M.}~\bibnamefont {{Meth}}},
  \bibinfo {author} {\bibfnamefont {R.}~\bibnamefont {{Stricker}}}, \bibinfo
  {author} {\bibfnamefont {E.~A.}\ \bibnamefont {{Martinez}}}, \bibinfo
  {author} {\bibfnamefont {P.}~\bibnamefont {{Schindler}}}, \bibinfo {author}
  {\bibfnamefont {T.}~\bibnamefont {{Monz}}}, \bibinfo {author} {\bibfnamefont
  {J.}~\bibnamefont {{Emerson}}},\ and\ \bibinfo {author} {\bibfnamefont
  {R.}~\bibnamefont {{Blatt}}},\ }\bibfield  {title} {\bibinfo {title}
  {Characterizing large-scale quantum computers via cycle benchmarking},\
  }\href {https://doi.org/10.1038/s41467-019-13068-7} {\bibfield  {journal}
  {\bibinfo  {journal} {Nat. Commun.}\ }\textbf {\bibinfo {volume} {10}},\
  \bibinfo {eid} {5347} (\bibinfo {year} {2019})},\ \Eprint
  {https://arxiv.org/abs/1902.08543} {arXiv:1902.08543 [quant-ph]} \BibitemShut
  {NoStop}%
\bibitem [{\citenamefont {{Zhang}}\ \emph {et~al.}(2022)\citenamefont
  {{Zhang}}, \citenamefont {{Yu}}, \citenamefont {{Zeng}}, \citenamefont
  {{Liu}},\ and\ \citenamefont {{Ma}}}]{Zhang22ScalableFastBenchmarking}%
  \BibitemOpen
  \bibfield  {author} {\bibinfo {author} {\bibfnamefont {Y.}~\bibnamefont
  {{Zhang}}}, \bibinfo {author} {\bibfnamefont {W.}~\bibnamefont {{Yu}}},
  \bibinfo {author} {\bibfnamefont {P.}~\bibnamefont {{Zeng}}}, \bibinfo
  {author} {\bibfnamefont {G.}~\bibnamefont {{Liu}}},\ and\ \bibinfo {author}
  {\bibfnamefont {X.}~\bibnamefont {{Ma}}},\ }\href@noop {} {\bibinfo {title}
  {Scalable fast benchmarking for individual quantum gates with local
  twirling}} (\bibinfo {year} {2022}),\ \Eprint
  {https://arxiv.org/abs/2203.10320} {arXiv:2203.10320 [quant-ph]} \BibitemShut
  {NoStop}%
\bibitem [{fla(2021)}]{flammia_averaged_2021}%
  \BibitemOpen
  \href {http://arxiv.org/abs/2108.05803} {\bibinfo {title} {Averaged circuit
  eigenvalue sampling}} (\bibinfo {year} {2021}),\ \Eprint
  {https://arxiv.org/abs/2108.05803} {arXiv:2108.05803 [quant-ph]} \BibitemShut
  {NoStop}%
\bibitem [{\citenamefont {{Proctor}}\ \emph {et~al.}(2022)\citenamefont
  {{Proctor}}, \citenamefont {{Seritan}}, \citenamefont {{Rudinger}},
  \citenamefont {{Nielsen}}, \citenamefont {{Blume-Kohout}},\ and\
  \citenamefont {{Young}}}]{Proctor21ScalableRB}%
  \BibitemOpen
  \bibfield  {author} {\bibinfo {author} {\bibfnamefont {T.}~\bibnamefont
  {{Proctor}}}, \bibinfo {author} {\bibfnamefont {S.}~\bibnamefont
  {{Seritan}}}, \bibinfo {author} {\bibfnamefont {K.}~\bibnamefont
  {{Rudinger}}}, \bibinfo {author} {\bibfnamefont {E.}~\bibnamefont
  {{Nielsen}}}, \bibinfo {author} {\bibfnamefont {R.}~\bibnamefont
  {{Blume-Kohout}}},\ and\ \bibinfo {author} {\bibfnamefont {K.}~\bibnamefont
  {{Young}}},\ }\bibfield  {title} {\bibinfo {title} {Scalable randomized
  benchmarking of quantum computers using mirror circuits},\ }\href
  {https://doi.org/10.1103/PhysRevLett.129.150502} {\bibfield  {journal}
  {\bibinfo  {journal} {\prl}\ }\textbf {\bibinfo {volume} {129}},\ \bibinfo
  {eid} {150502} (\bibinfo {year} {2022})},\ \Eprint
  {https://arxiv.org/abs/2112.09853} {arXiv:2112.09853 [quant-ph]} \BibitemShut
  {NoStop}%
\bibitem [{hel(2021)}]{helsen_estimating_2021}%
  \BibitemOpen
  \href {http://arxiv.org/abs/2110.13178} {\bibinfo {title} {Estimating
  gate-set properties from random sequences}} (\bibinfo {year} {2021}),\
  \Eprint {https://arxiv.org/abs/2110.13178} {arXiv:2110.13178 [quant-ph]}
  \BibitemShut {NoStop}%
\bibitem [{\citenamefont {Figueroa-Romero}\ \emph {et~al.}(2022)\citenamefont
  {Figueroa-Romero}, \citenamefont {Modi},\ and\ \citenamefont
  {Hsieh}}]{figueroa-romero_towards_2022}%
  \BibitemOpen
  \bibfield  {author} {\bibinfo {author} {\bibfnamefont {P.}~\bibnamefont
  {Figueroa-Romero}}, \bibinfo {author} {\bibfnamefont {K.}~\bibnamefont
  {Modi}},\ and\ \bibinfo {author} {\bibfnamefont {M.-H.}\ \bibnamefont
  {Hsieh}},\ }\bibfield  {title} {\bibinfo {title} {Towards a general framework
  of randomized benchmarking incorporating non-{Markovian} noise},\ }\href
  {https://doi.org/10.22331/q-2022-12-01-868} {\bibfield  {journal} {\bibinfo
  {journal} {Quantum}\ }\textbf {\bibinfo {volume} {6}},\ \bibinfo {pages}
  {868} (\bibinfo {year} {2022})},\ \Eprint {https://arxiv.org/abs/2202.11338}
  {arXiv:2202.11338 [quant-ph]} \BibitemShut {NoStop}%
\bibitem [{\citenamefont {Dalzell}\ \emph {et~al.}(2022)\citenamefont
  {Dalzell}, \citenamefont {Hunter-Jones},\ and\ \citenamefont
  {Brandão}}]{dalzell_random_2022}%
  \BibitemOpen
  \bibfield  {author} {\bibinfo {author} {\bibfnamefont {A.~M.}\ \bibnamefont
  {Dalzell}}, \bibinfo {author} {\bibfnamefont {N.}~\bibnamefont
  {Hunter-Jones}},\ and\ \bibinfo {author} {\bibfnamefont {F.~G. S.~L.}\
  \bibnamefont {Brandão}},\ }\bibfield  {title} {\bibinfo {title} {Random
  quantum circuits anti-concentrate in log depth},\ }\bibfield  {journal}
  {\bibinfo  {journal} {{PRX} Quantum}\ }\textbf {\bibinfo {volume} {3}},\
  \href {https://doi.org/10.1103/PRXQuantum.3.010333}
  {10.1103/PRXQuantum.3.010333} (\bibinfo {year} {2022}),\ \Eprint
  {https://arxiv.org/abs/2011.12277} {arXiv:2011.12277 [quant-ph]} \BibitemShut
  {NoStop}%
\bibitem [{\citenamefont {Hunter-Jones}(2023)}]{hunter_hones_private}%
  \BibitemOpen
  \bibfield  {author} {\bibinfo {author} {\bibfnamefont {N.}~\bibnamefont
  {Hunter-Jones}},\ }\href@noop {} {}\bibinfo {howpublished} {Private
  Communication} (\bibinfo {year} {2023})\BibitemShut {NoStop}%
\bibitem [{\citenamefont {{Hu}}\ \emph {et~al.}(2021)\citenamefont {{Hu}},
  \citenamefont {{Choi}},\ and\ \citenamefont {{You}}}]{hu_classical_2021}%
  \BibitemOpen
  \bibfield  {author} {\bibinfo {author} {\bibfnamefont {H.-Y.}\ \bibnamefont
  {{Hu}}}, \bibinfo {author} {\bibfnamefont {S.}~\bibnamefont {{Choi}}},\ and\
  \bibinfo {author} {\bibfnamefont {Y.-Z.}\ \bibnamefont {{You}}},\ }\href@noop
  {} {\bibinfo {title} {Classical shadow tomography with locally scrambled
  quantum dynamics}} (\bibinfo {year} {2021}),\ \Eprint
  {https://arxiv.org/abs/2107.04817} {arXiv:2107.04817 [quant-ph]} \BibitemShut
  {NoStop}%
\bibitem [{\citenamefont {Bu}\ \emph {et~al.}(2022)\citenamefont {Bu},
  \citenamefont {Koh}, \citenamefont {Garcia},\ and\ \citenamefont
  {Jaffe}}]{bu_classical_2022}%
  \BibitemOpen
  \bibfield  {author} {\bibinfo {author} {\bibfnamefont {K.}~\bibnamefont
  {Bu}}, \bibinfo {author} {\bibfnamefont {D.~E.}\ \bibnamefont {Koh}},
  \bibinfo {author} {\bibfnamefont {R.~J.}\ \bibnamefont {Garcia}},\ and\
  \bibinfo {author} {\bibfnamefont {A.}~\bibnamefont {Jaffe}},\ }\href@noop {}
  {\bibinfo {title} {Classical shadows with {Pauli}-invariant unitary
  ensembles}} (\bibinfo {year} {2022}),\ \Eprint
  {https://arxiv.org/abs/2202.03272} {arXiv:2202.03272 [quant-ph]} \BibitemShut
  {NoStop}%
\bibitem [{\citenamefont {Akhtar}\ \emph {et~al.}(2022)\citenamefont {Akhtar},
  \citenamefont {Hu},\ and\ \citenamefont {You}}]{akhtar_scalable_2022}%
  \BibitemOpen
  \bibfield  {author} {\bibinfo {author} {\bibfnamefont {A.~A.}\ \bibnamefont
  {Akhtar}}, \bibinfo {author} {\bibfnamefont {H.-Y.}\ \bibnamefont {Hu}},\
  and\ \bibinfo {author} {\bibfnamefont {Y.-Z.}\ \bibnamefont {You}},\
  }\href@noop {} {\bibinfo {title} {Scalable and flexible classical shadow
  tomography with tensor networks}} (\bibinfo {year} {2022}),\ \Eprint
  {https://arxiv.org/abs/2209.02093} {arXiv:2209.02093 [quant-ph]} \BibitemShut
  {NoStop}%
\bibitem [{\citenamefont {Bertoni}\ \emph {et~al.}(2022)\citenamefont
  {Bertoni}, \citenamefont {Haferkamp}, \citenamefont {Hinsche}, \citenamefont
  {Ioannou}, \citenamefont {Eisert},\ and\ \citenamefont
  {Pashayan}}]{bertoni_shallow_2022}%
  \BibitemOpen
  \bibfield  {author} {\bibinfo {author} {\bibfnamefont {C.}~\bibnamefont
  {Bertoni}}, \bibinfo {author} {\bibfnamefont {J.}~\bibnamefont {Haferkamp}},
  \bibinfo {author} {\bibfnamefont {M.}~\bibnamefont {Hinsche}}, \bibinfo
  {author} {\bibfnamefont {M.}~\bibnamefont {Ioannou}}, \bibinfo {author}
  {\bibfnamefont {J.}~\bibnamefont {Eisert}},\ and\ \bibinfo {author}
  {\bibfnamefont {H.}~\bibnamefont {Pashayan}},\ }\href@noop {} {\bibinfo
  {title} {Shallow shadows: Expectation estimation using low-depth random
  {Clifford} circuits}} (\bibinfo {year} {2022}),\ \Eprint
  {https://arxiv.org/abs/2209.12924} {arXiv:2209.12924 [quant-ph]} \BibitemShut
  {NoStop}%
\bibitem [{\citenamefont {Arienzo}\ \emph {et~al.}(2022)\citenamefont
  {Arienzo}, \citenamefont {Heinrich}, \citenamefont {Roth},\ and\
  \citenamefont {Kliesch}}]{arienzo_closed-form_2022}%
  \BibitemOpen
  \bibfield  {author} {\bibinfo {author} {\bibfnamefont {M.}~\bibnamefont
  {Arienzo}}, \bibinfo {author} {\bibfnamefont {M.}~\bibnamefont {Heinrich}},
  \bibinfo {author} {\bibfnamefont {I.}~\bibnamefont {Roth}},\ and\ \bibinfo
  {author} {\bibfnamefont {M.}~\bibnamefont {Kliesch}},\ }\href
  {https://doi.org/10.48550/arXiv.2211.09835} {\bibinfo {title} {Closed-form
  analytic expressions for shadow estimation with brickwork circuits}}
  (\bibinfo {year} {2022}),\ \Eprint {https://arxiv.org/abs/2211.09835
  [quant-ph]} {arXiv:2211.09835 [quant-ph]} \BibitemShut {NoStop}%
\bibitem [{\citenamefont {Fulton}\ and\ \citenamefont
  {Harris}(2013)}]{FouHar91}%
  \BibitemOpen
  \bibfield  {author} {\bibinfo {author} {\bibfnamefont {W.}~\bibnamefont
  {Fulton}}\ and\ \bibinfo {author} {\bibfnamefont {J.}~\bibnamefont
  {Harris}},\ }\href {https://doi.org/10.1007/978-1-4612-0979-9} {\emph
  {\bibinfo {title} {Representation theory}}},\ Vol.\ \bibinfo {volume} {129}\
  (\bibinfo  {publisher} {Springer Science \& Business Media},\ \bibinfo {year}
  {2013})\BibitemShut {NoStop}%
\bibitem [{\citenamefont {Bröcker}\ and\ \citenamefont
  {Dieck}(1985)}]{brocker_representations_1985}%
  \BibitemOpen
  \bibfield  {author} {\bibinfo {author} {\bibfnamefont {T.}~\bibnamefont
  {Bröcker}}\ and\ \bibinfo {author} {\bibfnamefont {T.~t.}\ \bibnamefont
  {Dieck}},\ }\href {https://doi.org/10.1007/978-3-662-12918-0} {\emph
  {\bibinfo {title} {Representations of Compact Lie Groups}}},\ Graduate Texts
  in Mathematics\ (\bibinfo  {publisher} {Springer-Verlag},\ \bibinfo {year}
  {1985})\BibitemShut {NoStop}%
\bibitem [{\citenamefont {Goodman}\ and\ \citenamefont
  {Wallach}(2009)}]{goodman_symmetry_2009}%
  \BibitemOpen
  \bibfield  {author} {\bibinfo {author} {\bibfnamefont {R.}~\bibnamefont
  {Goodman}}\ and\ \bibinfo {author} {\bibfnamefont {N.~R.}\ \bibnamefont
  {Wallach}},\ }\href {https://doi.org/10.1007/978-0-387-79852-3_1} {\emph
  {\bibinfo {title} {Symmetry, Representations, and Invariants}}},\ Graduate
  Texts in Mathematics\ (\bibinfo  {publisher} {Springer},\ \bibinfo {year}
  {2009})\BibitemShut {NoStop}%
\bibitem [{\citenamefont {Bump}(2004)}]{bump_lie_2004}%
  \BibitemOpen
  \bibfield  {author} {\bibinfo {author} {\bibfnamefont {D.}~\bibnamefont
  {Bump}},\ }\href {https://doi.org/10.1007/978-1-4757-4094-3} {\emph {\bibinfo
  {title} {Lie Groups}}},\ \bibinfo {edition} {1st}\ ed.,\ Graduate Texts in
  Mathematics\ (\bibinfo  {publisher} {Springer New York, {NY}},\ \bibinfo
  {year} {2004})\BibitemShut {NoStop}%
\bibitem [{\citenamefont {Folland}(2015)}]{folland_course_2015}%
  \BibitemOpen
  \bibfield  {author} {\bibinfo {author} {\bibfnamefont {G.~B.}\ \bibnamefont
  {Folland}},\ }\href {https://doi.org/10.1201/b19172} {\emph {\bibinfo {title}
  {A Course in Abstract Harmonic Analysis}}},\ \bibinfo {edition} {2nd}\ ed.\
  (\bibinfo  {publisher} {Chapman and Hall/{CRC}},\ \bibinfo {year}
  {2015})\BibitemShut {NoStop}%
\bibitem [{\citenamefont {Gowers}\ and\ \citenamefont
  {Hatami}(2017)}]{gowers_inverse_2017}%
  \BibitemOpen
  \bibfield  {author} {\bibinfo {author} {\bibfnamefont {W.~T.}\ \bibnamefont
  {Gowers}}\ and\ \bibinfo {author} {\bibfnamefont {O.}~\bibnamefont
  {Hatami}},\ }\bibfield  {title} {\bibinfo {title} {Inverse and stability
  theorems for approximate representations of finite groups},\ }\href
  {https://doi.org/10.1070/SM8872} {\bibfield  {journal} {\bibinfo  {journal}
  {Sb. Math.}\ }\textbf {\bibinfo {volume} {208}},\ \bibinfo {pages} {1784}
  (\bibinfo {year} {2017})},\ \Eprint {https://arxiv.org/abs/1510.04085}
  {arXiv:1510.04085 [math.GR]} \BibitemShut {NoStop}%
\bibitem [{\citenamefont {Bannai}\ \emph
  {et~al.}(2020{\natexlab{a}})\citenamefont {Bannai}, \citenamefont {Nakata},
  \citenamefont {Okuda},\ and\ \citenamefont {Zhao}}]{bannai_explicit_2020}%
  \BibitemOpen
  \bibfield  {author} {\bibinfo {author} {\bibfnamefont {E.}~\bibnamefont
  {Bannai}}, \bibinfo {author} {\bibfnamefont {Y.}~\bibnamefont {Nakata}},
  \bibinfo {author} {\bibfnamefont {T.}~\bibnamefont {Okuda}},\ and\ \bibinfo
  {author} {\bibfnamefont {D.}~\bibnamefont {Zhao}},\ }\href
  {http://arxiv.org/abs/2009.11170} {\bibinfo {title} {Explicit construction of
  exact unitary designs}} (\bibinfo {year} {2020}{\natexlab{a}}),\ \Eprint
  {https://arxiv.org/abs/2009.11170} {arXiv:2009.11170 [math.CO]} \BibitemShut
  {NoStop}%
\bibitem [{\citenamefont {Stembridge}(1987)}]{stembridge_rational_1987}%
  \BibitemOpen
  \bibfield  {author} {\bibinfo {author} {\bibfnamefont {J.~R.}\ \bibnamefont
  {Stembridge}},\ }\bibfield  {title} {\bibinfo {title} {Rational tableaux and
  the tensor algebra of $\mathfrak{gl}_n$},\ }\href
  {https://doi.org/10.1016/0097-3165(87)90077-X} {\bibfield  {journal}
  {\bibinfo  {journal} {Journal of Combinatorial Theory, Series A}\ }\textbf
  {\bibinfo {volume} {46}},\ \bibinfo {pages} {79} (\bibinfo {year}
  {1987})}\BibitemShut {NoStop}%
\bibitem [{\citenamefont {Stembridge}(1989)}]{stembridge_combinatorial_1989}%
  \BibitemOpen
  \bibfield  {author} {\bibinfo {author} {\bibfnamefont {J.~R.}\ \bibnamefont
  {Stembridge}},\ }\bibfield  {title} {\bibinfo {title} {A combinatorial theory
  for rational actions of $\mathrm{GL}_n$},\ }in\ \href
  {https://doi.org/10.1090/conm/088/999989} {\emph {\bibinfo {booktitle}
  {Invariant theory ({D}enton, {TX}, 1986)}}},\ \bibinfo {series} {Contemp.
  Math.}, Vol.~\bibinfo {volume} {88}\ (\bibinfo  {publisher} {Amer. Math.
  Soc., Providence, RI},\ \bibinfo {year} {1989})\ pp.\ \bibinfo {pages}
  {163--176}\BibitemShut {NoStop}%
\bibitem [{\citenamefont {Benkart}\ \emph {et~al.}(1994)\citenamefont
  {Benkart}, \citenamefont {Chakrabarti}, \citenamefont {Halverson},
  \citenamefont {Leduc}, \citenamefont {Lee},\ and\ \citenamefont
  {Stroomer}}]{benkart_tensor_1994}%
  \BibitemOpen
  \bibfield  {author} {\bibinfo {author} {\bibfnamefont {G.}~\bibnamefont
  {Benkart}}, \bibinfo {author} {\bibfnamefont {M.}~\bibnamefont
  {Chakrabarti}}, \bibinfo {author} {\bibfnamefont {T.}~\bibnamefont
  {Halverson}}, \bibinfo {author} {\bibfnamefont {R.}~\bibnamefont {Leduc}},
  \bibinfo {author} {\bibfnamefont {C.~Y.}\ \bibnamefont {Lee}},\ and\ \bibinfo
  {author} {\bibfnamefont {J.}~\bibnamefont {Stroomer}},\ }\bibfield  {title}
  {\bibinfo {title} {Tensor product representations of general linear groups
  and their connections with brauer algebras},\ }\href
  {https://doi.org/10.1006/jabr.1994.1166} {\bibfield  {journal} {\bibinfo
  {journal} {Journal of Algebra}\ }\textbf {\bibinfo {volume} {166}},\ \bibinfo
  {pages} {529} (\bibinfo {year} {1994})}\BibitemShut {NoStop}%
\bibitem [{\citenamefont {Roy}\ and\ \citenamefont
  {Scott}(2009)}]{roy_unitary_2009}%
  \BibitemOpen
  \bibfield  {author} {\bibinfo {author} {\bibfnamefont {A.}~\bibnamefont
  {Roy}}\ and\ \bibinfo {author} {\bibfnamefont {A.~J.}\ \bibnamefont
  {Scott}},\ }\bibfield  {title} {\bibinfo {title} {Unitary designs and
  codes},\ }\href {https://doi.org/10.1007/s10623-009-9290-2} {\bibfield
  {journal} {\bibinfo  {journal} {Des. Codes Cryptogr.}\ }\textbf {\bibinfo
  {volume} {53}},\ \bibinfo {pages} {13} (\bibinfo {year} {2009})},\ \Eprint
  {https://arxiv.org/abs/0809.3813} {arXiv:0809.3813 [math.CO]} \BibitemShut
  {NoStop}%
\bibitem [{\citenamefont {Diaconis}\ and\ \citenamefont
  {Saloff-Coste}(1993)}]{diaconis_comparison_1993}%
  \BibitemOpen
  \bibfield  {author} {\bibinfo {author} {\bibfnamefont {P.}~\bibnamefont
  {Diaconis}}\ and\ \bibinfo {author} {\bibfnamefont {L.}~\bibnamefont
  {Saloff-Coste}},\ }\bibfield  {title} {\bibinfo {title} {Comparison
  techniques for random walk on finite groups},\ }\href
  {https://doi.org/10.1214/aop/1176989013} {\bibfield  {journal} {\bibinfo
  {journal} {The Annals of Probability}\ }\textbf {\bibinfo {volume} {21}},\
  \bibinfo {pages} {2131} (\bibinfo {year} {1993})},\ \bibinfo {note}
  {publisher: Institute of Mathematical Statistics}\BibitemShut {NoStop}%
\bibitem [{\citenamefont {Varjú}(2012)}]{varju_random_2015}%
  \BibitemOpen
  \bibfield  {author} {\bibinfo {author} {\bibfnamefont {P.~P.}\ \bibnamefont
  {Varjú}},\ }\bibfield  {title} {\bibinfo {title} {Random walks in compact
  groups},\ }\href@noop {} {\bibfield  {journal} {\bibinfo  {journal} {Doc.
  Math.}\ }\textbf {\bibinfo {volume} {18}},\ \bibinfo {pages} {1137} (\bibinfo
  {year} {2012})},\ \Eprint {https://arxiv.org/abs/1209.1745} {arXiv:1209.1745
  [math.GR]} \BibitemShut {NoStop}%
\bibitem [{\citenamefont {Heinrich}(2021)}]{heinrich_2021}%
  \BibitemOpen
  \bibfield  {author} {\bibinfo {author} {\bibfnamefont {M.}~\bibnamefont
  {Heinrich}},\ }\emph {\bibinfo {title} {On stabiliser techniques and their
  application to simulation and certification of quantum devices}},\ \href
  {https://kups.ub.uni-koeln.de/50465/} {Ph.D. thesis},\ \bibinfo  {school}
  {University of Cologne} (\bibinfo {year} {2021})\BibitemShut {NoStop}%
\bibitem [{\citenamefont {Neuhauser}(2002)}]{neuhauser_explicit_2002}%
  \BibitemOpen
  \bibfield  {author} {\bibinfo {author} {\bibfnamefont {M.}~\bibnamefont
  {Neuhauser}},\ }\bibfield  {title} {\bibinfo {title} {An explicit
  construction of the metaplectic representation over a finite field},\ }\href
  {http://eudml.org/doc/123051} {\bibfield  {journal} {\bibinfo  {journal}
  {Journal of Lie Theory}\ }\textbf {\bibinfo {volume} {12}},\ \bibinfo {pages}
  {15} (\bibinfo {year} {2002})}\BibitemShut {NoStop}%
\bibitem [{\citenamefont {Nebe}\ \emph {et~al.}(2001)\citenamefont {Nebe},
  \citenamefont {Rains},\ and\ \citenamefont {Sloane}}]{nebe_invariants_2001}%
  \BibitemOpen
  \bibfield  {author} {\bibinfo {author} {\bibfnamefont {G.}~\bibnamefont
  {Nebe}}, \bibinfo {author} {\bibfnamefont {E.~M.}\ \bibnamefont {Rains}},\
  and\ \bibinfo {author} {\bibfnamefont {N.~J.~A.}\ \bibnamefont {Sloane}},\
  }\bibfield  {title} {\bibinfo {title} {The invariants of the {Clifford}
  groups},\ }\href@noop {} {\bibfield  {journal} {\bibinfo  {journal}
  {Des.~Codes Cryptogr.}\ }\textbf {\bibinfo {volume} {24}} (\bibinfo {year}
  {2001})},\ \Eprint {https://arxiv.org/abs/math/0001038} {arXiv:math/0001038
  [math.CO]} \BibitemShut {NoStop}%
\bibitem [{\citenamefont {{Zhu}}\ \emph {et~al.}(2016)\citenamefont {{Zhu}},
  \citenamefont {{Kueng}}, \citenamefont {{Grassl}},\ and\ \citenamefont
  {{Gross}}}]{ZhuKueGra16}%
  \BibitemOpen
  \bibfield  {author} {\bibinfo {author} {\bibfnamefont {H.}~\bibnamefont
  {{Zhu}}}, \bibinfo {author} {\bibfnamefont {R.}~\bibnamefont {{Kueng}}},
  \bibinfo {author} {\bibfnamefont {M.}~\bibnamefont {{Grassl}}},\ and\
  \bibinfo {author} {\bibfnamefont {D.}~\bibnamefont {{Gross}}},\ }\href@noop
  {} {\bibinfo {title} {{The {Clifford} group fails gracefully to be a unitary
  4-design}}} (\bibinfo {year} {2016}),\ \Eprint
  {https://arxiv.org/abs/1609.08172} {arXiv:1609.08172 [quant-ph]} \BibitemShut
  {NoStop}%
\bibitem [{\citenamefont {DiVincenzo}\ \emph {et~al.}(2002)\citenamefont
  {DiVincenzo}, \citenamefont {Leung},\ and\ \citenamefont
  {Terhal}}]{DiVincenzo985948}%
  \BibitemOpen
  \bibfield  {author} {\bibinfo {author} {\bibfnamefont {D.~P.}\ \bibnamefont
  {DiVincenzo}}, \bibinfo {author} {\bibfnamefont {D.~W.}\ \bibnamefont
  {Leung}},\ and\ \bibinfo {author} {\bibfnamefont {B.~M.}\ \bibnamefont
  {Terhal}},\ }\bibfield  {title} {\bibinfo {title} {Quantum data hiding},\
  }\href {https://doi.org/10.1109/18.985948} {\bibfield  {journal} {\bibinfo
  {journal} {IEEE Trans. Inf. Th.}\ }\textbf {\bibinfo {volume} {48}},\
  \bibinfo {pages} {580} (\bibinfo {year} {2002})},\ \Eprint
  {https://arxiv.org/abs/quant-ph/0103098} {arXiv:quant-ph/0103098 [quant-ph]}
  \BibitemShut {NoStop}%
\bibitem [{\citenamefont {Webb}(2016)}]{Web15}%
  \BibitemOpen
  \bibfield  {author} {\bibinfo {author} {\bibfnamefont {Z.}~\bibnamefont
  {Webb}},\ }\bibfield  {title} {\bibinfo {title} {The {Clifford} group forms a
  unitary 3-design},\ }\href
  {http://dl.acm.org/citation.cfm?id=3179439.3179447} {\bibfield  {journal}
  {\bibinfo  {journal} {Quantum Info. Comput.}\ }\textbf {\bibinfo {volume}
  {16}},\ \bibinfo {pages} {1379} (\bibinfo {year} {2016})},\ \Eprint
  {https://arxiv.org/abs/1510.02769} {arXiv:1510.02769 [quant-ph]} \BibitemShut
  {NoStop}%
\bibitem [{\citenamefont {Zhu}(2017)}]{Zhu15}%
  \BibitemOpen
  \bibfield  {author} {\bibinfo {author} {\bibfnamefont {H.}~\bibnamefont
  {Zhu}},\ }\bibfield  {title} {\bibinfo {title} {Multiqubit {Clifford} groups
  are unitary 3-designs},\ }\href {https://doi.org/10.1103/PhysRevA.96.062336}
  {\bibfield  {journal} {\bibinfo  {journal} {\pra}\ }\textbf {\bibinfo
  {volume} {96}},\ \bibinfo {pages} {062336} (\bibinfo {year} {2017})},\
  \Eprint {https://arxiv.org/abs/1510.02619} {arXiv:1510.02619 [quant-ph]}
  \BibitemShut {NoStop}%
\bibitem [{\citenamefont {Helsen}\ \emph {et~al.}(2018)\citenamefont {Helsen},
  \citenamefont {Wallman},\ and\ \citenamefont {Wehner}}]{HelWalWeh16}%
  \BibitemOpen
  \bibfield  {author} {\bibinfo {author} {\bibfnamefont {J.}~\bibnamefont
  {Helsen}}, \bibinfo {author} {\bibfnamefont {J.~J.}\ \bibnamefont
  {Wallman}},\ and\ \bibinfo {author} {\bibfnamefont {S.}~\bibnamefont
  {Wehner}},\ }\bibfield  {title} {\bibinfo {title} {Representations of the
  multi-qubit {Clifford} group},\ }\href {https://doi.org/10.1063/1.4997688}
  {\bibfield  {journal} {\bibinfo  {journal} {J. Math. Phys.}\ }\textbf
  {\bibinfo {volume} {59}},\ \bibinfo {pages} {072201} (\bibinfo {year}
  {2018})},\ \Eprint {https://arxiv.org/abs/1609.08188} {arXiv:1609.08188
  [quant-ph]} \BibitemShut {NoStop}%
\bibitem [{\citenamefont {Guralnick}\ and\ \citenamefont
  {Tiep}(2005)}]{guralnick_decompositions_2005}%
  \BibitemOpen
  \bibfield  {author} {\bibinfo {author} {\bibfnamefont {R.}~\bibnamefont
  {Guralnick}}\ and\ \bibinfo {author} {\bibfnamefont {P.}~\bibnamefont
  {Tiep}},\ }\bibfield  {title} {\bibinfo {title} {Decompositions of small
  tensor powers and {Larsen}’s conjecture},\ }\href
  {https://doi.org/10.1090/S1088-4165-05-00192-5} {\bibfield  {journal}
  {\bibinfo  {journal} {Represent. Theory}\ }\textbf {\bibinfo {volume} {9}},\
  \bibinfo {pages} {138} (\bibinfo {year} {2005})},\ \Eprint
  {https://arxiv.org/abs/math/0502080} {arXiv:math/0502080 [math.GR]}
  \BibitemShut {NoStop}%
\bibitem [{\citenamefont {Bannai}\ \emph
  {et~al.}(2020{\natexlab{b}})\citenamefont {Bannai}, \citenamefont {Navarro},
  \citenamefont {Rizo},\ and\ \citenamefont {Tiep}}]{BannaiEtAl:2020:tgroups}%
  \BibitemOpen
  \bibfield  {author} {\bibinfo {author} {\bibfnamefont {E.}~\bibnamefont
  {Bannai}}, \bibinfo {author} {\bibfnamefont {G.}~\bibnamefont {Navarro}},
  \bibinfo {author} {\bibfnamefont {N.}~\bibnamefont {Rizo}},\ and\ \bibinfo
  {author} {\bibfnamefont {P.~H.}\ \bibnamefont {Tiep}},\ }\bibfield  {title}
  {\bibinfo {title} {Unitary $t$-groups},\ }\href
  {https://doi.org/10.2969/jmsj/82228222} {\bibfield  {journal} {\bibinfo
  {journal} {J. Math. Soc. Japan}\ }\textbf {\bibinfo {volume} {72}},\ \bibinfo
  {pages} {909} (\bibinfo {year} {2020}{\natexlab{b}})},\ \Eprint
  {https://arxiv.org/abs/1810.02507} {arXiv:1810.02507 [math.RT]} \BibitemShut
  {NoStop}%
\bibitem [{\citenamefont {Montealegre-Mora}\ and\ \citenamefont
  {Gross}(2021)}]{montealegre-mora_rank-deficient_2021}%
  \BibitemOpen
  \bibfield  {author} {\bibinfo {author} {\bibfnamefont {F.}~\bibnamefont
  {Montealegre-Mora}}\ and\ \bibinfo {author} {\bibfnamefont {D.}~\bibnamefont
  {Gross}},\ }\bibfield  {title} {\bibinfo {title} {Rank-deficient
  representations in the theta correspondence over finite fields arise from
  quantum codes},\ }\href {https://doi.org/10.1090/ert/563} {\bibfield
  {journal} {\bibinfo  {journal} {Represent. Theory}\ }\textbf {\bibinfo
  {volume} {25}},\ \bibinfo {pages} {193} (\bibinfo {year} {2021})},\ \Eprint
  {https://arxiv.org/abs/1906.07230} {arXiv:1906.07230 [math.RT]} \BibitemShut
  {NoStop}%
\bibitem [{\citenamefont {Montealegre-Mora}\ and\ \citenamefont
  {Gross}(2022)}]{montealegre-mora_duality_2022}%
  \BibitemOpen
  \bibfield  {author} {\bibinfo {author} {\bibfnamefont {F.}~\bibnamefont
  {Montealegre-Mora}}\ and\ \bibinfo {author} {\bibfnamefont {D.}~\bibnamefont
  {Gross}},\ }\href {https://doi.org/10.48550/arXiv.2208.01688} {\bibinfo
  {title} {Duality theory for {Clifford} tensor powers}} (\bibinfo {year}
  {2022}),\ \Eprint {https://arxiv.org/abs/2208.01688} {arXiv:2208.01688}
  \BibitemShut {NoStop}%
\bibitem [{\citenamefont {Waldron}(2018)}]{waldron_introduction_2018}%
  \BibitemOpen
  \bibfield  {author} {\bibinfo {author} {\bibfnamefont {S.~F.~D.}\
  \bibnamefont {Waldron}},\ }\href {https://doi.org/10.1007/978-0-8176-4815-2}
  {\emph {\bibinfo {title} {An Introduction to Finite Tight Frames}}},\
  \bibinfo {edition} {1st}\ ed.,\ Applied and Numerical Harmonic Analysis\
  (\bibinfo  {publisher} {Birkhäuser New York, {NY}},\ \bibinfo {year}
  {2018})\BibitemShut {NoStop}%
\bibitem [{\citenamefont {Stewart}\ and\ \citenamefont
  {Sun}(1990)}]{stewart_matrix_1990}%
  \BibitemOpen
  \bibfield  {author} {\bibinfo {author} {\bibfnamefont {G.~W.}\ \bibnamefont
  {Stewart}}\ and\ \bibinfo {author} {\bibfnamefont {J.-g.}\ \bibnamefont
  {Sun}},\ }\href
  {https://shop.elsevier.com/books/matrix-perturbation-theory/stewart/978-0-08-092613-1}
  {\emph {\bibinfo {title} {Matrix Perturbation Theory}}}\ (\bibinfo
  {publisher} {Elsevier Science},\ \bibinfo {year} {1990})\BibitemShut
  {NoStop}%
\bibitem [{\citenamefont {Aharonov}\ \emph {et~al.}(2022)\citenamefont
  {Aharonov}, \citenamefont {Gao}, \citenamefont {Landau}, \citenamefont
  {Liu},\ and\ \citenamefont {Vazirani}}]{aharonov_polynomial-time_2022}%
  \BibitemOpen
  \bibfield  {author} {\bibinfo {author} {\bibfnamefont {D.}~\bibnamefont
  {Aharonov}}, \bibinfo {author} {\bibfnamefont {X.}~\bibnamefont {Gao}},
  \bibinfo {author} {\bibfnamefont {Z.}~\bibnamefont {Landau}}, \bibinfo
  {author} {\bibfnamefont {Y.}~\bibnamefont {Liu}},\ and\ \bibinfo {author}
  {\bibfnamefont {U.}~\bibnamefont {Vazirani}},\ }\href
  {https://doi.org/10.48550/arXiv.2211.03999} {\bibinfo {title} {A
  polynomial-time classical algorithm for noisy random circuit sampling}}
  (\bibinfo {year} {2022}),\ \Eprint {https://arxiv.org/abs/2211.03999
  [quant-ph]} {arXiv:2211.03999 [quant-ph]} \BibitemShut {NoStop}%
\bibitem [{\citenamefont {{Harper}}\ \emph {et~al.}(2019)\citenamefont
  {{Harper}}, \citenamefont {{Hincks}}, \citenamefont {{Ferrie}}, \citenamefont
  {{Flammia}},\ and\ \citenamefont {{Wallman}}}]{HarHinFer19}%
  \BibitemOpen
  \bibfield  {author} {\bibinfo {author} {\bibfnamefont {R.}~\bibnamefont
  {{Harper}}}, \bibinfo {author} {\bibfnamefont {I.}~\bibnamefont {{Hincks}}},
  \bibinfo {author} {\bibfnamefont {C.}~\bibnamefont {{Ferrie}}}, \bibinfo
  {author} {\bibfnamefont {S.~T.}\ \bibnamefont {{Flammia}}},\ and\ \bibinfo
  {author} {\bibfnamefont {J.~J.}\ \bibnamefont {{Wallman}}},\ }\bibfield
  {title} {\bibinfo {title} {Statistical analysis of randomized benchmarking},\
  }\href {https://doi.org/10.1103/PhysRevA.99.052350} {\bibfield  {journal}
  {\bibinfo  {journal} {Phys. Rev. A}\ }\textbf {\bibinfo {volume} {99}},\
  \bibinfo {pages} {052350} (\bibinfo {year} {2019})},\ \Eprint
  {https://arxiv.org/abs/1901.00535} {arXiv:1901.00535 [quant-ph]} \BibitemShut
  {NoStop}%
\bibitem [{\citenamefont {Watrous}(2018)}]{Wat18}%
  \BibitemOpen
  \bibfield  {author} {\bibinfo {author} {\bibfnamefont {J.}~\bibnamefont
  {Watrous}},\ }\href {https://doi.org/10.1017/9781316848142} {\emph {\bibinfo
  {title} {The Theory of Quantum Information}}}\ (\bibinfo  {publisher}
  {Cambridge University Press},\ \bibinfo {year} {2018})\BibitemShut {NoStop}%
\bibitem [{\citenamefont {Wallman}\ \emph {et~al.}(2016)\citenamefont
  {Wallman}, \citenamefont {Barnhill},\ and\ \citenamefont
  {Emerson}}]{WalBarEme16b}%
  \BibitemOpen
  \bibfield  {author} {\bibinfo {author} {\bibfnamefont {J.~J.}\ \bibnamefont
  {Wallman}}, \bibinfo {author} {\bibfnamefont {M.}~\bibnamefont {Barnhill}},\
  and\ \bibinfo {author} {\bibfnamefont {J.}~\bibnamefont {Emerson}},\
  }\bibfield  {title} {\bibinfo {title} {Robust characterization of leakage
  errors},\ }\href {https://doi.org/10.1088/1367-2630/18/4/043021} {\bibfield
  {journal} {\bibinfo  {journal} {New J. Phys.}\ }\textbf {\bibinfo {volume}
  {18}},\ \bibinfo {pages} {043021} (\bibinfo {year} {2016})},\ \Eprint
  {https://arxiv.org/abs/1412.4126} {arXiv:1412.4126 [quant-ph]} \BibitemShut
  {NoStop}%
\bibitem [{\citenamefont {Helsen}\ and\ \citenamefont
  {Walter}(2022)}]{helsen_thrifty_2022}%
  \BibitemOpen
  \bibfield  {author} {\bibinfo {author} {\bibfnamefont {J.}~\bibnamefont
  {Helsen}}\ and\ \bibinfo {author} {\bibfnamefont {M.}~\bibnamefont
  {Walter}},\ }\href {https://doi.org/10.48550/arXiv.2212.06240} {\bibinfo
  {title} {Thrifty shadow estimation: re-using quantum circuits and bounding
  tails}} (\bibinfo {year} {2022}),\ \Eprint {https://arxiv.org/abs/2212.06240
  [quant-ph]} {arXiv:2212.06240 [quant-ph]} \BibitemShut {NoStop}%
\bibitem [{\citenamefont {Roy}\ \emph {et~al.}(1986)\citenamefont {Roy},
  \citenamefont {Paulraj},\ and\ \citenamefont
  {Kailath}}]{roy_estimation_1986}%
  \BibitemOpen
  \bibfield  {author} {\bibinfo {author} {\bibfnamefont {R.}~\bibnamefont
  {Roy}}, \bibinfo {author} {\bibfnamefont {A.}~\bibnamefont {Paulraj}},\ and\
  \bibinfo {author} {\bibfnamefont {T.}~\bibnamefont {Kailath}},\ }\bibfield
  {title} {\bibinfo {title} {Estimation of signal parameters via rotational
  invariance techniques - {ESPRIT}},\ }in\ \href
  {https://doi.org/10.1109/MILCOM.1986.4805850} {\emph {\bibinfo {booktitle}
  {{MILCOM} 1986 - {IEEE} Military Communications Conference:
  Communications-Computers: Teamed for the 90's}}},\ Vol.~\bibinfo {volume}
  {3}\ (\bibinfo {year} {1986})\ pp.\ \bibinfo {pages}
  {41.6.1--41.6.5}\BibitemShut {NoStop}%
\bibitem [{\citenamefont {Znidaric}(2008)}]{znidaric_exact_2008}%
  \BibitemOpen
  \bibfield  {author} {\bibinfo {author} {\bibfnamefont {M.}~\bibnamefont
  {Znidaric}},\ }\bibfield  {title} {\bibinfo {title} {Exact convergence times
  for generation of random bipartite entanglement},\ }\href
  {https://doi.org/10.1103/PhysRevA.78.032324} {\bibfield  {journal} {\bibinfo
  {journal} {Phys. Rev. A}\ }\textbf {\bibinfo {volume} {78}},\ \bibinfo
  {pages} {032324} (\bibinfo {year} {2008})},\ \Eprint
  {https://arxiv.org/abs/0809.0554} {arXiv:0809.0554} \BibitemShut {NoStop}%
\bibitem [{\citenamefont {Ćwikliński}\ \emph {et~al.}(2013)\citenamefont
  {Ćwikliński}, \citenamefont {Horodecki}, \citenamefont {Mozrzymas},
  \citenamefont {Pankowski},\ and\ \citenamefont
  {Studziński}}]{cwiklinski_local_2013}%
  \BibitemOpen
  \bibfield  {author} {\bibinfo {author} {\bibfnamefont {P.}~\bibnamefont
  {Ćwikliński}}, \bibinfo {author} {\bibfnamefont {M.}~\bibnamefont
  {Horodecki}}, \bibinfo {author} {\bibfnamefont {M.}~\bibnamefont
  {Mozrzymas}}, \bibinfo {author} {\bibfnamefont {L.}~\bibnamefont
  {Pankowski}},\ and\ \bibinfo {author} {\bibfnamefont {M.}~\bibnamefont
  {Studziński}},\ }\bibfield  {title} {\bibinfo {title} {Local random quantum
  circuits are approximate polynomial-designs - numerical results},\ }\href
  {https://doi.org/10.1088/1751-8113/46/30/305301} {\bibfield  {journal}
  {\bibinfo  {journal} {J. Phys. A: Math. Theor.}\ }\textbf {\bibinfo {volume}
  {46}},\ \bibinfo {pages} {305301} (\bibinfo {year} {2013})},\ \Eprint
  {https://arxiv.org/abs/1212.2556 [quant-ph]} {arXiv:1212.2556 [quant-ph]}
  \BibitemShut {NoStop}%
\bibitem [{\citenamefont {Aharonov}\ \emph {et~al.}(2009)\citenamefont
  {Aharonov}, \citenamefont {Arad}, \citenamefont {Landau},\ and\ \citenamefont
  {Vazirani}}]{aharonov_detectability_2008}%
  \BibitemOpen
  \bibfield  {author} {\bibinfo {author} {\bibfnamefont {D.}~\bibnamefont
  {Aharonov}}, \bibinfo {author} {\bibfnamefont {I.}~\bibnamefont {Arad}},
  \bibinfo {author} {\bibfnamefont {Z.}~\bibnamefont {Landau}},\ and\ \bibinfo
  {author} {\bibfnamefont {U.}~\bibnamefont {Vazirani}},\ }\bibfield  {title}
  {\bibinfo {title} {The detectability lemma and quantum gap amplification},\
  }in\ \href@noop {} {\emph {\bibinfo {booktitle} {Proceedings of the
  Forty-First Annual ACM Symposium on Theory of Computing}}}\ (\bibinfo
  {publisher} {STOC `09},\ \bibinfo {year} {2009})\ p.\ \bibinfo {pages}
  {417},\ \Eprint {https://arxiv.org/abs/0811.3412} {arXiv:0811.3412
  [quant-ph]} \BibitemShut {NoStop}%
\bibitem [{\citenamefont {Anshu}\ \emph {et~al.}(2016)\citenamefont {Anshu},
  \citenamefont {Arad},\ and\ \citenamefont {Vidick}}]{anshu_simple_2016}%
  \BibitemOpen
  \bibfield  {author} {\bibinfo {author} {\bibfnamefont {A.}~\bibnamefont
  {Anshu}}, \bibinfo {author} {\bibfnamefont {I.}~\bibnamefont {Arad}},\ and\
  \bibinfo {author} {\bibfnamefont {T.}~\bibnamefont {Vidick}},\ }\bibfield
  {title} {\bibinfo {title} {A simple proof of the detectability lemma and
  spectral gap amplification},\ }\href
  {https://doi.org/10.1103/PhysRevB.93.205142} {\bibfield  {journal} {\bibinfo
  {journal} {Phys. Rev. B}\ }\textbf {\bibinfo {volume} {93}},\ \bibinfo
  {pages} {205142} (\bibinfo {year} {2016})},\ \Eprint
  {https://arxiv.org/abs/1602.01210} {arXiv:1602.01210} \BibitemShut {NoStop}%
\bibitem [{\citenamefont {Haferkamp}(2022{\natexlab{a}})}]{jonas}%
  \BibitemOpen
  \bibfield  {author} {\bibinfo {author} {\bibfnamefont {J.}~\bibnamefont
  {Haferkamp}},\ }\href@noop {} {}\bibinfo {howpublished} {{Private
  Communication}} (\bibinfo {year} {2022}{\natexlab{a}})\BibitemShut {NoStop}%
\bibitem [{\citenamefont {{Onorati}}\ \emph {et~al.}(2017)\citenamefont
  {{Onorati}}, \citenamefont {{Buerschaper}}, \citenamefont {{Kliesch}},
  \citenamefont {{Brown}}, \citenamefont {{Werner}},\ and\ \citenamefont
  {{Eisert}}}]{OnoBueKli17}%
  \BibitemOpen
  \bibfield  {author} {\bibinfo {author} {\bibfnamefont {E.}~\bibnamefont
  {{Onorati}}}, \bibinfo {author} {\bibfnamefont {O.}~\bibnamefont
  {{Buerschaper}}}, \bibinfo {author} {\bibfnamefont {M.}~\bibnamefont
  {{Kliesch}}}, \bibinfo {author} {\bibfnamefont {W.}~\bibnamefont {{Brown}}},
  \bibinfo {author} {\bibfnamefont {A.~H.}\ \bibnamefont {{Werner}}},\ and\
  \bibinfo {author} {\bibfnamefont {J.}~\bibnamefont {{Eisert}}},\ }\bibfield
  {title} {\bibinfo {title} {Mixing properties of stochastic quantum
  {Hamiltonians}},\ }\href {https://doi.org/10.1007/s00220-017-2950-6}
  {\bibfield  {journal} {\bibinfo  {journal} {Comm. Math. Phys.}\ }\textbf
  {\bibinfo {volume} {355}},\ \bibinfo {pages} {905} (\bibinfo {year}
  {2017})},\ \Eprint {https://arxiv.org/abs/1606.01914} {arXiv:1606.01914
  [quant-ph]} \BibitemShut {NoStop}%
\bibitem [{\citenamefont {Bravyi}\ \emph {et~al.}(2022)\citenamefont {Bravyi},
  \citenamefont {Latone},\ and\ \citenamefont {Maslov}}]{bravyi_6-qubit_2022}%
  \BibitemOpen
  \bibfield  {author} {\bibinfo {author} {\bibfnamefont {S.}~\bibnamefont
  {Bravyi}}, \bibinfo {author} {\bibfnamefont {J.~A.}\ \bibnamefont {Latone}},\
  and\ \bibinfo {author} {\bibfnamefont {D.}~\bibnamefont {Maslov}},\
  }\bibfield  {title} {\bibinfo {title} {6-qubit {Optimal} {Clifford}
  {Circuits}},\ }\href {https://doi.org/10.1038/s41534-022-00583-7} {\bibfield
  {journal} {\bibinfo  {journal} {npj Quantum Information}\ }\textbf {\bibinfo
  {volume} {8}},\ \bibinfo {pages} {79} (\bibinfo {year} {2022})},\ \Eprint
  {https://arxiv.org/abs/2012.06074} {arXiv:2012.06074 [quant-ph]} \BibitemShut
  {NoStop}%
\bibitem [{\citenamefont {Bravyi}\ and\ \citenamefont
  {Maslov}(2021)}]{bravyi_hadamard-free_2021}%
  \BibitemOpen
  \bibfield  {author} {\bibinfo {author} {\bibfnamefont {S.}~\bibnamefont
  {Bravyi}}\ and\ \bibinfo {author} {\bibfnamefont {D.}~\bibnamefont
  {Maslov}},\ }\bibfield  {title} {\bibinfo {title} {Hadamard-free circuits
  expose the structure of the {Clifford} group},\ }\href
  {https://doi.org/10.1109/TIT.2021.3081415} {\bibfield  {journal} {\bibinfo
  {journal} {IEEE Transactions on Information Theory}\ }\textbf {\bibinfo
  {volume} {67}},\ \bibinfo {pages} {4546} (\bibinfo {year} {2021})},\ \Eprint
  {https://arxiv.org/abs/2003.09412} {arXiv:2003.09412 [quant-ph]} \BibitemShut
  {NoStop}%
\bibitem [{\citenamefont {{Boixo}}\ \emph {et~al.}(2018)\citenamefont
  {{Boixo}}, \citenamefont {{Isakov}}, \citenamefont {{Smelyanskiy}},
  \citenamefont {{Babbush}}, \citenamefont {{Ding}}, \citenamefont {{Jiang}},
  \citenamefont {{Bremner}}, \citenamefont {{Martinis}},\ and\ \citenamefont
  {{Neven}}}]{Boixo2018CharacterizingQuantum}%
  \BibitemOpen
  \bibfield  {author} {\bibinfo {author} {\bibfnamefont {S.}~\bibnamefont
  {{Boixo}}}, \bibinfo {author} {\bibfnamefont {S.~V.}\ \bibnamefont
  {{Isakov}}}, \bibinfo {author} {\bibfnamefont {V.~N.}\ \bibnamefont
  {{Smelyanskiy}}}, \bibinfo {author} {\bibfnamefont {R.}~\bibnamefont
  {{Babbush}}}, \bibinfo {author} {\bibfnamefont {N.}~\bibnamefont {{Ding}}},
  \bibinfo {author} {\bibfnamefont {Z.}~\bibnamefont {{Jiang}}}, \bibinfo
  {author} {\bibfnamefont {M.~J.}\ \bibnamefont {{Bremner}}}, \bibinfo {author}
  {\bibfnamefont {J.~M.}\ \bibnamefont {{Martinis}}},\ and\ \bibinfo {author}
  {\bibfnamefont {H.}~\bibnamefont {{Neven}}},\ }\bibfield  {title} {\bibinfo
  {title} {Characterizing quantum supremacy in near-term devices},\ }\href
  {https://doi.org/10.1038/s41567-018-0124-x} {\bibfield  {journal} {\bibinfo
  {journal} {Nature Physics}\ }\textbf {\bibinfo {volume} {14}},\ \bibinfo
  {pages} {595} (\bibinfo {year} {2018})},\ \Eprint
  {https://arxiv.org/abs/1608.00263} {arXiv:1608.00263 [quant-ph]} \BibitemShut
  {NoStop}%
\bibitem [{\citenamefont {Brand\~ao}\ \emph {et~al.}(2016)\citenamefont
  {Brand\~ao}, \citenamefont {Harrow},\ and\ \citenamefont
  {Horodecki}}]{brandao_local_2016}%
  \BibitemOpen
  \bibfield  {author} {\bibinfo {author} {\bibfnamefont {F.~G. S.~L.}\
  \bibnamefont {Brand\~ao}}, \bibinfo {author} {\bibfnamefont {A.~W.}\
  \bibnamefont {Harrow}},\ and\ \bibinfo {author} {\bibfnamefont
  {M.}~\bibnamefont {Horodecki}},\ }\bibfield  {title} {\bibinfo {title} {Local
  {random} {quantum} {circuits} are {approximate} {polynomial}-{designs}},\
  }\href {https://doi.org/10.1007/s00220-016-2706-8} {\bibfield  {journal}
  {\bibinfo  {journal} {Commun. Math. Phys.}\ }\textbf {\bibinfo {volume}
  {346}},\ \bibinfo {pages} {397} (\bibinfo {year} {2016})},\ \Eprint
  {https://arxiv.org/abs/1208.0692} {arXiv:1208.0692} \BibitemShut {NoStop}%
\bibitem [{\citenamefont
  {Haferkamp}(2022{\natexlab{b}})}]{haferkamp_random_2022}%
  \BibitemOpen
  \bibfield  {author} {\bibinfo {author} {\bibfnamefont {J.}~\bibnamefont
  {Haferkamp}},\ }\bibfield  {title} {\bibinfo {title} {Random quantum circuits
  are approximate unitary $t$-designs in depth $o(nt^{5+o(1)})$},\ }\href
  {https://doi.org/10.22331/q-2022-09-08-795} {\bibfield  {journal} {\bibinfo
  {journal} {Quantum}\ }\textbf {\bibinfo {volume} {6}},\ \bibinfo {pages}
  {795} (\bibinfo {year} {2022}{\natexlab{b}})}\BibitemShut {NoStop}%
\bibitem [{\citenamefont {{Aaronson}}(2018)}]{Aaronson2018ShadowTomography}%
  \BibitemOpen
  \bibfield  {author} {\bibinfo {author} {\bibfnamefont {S.}~\bibnamefont
  {{Aaronson}}},\ }\bibfield  {title} {\bibinfo {title} {Shadow tomography of
  quantum states},\ }in\ \href@noop {} {\emph {\bibinfo {booktitle} {Proc. 50th
  Ann. ACM SIGACT Symp. Th. Comput.}}}\ (\bibinfo {year} {2018})\ pp.\ \bibinfo
  {pages} {325--338},\ \Eprint {https://arxiv.org/abs/1711.01053}
  {arXiv:1711.01053 [quant-ph]} \BibitemShut {NoStop}%
\bibitem [{\citenamefont {Gu}\ \emph {et~al.}(2021)\citenamefont {Gu},
  \citenamefont {Mishra}, \citenamefont {Englert},\ and\ \citenamefont
  {Ng}}]{Gu2021RandomizedLinearGate}%
  \BibitemOpen
  \bibfield  {author} {\bibinfo {author} {\bibfnamefont {Y.}~\bibnamefont
  {Gu}}, \bibinfo {author} {\bibfnamefont {R.}~\bibnamefont {Mishra}}, \bibinfo
  {author} {\bibfnamefont {B.-G.}\ \bibnamefont {Englert}},\ and\ \bibinfo
  {author} {\bibfnamefont {H.~K.}\ \bibnamefont {Ng}},\ }\bibfield  {title}
  {\bibinfo {title} {Randomized linear gate-set tomography},\ }\href
  {https://doi.org/10.1103/PRXQuantum.2.030328} {\bibfield  {journal} {\bibinfo
   {journal} {PRX Quantum}\ }\textbf {\bibinfo {volume} {2}},\ \bibinfo {pages}
  {030328} (\bibinfo {year} {2021})},\ \Eprint
  {https://arxiv.org/abs/2010.12235} {arXiv:2010.12235 [quant-ph]} \BibitemShut
  {NoStop}%
\bibitem [{\citenamefont {{Brieger}}\ \emph {et~al.}(2023)\citenamefont
  {{Brieger}}, \citenamefont {{Roth}},\ and\ \citenamefont
  {{Kliesch}}}]{Brieger21CompressiveGateSet}%
  \BibitemOpen
  \bibfield  {author} {\bibinfo {author} {\bibfnamefont {R.}~\bibnamefont
  {{Brieger}}}, \bibinfo {author} {\bibfnamefont {I.}~\bibnamefont {{Roth}}},\
  and\ \bibinfo {author} {\bibfnamefont {M.}~\bibnamefont {{Kliesch}}},\
  }\bibfield  {title} {\bibinfo {title} {Compressive gate set tomography},\
  }\href {https://doi.org/10.1103/PRXQuantum.4.010325} {\bibfield  {journal}
  {\bibinfo  {journal} {PRX Quantum}\ }\textbf {\bibinfo {volume} {4}},\
  \bibinfo {pages} {010325} (\bibinfo {year} {2023})},\ \Eprint
  {https://arxiv.org/abs/2112.05176} {arXiv:2112.05176 [quant-ph]} \BibitemShut
  {NoStop}%
\bibitem [{\citenamefont {Diaconis}\ and\ \citenamefont
  {Shahshahani}(1994)}]{diaconis_eigenvalues_1994}%
  \BibitemOpen
  \bibfield  {author} {\bibinfo {author} {\bibfnamefont {P.}~\bibnamefont
  {Diaconis}}\ and\ \bibinfo {author} {\bibfnamefont {M.}~\bibnamefont
  {Shahshahani}},\ }\bibfield  {title} {\bibinfo {title} {On the eigenvalues of
  random matrices},\ }\href {https://doi.org/10.2307/3214948} {\bibfield
  {journal} {\bibinfo  {journal} {Journal of Applied Probability}\ }\textbf
  {\bibinfo {volume} {31}},\ \bibinfo {pages} {49} (\bibinfo {year}
  {1994})}\BibitemShut {NoStop}%
\bibitem [{\citenamefont {Rains}(1998)}]{rains_increasing_1998}%
  \BibitemOpen
  \bibfield  {author} {\bibinfo {author} {\bibfnamefont {E.~M.}\ \bibnamefont
  {Rains}},\ }\bibfield  {title} {\bibinfo {title} {Increasing subsequences and
  the classical groups},\ }\href {https://doi.org/10.37236/1350} {\bibfield
  {journal} {\bibinfo  {journal} {The Electronic Journal of Combinatorics}\ ,\
  \bibinfo {pages} {R12}} (\bibinfo {year} {1998})}\BibitemShut {NoStop}%
\bibitem [{\citenamefont {Scott}(2008)}]{scott_optimizing_2008}%
  \BibitemOpen
  \bibfield  {author} {\bibinfo {author} {\bibfnamefont {A.~J.}\ \bibnamefont
  {Scott}},\ }\bibfield  {title} {\bibinfo {title} {Optimizing quantum process
  tomography with unitary 2-designs},\ }\href
  {https://doi.org/10.1088/1751-8113/41/5/055308} {\bibfield  {journal}
  {\bibinfo  {journal} {J. Phys. A: Math. Theor.}\ }\textbf {\bibinfo {volume}
  {41}},\ \bibinfo {pages} {055308} (\bibinfo {year} {2008})},\ \Eprint
  {https://arxiv.org/abs/0711.1017} {arXiv:0711.1017 [quant-ph]} \BibitemShut
  {NoStop}%
\end{thebibliography}%
